\documentclass[12pt, letterpaper]{article}

\usepackage{arxiv}

\usepackage[utf8]{inputenc} 
\usepackage[T1]{fontenc}    

\usepackage{amsmath, amsfonts, amsthm, bm, bbm}
\usepackage{mathtools}
\usepackage{mathrsfs}
\usepackage{wasysym}

\usepackage{graphicx}
\usepackage{float}
\usepackage{adjustbox}
\usepackage{array}
\usepackage{booktabs, multirow, tabularx, colortbl}
\usepackage{rotating}
\usepackage{dirtree}

\usepackage{caption}
\usepackage{subcaption}

\usepackage{enumerate}
\usepackage[inline]{enumitem}

\usepackage{xcolor}      
\usepackage{soul}
\usepackage{indentfirst}

\usepackage{algorithm}
\usepackage{algpseudocode}

\usepackage{listings}
\usepackage{cancel}
\usepackage{environ}

\usepackage{hyperref}

\usepackage[colorinlistoftodos]{todonotes}

\usepackage[authoryear]{natbib}

\newtheorem{prop}{Proposition}

\theoremstyle{remark}
\newtheorem*{remark}{Remark}

\newcommand{\y}{\checkmark}

\title{Changepoint Detection As Model Selection: A General Framework}

\author{
Michael Grantham \\
Department of Statistics,\\
University of Nebraska-Lincoln\\
\texttt{michaelgrantham@live.com} \\
\And
Bertrand Clarke \\
Department of Statistics,\\
University of Nebraska-Lincoln\\
\texttt{bclarke3@unl.edu} 
\And
Xueheng Shi \\
Department of Statistics,\\
University of Nebraska-Lincoln\\
\texttt{shixueheng@gmail.com} \\
}

\begin{document}

\maketitle

\begin{abstract}

This dissertation presents a general framework for changepoint detection based on $\ell_0$ model selection. The core method, Iteratively Reweighted Fused Lasso (IRFL), improves upon the generalized lasso by adaptively reweighting penalties to enhance support recovery and minimize criteria such as the Bayesian Information Criterion (BIC). The approach allows for flexible modeling of seasonal patterns, linear and quadratic trends, and autoregressive dependence in the presence of changepoints.

Simulation studies demonstrate that IRFL achieves accurate changepoint detection across a wide range of challenging scenarios, including those involving nuisance factors such as trends, seasonal patterns, and serially correlated errors. The framework is further extended to image data, where it enables edge-preserving denoising and segmentation, with applications spanning medical imaging and high-throughput plant phenotyping.

Applications to real-world data demonstrate IRFL’s utility. In particular, analysis of the Mauna Loa CO\textsubscript{2} time series reveals changepoints that align with volcanic eruptions and ENSO events, yielding a more accurate trend decomposition than ordinary least squares. Overall, IRFL provides a robust, extensible tool for detecting structural change in complex data.

\end{abstract}

\section{Changepoint Problems}

\subsection*{Changepoint Analysis Intuitions}
Changepoint detection is a fundamental problem in time series analysis, with applications across fields such as climate science, finance, genomics, and engineering. A changepoint refers to a time at which the statistical properties of a sequence of observations undergo an abrupt and meaningful shift. Such changes may occur in the mean, trend, variance, distribution, or correlation structure of a stochastic process. Detecting these changes is essential for accurate modeling, forecasting, and decision-making, as neglecting them can result in misleading inferences about the underlying data-generating mechanism and unreliable predictions.

Changepoint problems can be categorized based on the type of structural shift present in the data. One of the most widely studied cases is the mean-shift changepoint, where the expected value of a process changes at an unknown time. For instance, in the United States, weather stations typically upgrade or replace their gauges every 25 years. Such replacements often introduce unintentional mean shifts into the recorded measurements. Similarly, modifications in measurement procedures—such as changes in observation time, relocation of instruments or entire stations, or the occurrence of extreme meteorological or geological events—can also result in mean shifts. In the climate literature, these artificial disturbances are referred to as inhomogeneities. They introduce unintended biases into time series analyses, compromising estimation accuracy and predictive reliability. Detecting and correcting for such changepoints—through a process known as homogenization \citep{domonkos-2021-homogenisation, ribeiro-2016-homogenisation} — is therefore essential for producing consistent and reliable climate data records.

The following examples illustrate common forms of changepoints and highlight why their detection and proper treatment are essential. Subsection \ref{oil drilling example} presents a real-world case demonstrating the impact of mean shifts. However, changepoints can also occur in other structural components of a time series. Subsection \ref{global warming example} explores a trend shift in a climatological context, while subsection \ref{autocorrelation example} emphasizes the importance of accounting for autocorrelation in the same dataset.

In addition to mean and trend shifts, changes in variability play a critical role in understanding system dynamics, as they reveal transitions in volatility or stability over time. Subsection \ref{liver example} illustrates such a changepoint through a practical example involving a shift in variability.

\subsection{Oil Drilling} \label{oil drilling example}
One real-world example of changepoint analysis comes from oil well drilling. Nuclear magnetic resonance (NMR) measurements are gathered using specialized NMR logging tools which are integrated into the drill bit or the drill string. These tools generate a strong magnetic field around the surrounding rock, aligning the magnetic moments of hydrogen nuclei within (typically in water or hydrocarbons within the rock pores). After this alignment, a radio frequency pulse is applied to disturb the equilibrium. As the hydrogen nuclei return to their original alignment (a process known as relaxation), they emit a detectable radiofrequency signal \citep{bloch-1946-nmr,kleinberg-1992-nmr}. 

The NMR logging tool measures the relaxation times, which provide valuable information about the rock's porosity, permeability, and fluid content. Changepoints in the mean structure of these responses correspond to changes in the strata of rock \citep{fearnhead-2019-detecting}, making the accurate identification of these changepoints crucial for understanding the underlying geological formations and for optimizing drilling strategies. 

Figure \ref{fig:changepoints_example}  demonstrates how changes in the mean structure of nuclear magnetic resonance measurements can indicate differing strata in the subsurface rock. The figure includes a plot of the example data, with a depiction of the changing mean response. 

However, mean shifts like those found in NMR data represent only one kind of changepoint. Another key changepoint type is a trend shift changepoint, where the rate of change of a time series suddenly shifts. The next subsection deals with changepoints of this type.
\begin{figure}[H]
    \centering
    \includegraphics[width=\textwidth]{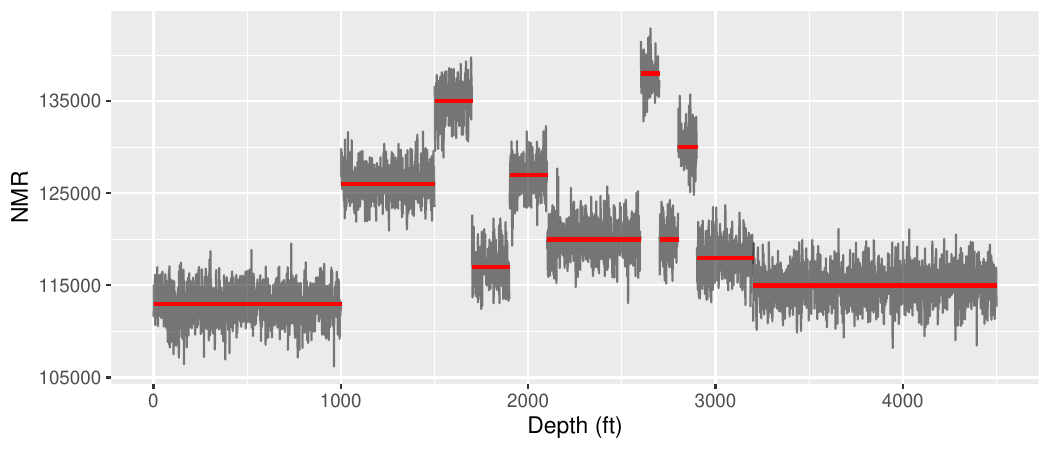}
    \caption[Plot of NMR responses in rock strata]{Example plot illustrating changes in nuclear magnetic resonance and corresponding changepoints in rock strata. In this plot, shifts in NMR responses indicate variations in these properties, corresponding to distinct layers or changes in the rock strata, highlighted by the changing mean NMR response.}
    \label{fig:changepoints_example}
\end{figure}

\subsection{Global Warming Surge} \label{global warming example}

A trend shift changepoint occurs when the rate of change itself undergoes a structural break, often due to environmental factors or other external influences. This is particularly important in climate data, where identifying such points allows for models that distinguish between periods governed by different warming rates, thereby providing insight into when accelerations or slowdowns in the trend occur. Without accounting for these shifts, models that assume a constant trend risk obscuring important features of the data and producing misleading inferences.

To illustrate, global warming time series data from the NASA GISS (Goddard Institute for Space Studies; \cite{nasa_gis}) are analyzed both with and without accounting for a changepoint. In the na\"ive case, a simple linear regression assumes a constant rate of warming, missing temporal shifts in the slope and failing to reveal when the rate of global warming increased or decreased. Figure \ref{fig:GGIS_trend} is included below for reference. The $y-$axis is the deviation from the global mean temperature (taken over both land and water) from the years 1951 through 1980 (referred to as the ``temperature anomaly''), while the $x-$axis is the year corresponding to the anomaly.

\begin{figure}[H]
    \centering
    \textbf{Global Mean Estimates based on Land and Ocean Data from NASA/GISS/GISTEMP v4}
    \includegraphics[width=\textwidth]{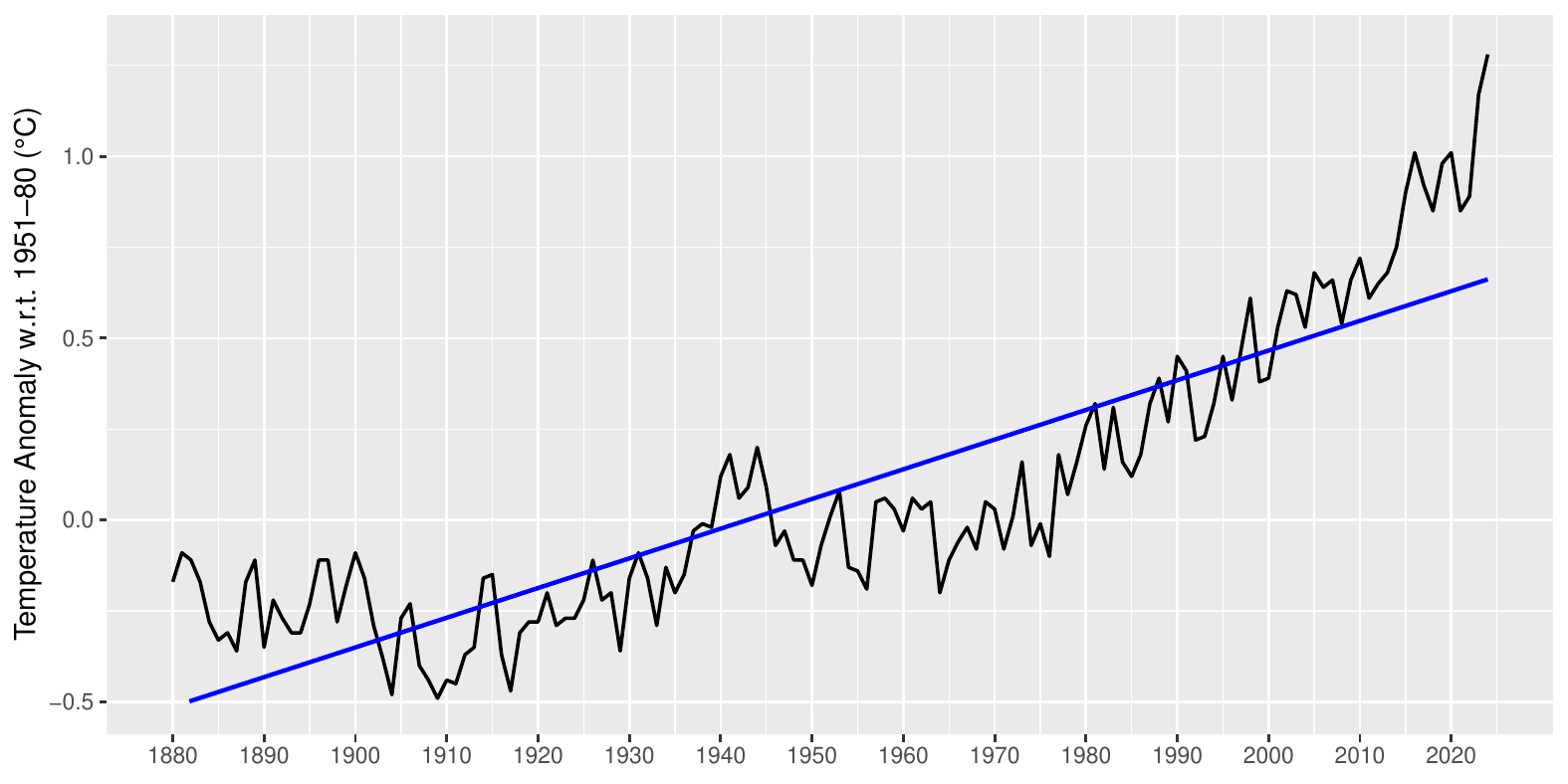}
    \caption[Naive plot of global warming trend]{Example plot of a na\"ive linear regression of global mean land and ocean temperature.}
    \label{fig:GGIS_trend}
\end{figure}
A plot of the residuals clearly reveals that this simple linear regression is missing key features of the data. See Figure \ref{fig:GGIS_residuals} below:

\begin{figure}[H]
    \centering
    \textbf{Global Mean Linear Model Residuals based on Land and Ocean Data from NASA/GISS/GISTEMP v4}
    \includegraphics[width=\textwidth]{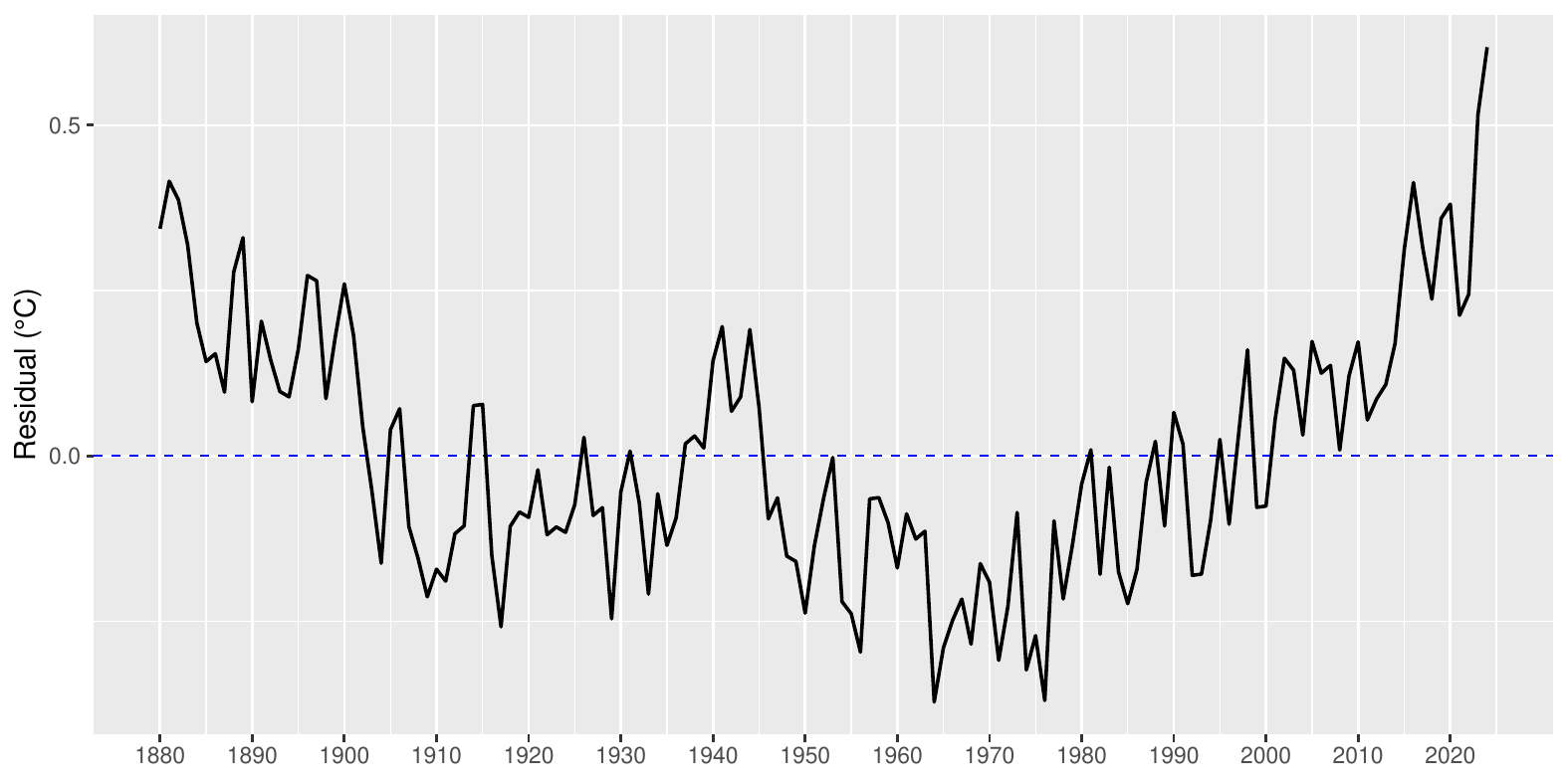}
    \caption[Plot of residuals from linear model of global warming trend]{Plot of residuals from linear model in Figure \ref{fig:GGIS_trend}.}
    \label{fig:GGIS_residuals}
\end{figure}
Clearly, a simple linear regression is insufficient to model the climate data in question. By contrast, including a changepoint in the model reveals valuable insights. Not only does it pinpoint when the trend shift occurred, offering a clearer understanding of the timing of significant shifts in the environment, but it also allows for a more accurate estimation of the effect size associated with that change. In the particular case of this GISS global warming data, a changepoint is detected in the year 1974, with a jump in trend from 0.0037 $^\circ$C/year prior to 1974 to 0.02 $^\circ$C/year after (see Figure \ref{fig:GGIS_changepoint}). Both the location of the changepoint (the year 1974) and the magnitude of the change in trend are useful to climate researchers.

\begin{figure}[H]
    \centering
        \textbf{Global Mean Estimates based on Land and Ocean Data from NASA/GISS/GISTEMP v4}
        \includegraphics[width=\textwidth]{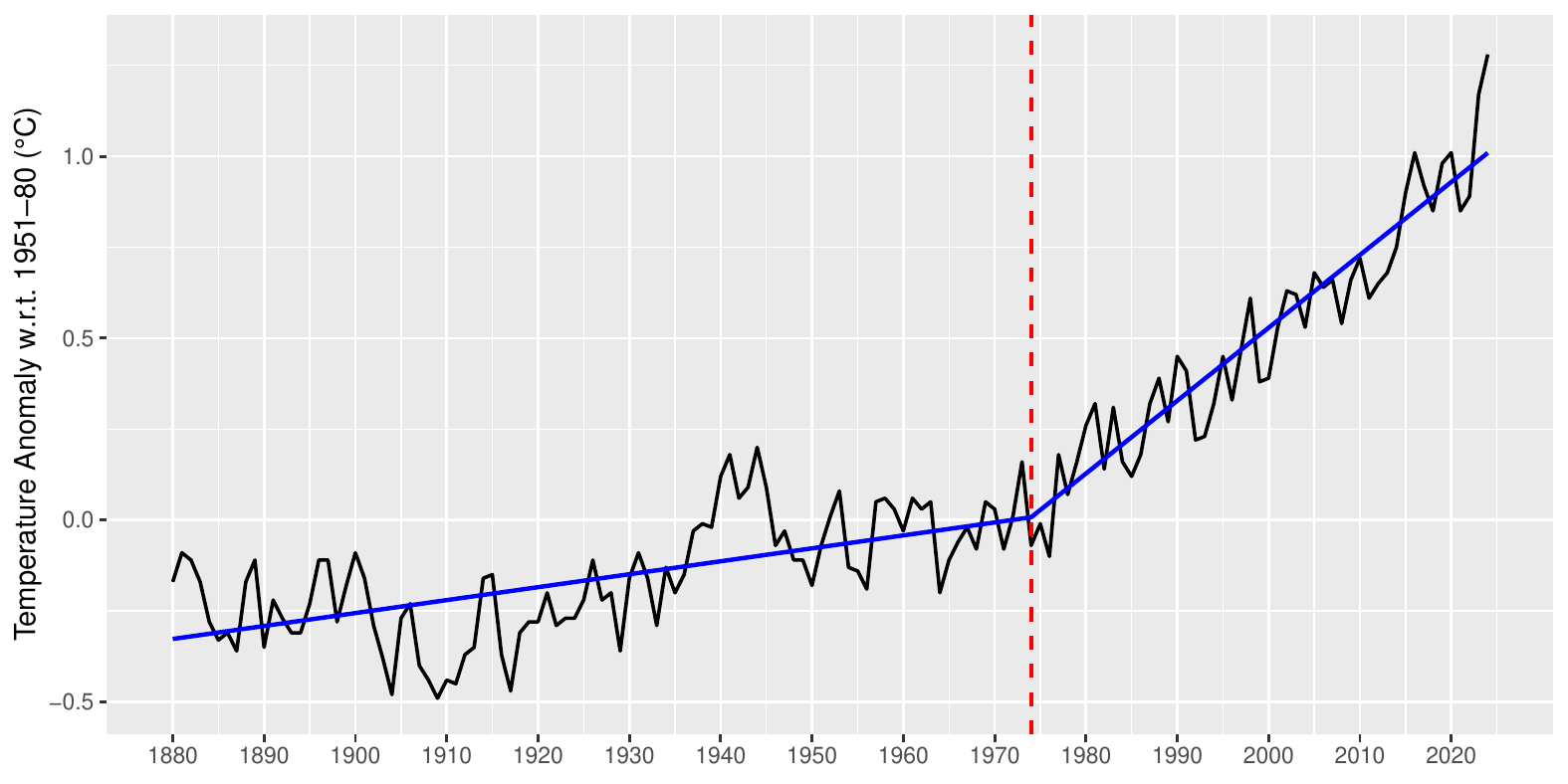}
    \caption[Plot of global warming trend with changepoint]{By including a detected changepoint before running the regression, one can see that the estimate of trend radically changes.}
    \label{fig:GGIS_changepoint}
\end{figure}

\subsection{Global Warming Data with Autocorrelation}\label{autocorrelation example}

While the detection of a trend shift, as in the previous example, is useful for climate scientists analyzing global temperature data, accounting for serial correlation in such a time series can significantly alter both the estimation and interpretation of a changepoint; see \citet{Shi-2022-CET}. Serial correlation (also called autocorrelation) refers to the dependence of an observation on its past values, often captured through correlations at different lags. A simple example is noise generated by an autoregressive process of order one, expressed as
\begin{equation}
\label{autoregressive}
X_t = \mu+\epsilon_t,\quad\epsilon_t=\phi \epsilon_{t-1}+\eta_t, \quad \eta_t\overset{\text{iid}}{\sim}WN(0,\sigma^2),
\end{equation}
where $WN(\cdot,\cdot)$ denotes a white noise process. Here, $X_t$ exhibits autocorrelation that decays with lag, and because each value depends only on the most recent one, it is also Markov.
Accounting for such dependence is crucial in practice: \citet{beaulieu-2024-globalwarmingsurge} report a trend change in the global temperature record around 1974, with additional surges in the 1990s and 2010s. Such findings, if taken at face value, suggest an alarming acceleration in warming that could drive costly shifts in public policy. However, once the serial correlation inherent in the data is taken into account (the maximum likelihood estimate of an AR(1) parameter is about $0.53$), the apparent post-1974 surge becomes statistically insignificant and would need to be at least 55\% larger in order to be detectable. This exact situation---the conflation between ``extra changepoints" and ``larger autocorrelation" is a very nuanced and important issue, but we delay discussion of this problem until subsection  1.5.

For visualization purposes, a graph of the global mean estimate of temperature, with a supposed trend shift fit at the year 2015, is included in Figure \ref{fig:GGIS_extra_changepoint}. Using OLS estimates, the rate of increase before 2015 is approximately $0.015^\circ$C/year, while it appears to be approximately $0.021^\circ$C/year after.

\begin{figure}[H]
    \centering
        \textbf{Global Mean Estimates based on Land and Ocean Data from NASA/GISS/GISTEMP v4}
        \includegraphics[width=\textwidth]{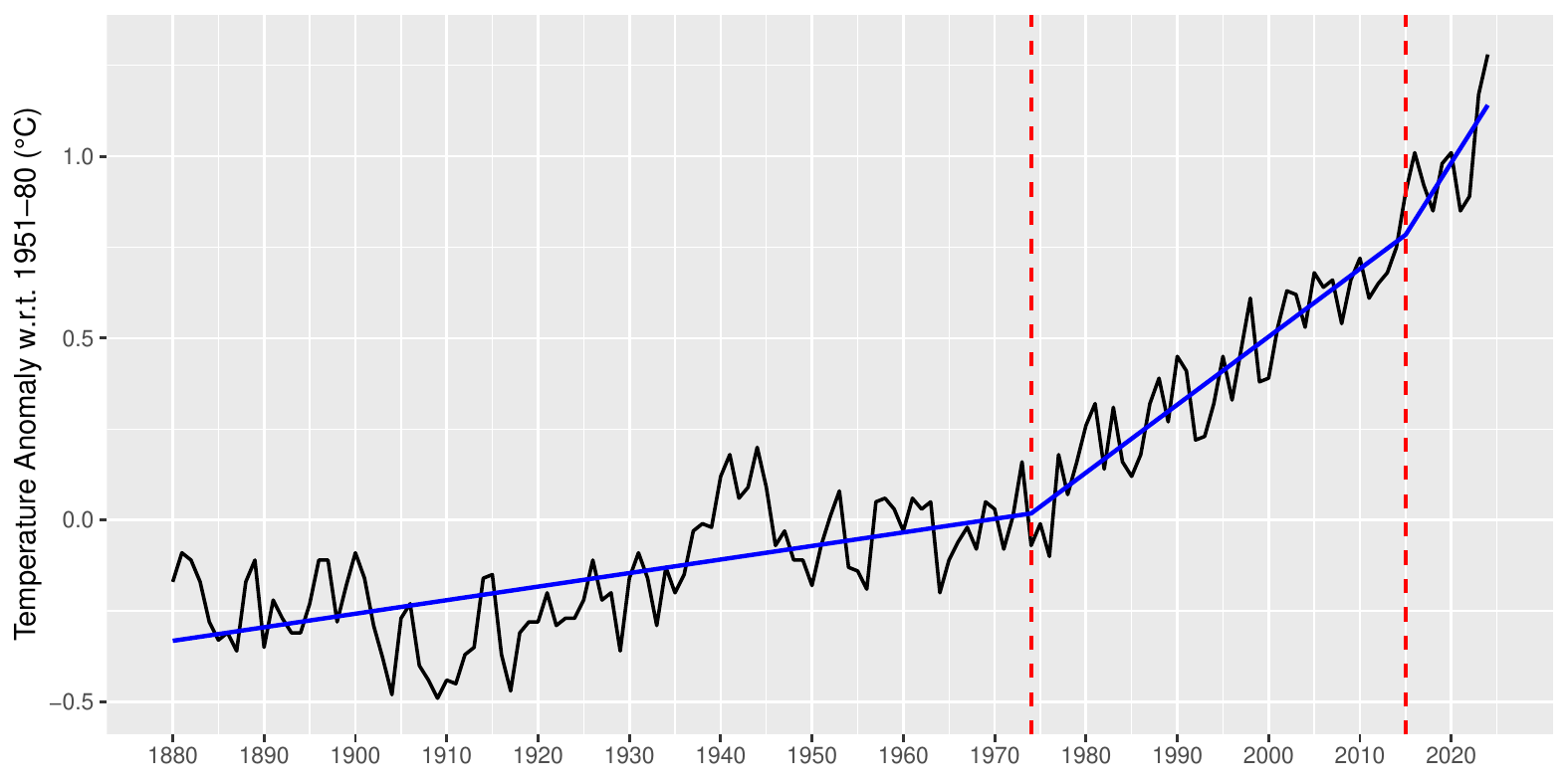}
    \caption[Plot of global warming trend with spurious changepoint]{There is speculation of a ``surge'' in the rate of global warming (a trend shift) at approximately the year 2015.}
    \label{fig:GGIS_extra_changepoint}
\end{figure}

Therefore, due to the presence of serial correlation in the time series, there is not strong enough evidence at this point in time to believe there has been a surge in the global warming rate after 1974. 

\subsection{Liver transplantation}\label{liver example}
Another salient example of changepoint analysis comes from the world of organ procurement for the purpose of transplantation. In \citet{gao-2020-variance}, a severed porcine liver was monitored upon the infusion of perfusion liquid (a liquid intended to imitate the life-sustaining properties of blood) into the organ. The surface temperature of the liver was measured over a dense grid on the organ every ten minutes for 24 hours, and the recorded temperature was reported for one randomly selected location. Initially the variability in surface temperature measurements was high as the organ resisted changes in the ambient temperature. However, after about ten hours the organ lost its viability for transplant---a change which is reflected in the somewhat abrupt loss of variability in the surface temperature. Locating this changepoint in the variability of the time series is key to understanding the point at which the liver loses its viability. A simulation mimicking this data is included in Figure \ref{fig:Liver Plot} below for reference, because the original data was not available.

\begin{figure}[H]
    \centering
        \textbf{Surface Temperature of Severed Porcine Liver over 24 Hours}
        \includegraphics[width=\textwidth]{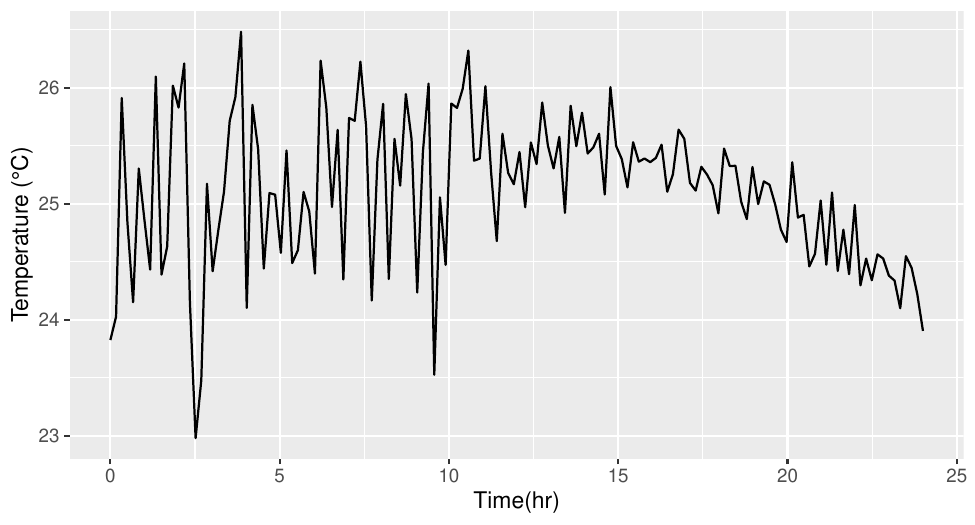}
    \caption[Plot of Porcine Liver Temperature over 24 Hours]{The variability of the surface temperature of the liver decreases at roughly the 10-hour mark, corresponding to the loss of the organ's viability for transplant.}
    \label{fig:Liver Plot}
\end{figure}
A variance changepoint can be represented by allowing the noise variance to differ before and after some unknown time $\tau$. 
Formally, one may write
\[
y_t = \mu_t + \epsilon_t, \qquad 
\epsilon_t \sim 
\begin{cases}
\mathcal{N}(0, \sigma_1^2), & t < \tau,\\[4pt]
\mathcal{N}(0, \sigma_2^2), & t \ge \tau,
\end{cases}
\]
where $\mu_t$ denotes the mean structure (not necessarily constant, and clearly not constant in Figure \ref{fig:Liver Plot}). 
For each candidate $\tau$, the segment variances can be estimated as
\[
\hat{\sigma}_1^2(\tau) = \frac{1}{n_1}\sum_{t < \tau}(y_t - \hat{\mu}_t)^2, 
\qquad 
\hat{\sigma}_2^2(\tau) = \frac{1}{n_2}\sum_{t \ge \tau}(y_t - \hat{\mu}_t)^2,
\]
with $n_1 = \tau-1$ and $n_2 = n - \tau+1$.
The changepoint time $\hat{\tau}$ is then selected as the estimate of $\tau$ that maximizes the log-likelihood
\[
\ell(\tau) = -\frac{n_1}{2}\log \hat{\sigma}_1^2(\tau)
             -\frac{n_2}{2}\log \hat{\sigma}_2^2(\tau),
\]
or equivalently minimizes the total within-segment variance.  
In this example, $\hat{\tau}$ corresponds to the time at which the surface-temperature variability of the liver markedly decreases, indicating loss of viability. A further motivational treatment of model selection is contained in subsection \ref{model selection and objective criteria}.

These real-world examples highlight the importance of accurately detecting structural changes in the mean, trend, and volatility of time series data for subsequent analysis and prediction. 

\subsection{Model Selection and Objective Criteria}\label{model selection and objective criteria}
All of the examples considered so far admit competing plausible explanations (e.g., mean shift versus no mean shift, or trend shift versus autocorrelation effects). This illustrates a central difficulty: in changepoint analysis, as in time series modeling more generally, candidate models often differ not only in the number and type of changepoints they include but also in the number and complexity of other parameters (see, for instance, the competition between autocorrelation and changepoints in subsection 1.3, or the following from later in this dissertation: Section 2’s Wild Contrast Maximization, Section 4’s scenarios 4 and 5, and Section 5’s Mauna Loa discussion). A crucial question is therefore how to decide objectively which model provides the best balance among explanatory power, parsimony, and accuracy. Simply choosing the model with the best fit (for example, the one with the largest maximized likelihood) is inadequate, as models with more parameters tend to overfit. 

One of the most used remedies is the Bayesian Information Criterion (BIC), also known as the Schwarz Criterion \citep{schwarz-1978-bic}. For a fitted model with maximized likelihood $\mathcal{L}(\hat\theta)$, number of estimated parameters $k$, and sample size $n$, its BIC is computed as
\begin{equation}
\label{BIC}
\text{BIC} = -2 \log \mathcal{L}(\hat\theta) + k \log n.
\end{equation}
The first term, $-2 \log \mathcal{L}(\hat\theta)$, rewards models that fit the data well, while the second term, $k \log n$, penalizes models with too many parameters. This balance ensures that the chosen model is not only statistically adequate, but also parsimonious. 

In practice, among all candidate models considered, the preferred model is the one that minimizes the BIC scores. This provides a consistent and objective basis for deciding between alternative changepoint configurations and is the standard approach adopted in much of the changepoint literature. Information criteria such as the AIC and BIC have a long history in model selection, and a detailed treatment of their theoretical foundations and practical use can be found in \citet{burnham-2002-model}.

\subsection{Dissertation Organization}

The importance of changepoint analysis being evident across multiple disciplines, extensive research has focused on changepoint analysis since the pioneering work in \citet{page-1955-test}. For example, notable methods for detecting mean shifts include CUSUM \citep{page-1955-test}, SCUSUM \citep{kirch-2006-SCUSUM}, binary segmentation \citep{scott-1974-cluster, sen-1975-tests}, and optimal partition \citep{jackson-2005-op}. In this dissertation the problem will be approached from an $\ell_0$ model selection perspective which has been largely overlooked in the changepoint literature. Existing manuscripts that employ model selection approaches for changepoint detection are relatively rare, for example, \citet{harchaoui-2010-tv, rojas-2014-fusedlasso, gallagher-2014-adalasso}. 

The principal contribution of this dissertation is the development of \emph{Iteratively Reweighted Fused Lasso (IRFL)}, a unified framework that consolidates many existing changepoint detection methods within a single coherent model selection perspective. The IRFL introduces an iterative reweighting procedure that adaptively strengthens or relaxes regularization across potential changepoints, enabling the method to distinguish genuine structural shifts from spurious fluctuations. In doing so, it offers a practical approximation to the best-subset-selection approaches ($\ell_0$) that would otherwise be computationally infeasible. This unification bridges heuristic, exact, and adaptive segmentation techniques, providing a general framework for estimating changepoint models with improved interpretability and robustness, from one-dimensional time series to two-dimensional images.

We have organized the remainder of this dissertation as follows. Section 2 surveys heuristic and exact segmentation algorithms for estimating changepoint models. Section 3 develops the model selection framework—focusing on the fused lasso and its generalization in the genlasso—that provides the foundation for the methods introduced next. In Section 4 we present the Iteratively Reweighted Fused Lasso (IRFL), an original extension that advances the genlasso framework through repeated iteration inspired by the adalasso---a principle we have called \emph{iterate to isolate}. Its effectiveness is demonstrated through simulation studies in a variety of modeling scenarios, highlighting its adaptability to a broad class of changepoint problems. Finally, we conclude Section 5 with two real-world use-cases of the IRFL, in climatology and in image denoising and segmentation, illustrating the IRFL’s versatility in high-impact arenas.

\section{Segmentation Algorithms: Heuristic and Exact Approaches}

This section presents a range of algorithms for detecting changepoints in time series that segregate into two broad classes: heuristic methods and exact methods. The first half of the section focuses on heuristic approaches, which rely on greedy or randomized search strategies to identify changepoints in a computationally efficient manner. We begin with the simplest setting of a single mean shift, introducing at-most-one-changepoint (AMOC) tests, and then extend these ideas to multiple mean shift changepoint detection through methods such as Binary Segmentation (BS) and Wild Binary Segmentation (WBS). These ideas are further refined to handle complications such as autocorrelation and linear trends in Wild Contrast Maximization (WCM), and Narrowest-Over-Threshold (NOT). A brief treatment of mean shifts in the presence of seasonality is also considered.

In the second half of the section, we turn to exact methods which use dynamic programming to evaluate all possible segmentations and yield globally optimal solutions under a given cost function. While these methods are more computationally intensive, they overcome the inherent limitations of heuristics by guaranteeing optimality. We discuss classical dynamic programming approaches as well as more recent variants designed to improve efficiency, and highlight the trade-offs between computational cost and segmentation accuracy.

\subsection{Heuristic Methods}
\subsubsection*{Binary Segmentation (BS) and CUSUM}

In the simplest case of changepoint detection, the task is to determine whether a single mean shift occurs within a time series. A variety of AMOC tests have been proposed for this setting, including $T_{\text{max}}$, $F_{\text{max}}$, SCUSUM, the standard-normal-homogeneity test (SNHT), and several nonparametric approaches. Each of these is designed to assess the null hypothesis of no change against the alternative of exactly one change in mean. Among these methods, the Cumulative Sum (CUSUM) test has achieved particularly broad usage, both in classical statistical process control and in modern changepoint applications, and will therefore be the focus here.

Formally, let $\{y_t\}_{t=1}^n$ denote a time series, and choose integer $k$ so that $2 < k \leq n$. Assuming there is no underlying trend, the model for a single changepoint is
\begin{equation}
\label{singlechangepoint}
y_t = \begin{cases}
\mu + \epsilon_t, & 1 \leq t < k, \\
\mu + \Delta + \epsilon_t, & k \leq t \leq n,
\end{cases}
\end{equation}
where $\mu$ is an unknown mean, $\Delta$ is the magnitude of the mean shift at time $k$, and $\{\epsilon_t\}_{t=1}^n$ are independent and identically distributed random variables with mean zero and variance $\sigma^2$. We then form the hypothesis test 
\begin{align*}
H_0: \Delta &= 0 \qquad \text{versus} \qquad H_A: \Delta \neq 0,
\end{align*}
for any time $k$.

The corresponding CUSUM statistic is defined as
\begin{equation}
\label{CUSUM}
C_k = \sum_{t=1}^k y_t - \frac{k}{n}\sum_{t=1}^n y_t,
\end{equation}
the difference between the partial sum of the first $k$ observations and $k$ times the overall sample mean $\bar{y}_{1:n}$. Under $H_0$, both quantities have expected value $k\mu$, so their difference should be close to zero. A mean shift induces a systematic deviation, causing $C_k$ to diverge near the true changepoint.

Although \eqref{CUSUM} is a standard form, many authors scale the statistic differently, defining
\begin{equation}
\label{CUSUM-scaled}
\text{CUSUM}(k) = \frac{C_k}{\sigma\sqrt{n}}.
\end{equation}
In practice, $\sigma^2$ is unknown and is often replaced by its null-hypothesis estimate,
\begin{equation}
\label{sigma2-hat}
\hat{\sigma}^2_{H_0} = \frac{\sum_{t=1}^n \left(y_t-\bar{y}_{1:n}\right)^2}{n-1},
\end{equation}
where $\bar{y}_{1:n} = \frac{1}{n}\sum_{t=1}^n y_t$. The test statistic is then defined by taking the \emph{maximum} absolute CUSUM value across all possible changepoint locations,
\begin{equation}
\label{CUSUM-test-statistic}
R(y) = \max_{1<k\leq n}|\text{CUSUM}(k;y)|.
\end{equation}
This definition emphasizes that there are many candidate indices $k$ where deviations might occur, but the procedure selects the one with the largest discrepancy. The maximization is essential: it ensures that $R(y)$ reflects the strongest evidence of a shift anywhere in the sequence, a fact that will connect directly to later extensions. If $R(y)$ is larger than an asymptotic critical value (explained shortly), then a changepoint is determined at time $k$.
As shown in \citet{csorgo-1997-CUSUM}, because $\hat{\sigma}^2_{H_0}$ is a consistent estimator of $\sigma^2$, the scaled CUSUM statistic converges in distribution under $H_0$ to the supremum of a standard Brownian bridge:
\begin{equation}
\label{CUSUM-convergence}
\max_{1<k\leq n}|\text{CUSUM}(k;y)| \overset{\mathcal{D}}{\to} \underset{0\leq t \leq 1}{\text{sup}}|B(t)|,
\end{equation}
where ${B(t)}_{t\in[0,1]}$ denotes the Brownian bridge (a stochastic process that starts and ends at zero but otherwise behaves like a Brownian motion in between). This connection provides the theoretical basis for the test’s critical values, which can be obtained from the distribution of the Brownian bridge supremum. A table of asymptotic critical values is provided in Appendix~\ref{app:crit}.

Despite their effectiveness in detecting a single mean shift, the CUSUM test, along with other AMOC tests such as SCUSUM, $F_\text{max}$, $T_\text{max}$, and SNHT, struggle to identify multiple mean shifts in a time series. To formalize this, consider a time series $\{y_t\}_{t=1}^n$ with $m$ unknown mean shifts indexed by $\boldsymbol{\tau} = \tau_{1:m} = \{\tau_1, \tau_2, \ldots, \tau_m\}$, located at $2 \le \tau_1 < \tau_2 < \cdots < \tau_m \le n$. The model is 
\begin{equation}
\label{totalvariation}
y_t = \begin{cases}
    \mu_1 + \epsilon_t & \text{for } \tau_0 \leq t \leq \tau_1 - 1, \\
    \mu_2 + \epsilon_t & \text{for } \tau_1 \leq t \leq \tau_2 - 1, \\
    \hspace{.6cm}\vdots \\
    \mu_{m+1} + \epsilon_t & \text{for } \tau_m \leq t \leq \tau_{m+1} - 1,
\end{cases}
\end{equation}
where $\{\epsilon_t\}$ are i.i.d.\ errors with mean zero and variance $\sigma^2$, and $\tau_0 = 1$ and $\tau_{m+1} = n + 1$ are boundary conditions. The $m$ changepoints $\tau_1, \tau_2, \ldots, \tau_m$ divide the series $\{y_t\}$ into $m + 1$ segments, where the $i^\text{th}$ segment contains data points $\{y_{\tau_{i-1}}, \ldots, y_{\tau_i - 1}\}$ with mean $\mu_i$. Both the number $m$ and the configuration $\boldsymbol{\tau}$ of changepoints, as well as the segment means $\{\mu_i\}_{i=1}^{m+1}$, are unknown.

Here the goal of changepoint analysis is to identify the configuration  
\(\bm{\tau} = \{\tau_1, \tau_2, \dots, \tau_m\}\) that minimizes a cumulative loss:
\begin{equation}\label{cost}
\sum_{i=1}^{m+1} \mathcal{C}\big(y_{\tau_{i-1} : (\tau_i - 1)}\big)
\;+\;
\beta \cdot f(\tau_{1:m}),
\end{equation}
where \(\mathcal{C}(\cdot)\) is a cost function. Typical examples include the residual sum of squares (RSS) and negative twice the log-likelihood, defined respectively as

\[
\mathcal{C}_{\text{RSS}}\big(y_{\tau_{i-1} : (\tau_i - 1)}\big)
=
\sum_{t=\tau_{i-1}}^{\tau_i - 1} 
\big(y_t - \hat{\mu}_i\big)^2,
\]
and
\[
\mathcal{C}_{\text{NLL}}\big(y_{\tau_{i-1} : (\tau_i - 1)}\big)
=
-2 \sum_{t=\tau_{i-1}}^{\tau_i - 1}
\log p(y_t \mid \hat{\mu}_i),
\]
where \(\hat{\mu}_i\) denotes the estimated mean within segment \([\tau_{i-1}, \tau_i - 1]\), typically the segment mean,
\[
\hat{\mu}_i
=
\frac{1}{\tau_i - \tau_{i-1}}
\sum_{t=\tau_{i-1}}^{\tau_i - 1} y_t
\]
and where \(p(y_t \mid \hat{\mu}_i)\) denotes the fitted likelihood within segment \([\tau_{i-1}, \tau_i - 1]\) under the estimated mean parameter \(\hat{\mu}_i\).

The penalty term \(f(\tau_{1:m})\), a function depending on the changepoint configuration, is a measure of model complexity. This measure is multiplied by the penalty factor \(\beta > 0\), which scales the model complexity relative to the cumulative cost and therefore encourages parsimony. 


The inability of AMOC tests to determine model~\eqref{totalvariation} has motivated the development of procedures designed to handle multiple changepoints. One of the earliest such methods is Binary Segmentation (BS), which extends AMOC tests by embedding them into a recursive partitioning framework. The idea is straightforward: an AMOC test (such as CUSUM) is applied to the full series to detect the single most significant changepoint. If a changepoint is found, the series is split into two subsegments at that location, and the same procedure is applied recursively to each subsegment. This process continues until no further changepoints are detected, based on a chosen stopping criterion such as a significance threshold (i.e., a size-$\alpha$ critical value for the chosen AMOC test) or a minimum segment length.

From an intuitive perspective, BS assumes that even if a segment contains multiple changepoints, the largest mean shift within that segment will dominate the AMOC statistic and be detected first. By recursively applying this strategy, BS builds up a set of changepoints that partitions the series. The appeal of this approach lies in its simplicity and its computational efficiency: by avoiding the need to examine every possible configuration $\tau_{1:m}$, BS drastically reduces computation compared to exhaustive methods.




However, because AMOC tests are applied to segments that may contain more than one changepoint, the test used to detect changepoints does not match the model being estimated in \eqref{totalvariation}. This can cause changepoints to be placed suboptimally, and early errors cannot be corrected later. Consequently, while BS often works well when changepoints are well-separated and have large associated mean shifts, it can perform poorly when changes are subtle or closely spaced \citep{Shi-2022-comparison,fryzlewicz-2014-WBS}.

In summary, BS represents a natural first step toward adapting single-changepoint tests to the multiple-changepoint setting. It builds directly on the same principles that make AMOC tests effective but inherits their limitations when applied outside their intended scope.

\subsubsection*{Wild Binary Segmentation (WBS)}

WBS, introduced by \citet{fryzlewicz-2014-WBS}, is an improvement of BS designed to address its main shortcomings. Because BS partitions the data after each detected changepoint and tests only within those resulting segments, early errors can propagate, and closely spaced or subtle changes may be missed. In principle, this limitation could be overcome by testing every possible sub-interval for changepoints, but that approach is computationally infeasible.  

WBS resolves this by instead applying the same CUSUM-based test over a large collection of randomly drawn subintervals. This randomized design maintains computational efficiency while improving sensitivity: with a sufficient number of subintervals \(M\), each true changepoint is likely to appear in at least one interval where it dominates the local contrast, enabling the detection of smaller or overlapping shifts that BS might overlook.

Formally, denote $F_n^M$ as a set of $M$ randomly chosen sub-intervals $[s_m,e_m]$ for $m=1,\ldots,M$, where $s_m$ and $e_m$ are drawn independently and with replacement from $\{1,\dots,n\}$, subject to $s_m<e_m$, and $e_m-s_m>\delta$ where $\delta>0$ is the user-specificied minimum segment length. A CUSUM-type statistic on an interval $[s,e]$ is defined by
\begin{equation*}
    \label{CUSUM-Frylewicz-WBS}
    \mathcal{C}_{s,e}^{b} = \sqrt{\frac{e-b}{(e-s+1)(b-s+1)}}\sum_{t=s}^{b-1}y_t - \sqrt{\frac{b-s+1}{(e-s+1)(e-b)}}\sum_{t=b}^e y_t,
\end{equation*}
where $s < b < e$. 

For a given segment $[s,e]$, only those random intervals $[s_m,e_m] \in F_n^M$ fully contained in $[s,e]$ are considered. Within each such interval, the AMOC statistic is maximized over $b$ to obtain a candidate changepoint $b_m$ with value $|\mathcal{C}_{s_m,e_m}^{b_m}|$. Intervals exceeding the threshold $\zeta_n$ contribute candidates, and as in BS \eqref{CUSUM-test-statistic} the recursion then selects the single location with the strongest evidence of change as a changepoint:

\begin{equation}
\label{WBS-max}
k = \underset{[s_m,e_m]\subseteq [s,e]}{\text{arg max }}|\mathcal{C}_{s_m,e_m}^{b_m}|.
\end{equation}

If the maximum in \eqref{WBS-max} exceeds $\zeta_n$, the recursion continues on $[s,k]$ and $[k,e]$; otherwise, it terminates. Because a large number of intervals is tested, the probability is high that any changepoint within $[s,e]$ will be isolated in at least one sub-interval, allowing it to be detected even if it is weak relative to others. Indeed, in \citet{fryzlewicz-2014-WBS} it is shown that WBS can detect changepoints with greater accuracy than BS in many settings, without a large increase in runtime. Guidance on the choice of $M$ and $\zeta_n$ is given in \citet{fryzlewicz-2014-WBS} and included with pseudo code of the algorithm in Appendix ~\ref{app:pseudocode}.


While WBS substantially improves upon BS in many scenarios, it is not without limitations. In general, it is an aggressive method that tends to overestimate the number of changepoints \citep{Lund-2020-WBScomments}. Its reliance on random intervals introduces variability into the results: the set of detected changepoints may differ across runs, especially when $M$ is small. Moreover, like BS, WBS does not revisit earlier decisions—once a changepoint is detected, its location is fixed, and no later step can adjust it. This can limit performance when changepoints occur in rapid succession or when the signal is ambiguous. Despite these limitations, WBS remains an efficient and effective refinement of the BS framework.

\subsubsection*{Wild Contrast Maximization (WCM)}

One of the key challenges in changepoint detection is identifying changepoint locations when the data exhibit dependence rather than being independently distributed. A particularly relevant case is when the data follow a white noise process with an AR(1) dependence structure. In this setting, the model can be written as
\begin{equation}\label{AR1 model} 
y_t = 
\begin{cases} 
\mu_1 + \epsilon_t, & t = \tau_0, \\ 
\phi y_{t-1} + \mu_1 + \epsilon_t, & \tau_0 + 1 \leq t \leq \tau_1 - 1, \\ 
\phi y_{t-1} + \mu_2 + \epsilon_t, & \tau_1 \leq t \leq \tau_2 - 1, \\ 
\hspace{1.1cm} \vdots \\ 
\phi y_{t-1} + \mu_m + \epsilon_t, & \tau_m \leq t \leq \tau_{m+1} - 1,
\end{cases} 
\end{equation}
where the errors ${\epsilon_t}$ are i.i.d.\ with zero mean and variance $\sigma^2$, and the AR$(1)$ autocorrelation parameter $\phi$ satisfies $|\phi| < 1$ to ensure stationarity \citep{brockwell-davis-2016-forecasting, hamilton-1994-series}.
This model is commonly referred to as a Level-Shift (LS) model \citep{chen-1993-joint}. In this formulation, the asymptotic mean of the $i^\text{th}$ segment is $\mu_i^\prime := \mu_i/(1-\phi)$ instead of $\mu_i$. We have that for $\tau_{i-1} \leq t \leq \tau_i - 1$,
\[
    \mathbb{E}[y_t] = \phi \mathbb{E}[y_{t-1}] + \mu_i,
\]
and under stationarity $\mathbb{E}[y_t] = \mathbb{E}[y_{t-1}]$, which yields
\begin{equation}\label{expected segment mean} 
\mathbb{E}[y_t] = \frac{\mu_i}{1 - \phi} = \mu_i^\prime , \qquad \tau_{i-1} \leq t \leq \tau_i - 1.
\end{equation}

Estimating changepoint models under an AR dependence structure presents additional challenges, as serial correlation can produce fluctuations that mimic changepoint-like shifts. Positive autocorrelation, in particular, induces clusters of residuals with the same sign, which can resemble a mean shift or structural break when none is present. When this dependence is ignored, changepoint procedures may incorrectly interpret correlated noise as genuine structural change, leading to spurious detections.


One heuristic approach to mitigate this problem is the {Wild Contrast Maximization} (WCM) method of \citet{cho-2024-WCM}. While parameterized slightly differently, WCM is loosely equivalent to \eqref{AR1 model}. Specifically, WCM uses the formulation
\begin{align}\label{WCM equation}
    y_t &= \begin{cases}
        \mu_i^\prime + \eta_t, & t = 1, \\
        \mu_i^\prime + \eta_t,  & \tau_{i-1} \le t \le \tau_i - 1,\ t \ne 1,
    \end{cases}\\
    \eta_t &= \phi \eta_{t-1} + \epsilon_t,\hspace{1cm} \epsilon_t \sim \text{WN}(0,\sigma^2)\label{autocorr}
\end{align}
where $\text{WN}(0,\sigma^2)$ denotes a white noise process with mean 0 and variance $\sigma^2$. Note that the dependence structure induced by $\phi$ enters the model only through the innovations, and for this reason \eqref{WCM equation} is often referred to in the literature as an Innovation Outlier (IO) model \citep{box-1975-intervention, Perron-1991_ERPmemo}.

Under the assumption that changepoints are relatively rare, i.e.\ $\mu_t^\prime = \mu_{t-1}^\prime$ for most $t$, \eqref{WCM equation} admits the relationship
\begin{align}\label{Innovation Outlier}
y_t &= \mu_t^\prime + \eta_t\nonumber \\
    &\overset{(\ref{autocorr})}{=} \mu^\prime_t +\phi\eta_{t-1}+\epsilon_t \nonumber\\
    &\overset{(\ref{WCM equation})}{=}  \mu_t^\prime + \phi(y_{t-1} - \mu_{t-1}^\prime) + \epsilon_t \nonumber\\
    &= \mu_t^\prime + \phi(y_{t-1} - \mu_t^\prime) + \epsilon_t \qquad\text{(most of the time)}\nonumber\\
    &= \phi y_{t-1} + (1 - \phi)\mu_t^\prime + \epsilon_t \nonumber\\
    &\overset{\eqref{expected segment mean}}{=}\phi y_{t-1} + \mu_t + \epsilon_t.
\end{align}
The approximate equality between \eqref{AR1 model} and \eqref{Innovation Outlier} is immediate once the parameters are expressed on the same scale, showing that the two handle serial dependence in much the same way. In fact, the only difference is that the mean shift modeled by \eqref{AR1 model} is not ``abrupt," but instead the mean gradually (with geometric decay rate) adjusts across changepoints. A more complete explanation and discussion of the differences between the LS and IO models is contained in Appendix~\ref{IO LS}.

WCM proceeds in two phases. In the first phase, starting from the full dataset $y_{1:t}$, the method searches for the subinterval $(s,e)$ with interior point $k \in (s,e)$ that maximizes an absolute CUSUM-type statistic:
\begin{align}
(s_0,k_0,e_0) &= \underset{s \le \ell < k < r \le e}{\mathrm{arg\,max}} \ \big| \mathcal{C}^k_{\ell,r} \big|, \quad \text{where} \label{WCM-max} \\
\nonumber\\
\mathcal{C}^k_{\ell,r} &= \sqrt{\frac{(k-\ell)(r-k)}{r-\ell}} \left( \frac{1}{k-\ell} \sum_{t=\ell}^{k-1} y_t - \frac{1}{r-k} \sum_{t=k}^r y_t \right).
\label{max cusum wcm}
\end{align}
The emphasis in \eqref{WCM-max} is that, like BS and WBS, the procedure advances by selecting the single triplet \((\ell,k,r)\) that yields the largest absolute contrast value $|\mathcal{C}^k_{\ell,r}|$ over all admissible subintervals. The location $k_0$ corresponding to this maximum is declared a candidate changepoint, and the data are split at $k_0$. The process is then applied recursively to each subsegment until $n-1$ candidate changepoints are found, each with an associated CUSUM score.

These changepoints, taken in decreasing order of the maximum statistic $\big|\mathcal{C}^k_{\ell,r}\big|$, define a nested sequence of models $\mathcal{M}_0 \subseteq \mathcal{M}_1 \subseteq \cdots \subseteq \mathcal{M}_{n-1}$, where $\mathcal{M}_0$ is the null (no changepoints), $\mathcal{M}_1$ corresponds to the changepoint maximizing $|\mathcal{C}^k_{\ell,r}|$, and $\mathcal{M}_{n-1}$ is the full $n-1$ changepoint model. These models are subsequently ranked by the Schwarz Criterion or BIC in \eqref{BIC}.

Because serial dependence can produce changepoint-like patterns in the statistic, models with falsely detected changepoints often have BIC scores of similar magnitude, while models with genuine changepoints tend to be separated by larger BIC scores. This results in what \citet{cho-2024-WCM} describe as a `gappy' sequence of BIC scores. The second phase of WCM exploits this by selecting the models corresponding to the $M$ largest gaps in the sorted BIC score sequence, where $M$ is the user-specified maximum number of changepoints.

In this second phase, the lag and estimate of autocorrelation are obtained via least squares, and a backward-selection procedure is applied using the BIC criterion to refine the segmentation. This involves comparing the BIC of segments with and without newly proposed changepoints, with the AR parameter $\phi$ estimated \emph{from the model that includes the greater number of changepoints}. This estimation choice is deliberate, as simulations show that estimating $\phi$ with fewer changepoints present can lead to a BIC which favors the omission of real changepoints in favor of explaining the structure through stronger autocorrelation. 

The underlying tension is that serial correlation and mean shifts can explain the same patterns in the data: increasing $\phi$ can smooth over small structural breaks, while adding changepoints can absorb what would otherwise appear as persistence. In the final analysis, one must decide which to prioritize—a model that attributes most variation to dependence, or one that attributes it to structural change. WCM lands firmly and explicitly on the latter of the two options, though for emphasis this is a choice made in \citet{cho-2024-WCM} which is explicitly acknowledged not to minimize the BIC. Pseudo code for WCM is provided in Appendix~\ref{app:pseudocode}.

The WCM approach can be viewed as an extension of WBS to settings with serial dependence, but with a distinct and more elaborate model selection stage. While WCM performs well in simulations (see Section 4), it inherits from BS and WBS the limitation that each segment is tested under an AMOC assumption, which may be inappropriate when multiple changepoints exist in the segment. This can lead to reduced accuracy when changepoints are subtle or closely spaced. 

\subsubsection*{Narrowest-Over-Threshold (NOT)}

The {Narrowest-Over-Threshold} (NOT) algorithm \citep{baranowski-2019-narrowest} modifies the identification criterion at each recursive step to address this limitation. Like WCM, NOT is closely related to the WBS framework in that both rely on local detection combined with global aggregation to produce the final segmentation. However, NOT introduces an important refinement: it explicitly seeks to ensure that the segments being tested for changepoints contain at most one changepoint. This design allows for more accurate localization of changepoints, as the fitted model on each segment better matches the statistical assumptions of the test being applied---namely, that the data within a segment are homogeneous, with no additional mean shifts confounding the contrast statistic.

The NOT algorithm begins by generating $M$ subsequences of the data $\{y_{s:e} : 1 \le s < e \le n\}$, with the restriction that $e - s \ge 2d$, where $d$ is the minimum number of observations required to fit the chosen model. For each subsequence $y_{s:e}$ and each candidate changepoint location $b \in \{s + d, \dots, e - d\}$, a generalized likelihood ratio (GLR) statistic is computed:
\begin{equation}
\label{GLR}
\mathcal{R}^b_{[s,e]}(y) = 2 \log \left[ \frac{\underset{\Theta_1, \Theta_2}{\sup} \ L(y_{s:(b-1)}; \Theta_1) \ L(y_{b:e}; \Theta_2)}{\underset{\Theta}{\sup} \ L(y_{s:e}; \Theta)} \right],
\end{equation}
where $L(y_{s:e}; \Theta)$ denotes the likelihood of the data $y_{s:e}$ under parameter vector $\Theta$. 

For example, under a Gaussian mean-shift model with independent errors,
\[
L(y_{s:e}; \Theta) 
= (2\pi\sigma^2)^{-\frac{e-s+1}{2}}
\exp\!\left\{-\frac{1}{2\sigma^2}\sum_{t=s}^{e}(y_t - \mu)^2\right\},
\]
where $\Theta = (\mu, \sigma^2)$. 

The index $b$ that maximizes \eqref{GLR} is the best candidate changepoint for that interval:
\begin{equation}
\label{NOT-best-changepoint-subsample}
\mathcal{R}_{[s,e]}(y) = \underset{b \in \{s + d, \dots, e - d\}}{\max} \ \mathcal{R}^b_{[s,e]}(y).
\end{equation}

In the case of Gaussian noise, the GLR statistic can be expressed as a {contrast function}—the inner product between the observed data and an appropriate weight vector. For computational convenience, \citet{baranowski-2019-narrowest} define the contrast $\mathcal{C}^b_{s,e}$ such that
\[
    \underset{b}{\mathrm{arg\,max}} \ \mathcal{C}^b_{s,e} = \underset{b}{\mathrm{arg\,max}} \ \mathcal{R}^b_{[s,e]}(y),
\]
and
\[
    \mathcal{C}_{s,e} = \underset{b}{\max} \ \mathcal{C}^b_{s,e}.
\]
Here, $\mathcal{C}^b_{s,e}$ denotes the contrast value at a specific candidate changepoint $b$ within interval $[s,e]$, while $\mathcal{C}_{s,e}$ is the maximum contrast attained within that interval. In other words, $\mathcal{C}^b_{s,e}$ measures the local evidence for a changepoint at position $b$, and $\mathcal{C}_{s,e}$ summarizes the strongest such evidence over all possible $b \in (s,e)$.

As in BS, WBS, and WCM, each interval produces a maximizer $\hat{b}$ where the contrast is largest, and the corresponding value $\mathcal{C}_{s,e}$ summarizes the evidence for a change. In this sense, the local detection mechanism aligns with the maximization principle used by the other methods: NOT also relies on locating the point of maximal local contrast within each interval and evaluating its strength against a common threshold.

Where NOT diverges is in how it chooses among  candidates exceeding the common threshold. Once all $\mathcal{C}_{s_m,e_m}$ have been computed for the $M$ generated intervals, they are compared to a user-specified threshold $\zeta_n$. Instead of simply taking the global maximum, NOT selects the interval of minimal length whose contrast exceeds the threshold. The corresponding maximizer $b_m$ within that narrowest qualifying interval is declared the changepoint. This “narrowest-over-threshold” rule gives the algorithm its name and ensures, with high probability, that the selected segment contains at most one changepoint, so that $b_m$ is a close approximation to the maximum-likelihood estimate of its location.

Once a changepoint is selected, the data are split into two segments at that location, and the procedure is applied recursively to each segment. Because NOT prioritizes intervals that are both short and highly significant, it reduces the risk of splitting on segments that contain multiple changepoints—an error that can propagate in methods like BS and WBS. Pseudo code for the NOT procedure is provided in Appendix~\ref{app:pseudocode}.

The GLR framework in \eqref{GLR} is flexible because it separates the data-generating assumptions from the detection procedure: once the likelihood $L(y_{s:e};\Theta)$ is specified, the same maximization principle applies regardless of the underlying model. This allows the test statistic to adapt naturally to different types of structural changes. There are four specific cases described in \citet{baranowski-2019-narrowest}:
\begin{itemize}
    \item Constant variance, piecewise constant mean.
    \item Constant variance, continuous and piecewise linear mean.
    \item Constant variance, piecewise linear (not necessarily continuous) mean.
    \item Piecewise constant variance, piecewise constant mean.
\end{itemize}
In summary, NOT retains the efficiency and local focus of WBS-style methods but enhances localization accuracy by implicitly enforcing a one-changepoint-per-segment condition. This design makes it particularly effective in scenarios where changepoints are closely spaced or where accurate localization is critical. However, like other segmentation algorithms in this family, NOT assumes that residual variability is explained by random noise, dependence, or trend.

\subsubsection*{Mean Shift Models with Seasonality}

In many real-world applications—particularly in environmental, economic, and biomedical time series—seasonality constitutes a dominant source of structured variability. When unaccounted for, seasonal patterns can obscure or mimic mean shifts, leading to false detections or missed changepoints. To address this, changepoint models can be extended to incorporate periodic components, yielding a mean shift formulation with seasonality. 

In this setting, the mean structure is augmented to include a seasonal component $s_t$, producing the model
\begin{equation}
\label{full model}
y_t =
\begin{cases}
    \mu_1 + s_t + \epsilon_t, & 1 \le t \le p - 1, \\
    \mu_1 + s_t + \epsilon_t, & \tau_0 \le t \le \tau_1 - 1, \\
    \mu_2 + s_t + \epsilon_t, & \tau_1 \le t \le \tau_2 - 1, \\
    \hspace{1cm} \vdots & \\
    \mu_m + s_t + \epsilon_t, & \tau_m \le t \le \tau_{m+1} - 1,
\end{cases}
\end{equation}
where the errors $\{\epsilon_t\}$ are i.i.d.\ with mean zero and variance $\sigma^2$. To ensure the number of available observations exceeds the number of free parameters, the constraints $\tau_0 = p$ and $\tau_{m+1} = n+1$ are also imposed as boundary conditions. This guarantees that the seasonal effects can be uniquely determined. The seasonal components are constrained to satisfy
\[
    s_t = s_{t+p} \quad \text{for} \quad 1 \le t \le n-p,
\]
and 
\[
    \sum_{t=1}^p s_t = 0.
\]
These conditions mean that there are sufficient observations to estimate a model uniquely with a changepoint placed at every value of $t$ for $t=p$, $p+1,\ldots,n-1,n$ (though a full model like this would be undesirable for analytic reasons).

Detecting changepoints in the presence of seasonality is considerably more 
challenging than in the pure mean-shift setting. 
The regular oscillation of a seasonal pattern introduces structured variability 
that can easily be mistaken for a change in mean—
a phenomenon analogous to the effect of positive autocorrelation discussed 
in the Wild Contrast Maximization subsection of Section~2, 
where clusters of residuals with the same sign may mimic a mean shift 
or structural break when none is present. 
What appears to be a sudden rise or fall in the data may, in fact, 
be nothing more than the expected crest or trough of the seasonal cycle. 
Conversely, a genuine mean shift can be obscured if it occurs near such an extremum, 
as the seasonal oscillation tends to dominate the observed variation.

This dual problem—spurious detections caused by regular oscillations and missed detections caused by masking—means that standard changepoint methods that ignore seasonality are unreliable. To distinguish true shifts from predictable variation, one must explicitly account for the seasonal structure, ensuring that the test statistic reflects departures from the expected oscillation rather than the oscillation itself.

While there is no extensive body of literature devoted to the problem of multiple mean shifts under seasonality, \citet{Lund-2007-seasonality} propose an $F_{\text{max}}$ test for detecting at most one undocumented mean shift when both autocorrelation and seasonality are present. Their test statistic is defined as
\begin{equation}
\label{fmax}
F_{\text{max}} = \underset{1 < \tau \le n}{\max} \; F_\tau,
\end{equation}
where
\begin{equation}
\label{ftau}
F_\tau = \frac{\mathrm{SSE}_0 - \mathrm{SSE}_A(\tau)}{\mathrm{SSE}_A(\tau) / (\mathrm{Error \ df})}.
\end{equation}
Here, $\mathrm{SSE}_0$ is the residual sum of squares under the null hypothesis of no changepoint, and $\mathrm{SSE}_A(\tau)$ is the residual sum of squares under the alternative hypothesis of a changepoint at time $\tau$. The ``Error df'' term denotes the residual degrees of freedom after accounting for all estimated parameters. Under $H_0$, $F_\tau \sim F_{1, \ \mathrm{Error \ df}}$. If $F_\tau$ exceeds the chosen critical value for some $\tau$, $H_0$ is rejected, and the changepoint location is taken to be the $\tau$ that maximizes $F_\tau$.

While $F_{\text{max}}$ provides a principled way to test for a single changepoint in seasonal data, it is, like other AMOC procedures, inherently limited to the at-most-one-changepoint case.  Moreover, no published work appears to have implemented an extension to the multiple-mean-shift case.
\subsubsection*{Summary}
Mean-shift models illustrate how quickly changepoint problems grow in complexity once additional structure—such as seasonality or autocorrelation—is introduced. Standard AMOC-based procedures serve as useful building blocks but can detect only one change at a time, requiring recursive application to handle multiple changepoints. Consequently, heuristic methods such as BS, WBS, WCM, and NOT are valued for their speed and adaptability but remain approximate, since early decisions in the recursion cannot be revised. 

\subsection{Exact Methods for Segmentation}
Exact segmentation algorithms approach changepoint detection as a global optimization problem. Rather than identifying changepoints sequentially, they evaluate all possible segmentations and select the configuration that minimizes a penalized cost function, guaranteeing optimality with respect to the chosen model and penalty. Because their solutions are constructed recursively from optimal sub-solutions, these algorithms fall under the broader framework of {dynamic programming}. In a computer science setting, ``dynamic'' refers to the recursive decomposition of the segmentation problem: the optimal segmentation of the first $t$ observations is built from the optimal segmentation of a shorter initial subsequence, plus the cost of the final segment.

\subsubsection*{Optimal Partitioning (OP)}
We begin with the simplest dynamic programming formulation, {Optimal Partitioning (OP)}, which seeks to solve the penalized minimization problem already introduced in~(\ref{cost}):
\[
\min_{\tau_{1:m}} \ \sum_{i=1}^{m+1} \mathcal{C}\big(y_{\tau_{i-1}:(\tau_i-1)}\big) + \beta \cdot f(\tau_{1:m}),
\]
where, as before, $\tau_0 = 1$, $\tau_{m+1} = n+1$, and $\beta > 0$ is a penalty factor controlling the number of changepoints. Unlike binary segmentation which does not have a global optimum guarantee, optimal partitioning guarantees identification of the best set of changepoints by the clever use of a dynamic programming technique.

The model assumed here is the multiple mean shift model~(\ref{totalvariation}):
\[
y_t =
\begin{cases}
\mu_1 + \epsilon_t, & \tau_0 \leq t \leq \tau_1 - 1, \\
\mu_2 + \epsilon_t, & \tau_1 \leq t \leq \tau_2 - 1, \\
\hspace{.6cm}\vdots \\
\mu_m + \epsilon_t, & \tau_m \leq t \leq \tau_{m+1} - 1,
\end{cases}
\]
where $\{\epsilon_t\}$ are i.i.d. zero-mean random variables with variance $\sigma^2$ and the $\mu_i$ are segment means. The necessary recursion is given in Proposition \ref{prop:op-recursion}.

\begin{prop}[Optimal partition recursion]\label{prop:op-recursion}
Let $\mathcal{T}_s=\{\tau:1=\tau_0<\tau_1<\cdots<\tau_m<\tau_{m+1}=s\}$ 
denote the set of changepoint vectors for $y_{1:(s-1)}$, and define
\[
F(s) = \min_{\tau \in \mathcal{T}_s} \sum_{i=1}^{m+1}
    \left[\mathcal{C}\bigl(y_{\tau_{i-1}:(\tau_i-1)}\bigr)+\beta\right],
    \qquad F(1)=-\beta.
\]
Then for any $s\ge 2$,
\begin{equation}\label{eq:op-recursion}
F(s) \;=\; \min_{1 \le r < s}\Bigl\{F(r)+\mathcal{C}\bigl(y_{r:(s-1)}\bigr)+\beta\Bigr\}.
\end{equation}
\end{prop}
In words: the optimal partition of $y_{1:(s-1)}$ must end with an optimal last changepoint at $r<s$.
\begin{proof}
\begin{align}\label{op-recursion} F(s) &= \min_{\tau \in \mathcal{T}_s} \left\{\sum_{i=1}^{m+1}[\mathcal{C}(y_{\tau_{i-1}:(\tau_i-1)})+\beta]\right\}\nonumber\\ &=\min_{1\le r<s} \left\{\min_{\tau \in \mathcal{T}_r}\sum_{i=1}^{m}[\mathcal{C}(y_{\tau_{i-1}:(\tau_i-1)})+\beta] + \mathcal{C}(y_{r:(s-1)})+\beta\right\}\nonumber\\ &=\min_{1\le r<s} \left\{F(r)+\mathcal{C}(y_{r:(s-1)})+\beta\right\}. \end{align}
\end{proof}
To solve this minimization therefore, the algorithm iteratively computes the minimum cost for segmenting the first 
$s-1$ points in the series $F(s)$ for $s = 2,3,\ldots,n+1$, using previously computed costs to build the optimal solution. At each step of the recursion, the optimal cost $F(s)$ is recorded along with the set of optimal changepoints for the first $s-1$ points, $\hat\tau(s)$. After completing the entire recursion, the set of optimal changepoints prior to the $n+1$ data point, $\hat\tau(n+1)$, is returned.

\subsubsection*{Pruned Exact Linear Time (PELT)}
PELT~\citep{killick-2012-pelt} builds directly on the OP recursion \eqref{eq:op-recursion}, still assuming the multiple mean shift model~\eqref{totalvariation} or other segment models compatible with~\eqref{cost}. The model is identical to that in OP; the difference lies in the efficiency of the search. This is accommodated by the use of a pruning step which often in practice reduces the running time of OP dramatically.

In \eqref{op-recursion}, the optimal $r<s$ at which to place (or not place) a changepoint is not known a priori, and so for each $1\le s\le n+1$, OP checks to see which $r<s$ is the optimal location for a changepoint prior to $s$. The number of checks grows linearly with $s$ and therefore quadratically with $n$. PELT implements a pruning step so that many, or even most, of the $r<s$ do not need to be checked at every recursive step because they are provably not able to be included in any optimal changepoint configuration. This fact is formalized in Proposition \ref{prop:pelt}.

\begin{prop}[PELT Pruning Criterion~\citep{killick-2012-pelt}]\label{prop:pelt}
Consider the recursion in \eqref{op-recursion}.  
Suppose there exists a constant $K$ such that for all $r<s<t\le n+1$,
\begin{equation}\label{eq:pelt-assumption}
    \mathcal{C}\bigl(y_{r:(s-1)}\bigr)
    + \mathcal{C}\bigl(y_{s:(t-1)}\bigr)
    + K
    \;\le\;
    \mathcal{C}\bigl(y_{r:(t-1)}\bigr).
\end{equation}
This assumption states that splitting any segment $y_{r:(t-1)}$ into two shorter segments at an intermediate index $s$ reduces the cost by at least $K$ (before accounting for penalties). Then, for any $t>s$, if
\begin{equation}\label{criterion}
    F(r)\;+\;\mathcal{C}\bigl(y_{r:(s-1)}\bigr)\;+\;K \;\ge\; F(s),
\end{equation}
the index $r$ can never be the optimal last changepoint prior to $t$.  
\end{prop}

\begin{proof}
Fix $r<s<t\le n+1$. By \eqref{op-recursion},
\begin{equation}\label{eq:pelt-upper}
F(t)\;\le\;F(s)+\mathcal{C}\bigl(y_{s:(t-1)}\bigr)+\beta.
\end{equation}
If $r$ is used as the last changepoint before $t$, the resulting cost is
\begin{equation}\label{eq:pelt-via-r}
F(r)+\mathcal{C}\bigl(y_{r:(t-1)}\bigr)+\beta.
\end{equation}
By \eqref{eq:pelt-assumption},
\begin{equation}\label{eq:pelt-split}
\mathcal{C}\bigl(y_{r:(t-1)}\bigr)
\;\ge\;
\mathcal{C}\bigl(y_{r:(s-1)}\bigr)
+
\mathcal{C}\bigl(y_{s:(t-1)}\bigr)
+
K.
\end{equation}
Combining \eqref{eq:pelt-via-r} and \eqref{eq:pelt-split} yields
\[
F(r)+\mathcal{C}\bigl(y_{r:(t-1)}\bigr)+\beta
\;\ge\;
F(r)+\mathcal{C}\bigl(y_{r:(s-1)}\bigr)+K+\mathcal{C}\bigl(y_{s:(t-1)}\bigr)+\beta.
\]
If \eqref{criterion} holds, then
\[
F(r)+\mathcal{C}\bigl(y_{r:(s-1)}\bigr)+K
\;\ge\;
F(s),
\]
and therefore
\[
F(r)+\mathcal{C}\bigl(y_{r:(t-1)}\bigr)+\beta
\;\ge\;
F(s)+\mathcal{C}\bigl(y_{s:(t-1)}\bigr)+\beta.
\]

Thus, for this $t$, using $r$ as the last changepoint yields a cost \emph{no smaller} than using $s$ as the last changepoint. Since $t>s$ was arbitrary, the same holds for all future $t$: $r$ never gives a strictly smaller value of $F(t)$ than $s$. Hence, $r$ can be removed from the candidate set without loss of optimality.
\end{proof}

\begin{remark}
If a time index $r$ satisfies the pruning condition in \eqref{criterion} for any value of $t$, it can be permanently removed from consideration as a potential last changepoint in all subsequent iterations~\citep{killick-2012-pelt}. 
\end{remark}

\subsubsection*{Conjugate Pruned Optimal Partitioning (CPOP)}

While PELT is highly effective for detecting changes in mean, many real-world processes are better modeled by changes in slope rather than abrupt mean shifts. Conjugate Pruned Optimal Partitioning (CPOP), introduced in \citet{fearnhead-2019-cpop}, addresses this case by fitting a continuous piecewise linear function to the data, with changepoints corresponding to slope changes. Continuity is enforced at each changepoint, so the fitted value at the end of one segment equals the fitted value at the start of the next. This makes CPOP particularly useful for processes where the rate of change is more informative than absolute levels. We adopt the notation of \citet{fearnhead-2019-cpop} below.

Formally, let $\{y_t\}_{t=1}^n$ be a time series with $m$ changepoints $2\le\tau_1<\tau_2<\cdots<\tau_m\le n$, and set $\tau_0=1$ and $\tau_{m+1}=n+1$ as boundary conditions. These changepoints divide the data into $m+1$ contiguous segments. Each changepoint $\tau_i$ is associated with a fitted value $\phi_{\tau_i}$, so that the fitted function is determined by the ordered pairs $\{(\tau_i,\phi_{\tau_i})\}_{i=0}^{m+1}$. In plain words, these ordered pairs are points in $(t,y_t)$-space corresponding in the first index to the temporal location of the changepoint, and in the second index to the $y$-value of the estimated signal at that changepoint. Collecting these values gives $\phi=(\phi_{\tau_0},\dots,\phi_{\tau_{m+1}})$, and for $j\le k$, the subset $\phi_{j:k}=(\phi_{\tau_j},\dots,\phi_{\tau_k})$ corresponds to the fitted values at changepoints $\tau_j$ through $\tau_k$.

On each segment, the fitted function interpolates linearly between adjacent fitted values. The model is therefore
\[
y_t =
\begin{cases}
\phi_{\tau_0} + \dfrac{\phi_{\tau_1}-\phi_{\tau_0}}{\tau_1-\tau_0}(t-\tau_0) + \epsilon_t, & \tau_0 \le t \le \tau_1-1, \\
\phi_{\tau_1} + \dfrac{\phi_{\tau_2}-\phi_{\tau_1}}{\tau_2-\tau_1}(t-\tau_1) + \epsilon_t, & \tau_1 \le t \le \tau_2-1, \\
\qquad\vdots & \\
\phi_{\tau_m} + \dfrac{\phi_{\tau_{m+1}}-\phi_{\tau_m}}{\tau_{m+1}-\tau_m}(t-\tau_m) + \epsilon_t, & \tau_m \le t \le \tau_{m+1}-1,
\end{cases}
\]
where the errors $\{\epsilon_t\}$ are i.i.d.\ with mean zero and variance $\sigma^2$. Thus, each segment is a straight line constrained to connect 
$(\tau_i, \phi_{\tau_i})$ and $(\tau_{i+1}, \phi_{\tau_{i+1}})$, 
thereby ensuring global continuity.

The penalized objective is
\[
\sum_{i=0}^m\left[
\frac{1}{\sigma^2}\sum_{t=\tau_i+1}^{\tau_{i+1}}
\left(y_t-\phi_{\tau_i}-\frac{\phi_{\tau_{i+1}}-\phi_{\tau_i}}{\tau_{i+1}-\tau_i}(t-\tau_i)\right)^2
+ h(\tau_{i+1}-\tau_i)
\right] + \beta m,
\]
where $h(\cdot)$ penalizes segment length and $\beta>0$ penalizes the number of changepoints. For a segment $y_{(s+1):t}$ with fitted value $\phi'$ at $s$ and $\phi$ at $t$, the segment cost is
\[
\mathcal{C}(y_{(s+1):t},\phi',\phi)
= \frac{1}{\sigma^2}\sum_{i=s+1}^t\left(y_i-\phi'-\frac{\phi-\phi'}{t-s}(i-s)\right)^2.
\]

To derive the recursion used in CPOP, let $f^t(\phi)$ denote the minimum penalized cost of segmenting $y_{1:t}$ subject to the constraint that the fitted value at time $t$ equals $\phi$. The recursion used in CPOP is formalized in Proposition \ref{prop:cpop}.
\begin{prop}[CPOP recursion \citep{fearnhead-2019-cpop}]\label{prop:cpop}
For $t\ge 1$ and fitted value $\phi$ at time $t$,
\[
f^t(\phi) \;=\; \min_{s<t,\;\phi'} \Bigl\{ f^s(\phi') + \mathcal{C}\bigl(y_{(s+1):t},\phi',\phi\bigr) + h(t-s) + \beta \Bigr\}.
\]
\end{prop}

\begin{proof}
By definition of $f^t(\phi)$, 
\begin{align}
f^t(\phi)
&= \min_{\tau,\,k,\,\phi_{0:k}} \Bigg\{
\sum_{i=0}^{k-1}\Big[\mathcal{C}\bigl(y_{(\tau_i+1):\tau_{i+1}},\phi_{\tau_i},\phi_{\tau_{i+1}}\bigr) + h(\tau_{i+1}-\tau_i)\Big] \nonumber\\
&\qquad\qquad\qquad\qquad\quad + \Big[\mathcal{C}\bigl(y_{(\tau_k+1):t},\phi_{\tau_k},\phi\bigr) + h(t-\tau_k)\Big] + \beta(k+1)\Bigg\}.
\label{eq:cpop-full}
\end{align}
Introduce the last changepoint $s:=\tau_k$ and its fitted value $\phi':=\phi_{\tau_k}$, and separate the final segment:
\begin{align}
f^t(\phi)
&= \min_{\phi',\,s} \Bigg\{
\min_{\tau_{0:(k-1)},\,k,\,\phi_{0:(k-1)}} \Bigg\{
\sum_{i=0}^{k-2}\Big[\mathcal{C}\bigl(y_{(\tau_i+1):\tau_{i+1}},\phi_{\tau_i},\phi_{\tau_{i+1}}\bigr) + h(\tau_{i+1}-\tau_i)\Big] \nonumber\\
&\qquad\qquad\qquad\qquad\qquad\qquad\quad + \mathcal{C}\bigl(y_{(\tau_{k-1}+1):s},\phi_{\tau_{k-1}},\phi'\bigr) + h(s-\tau_{k-1}) + \beta k \Bigg\} \nonumber\\
&\qquad\qquad\qquad\qquad\qquad\qquad\quad + \mathcal{C}\bigl(y_{(s+1):t},\phi',\phi\bigr) + h(t-s) + \beta \Bigg\}.
\label{eq:cpop-nested}
\end{align}
The inner minimum in \eqref{eq:cpop-nested} is precisely $f^s(\phi')$. Substituting yields
\[
f^t(\phi) \;=\; \min_{s<t,\;\phi'} \Bigl\{ f^s(\phi') + \mathcal{C}\bigl(y_{(s+1):t},\phi',\phi\bigr) + h(t-s) + \beta \Bigr\},
\]
which is the claimed recursion.
\end{proof}

Thus, as in \eqref{op-recursion} for OP, the optimal segmentation at $t$ can always be expressed in terms of an optimal segmentation ending at an earlier changepoint $s$. Unlike OP, in CPOP both changepoint locations and their fitted values $\phi$ must be carried through the recursion to maintain continuity. Practical implementations combine this recursion with pruning strategies to reduce the number of candidates that must be retained. A detailed algorithm and pseudo code are given in Appendix~\ref{app:pseudocode}.

\subsubsection*{Segment Neighbourhood (SN)}
Segment Neighbourhood (SN) takes a slightly different perspective: rather than fixing the penalty $\beta$ and optimizing over $m$, SN fixes $m$ and finds the optimal segmentation with exactly $m$ changepoints. The model is again the multiple mean shift model~(\ref{totalvariation}), but for each $m\leq m_\text{max}$ with $m_\text{max}$ being a user-specified maximium number of changepoints, SN minimizes the following cumulative loss:\begin{equation}
    \label{SN cost}
    \sum_{i=1}^{m+1}\mathcal{C}(y_{\tau_{i-1}:(\tau_i-1)})
\end{equation}
where once again $\tau_0=1$ and $\tau_{m+1}=n+1$ as boundary conditions. Note that (\ref{SN cost}) is equivalent to (\ref{cost}) with the penalty $f(\tau_{1:m})=0$.

Similar to OP, SN makes use of a dynamic programming technique wherein optimal segmentations are built recursively by using the relationship between the optimal segmentation of $y_{1:(s-1)}$ and the optimal segmentation for $y_{1:(r-1)}$ for each $r<s$. This recursion is formalized in Proposition \ref{prop:sn}.

\begin{prop}[Segment Neighbourhood Recursion \citep{Maidstone-2017-algorithms}]\label{prop:sn}
Let $F_m(s)$ denote the minimum cost of segmenting $y_{1:(s-1)}$ with exactly $m$ changepoints, i.e.
\[
F_m(s) \;=\; \min_{\bm\tau}\;\sum_{i=1}^{m+1}\mathcal{C}\bigl(y_{\tau_{i-1}:(\tau_i-1)}\bigr),
\qquad m<s.
\]
Then for all $s$ and $m<s$, we have that
\begin{equation}\label{eq:sn-recursion}
F_m(s) \;=\; \min_{m<r<s}\,\Bigl\{F_{m-1}(r) + \mathcal{C}\bigl(y_{r:(s-1)}\bigr)\Bigr\}.
\end{equation}
\end{prop}

\begin{proof}
By definition of $F_m(s)$, every candidate changepoint vector $\bm\tau$ has a last changepoint $\tau_m=r<s$.  
Separating the segment at time $r$ gives
\[
F_m(s) \;=\;\min_{r=\tau_m<s}\left\{
    \left[\min_{\tau_{1:(m-1)}}\sum_{i=1}^m \mathcal{C}\bigl(y_{\tau_{i-1}:(\tau_i-1)}\bigr)\right]
    \;+\;\mathcal{C}\bigl(y_{\tau_m:(s-1)}\bigr)
\right\}.
\]
The inner minimum is exactly $F_{m-1}(r)$, so
\[
F_m(s) \;=\;\min_{m<r<s}\,\Bigl\{F_{m-1}(r)+\mathcal{C}\bigl(y_{r:(s-1)}\bigr)\Bigr\},
\]
which is \eqref{eq:sn-recursion}.
\end{proof}

As in OP, the values of $F_m(s)$ and corresponding $\hat\tau_m (s)$ are built recursively for each $1<s\leq n+1$. Pseudo code and implementation details are included in Appendix \ref{app:pseudocode}.

\subsubsection*{Pruned Segment Neighbourhood}

The Pruned Segment Neighbourhood (pSN) algorithm~\citep{Rigaill-2010-PrunedSN} extends the Segment Neighbourhood (SN) approach by introducing a pruning mechanism, analogous in spirit to PELT. The underlying model is unchanged---typically~\eqref{totalvariation} with a fixed number of changepoints $m$---but the method exploits the additivity of the cost function to achieve substantial computational gains. Specifically, assume
\[
\mathcal{C}(y_{r:s}) \;=\; \sum_{i=r}^s \gamma(y_i,\mu),
\]
where $\gamma(\cdot)$ is parameterized by a segment-level parameter $\mu$. Here $\mu$ is a free variable representing the constant value associated with the segment $y_{r:s}$; the algorithm will later optimize over $\mu$ when evaluating the total cost. Commonly, $\gamma(\cdot)$ represents squared-error loss so that
\[
\mathcal{C}(y_{r:s}) \;=\; \sum_{i=r}^s (y_i-\mu)^2
\]
represents the squared error loss over the data $y_{r:s}$ \emph{subject to the mean of the segment from $r$ to $s$ being $\mu$}.

To take advantage of this structure, pSN introduces auxiliary functions $\mathcal{F}_m^s(r,\mu)$, which represent the cost of segmenting $y_{1:(s-1)}$ with $m$ changepoints, the most recent one at $r$, and with the final segment $y_{r:(s-1)}$ parameterized by $\mu$, which for emphasis is a scalar $\in\mathbbm{R}$. Formally,
\[
\mathcal{F}_m^s(r,\mu) \;=\; F_{m-1}(r) + \sum_{i=r}^{s-1}\gamma(y_i,\mu), \qquad r\leq s-1,
\]
with the convention $\mathcal{F}_m^s(s,\mu)=F_{m-1}(s)$.  

The function $\mathcal{F}_m^s(r,\mu)$ therefore stores the total cost of all segmentations of the first $s{-}1$ observations that end in a segment starting at $r$ and parameterized by $\mu$. These functions serve two purposes. First, they can be updated recursively as
\begin{equation}\label{eq:update}
\mathcal{F}_m^s(r,\mu)\;=\;\mathcal{F}_m^{s-1}(r,\mu)+\gamma(y_s,\mu),
\end{equation}
so that the cost of extending a segment to time $s$ can be computed directly from its value at $s-1$. Second, they reproduce the SN recursion itself:
\begin{align*}
\min_{r<s}\min_\mu \mathcal{F}_m^s(r,\mu)
&=\min_{r<s}\min_\mu \Bigl[F_{m-1}(r)+\sum_{i=r}^{s-1}\gamma(y_i,\mu)\Bigr] \\
&=\min_{r<s}\Bigl[F_{m-1}(r)+\mathcal{C}(y_{r:(s-1)})\Bigr] \\
&=F_m(s).
\end{align*}
Thus the $\mathcal{F}_m^s(r,\mu)$ functions both update efficiently and encode the exact SN objective.

A key step in pSN is to express the cost for a fixed $\mu$ as the minimum over all possible most recent changepoints:
\[
\mathcal{F}_m^s(\mu)=\min_{r<s}\mathcal{F}_m^s(r,\mu).
\]
This leads to the recursive formulation in Proposition~\ref{prop:pruned-sn-structure}.

\begin{prop}[Recursive Structure in pSN \citep{Rigaill-2010-PrunedSN}]\label{prop:pruned-sn-structure}
For $m<s$,
\[
\mathcal{F}_m^s(\mu)
=\min\Bigl\{\mathcal{F}_m^{s-1}(\mu)+\gamma(y_s,\mu),\;F_{m-1}(s)\Bigr\}.
\]
Moreover, for any fixed $\mu$, there exists some $r<s$ such that
\[
\mathcal{F}_m^s(\mu)=\mathcal{F}_m^s(r,\mu).
\]
\end{prop}

\begin{proof}
From the definition,
\[
\mathcal{F}_m^s(\mu)=\min_{r<s}\mathcal{F}_m^s(r,\mu).
\]
For $r<s$,
\[
\mathcal{F}_m^s(r,\mu)=\mathcal{F}_m^{s-1}(r,\mu)+\gamma(y_s,\mu),
\]
and for $r=s$,
\[
\mathcal{F}_m^s(s,\mu)=F_{m-1}(s).
\]
Therefore
\[
\mathcal{F}_m^s(\mu)=\min\Bigl\{\mathcal{F}_m^{s-1}(\mu)+\gamma(y_s,\mu),\;F_{m-1}(s)\Bigr\}.
\]
\end{proof}

Proposition~\ref{prop:pruned-sn-structure} implies that for each fixed $\mu$, the minimal cost $\mathcal{F}_m^s(\mu)$ is achieved by at least one candidate changepoint $r<s$. If for some $r$ there is no $\mu$ for which it is optimal at step $s$, then $r$ can never be optimal at any later time and may be pruned. This yields the following result.

\begin{prop}[pSN Pruning Criterion]\label{prop:pruned-sn}
If for some $r<s$,
\[
\mathcal{F}_m^s(r,\mu) > \mathcal{F}_m^s(\mu)\qquad\text{for all }\mu,
\]
then $r$ can never be optimal for any $t>s$ and may be discarded.
\end{prop}

\begin{proof}
Assume $\mathcal{F}_m^s(r,\mu)>\mathcal{F}_m^s(\mu)$ for all $\mu$. Then
\[
\mathcal{F}_m^{s+1}(r,\mu)=\mathcal{F}_m^s(r,\mu)+\gamma(y_{s+1},\mu)
> \mathcal{F}_m^s(\mu)+\gamma(y_{s+1},\mu)\ge\mathcal{F}_m^{s+1}(\mu).
\]
By induction, the inequality
\[
\mathcal{F}_m^t(r,\mu)>\mathcal{F}_m^t(\mu)\qquad\text{for all }\mu
\]
holds for every $t>s$. Hence $r$ can never be optimal.
\end{proof}

This pruning strategy greatly reduces the number of candidate changepoints that must be retained at each step, leading to substantial computational savings relative to the unpruned SN algorithm. A complete description and pseudo code are provided in Appendix~\ref{app:pseudocode}.

\subsubsection*{AR1Seg}
As noted previously, one of the key challenges in changepoint detection is identifying changepoint locations when the data exhibit dependence rather than being independently distributed. A particularly relevant case is when the data follow a white noise process with an AR(1) dependence structure as in (\ref{AR1 model}). 

In this setting, changepoints are not sought in the autoregressive dependence itself—which is assumed constant—but in the underlying mean structure. That is, the mean level of the process may shift across segments, while the autocorrelation parameter $\phi$ remains stationary within each segment. The difficulty arises because on the one hand, an accurate estimate of $\phi$ is crucial for correctly identifying changepoint locations, yet on the other hand, reliable changepoint estimates are needed to estimate $\phi$ accurately.
 The AR1Seg procedure, introduced by \citet{Chakar-2017-AR1Seg}, addresses this challenge by separating the estimation process into two distinct steps. This procedure, along with clarifying details and modifications for notational consistency, is presented below. 

First, an estimate of the AR(1) parameter $\phi$ is obtained using the robust median‐difference estimator proposed by \citet{Chakar-2017-AR1Seg}. This estimator was introduced as a modification of the robust autocorrelation estimator of \citet{Ma-2000-RobustAR}, designed to improve robustness in the presence of mean shifts. It is computed as
\begin{equation}
    \tilde{\phi} =
    \frac{(\text{med}_{1\leq i\leq n-2}|y_{i+2}-y_i|)^2}
         {(\text{med}_{1\leq i\leq n-1}|y_{i+1}-y_i|)^2} - 1.
\end{equation} 

Second, this estimate of $\phi$ is used to decorrelate the data. The act of ``decorrelating'' in this context means transforming the series so that the AR(1) dependence is removed, leaving residuals that are approximately independent and identically distributed. Specifically, given an estimate $\tilde{\phi}$, the transformed series  
\begin{equation*}
    z_t = y_t - \tilde{\phi} \, y_{t-1}
\end{equation*}  
can be interpreted as the sequence of one-step-ahead prediction errors under the fitted AR(1) model. Each $z_t$ is the observed value minus its best linear prediction from the immediately preceding observation, so the $z_t$ values should be uncorrelated (and, under Gaussian assumptions, independent) if the model is correct. It is noteworthy that the accuracy of the estimate of $\phi$ affects the quality of the de-correlation, and a more effective estimator of $\phi$ can be found in \citet{Shi-2022-autocovariance}.

Next, a dynamic programming approach is used to locate the changepoints and estimate the segment means of the (now independent) data. This is done with a dynamic programming solution found based on OP first articulated in \citet{auger-1989-op} and improved in \citet{Maidstone-2017-algorithms} and \citet{Rigaill-2012-criteria}, such as SN. Because spurious changepoints can be introduced as a result of the decorrelation procedure, an optional intermediate step is also included to trim back the number of located changepoints. A more thorough treatment of AR1Seg along with its pseudo code is contained in Appendix \ref{app:pseudocode}.

\section{Changepoint Detection via Model Selection}
In this section, we approach changepoint detection through the lens of model selection. At the center of this perspective is the generalized lasso (genlasso), a flexible framework for penalized regression that accommodates a wide range of model structures, along with an extension of the genlasso framework to efficiently estimate an $\ell_0$-approximated solution. We begin with the lasso and its adaptive variant the adalasso, reviewing their role as sparsity-inducing estimators and their interpretation as model selectors. We then extend to the fused lasso, which directly penalizes differences between adjacent parameters to detect changepoints, and to its adaptive variant, which improves consistency by incorporating reweighting. Each of these methods can be situated within the genlasso framework, which provides both conceptual clarity and methodological unification.

Because the genlasso can accommodate iterative reweighting, this section culminates in the guiding principle that carries forward into the next: iterate to isolate---a strategy for isolating \textit{true} changepoints that underpins the Iteratively Reweighted Fused Lasso (IRFL) introduced in Section 4.

\subsection{The Lasso and Fused Lasso}
Recall that the $\ell_1$ norm is given by  
\begin{equation*}
|| x||_1 = \sum_{j=1}^n | x_j|,
\end{equation*}
the sum of absolute values and the $\ell_\infty$ norm is
\begin{equation*}
|| x||_\infty = \max_{1 \leq j \leq n} | x_j|
\end{equation*}
 a measure of the largest component of a vector. The $\ell_0$ pseudo-norm is defined \begin{equation*}
     \label{l0 pseudonorm}
     \|x\|_0 = \sum_{j=1}^n\mathbbm{1}_{\{x_j\ne 0\}},
 \end{equation*} which counts the number of nonzero elements of a vector $x$. We see that the $\ell_0$ pseudo-norm is not a true norm because for any $\alpha>0$, $\alpha||x||_0\ne ||\alpha x||_0$, which is necessary for a norm.
 
The underlying changepoint problem that the fused lasso approximates can be expressed as a penalized empirical risk where the penalty is on the number of changepoints. Specifically, \eqref{cost} can be expressed as  
\begin{equation}\label{cpts}
    \hat{{\mu}}
    = \underset{{\mu}\in\mathbb{R}^n}{\arg\min}
    \left\{
        \| {y} - {\mu} \|_2^2
        + \lambda \, \| D {\mu} \|_0
    \right\},
\end{equation}
where \({\mu} = (\mu_1,\ldots,\mu_n)^\top\),  and 
\(D \in \mathbb{R}^{(n-1)\times n}\) is the first-order difference matrix
\begin{equation}\label{Difference matrix}
    D =
    \begin{bmatrix}
    -1 & 1 & 0 & \cdots & 0 \\
    0 & -1 & 1 & \cdots & 0 \\
    \vdots &\ddots & \ddots & \ddots & \vdots \\
    0 & \cdots & 0 & -1 & 1
    \end{bmatrix}.
\end{equation}
Here, the changepoints $\tau_k$ in (\ref{totalvariation}) correspond to the locations where 
\[
(D\mu)_t = \mu_t-\mu_{t-1}\neq0,
\]
naturally indicating a change in the mean between consecutive observations. In \eqref{cpts}, the Lagrange multiplier $\lambda$ is a regularization parameter that controls the trade-off between the goodness-of-fit and the number of changepoints, fulfilling the same role as $\beta$ in \eqref{cost}. As the number of changepoints increases in \eqref{cpts}, the empirical risk $\|y-\mu\|_2^2$ decreases, but each additional changepoint contributes a penalty of $\lambda$ to the overall cost. This penalized approach encourages parsimony by limiting the number of changes in the estimated mean, leading to a simpler, more interpretable model that avoids overfitting. 

While penalized regression introduces $\lambda$ to control the number of changes in the estimated mean, an alternate approach is to use the constrained form. Instead of applying a penalty to the number of changepoints, this method directly imposes a limit on the number of mean shifts allowed in the model. In this method, the minimization problem is conceived of as follows:
\begin{align}\label{minimize}
    &\hat{\mu} = \underset{{\mu}\in\mathbbm{R}^n}{\text{arg min }}\|y-\mu\|_2^2 \qquad \text{subject to}\quad \|D\mu\|_0\leq m,
\end{align}
for each possible number of changepoints $m \in \{0,1,2,\dots,n-1\}$.

Like the penalized form, the constrained form seeks to balance model complexity and fit by limiting the number of mean changes. This forces the model to choose only the most significant shifts in the data, effectively accomplishing the same goal of simplifying the model and preventing overfitting.

The relationship between the penalized and constrained forms can be understood through the paradigm of optimization theory, articulated below in Proposition \ref{prop:pen-con}.

\begin{prop}[Penalized Solution is Constrained-Optimal]
\label{prop:pen-con}
Fix $\lambda>0$ and let
\[
\hat\mu_\lambda \in \arg\min_{\mu\in\mathbb{R}^n}\Bigl\{\|y-\mu\|_2^2+\lambda\|D\mu\|_0\Bigr\},
\qquad
m_\lambda := \|D\hat\mu_\lambda\|_0.
\]
Then $\hat\mu_\lambda$ also solves the constrained problem with $m=m_\lambda$:
\[
\hat\mu_\lambda \in \arg\min\bigl\{\|y-\mu\|_2^2:\ \|D\mu\|_0\le m_\lambda\bigr\}.
\]
\end{prop}

\begin{proof}
Suppose not. Then there exists $\tilde\mu$ with $\|D\tilde\mu\|_0\le m_\lambda$ and
$\|y-\tilde\mu\|_2^2<\|y-\hat\mu_\lambda\|_2^2$. But then
\[
\|y-\tilde\mu\|_2^2+\lambda\|D\tilde\mu\|_0
\ \le\
\|y-\tilde\mu\|_2^2+\lambda m_\lambda
\ <\
\|y-\hat\mu_\lambda\|_2^2+\lambda m_\lambda
\ \le\
\|y-\hat\mu_\lambda\|_2^2+\lambda\|D\hat\mu_\lambda\|_0,
\]
contradicting the optimality of $\hat\mu_\lambda$ for the penalized problem.
\end{proof}

\begin{remark}
As noted in \citet{tibshirani-2014-closerlook}, the correspondence between $\lambda$ and $m$ is not one-to-one. 
For each value of $\lambda$, there is a unique corresponding model size $m_\lambda = \|D\hat\mu_\lambda\|_0$, 
but a given $m$ may remain optimal over a continuous interval of $\lambda$ values, producing a piecewise-constant 
solution path in $\lambda$.
\end{remark}

Notably, \eqref{cpts} can be expressed in a linear model form. Let\begin{equation}
\label{lineartransformation}
    \beta = \begin{bmatrix}
        \mu_1\\
        \mu_2-\mu_1\\
        \mu_3-\mu_2\\
        \vdots\\
        \mu_{n}-\mu_{n-1}
    \end{bmatrix},\qquad X=\begin{bmatrix}
        1 & 0 & 0 & \cdots & 0\\
        1 & 1 & 0 & \cdots & 0
        \\
        1 & 1 & 1 & \cdots & 0\\
        \vdots & \vdots & \vdots & \ddots & \vdots\\
        1 & 1 & 1 & \cdots & 1
    \end{bmatrix}_{n\times n}.
\end{equation} 
Now, $\beta$ and $\mu$ satisfy a linear relation:
\begin{equation}
\label{mu=xb}
    \mu = X\beta.
\end{equation} 
The nonzero $\beta_j$ therefore correspond to the changepoint locations; the value and sign of each $\beta_j$ corresponds to the difference between the segment means up to $y_t$ and after $y_t$. The estimate of $y_t$ is \begin{equation*}
    {\mu}_t=\sum_{j=1}^t \beta_j=(X\beta)_t.
\end{equation*} Using \eqref{lineartransformation} and \eqref{mu=xb}, \eqref{cpts} can be rewritten as \begin{equation}
    \label{L0-regularization}
    \hat\beta = \underset{\beta \in \mathbbm{R}^{n}}{\text{arg min }}\left\{\|y-X\beta\|_2^2 + \lambda||\beta||_0\right\}.
\end{equation}
Now, Equation (\ref{L0-regularization}) can then be viewed as the Lagrangian of the constrained formulation of the problem:
\begin{equation}
\label{L0-regularization-constrained}
    \hat\beta = \underset{\beta\in\mathbbm{R}^n}{\text{arg min }}\|y-X\beta\|_2^2\quad{\text{subject to}}\quad ||\beta||_0\leq m.
\end{equation}
The connection between \eqref{L0-regularization-constrained} and the original formulation in \eqref{cpts} is made explicit in Proposition \ref{prop:tv-increment}.

\begin{prop}\label{prop:tv-increment}
Under the reparameterization \eqref{lineartransformation}–\eqref{mu=xb}, the number of changepoints satisfies
\[
\|D\mu\|_0 \;=\; \bigl\|(\beta_2,\ldots,\beta_n)\bigr\|_0.
\]
Consequently, the constraint $\|D\mu\|_0\le m$ is equivalent to $\|(\beta_2,\ldots,\beta_n)\|_0\le m$, and 
\[
\min_{\mu}\;\|y-\mu\|_2^2+\lambda\|D\mu\|_0
\;=\;
\min_{\beta}\;\|y-X\beta\|_2^2+\lambda\,\bigl\|(\beta_2,\ldots,\beta_n)\bigr\|_0.
\]
In particular, the intercept $\beta_1=\mu_1$ is neither counted in $\|D\mu\|_0$ nor penalized in the equivalent $\beta$-formulation.
\end{prop}

\begin{proof}
From \eqref{mu=xb}, $(X\beta)_t=\sum_{j=1}^t\beta_j$. For $t\ge 2$,
\[
(D\mu)_t \;=\; \mu_t-\mu_{t-1} \;=\; (X\beta)_t-(X\beta)_{t-1}
\;=\; \Bigl(\sum_{j=1}^t\beta_j\Bigr)-\Bigl(\sum_{j=1}^{t-1}\beta_j\Bigr)
\;=\; \beta_t.
\]
Thus $(D\mu)_t=0$ if and only if $\beta_t=0$ for $t=2,\ldots,n$, yielding
$\|D\mu\|_0=\|(\beta_2,\ldots,\beta_n)\|_0$. Substituting $\mu=X\beta$ gives
$\|y-\mu\|_2^2=\|y-X\beta\|_2^2$, establishing the stated equivalences and the fact that $\beta_1$ is unpenalized.
\end{proof}

Even though a solution $\hat\beta$ can be estimated in polynomial time for the specific case of the design matrix $X$ as in \eqref{lineartransformation} (SN and pSN described in Section 2 admit a polynomial time solution), for a general $X$ \eqref{L0-regularization-constrained} is intractable. This is because it must shrink an unknown number of the $\beta_j$ coefficients to zero and therefore relies on a best-subset selection procedure. Such procedures are known to be NP-hard, as the size of a regressor set’s power set grows exponentially with the number of regressors (see \citet{garey-1979-computers, natarajan-1995-sparse}).

Because of the exponential time required to find a solution to (\ref{L0-regularization-constrained}), alternative methods for obtaining approximate solutions have been explored.
One such approach, explored in \citet{harchaoui-2010-tv}, is to replace the nonconvex $\ell_0$ penalty with its convex surrogate $\ell_1$ and solve 
\begin{equation}
    \label{L1-regularization}
    \hat\beta = \underset{\beta \in\mathbbm{R}^{n}}{\text{arg min }}\left\{\|y-X\beta\|_2^2 + \lambda\|\beta\|_1\right\},
\end{equation}
where $\|\cdot\|_1$ denotes the $\ell_1$ norm, and hope that the $\hat\beta$ which solves \eqref{L1-regularization} is ``close" to the solution to  \eqref{L0-regularization}.

This is simply the lasso estimator \citep{tibshirani-1996-lasso}. It was  shown in \citet{harchaoui-2010-tv} that when applied to the mean shift problem, the lasso is an approximately consistent estimator of the changepoint locations under appropriate regularity conditions. Here, ``approximately consistent’’ means that as the sample size goes to infinity, the estimated changepoints converge to neighborhoods of the true changepoints whose widths shrink to zero asymptotically, 
even if the estimated locations are not exactly equal to the true ones. In fact, the lasso estimator can be shown to be consistent estimator for $\beta$ under suitable regularity conditions on $X$. Specifically, consistency for estimation and prediction requires sparsity of the true parameter and conditions such as the restricted eigenvalue or compatibility conditions, while consistent variable selection additionally requires the irrepresentable condition and a minimal signal strength (see \citet{bickel-2009-lasso-dantzig, buhlmann-2011-hdstat, wainwright-2009-sharp-thresholds, zhao-2006-modelselection}).

Because of its convex formulation, computational efficiency, and regularization properties, the lasso has become a widely used alternative to forward and backward stepwise selection \citep{efron-2004-lar,hastie-2015-sparsity,hastie-2020-bestsubset}.

An equivalent formulation of \eqref{L1-regularization} using the familiar $\mu=X\beta$ linear transformation from \eqref{mu=xb} can be written as follows: 
\begin{equation}
    \label{fusedlasso}
        \hat{\mu}^{\text{FL}} = \underset{{\mu}\in\mathbbm{R}^n}{\text{arg min}}  \left\{\|y-\mu\|_2^2 + \lambda\|D\mu\|_1\right\}.
\end{equation}The estimator in \eqref{fusedlasso} is known as the fused lasso.  The major difference between (\ref{fusedlasso}) and (\ref{L1-regularization}) is the explicit articulation of the $\ell_1$ penalty on the difference of consecutive entries in $\mu$, which is known as a total variation penalty first discussed in \citet{harchaoui-2010-tv}. This penalty on the absolute differences between consecutive values of $\mu$ encourages the mean vector $\mu$ to fuse into long stretches of piecewise-constant values of $\mu_t$. 

The standard lasso penalizes the entries of $\beta$ directly, promoting sparsity in the coefficients rather than in the locations of changes in $\mu$. By contrast, the fused lasso applies its penalty to successive differences in $\mu$, inducing sparsity in $D\mu$. 
This structure corresponds to a piecewise-constant mean vector with relatively few changepoints, aligning naturally with the segmentation assumption underlying changepoint analysis. As a result, the fused lasso tends to produce estimates of $\mu$ consisting of long homogeneous stretches separated by abrupt shifts, making it the preferred formulation when the goal is changepoint detection rather than variable selection.

As will be seen in simulation studies in Section 4, the fused lasso tends to overestimate the number of changepoints present. To address this tendency, the adaptive lasso (adalasso) is introduced in the next subsection as a procedure which improves upon the lasso, and an equivalent procedure for the fused lasso---the adaptive fused lasso---is introduced and explained.

\subsection{The Adaptive Lasso and Adaptive Fused Lasso}

Despite its effectiveness for model selection, the lasso lacks the ``oracle" property. First discussed in \citet{fan-2001-variable}, the oracle property encapsulates an estimator's asymptotically ability to estimate the non-spurious regressors as well as if it had \textit{a priori} knowledge of what the set of true regressors actually was. More precisely, an estimator $\delta(y)$ has the oracle property if and only if the estimate of the regression coefficient by $\delta(y)$, $\hat{\beta}^\delta$, asymptotically selects the truly nonzero $\beta$ and has the optimal estimation rate. Formally, let $\beta^*$ be the true nonzero coefficients of $\beta$. Then if $\mathcal{A}=\{j:\beta_j^*\neq0\}$, $\delta(y)$ has the oracle property if and only if
\begin{enumerate}
    \item $\mathcal{P}(\{j:\hat\beta_j^\delta\neq0\}=\mathcal{A})\to 1 \text{ as } n \to \infty$, \text{ and }
    \item $\sqrt{n}(\hat{\beta}^\delta_\mathcal{A}-\beta^*)\overset{\mathcal{D}}{\to} \mathcal{N}(0,\Sigma^*)$.
\end{enumerate}

However, a variant of the lasso estimator does in fact obtain the coveted oracle property. 
This estimator is called the adaptive lasso (adalasso). 
It is defined as a two-step procedure:

\begin{enumerate}
    \item \textbf{First regression} 
    Obtain a root-$n$ consistent initial estimator $\hat{\beta}^{\hspace{2pt}\text{init}}$ of $\beta^*$ by regressing $y$ on $X$. 
    
    \item \textbf{Second regression} 
    Use $\hat{\beta}^{\hspace{2pt}\text{init}}$ to construct weights 
    \[
        \hat{w}_j = |\hat{\beta}^{\hspace{2pt}\text{init}}_j|^{-\gamma}, \quad \gamma > 0,
    \]
    and then solve the weighted lasso optimization
    \begin{equation}\label{adalasso definition}
        \hat{\beta}^{\hspace{2pt}\text{ada}} 
        = \underset{\beta}{\text{arg min}}
        \left\{
            \|y - X\beta\|_2^2 
            + \lambda \sum_{j=1}^p \hat{w}_j |\beta_j|
        \right\}.
    \end{equation}
\end{enumerate}
The adaptive lasso (adalasso) extends the lasso by reweighting coefficients according to an initial estimate $\hat{\beta}_j$, 
so that smaller coefficients incur larger penalties and are therefore more likely to be shrunk exactly to zero. This adaptive weighting enhances sparsity and enables consistent estimation of $\beta$ in settings where the standard lasso fails. \citet{zou-2006-adaptive} also demonstrate its superior ability to recover the true set of nonzero coefficients.

Formally, in mean shift changepoint problems, the adalasso solves
\begin{align}
    \label{adalasso regularization}
    \hat\beta 
    &= \underset{\beta \in\mathbbm{R}^{n}}{\text{arg min }}\left\{\|y-X\beta\|_2^2+ \lambda\sum_{t=2}^n\frac{|\beta_t|}{|\hat\beta^{\hspace{2pt}\text{init}}_t|}\right\}\nonumber\\
    &= \underset{\beta \in\mathbbm{R}^{n}}{\text{arg min }}\left\{\|y-X\beta\|_2^2 + \lambda\sum_{t=2}^n\hat w_t|\beta_t|\right\},
\end{align}
where $\hat w_t = |\hat\beta^{\hspace{2pt}\text{init}}_t|^{-\gamma}$ ($\gamma=1$ is commonly chosen), $X$ is a lower triangular matrix of ones, and 
\begin{equation}\label{first regression}
    \hat\beta^{\hspace{2pt}\text{init}} = \underset{\beta\in\mathbbm{R}^n}{\text{arg min }}\left\{\|y-X\beta\|_2^2+\lambda\sum_{t=2}^n|\beta_t|\right\}.
\end{equation}

This same reweighting principle can be extended to the fused lasso. 
Substituting the relationship among $\mu$, $X$, and $\beta$ from \eqref{lineartransformation} and \eqref{mu=xb} into 
\eqref{adalasso regularization} yields what we refer to as the \emph{adaptive fused lasso}, or adafused lasso:
\begin{align}
    \label{adafused lasso defn}
    \hat{\mu}^\text{Ada FL}
    &=\underset{\mu\in\mathbbm{R}^n}{\text{arg min }}\left\{\|y-X\beta\|_2^2+\lambda\sum_{t=2}^n\frac{|\mu_t-\mu_{t-1}|}{|\hat{\mu}_t^{\text{FL}}-\hat{\mu}_{t-1}^{\text{FL}}|}\right\}\nonumber\\
    &=\underset{\mu\in\mathbbm{R}^n}{\text{arg min }}\left\{\|y-X\beta\|_2^2+\lambda\sum_{t=2}^n\hat w_t|\mu_t-\mu_{t-1}|\right\},
\end{align}
where $\hat w_t=|\hat{\mu}_t^{\text{FL}}-\hat{\mu}_{t-1}^{\text{FL}}|^{-1}$ and $\hat{\mu}^{\text{FL}}$ is the fused lasso solution from \eqref{fusedlasso}. Note that the terms ``adaptive fused lasso’’ and ``fused adaptive lasso’’ appear elsewhere in the changepoint literature to describe related but distinct estimators. In \citet{ghassany-adafused-2010}, the adaptive fused lasso introduces adaptive weights on both the coefficient magnitudes and their successive differences, generalizing the lasso and fused lasso simultaneously. 
By contrast, the fused adaptive lasso of \citet{Rinaldo-2009-Refinements} applies adaptive weights only to the coefficient magnitudes, derived from a preliminary fusion step based on block averages. 
The formulation adopted here differs from both: adaptive weights are applied directly to the successive differences $|\hat{\mu}_t^{\text{FL}} - \hat{\mu}_{t-1}^{\text{FL}}|$, yielding a construction analogous to the adaptive lasso but expressed in terms of local mean shifts.

Reweighted iterations of the generalized lasso tend to isolate true changepoint locations—a principle we refer to as \textit{iterate to isolate}. 
Just as the adaptive lasso improves upon the lasso’s already-sparse solution by removing small, inconsequential coefficients $\beta_j$, 
the adaptive fused lasso refines the fused lasso by eliminating small, insignificant mean shifts $\hat{\mu}_t^{\text{FL}} - \hat{\mu}_{t-1}^{\text{FL}}$, 
yielding a more parsimonious segmentation. 
Simulation results in Section~4 will show that models obtained via the adaptive fused lasso achieve BIC minima that are never greater than those from the standard fused lasso, 
supporting the \emph{iterate to isolate} principle.
In summary, the adaptive fused lasso parallels the adaptive lasso in much the same way that the fused lasso parallels the standard lasso. 
The fused lasso solution corresponds to that of the lasso with a lower-triangular design matrix $X$ of ones under the transformation $\mu = X\beta$. Under this transformation, the adaptive fused lasso corresponds to the adaptive lasso, with weights $|\hat{\beta}^{\,\text{init}}_j|^{-1}$ mapping to $|\hat{\mu}^{\text{FL}}_t - \hat{\mu}^{\text{FL}}_{t-1}|^{-1}$. 
This equivalence is not coincidental: all four estimators—the lasso, adaptive lasso, fused lasso, and adaptive fused lasso—arise as special cases of the generalized lasso framework introduced in the next subsection.

\subsection{Generalized Lasso (Genlasso)}

The Generalized Lasso (genlasso) was first discussed in \citet{tibshirani-2011-solutiongenlasso} and solves the following:
\begin{equation}
    \label{genlasso-definition}
    \hat\beta = \underset{\beta}{\text{arg min }}\|y-X\beta\|_2^2 + \lambda||D\beta||_1
\end{equation} where $X$ is a general design matrix governing empirical risk and $D$ is a ``structure matrix" that identifies which linear combinations of the $\beta_j$'s are penalized by the $\ell_1$ norm.

Equation~(\ref{genlasso-definition}) is a general framework that encompasses several well-known estimators as special cases. 
By appropriate choices of the design matrix $X$ and structure matrix $D$, it reduces to the lasso, adaptive lasso, fused lasso, and adaptive fused lasso. Specifically:

\begin{enumerate}[label=(\roman*)]
    \item \textbf{Lasso.}  
    Setting $D = I_{n\times n}$ in (\ref{genlasso-definition}) yields the standard lasso estimator.

    \item \textbf{Adaptive Lasso.}  
    Let $\hat{\beta}^{\,\text{init}}$ be the solution to the lasso problem~(\ref{L1-regularization}).  
    Keeping $X$ fixed and choosing $D$ as a diagonal matrix with entries $D_{jj} = |\hat{\beta}^{\,\text{init}}_j|^{-1}$ produces the adaptive lasso estimator.

    \item \textbf{Fused Lasso.}  
    With $X = I_{n\times n}$, $\mu = \beta$, and $D = D_{1d}$ equal to the first-order difference matrix
    \begin{equation}
        \label{D Matrix}
        D_{1d}=\begin{bmatrix}
            -1 & 1 & 0 & \cdots & 0& 0\\
            0 & -1 & 1 & \cdots & 0&0\\
            \vdots & \vdots & \vdots & \ddots & \vdots&\vdots\\
            0 & 0 & 0 & \cdots & -1 & 1
        \end{bmatrix}_{(n-1)\times n},
    \end{equation}
    (\ref{genlasso-definition}) reduces to the fused lasso.

    \item \textbf{Adaptive Fused Lasso.}  
    If $\hat{\mu}^{\text{FL}}$ is the fused lasso solution~(\ref{fusedlasso}), then setting $X = I_{n\times n}$, $\beta = \mu$, and 
    \[
        D = 
        \begin{bmatrix}
            -\frac{1}{|\hat{\mu}^{\text{FL}}_2 - \hat{\mu}^{\text{FL}}_1|} &  \frac{1}{|\hat{\mu}^{\text{FL}}_2 - \hat{\mu}^{\text{FL}}_1|} & 0 & \ldots & 0\\
            0 & -\frac{1}{|\hat{\mu}^{\text{FL}}_3 - \hat{\mu}^{\text{FL}}_2|} & \frac{1}{|\hat{\mu}^{\text{FL}}_3 - \hat{\mu}^{\text{FL}}_2|} & \ldots & 0\\
            \vdots & \vdots & \vdots & \ddots & \vdots\\
            0 & 0 & 0 & \ldots & \frac{1}{|\hat{\mu}^{\text{FL}}_n - \hat{\mu}^{\text{FL}}_{n-1}|}
        \end{bmatrix}
    \]
    recovers the adaptive fused lasso.
\end{enumerate}
Thus, the flexibility of $(X, D)$ in (\ref{genlasso-definition}) unifies both regression-based and changepoint-based estimators within a single generalized lasso framework.

\subsubsection*{The Genlasso Solution Path}

An important property of the genlasso is that, for the given design matrix $X$ and structure matrix $D$, the estimator has a well-defined \textbf{solution path} as the tuning parameter $\lambda$ decreases from $+\infty$ to $0$. This path is piecewise linear in $\lambda$ \citep{tibshirani-2011-solutiongenlasso}, meaning that the set of nonzero entries in $D\hat\beta$ change only at a finite number of critical values of $\lambda$, called nodes. At each node, new structure is introduced into the model (for example, a new changepoint in the fused lasso), and between nodes the solution evolves linearly. 

This characterization of the solution path has two important consequences for changepoint detection. First, it allows efficient algorithms to trace the entire path of solutions without solving the optimization problem separately for each $\lambda$. Second, it provides a natural collection of candidate models from which a final estimate can be selected using information criteria such as BIC, AIC, mBIC and MDL---every node on the solution path corresponds to a unique model.

\subsection{Motivation for IRFL}
While the generalized lasso solution path provides a principled sequence of candidate changepoint models, selection using BIC often leads to oversegmentation. This occurs because the $\ell_1$ penalty produces gradual transitions rather than sharp, discrete ones, ``smearing" a changepoint across several indices, rather than locating it to a single discrete index. This limitation motivates refinements that better separate genuine from spurious changepoints.

The adaptive fused lasso \eqref{adafused lasso defn} offers one such refinement. By reweighting the penalty by the magnitude of the estimated changes, it penalizes small shifts (gradual transitions) more heavily while preserving large, genuine ones. This extra step prunes away many of the spurious changepoints left behind by the fused lasso. A single reweighting iteration can suppress noise and highlight the true structure more effectively, isolating the true changepoint locations. It is an enhancement of the $\ell_1$ relaxation, but it still does not provide an efficient approximation to the underlying $\ell_0$ changepoint problem.

Even this improvement does not go far enough. The adafused lasso solution often retains a few unnecessary changepoints, suggesting that one round of reweighting is not always sufficient. This observation naturally leads to the idea of iterating the adaptive step: if one reweighting improves the model, multiple reweightings may refine it further. The following examples illustrate this progression, showing how the fused lasso is good but not great, how the adafused lasso improves on it, and how further reweighting motivates the Iteratively Reweighted Fused Lasso (IRFL). 

\subsection{Iterate to Isolate}

To test the effectiveness of the fused lasso and demonstrate the need for an iterative reweighting scheme, we generated data as follows. A vector of length $1000$ was generated with the first 250 datapoints having a mean of 0. From time 251 to 500, the mean shifts up to 2. The mean then shifts back to 0 from index 501 to 750 and then returns to 2 from index 751 to 1000. The fused lasso solution was found using the $\texttt{genlasso}$ package in $\texttt{R}$. The results are found in Figure \ref{fig:fused lasso}. 
\begin{figure}[H]
    \centering
    \includegraphics[width=1\textwidth]{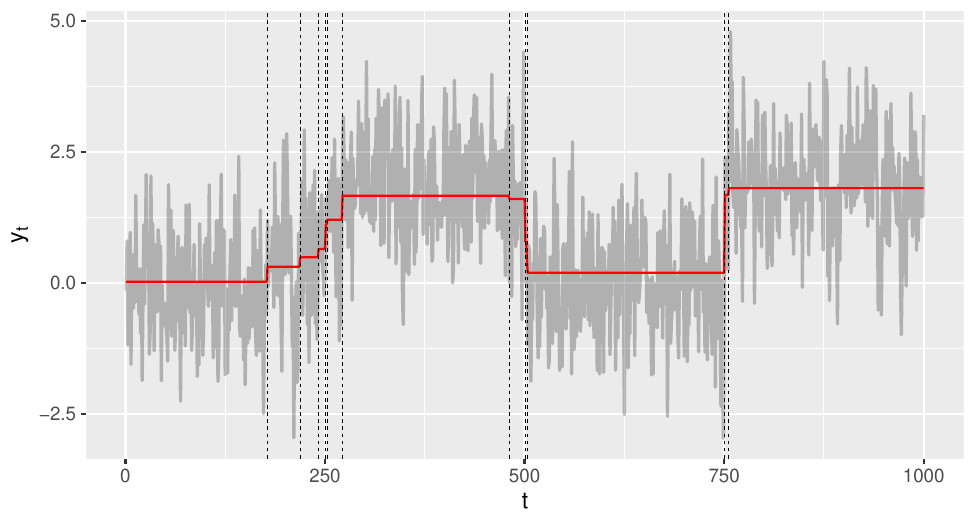}
    \caption[Spurious changepoints detected by fused lasso]{Changepoints detected by fused lasso. While the method identifies the correct changepoints, it also detects several spurious changepoints near the true changepoint locations.}
    \label{fig:fused lasso}
\end{figure}
As is evident from Figure \ref{fig:fused lasso}, the fused lasso solution does actually find the correct changepoints. However, it includes several spurious changepoints as well.

We also found the adafused lasso solution for the same dataset, demonstrating its usefulness in discarding spurious changepoints. The results are in Figure \ref{fig:adafused lasso}:
\begin{figure}[H]
    \centering
    \includegraphics[width=1\textwidth]{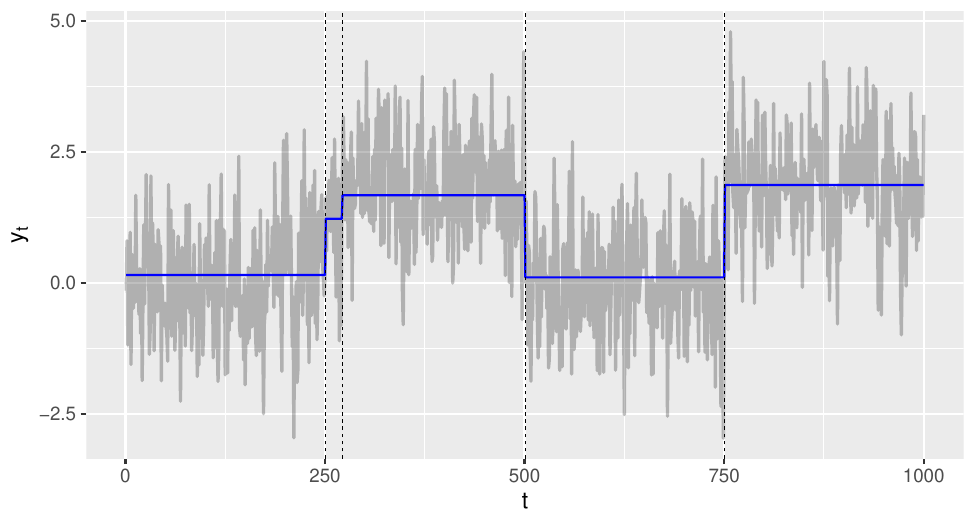}
    \caption[Adafused lasso improves upon fused lasso]{Changepoints detected by adafused lasso. The method correctly discards some of the spurious changepoints, but not all.}
    \label{fig:adafused lasso}
\end{figure}
We see in the example that the adaptive fused lasso significantly improves upon the fused lasso solution. It does so by eliminating most of the spurious changepoints, although one spurious changepoint still remains.

Upon examining the previous two figures, one might hope that another adaptive reweighting step could discard the extra changepoint. That is, if\begin{align*}
    \hat{\mu}^{\hspace{2pt}\text{Ada FL}}&=\underset{\mu\in\mathbbm{R}^n}{\text{arg min }}\left\{\sum_{t=1}^n(y_t-\mu_t)^2+\lambda\sum_{t=2}^n\hat w_t^{(2)}|\mu_t-\mu_{t-1}|\right\}
\end{align*} where $\hat w_t^{(2)}=|\hat{\mu}_t^{\text{FL}}-\hat{\mu}_{t-1}^{\text{FL}}|^{-1}$ and $\hat{\mu}^{\text{FL}}$ is a solution to the fused lasso problem in (\ref{fusedlasso}), then updating the weights $\hat w_t^{(3)}=|\hat{\mu}_t^{\text{Ada FL}}-\hat{\mu}_{t-1}^{\text{Ada FL}}|^{-1}$ and solving \begin{align*}
    \hat{\mu}^{{(3)}}&=\underset{\mu\in\mathbbm{R}^n}{\text{arg min }}\left\{\sum_{t=1}^n(y_t-\mu_t)^2+\lambda\sum_{t=2}^n\hat w_t^{(3)}|\mu_t-\mu_{t-1}|\right\}
\end{align*} 
would reduce or eliminate spurious changepoints. If intuition holds, the adaptive weights using the new estimated values from adafused lasso will improve on the adafused lasso estimate. As Figure \ref{fig:irfl} below demonstrates, this is indeed what happens.

\begin{figure}[H]
    \centering
    \includegraphics[width=1\textwidth]{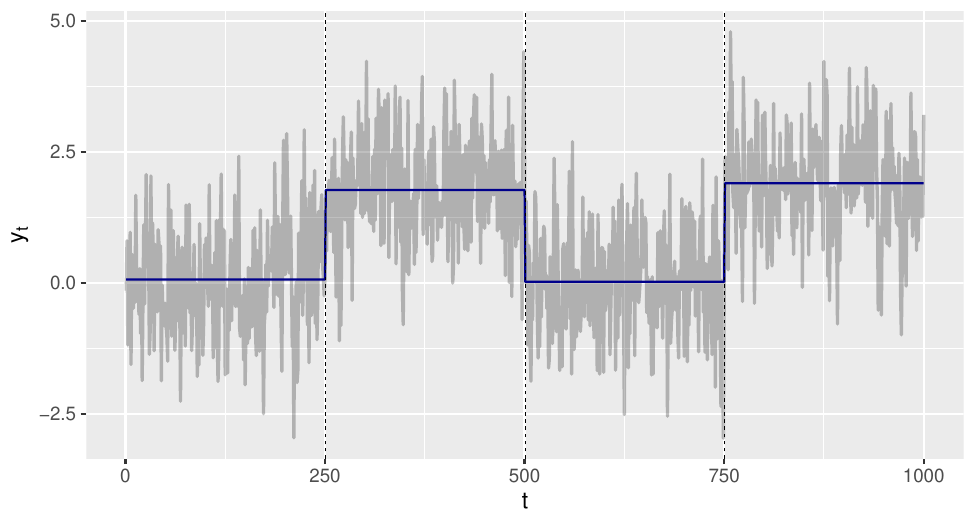}
    \caption[Additionally reweighting after adafused lasso improves results]{Using the adafused lasso solution as weights, a further iteration improves upon the adafused lasso solution by discarding the spurious changepoint.}
    \label{fig:irfl}
\end{figure}
The third regression successfully eliminated the spurious changepoint at ~$t=275$ while preserving the other three correct ones. The resulting model is more parsimonious and reflects the true model.

Motivated by this observation, the next section introduces the Iteratively Reweighted Fused Lasso (IRFL). At each iteration, the candidate changepoint model is re-estimated with updated adaptive weights and evaluated by its BIC score. The process continues only while successive iterations yield a significant improvement in model fit. 
This iterative refinement addresses the limitations of both the fused and adaptive fused lasso methods and establishes a flexible framework applicable to any model expressible within the generalized lasso formulation. 
The following section develops the IRFL algorithm in detail and demonstrates its effectiveness across a range of changepoint problems.

\section{Iteratively Reweighted Fused Lasso}

\subsection{IRFL}

The preceding examples showed that both the fused lasso and its adaptive extension, while able to locate the true changepoints, also introduce spurious ones near the boundaries between consecutive segments. Iterative reweighting offered a natural remedy—each round of reweighting progressively reduced false detections and improved model parsimony. 

Motivated by this observation, we now develop the Iteratively Reweighted Fused Lasso (IRFL) as a unifying approach to $\ell_0$ changepoint detection. The central idea is to place all of the problem types discussed in Section 2—mean shifts, autoregressive dependence, global trends, slope changes, and seasonal effects—under a single generalized lasso framework, with IRFL providing the common estimation scheme through iterative reweighting---the \emph{iterate to isolate} principle.

We will argue that IRFL can also be understood as a ``superior'' surrogate $\ell_0$ method: successive adaptive updates within the generalized lasso framework drive the solution toward the ideal, but computationally intractable, $\ell_0$-penalized model (see Section 3.1). The ideal $\ell_0$ estimator requires combinatorial search over all possible inclusion patterns of the $p$ parameters, corresponding to $2^p$ distinct models; exhaustive minimization over this discrete set is NP-hard and therefore infeasible except for trivially small $p$. This perspective is summarized heuristically in Equation \ref{0-1 convergence}.

We begin with a description of the IRFL method including the algorithm's pseudo code. We then introduce a distance metric used to evaluate changepoint recovery. Finally, we demonstrate IRFL’s performance across thirteen simulated scenarios, showing that it consistently improves upon existing methods while remaining broadly applicable across diverse changepoint settings.

\subsection{IRFL Algorithm Description}

Because the fused lasso is a special case of the generalized lasso, the method is presented in a $D$-matrix representation so that it applies broadly to changepoint problems. In the genlasso formulation, the pair $(X, D)$ together encode the design of the model being estimated. The columns of $X$ determine what structural patterns are included in the fit---for example, constant levels, linear/quadratic trends, or seasonality ---while the rows of $D$ specify which contrasts are penalized and therefore what kinds of changepoints can be detected. In effect, $X$ governs the functional form of the signal within each segment, and $D$ governs the rules by which segments can change. Readers are encouraged to see Section \ref{empirical results} for precise examples. 

This unified representation is a central contribution: it provides a single framework capable of estimating changepoint models of many different kinds under one banner. By varying only the specification of $(X, D)$, the method seamlessly adapts to a broad spectrum of model structures, from simple mean shifts to complex trend and seasonal structures. The algorithm proceeds as follows:

\vspace{.5cm}
\noindent \textbf{Step 0:} Initialize $i \gets 1$ and set $D^{(i)} \gets D_0$, 
where $D_0$ denotes the initial first‐order difference matrix.
\vspace{.5cm}

\noindent \textbf{Step 1:} Generate a full solution path for the generalized lasso problem:
\begin{equation}
    \label{IRFL-1}
    \hat\beta_\lambda 
    = \underset{\beta \in \mathbb{R}^n}{\mathrm{arg\,min}}
    \;\frac{1}{2}\|y - X\beta\|_2^2 
    + \lambda \|D^{(i)}\beta\|_1.
\end{equation}

\noindent \textbf{Step 2:} For each candidate along the solution path, compute the BIC score:
\begin{equation}
    \label{BIC-IRFL}
    \mathrm{BIC}(\hat\beta_\lambda) 
    = n\log\!\left(\frac{\mathrm{SSE}(\hat\beta_\lambda)}{n}\right) 
    + m_{\lambda}\log(n),
\end{equation}
where 
\[
m_\lambda = \big|\{ j : [D_0\hat\beta_\lambda]_j \neq 0 \}\big|
\]
is the number of detected changepoints. 
At iteration $i$, select the model
\begin{equation}
    \label{first beta}
    \hat{\beta}^{(i)} = \hat{\beta}_{\lambda^{(i)}}, 
    \qquad 
    \lambda^{(i)} = \arg\min_{\lambda} \mathrm{BIC}(\hat{\beta}_\lambda).
\end{equation}
Note: The arg min is unique except on a set of measure zero.

\noindent \textbf{Step 3:} Construct a diagonal weight matrix
\[
\hat W^{(i)} = \mathrm{diag}(\hat w_1^{(i)}, \ldots, \hat w_p^{(i)}),
\qquad
\hat w_j^{(i)} = \frac{1}{|(D_0\hat\beta^{(i)})_j| + \varepsilon},
\]
where $p = \mathrm{rank}(D_0)$ and $\varepsilon > 0$ is a small constant.

\noindent \textbf{Step 4:} Update the penalty matrix via
\[
D^{(i+1)} = \hat W^{(i)} D_0,
\]
and repeat Steps~1–4 until the BIC score no longer improves.

At the $i^\text{th}$ iteration for $i>1$, IRFL solves the weighted problem
\begin{equation}
    \label{IRFL-2}
    \hat\beta_\lambda^{(i+1)} 
    = \underset{\beta \in \mathbb{R}^n}{\text{arg min }} 
    \frac{1}{2}\|y - X\beta\|_2^2 
    + \lambda \sum_{j=1}^p 
    \frac{|(D_0\beta)_j|}{|(D_0\hat\beta^{(i)})_j| + \varepsilon}\,.
\end{equation}
Each iteration of IRFL can be viewed as a refinement of the previous one, driven by the adaptive weighting scheme. 

Small values of $D_0\beta$ are assigned large weights, which increases their penalty in the next iteration. This mechanism effectively “punishes” spurious or inconsequential changepoints, encouraging the corresponding coefficients to remain fused where they might otherwise separate.  

A concise formulation of the algorithm is included below.

\begin{algorithm}[H]
\caption{Iteratively Reweighted Fused Lasso (IRFL)}
\begin{algorithmic}[1]
\State \textbf{Input:} Data $y \in \mathbb{R}^{n}$; design matrix $X \in \mathbb{R}^{n\times n}$; initial difference matrix $D_0 \in \mathbb{R}^{p\times n}$; offset $\varepsilon>0$; convergence tolerance \textbf{tol}, maximum number of iterations $K$.
\vspace{4pt}

\State \textbf{Initialize:} 
$i \gets 1$; \; $D^{(1)} \gets D_0$; \; $\mathrm{BIC}_{\mathrm{prev}} \gets +\infty$
\vspace{4pt}

\Repeat
    \State Generate the full generalized lasso path:
    \begin{equation*}
        \hat{\beta}_\lambda^{(i)} 
        = \arg\min_{\beta \in \mathbb{R}^n}
        \frac{1}{2}\|y - X\beta\|_2^2 
        + \lambda\,\|D^{(i)}\beta\|_1,
        \qquad \lambda>0.
    \end{equation*}
    \For{each candidate $\beta$ on the path}
        \State Compute the BIC:
        \begin{equation*}
            \mathrm{BIC}(\beta) 
            = n\log\!\Big(\tfrac{\mathrm{SSE}(\beta)}{n}\Big)
            + \|D_0\beta\|_0 \log n,
        \end{equation*}
        where $\|D_0\beta\|_0 = |\{ j : (D_0\beta)_j \neq 0 \}|$.
    \EndFor
    \State \textbf{Select:} 
    \[
        \hat{\beta}^{(i)} = \arg\min_\beta \mathrm{BIC}(\beta),
        \qquad \mathrm{BIC}_{\mathrm{curr}} = \mathrm{BIC}(\hat{\beta}^{(i)}).
    \]
    \State \textbf{Reweight:} For $j=1,\ldots,p$ set
    \[
        \hat w^{(i+1)}_j = \frac{1}{|(D_0\hat{\beta}^{(i)})_j| + \varepsilon},
        \qquad 
        \hat W^{(i+1)} = \mathrm{diag}(\hat w^{(i+1)}_1,\ldots,\hat w^{(i+1)}_p).
    \]
    \State \textbf{Update:} 
    $D^{(i+1)} \gets \hat W^{(i+1)} D_0.$
    \State \textbf{Check convergence:} 
    stop if $|\mathrm{BIC}_{\mathrm{curr}} - \mathrm{BIC}_{\mathrm{prev}}| \le \textbf{tol}$.
    \State Update $\mathrm{BIC}_{\mathrm{prev}} \gets \mathrm{BIC}_{\mathrm{curr}}$, \; $i \gets i+1$.
    \State stop if $i>K$
\Until{converged}
\vspace{4pt}

\State \textbf{Output:} Final estimate $\hat{\beta} = \hat{\beta}^{(i)}$; changepoint set $\{ j : (D_0\hat{\beta})_j \neq 0 \}$.
\end{algorithmic}
\end{algorithm}

\subsection[A Surrogate $\ell_0$ Solution]{A Surrogate $\ell_0$ Solution}

The mechanism underlying IRFL can be described as the \textit{iterate to isolate} principle. With each reweighting, smaller changes are increasingly penalized, while genuine ones persist. This iterative reweighting has a tendency to isolate true changepoints, revealing a theoretical connection: repeated reweighting approximates an $\ell_0$-regularized solution.

Intuitively, as the iterations proceed, the adaptive penalty in \eqref{IRFL-2} begins to behave like an indicator, distinguishing whether a jump is truly present.  
For each fixed $\varepsilon>0$, taking the limit in iterations $i\to\infty$ gives
\[
\lim_{i\to\infty}
\frac{|(D\hat\beta^{(i+1)})_j|}{|(D\hat\beta^{(i)})_j|+\varepsilon}
=
\begin{cases}
0, & (D\hat\beta^{(i)})_j=0,\\[6pt]
1 - \dfrac{\varepsilon}{|(D\hat\beta^{(i)})_j|+\varepsilon}, & (D\hat\beta^{(i)})_j\neq0.
\end{cases}
\]
If $\varepsilon=\varepsilon_n$ is allowed to decrease with the sample size $n$, then as $n\to\infty$, $\varepsilon_n\downarrow0$ and asymptotically
\[
\lim_{n\to\infty}
\left(
1 - \frac{\varepsilon_n}{|(D\hat\beta^{(i)})_j|+\varepsilon_n}
\right)
= 1.
\]
Putting these pieces together, we obtain
\[
\lim_{n\to\infty}\left\{\lim_{i\to\infty}\sum_{j} \frac{|(D\hat\beta^{(i+1)})_j|}{|(D\hat\beta^{(i)})_j| + \varepsilon_n}\right\}
=
\lim_{n\to\infty}\sum_{j:(D\beta)_j \neq 0}
\left( 1 - \frac{\varepsilon_n}{|(D\hat\beta^{(i)}_j| + \varepsilon_n} \right)
=
\|D\beta\|_0
\] because at stationarity $\hat\beta^{(i+1)}=\hat\beta^{(i)}$.
Thus, as $n$ increases and $\varepsilon_n$ shrinks, the adaptive penalty becomes effectively binary, and the IRFL penalty behaves like a binary (0/1) on inactive / active differences, respectively, and thus approximates an $\ell_0$-type penalty on $D\beta$. In particular, letting
\[
\hat\beta_{\lambda}^{(\infty)}(\varepsilon_n)
:= \lim_{i\to\infty} \hat\beta_{\lambda}^{(i)}(\varepsilon_n),
\]
we have the approximation
\begin{equation}
    \label{0-1 convergence}
    \lim_{n\to\infty}\hat\beta_{\lambda}^{(\infty)}(\varepsilon_n)
\;\approx\;
\underset{\beta \in \mathbb{R}^n}{\arg\min}
\left\{
\frac{1}{2}\|y - X\beta\|_2^2
+ \lambda \|D\beta\|_0
\right\}.
\end{equation}

\noindent Thus, IRFL can be interpreted not merely as a reweighted lasso, but as an iterative surrogate for $\ell_0$-regularization in changepoint selection, achieving discrete model selection behavior through continuous reweighting.

\begin{remark}\label{rem:IRFL_reweighted_l1}
Iterative reweighting procedures have appeared previously in the variable–selection
literature, most notably in the reweighted $\ell_1$ scheme of \citet{candes-2008-reweighted}. This method updates
weights of the form $w_j = (|\beta_j| + \varepsilon)^{-1}$, where the parameter
$\varepsilon$ is chosen intentionally not to be too small, thereby preventing
coefficients shrunk to zero from becoming permanently inactive in subsequent iterations (as their corresponding weights do not ``blow up" to infinity). In that context,
$\varepsilon$ plays a conceptual role in ensuring that small coefficients which were killed off in a given iteration $i$ may be
``resurrected'' on iteration $i+1$, for example, and the resulting estimator behaves as a smooth approximation to an $\ell_0$ penalty.

The present approach differs in two key respects. First, IRFL applies the
reweighting principle not to regression coefficients, but to the structured differences $(D\beta)_j$, so that sparsity is imposed directly on
level shifts (or higher order differences) and therefore on the changepoint structure of the signal.
Second, the parameter $\varepsilon$ used in IRFL serves purely as a numerical
stabilizer to ensure a well–defined reweighting path and to avoid division by
zero; it does \emph{not} intentionally fulfill the resurrection role it occupies in \citet{candes-2008-reweighted}. In this sense, IRFL is
philosophically related to earlier iterative reweighting methods, but represents a
distinct extension of this idea to structured-difference penalties and the changepoint
setting, with a different interpretation of both the weights and the stabilizing
parameter~$\varepsilon$.
\end{remark}

\subsection{A Distance Metric}

Shortly, a series of simulation results comparing IRFL to its competitors will be given. To evaluate how closely the estimated changepoints of IRFL recover the true configuration, we employ a distance metric based on linear assignment \citep{kuhn-1955-hungarian}.  
This metric penalizes differences in both the number and the locations of detected changepoints, assigning small penalties for positional mismatches and unit penalties for missing or extra changepoints.

\subsubsection*{Definition}

Let $\bm{\tau} = \{\tau_1, \ldots, \tau_m\}$ denote the true changepoints and 
$\hat{\bm{\tau}} = \{\hat{\tau}_1, \ldots, \hat{\tau}_{\hat{m}}\}$ the estimated changepoints, both subsets of $\{1,\ldots,n\}$.  
A \emph{mismatch} occurs when an estimated changepoint $\hat{\tau}_i$ does not exactly coincide with any element of $\bm{\tau}$.  
A \emph{missing} changepoint is an element of $\bm{\tau}$ with no corresponding estimate,  
and an \emph{extra} changepoint is one that appears in $\hat{\bm{\tau}}$ but not in $\bm{\tau}$.

Formally, the total distance between $\hat{\bm{\tau}}$ and $\bm{\tau}$ is defined as
\begin{equation}
\label{eq:distance}
d(\hat{\bm{\tau}}, \bm{\tau})
=
\min_{\pi \in \Pi}
\left\{
   \sum_{i=1}^{k} \frac{1}{n}\, \big| \hat{\tau}_i - \tau_{\pi(i)} \big|
   + (m - k) + (\hat{m} - k)
\right\},
\end{equation}
where $\Pi$ represents all one-to-one matchings between subsets of $\hat{\bm{\tau}}$ and $\bm{\tau}$.  
Intuitively, $\Pi$ enumerates every possible way to pair up estimated and true changepoints,  
and the distance $d(\hat{\bm{\tau}}, \bm{\tau})$ corresponds to the \emph{best possible fit}—that is, the minimal total penalty over all such pairings.  
The first term measures how far matched changepoints are apart, while the second and third terms penalize missing and extra changepoints, respectively.  
This minimal cost is computed efficiently using the Hungarian algorithm \citep{kuhn-1955-hungarian}.

To make this definition concrete, the following examples illustrate how penalties accrue in typical cases.

\subsubsection*{Examples}

\paragraph{Example 1 (Undercount).}
Suppose the true changepoints are $\bm{\tau} = \{251, 501, 751\}$ and the estimated changepoints are $\hat{\bm{\tau}} = \{248, 745\}$ with $n=1000$.  
The best matching pairs $(248,251)$ and $(745,751)$ yield
\[
 \frac{1}{1000}\big(|248-251| + |745-751|\big) = 0.009.
\]
Because one true changepoint ($501$) is missing, a penalty of $1$ is added:
\[
d(\hat{\bm{\tau}}, \bm{\tau}) = 1.009.
\]

\paragraph{Example 2 (Overcount).}
Suppose the true changepoints are $\bm{\tau} = \{251, 501, 751\}$ and the estimated changepoints are $\hat{\bm{\tau}} = \{249, 498, 760, 905\}$ with $n=1000$.  
The best matching pairs $(249,251)$, $(498,501)$, and $(760,751)$ yield
\[
\frac{1}{1000}\big(|249-251| + |498-501| + |760-751|\big) = 0.013.
\]
Since one estimated changepoint ($905$) is extra, an additional penalty of $1$ is added:
\[
d(\hat{\bm{\tau}}, \bm{\tau}) = 1.013.
\]

\subsubsection*{Interpretation}

This distance combines location accuracy and changepoint count into a single measure.  
A perfect match gives $d(\hat{\bm{\tau}}, \bm{\tau}) = 0$, while small positional errors yield small fractional penalties, and missing or extra changepoints incur full unit penalties.  
In the simulation results that follow, this metric serves as a key measure of recovery accuracy.

\subsection{Empirical Results}  \label{empirical results}
To analyze the performance of the changepoint detection techniques in Section 2 compared with IRFL under various signal structures, we simulated data under fourteen scenarios. In the first ten cases, 1000 datasets were generated with either $n=1000$ or $n=1200$ observations (depending on the model), incorporating different changepoint configurations and global parameters. In all cases, the design and structured difference matrices used for estimation of the fused lasso and subsequent reweighted models required for the \texttt{genlasso} package have been included for reference.  Competing methods (FL, Ada FL, IRFL, BS, WBS, WCM, NOT, CPOP, and AR1Seg) are evaluated on their ability to avoid false detection under global dependence and to recover and localize true changepoints accurately, using both changepoint counts and distance to the true changepoint vector as evaluative metrics.

The results are divided into four groupings. Group 1 (scenarios 1-3) contains scenarios in which the goal is to detect mean shifts when there are no global (possibly confounding) variables present. Group 2 (scenarios 4-7) contain scenarios in which the goal is to detect mean shifts when there are global parameters present (trend and autoregressive structure). Group 3 (scenarios 8-10) contain scenarios in which the goal is to detect changes in slope, rather than in mean, and Group 4 (scenarios 11-14) contain scenarios in which the goal is to estimate mean shifts in the presence of seasonality and a consistent trend.

In each scenario, results are presented through a common set of plots. First, a representative simulation of the underlying data-generating process is displayed to illustrate the structure of the signal. In all cases, the subsequent boxplots show the sampling distribution of the number of changepoints estimated by each method relative to the true value (denoted $\hat m$ and $m$, respectively). For models containing one or more true changepoints, an additional boxplot of the distance metric (defined above in Section 4.4) is reported, which quantifies the discrepancy between the estimated changepoint vector $\hat{\bm\tau}$ and the true changepoint vector $\bm\tau$. Finally, boxplots are included displaying the sampling distributions of any global parameters or other quantities of interest by IRFL and any other methods appropriate for estimating the model in question (e.g., BIC, slope or AR(1) parameters, seasonality coefficients, etc).  

When interpreting the boxplots, readers should focus on both the spread and the central tendency of each distribution. Narrower boxes and medians close to the true parameter values indicate methods that are both stable and accurate. Across nearly all scenarios, IRFL exhibits these qualities to an exceptional degree, showing markedly tighter distributions and substantially lower error relative to competing approaches. The plots therefore demonstrate that IRFL not only avoids systematic bias but also achieves consistently stronger recovery of changepoint locations and parameters.

To clarify the range of simulation settings under which these results were obtained, Table~\ref{tab:scenario-summary} summarizes all fourteen experimental scenarios, organized by structural features and changepoint type. These results together highlight IRFL’s superiority as a practical and robust method.

\begin{table}[H]
\centering
\footnotesize
\setlength{\tabcolsep}{4pt}
\renewcommand{\arraystretch}{1.15}

\arrayrulecolor{gray!35}
\begin{tabularx}{\textwidth}{@{}c X c c c c@{}}
\toprule
\textbf{Scenario} & \textbf{Description} & \textbf{Seas.} & \textbf{Trend} & \textbf{AR(1)} & \textbf{CP Type} \\
\midrule[1.1pt]

\rowcolor{gray!10}
\multicolumn{6}{@{}l}{\textbf{Group 1: Mean shifts, no global params}}\\
1 & Mean-shift; no changepoints & -- & -- & -- & None \\
\cmidrule(lr){1-6}
2 & Three Mean Shifts  & -- & -- & -- & Mean \\
\cmidrule(lr){1-6}
3 & Mean shifts with uneven segment lengths & -- & -- & -- & Mean \\
\midrule[1.1pt]

\rowcolor{gray!10}
\multicolumn{6}{@{}l}{\textbf{Group 2: Mean shifts with global params}}\\
4 & AR(1); no changepoints & -- & -- & \y & None \\
\cmidrule(lr){1-6}
5 & AR(1) + three mean shifts  & -- & -- & \y & Mean \\
\cmidrule(lr){1-6}
6 & Constant linear trend; no changepoints & -- & \y & -- & None \\
\cmidrule(lr){1-6}
7 & Constant linear trend + three mean shifts  & -- & \y & -- & Mean \\
\midrule[1.1pt]

\rowcolor{gray!10}
\multicolumn{6}{@{}l}{\textbf{Group 3: Trend shifts}}\\
8 & Continuous trend model; no slope changes & -- & -- & -- & Trend \\
\cmidrule(lr){1-6}
9 & Three trend shifts  & -- & -- & -- & Trend \\
\cmidrule(lr){1-6}
10 & Trend Shifts with random segment lengths & -- & -- & -- & Trend \\
\cmidrule(lr){1-6}

\rowcolor{gray!10}
\multicolumn{6}{@{}l}{\textbf{Group 4: Mean shifts with seasonality and trend}}\\
11 & Seasonal pattern only; no changepoints & \y & -- & -- & None \\
\cmidrule(lr){1-6}
12 & Seasonality + three mean shifts & \y & -- & -- & Mean \\
\cmidrule(lr){1-6}
13 & Seasonality + constant trend; no changepoints & \y & \y & -- & None \\
\cmidrule(lr){1-6}
14 & Seasonality + trend + three mean shifts & \y & \y & -- & Mean \\
\bottomrule
\end{tabularx}

\arrayrulecolor{black}
\caption{Table of Simulation results by structural features and changepoint type.}
\label{tab:scenario-summary}
\end{table}
\begin{remark}\label{rem:randomization-scope}
In principle, the most comprehensive evaluation would randomize the number, locations, and magnitudes of changepoints in every scenario. In this section, full randomization was implemented only in the \emph{mean–shift + constant–trend} settings (with the randomized extension following Scenario~7) and in the \emph{trend–shift} settings (Scenario 10). This choice reflects the judgment that these cases adequately stress variability in segment length, jump size, and changepoint type. The fixed designs used elsewhere are not expected to bias the comparisons when reasonable assumptions regarding signal-to-noise ratio are met.
\end{remark}

\subsection{Group 1: Mean Shifts without Global Parameters}
\noindent This group provides a baseline for performance under the simplest conditions: pure mean-shift models with independent, identically distributed (IID) noise and no global structure such as trend or seasonality. Each model is formulated within the generalized lasso framework, using a first-difference penalty to enforce sparsity in level changes. The experiments examine two key behaviors: whether methods detect changepoints when none are present, and how accurately they recover changepoints when true shifts exist, including cases with uneven segment lengths. The methods compared are the fused lasso, adaptive fused lasso (adafused lasso), IRFL, binary segmentation (BS), wild binary segmentation (WBS), and PELT. Performance is summarized by the number of detected changepoints and by a distance metric quantifying the discrepancy between estimated and true changepoint locations.
\subsubsection*{Scenario 1: Mean Shift with No Changepoints}In scenario 1, we have generated data from a null model with no underlying changepoints and no trend or other global parameter. Here, the true mean remains constant across all time points. Any detected changepoints are false.

The model is
\begin{equation}\label{scenario1 equation}
    y_t = \mu_t + \epsilon_t,\qquad \epsilon_t \overset{IID}{\sim} WN(0,\sigma^2).
\end{equation} The design and structured difference matrices and parameter vector are

\begin{equation*}
\resizebox{\textwidth}{!}{$
X = 
\begin{bmatrix}
1 & 0 & 0 & \dots & 0 \\
0 & 1 & 0 & \dots & 0 \\
0 & 0 & 1 & \dots & 0 \\
\vdots & \vdots & \vdots & \ddots & \vdots \\
0 & 0 & 0 & \dots & 1
\end{bmatrix}_{n \times n}, \quad 
{\beta} = \begin{bmatrix}
    \mu_1\\
    \mu_2\\
    \vdots\\
    \mu_n
\end{bmatrix}_{n \times 1},\quad
D = 
\begin{bmatrix}
-1 & 1  & 0  & 0  & \dots & 0 & 0 \\
0  & -1 & 1  & 0  & \dots & 0 & 0 \\
0  & 0  & -1 & 1  & \dots & 0 & 0 \\
\vdots & \vdots & \vdots & \vdots & \ddots & \vdots & \vdots \\
0  & 0  & 0  & 0  & \dots & -1 & 1
\end{bmatrix}_{(n-1) \times n}.
$}
\end{equation*}

\bigskip
\noindent The results of the simulations in the bottom box of Figure \ref{fig:MSFL-Null-Group} reveal that IRFL consistently finds no changepoints, and  has superior false detection rate compared with WBS and PELT.
\begin{figure}[H]
    \centering
    \textbf{White Noise Simulation with No Changepoints}\\ 
    \vspace{0.5em}
    \includegraphics[width=1\textwidth]{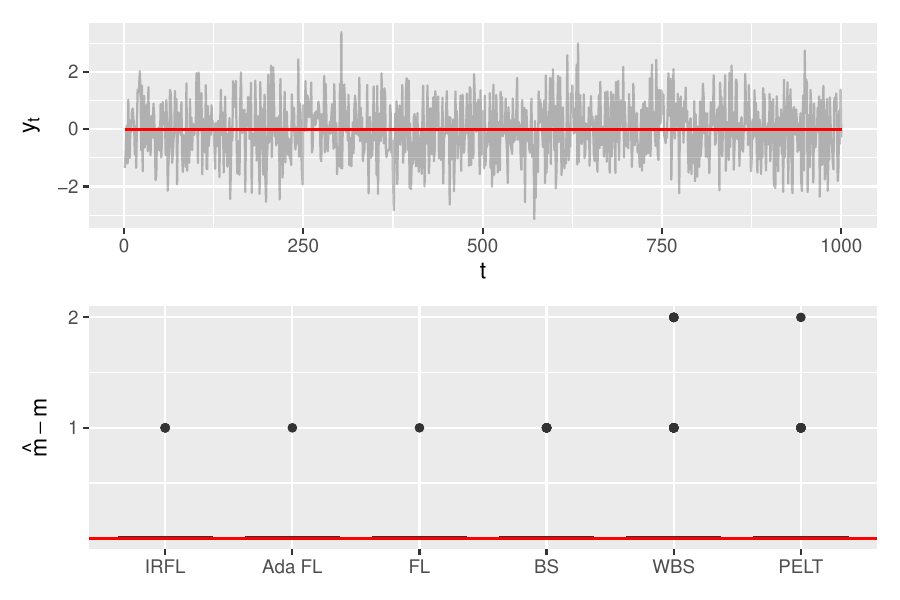}
    \caption[Comparison of mean shift detection methods with no changepoints]{Top: Representative white noise model. Bottom: Estimated number of changepoints detected by various methods for a null model with no changepoints and IID $ WN(0, \sigma^2) $ errors.}

    \label{fig:MSFL-Null-Group}
\end{figure}

The bottom plot is actually a boxplot; the single dots represent outliers. IRFL, Ada FL, and FL all have superior performance to the other methods, as they never falsely detect more than one changepoint---and even this is only done very infrequently.

\subsubsection*{Scenario 2: Mean Shift with Three Changepoints}
The second scenario has three changepoints at 
$t=251$, $t=501$, and $t=751$, but with no overall trend or other global parameter. The true mean alternates between 0 and 2 at the changepoints, creating distinct segments within the data. The model is in \eqref{scenario1 equation} given for Scenario 1, and the design, structured difference matrix, and parameter vector are also the same as in Scenario 1. 

Figure \ref{fig:MSFL-Group} shows that IRFL compares favorably to WBS and PELT both in the number of detected mean shifts as well as the distance between the estimated changepoint vector $\hat\tau$ and the true changepoint vector $\tau=[251,501,751]^\top $.

\begin{figure}[H]
    \centering
    \textbf{White Noise with Three Changepoints}\\ 
    \vspace{0.5em}
    \includegraphics[width=1\textwidth]{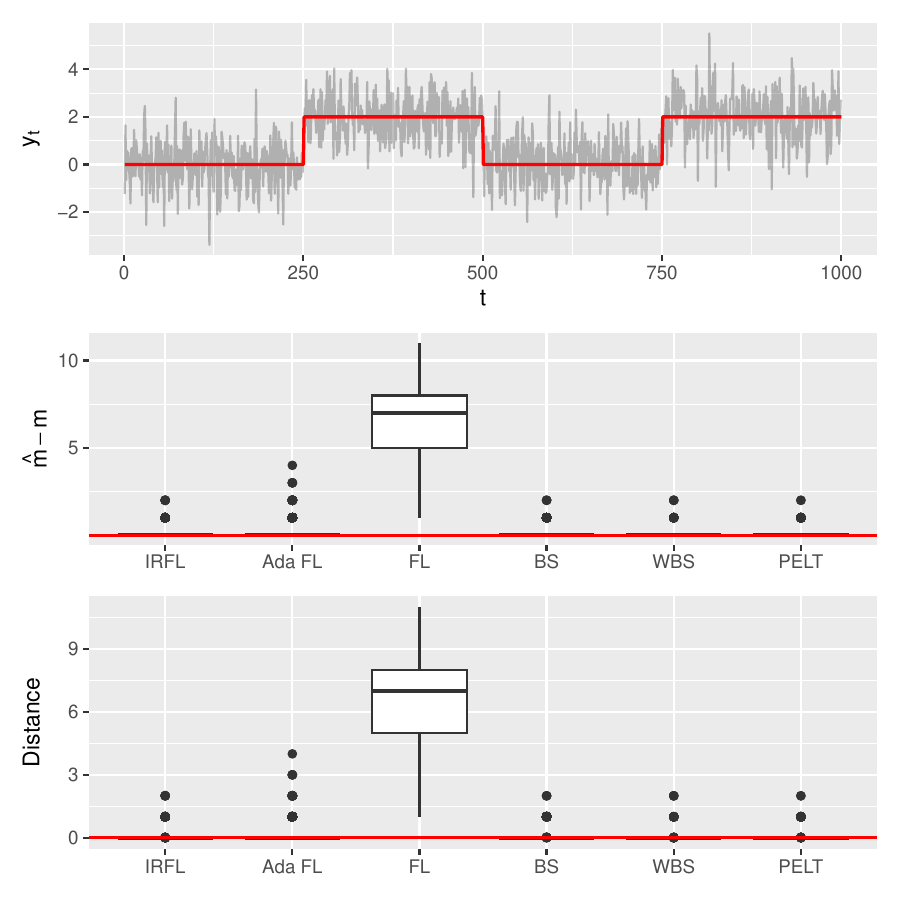}
    \caption[Comparison of mean shift detection methods with three changepoints]{Top: Representative mean shift model with three changepoints at $t = 251, 501, $ and $ 751 $.  Middle: Boxplot of sampling distribution of number of changepoints found using the competing methods. Bottom: Boxplot of sampling distribution of the distance between the estimated changepoint vector $\bm{\hat\tau}$ and the true vector $\bm\tau$ by method.}
    \label{fig:MSFL-Group}
\end{figure}

\subsubsection*{Scenario 3: Mean Shift with Uneven Segment Lengths}
In the third scenario, we place five changepoints so that the resulting segment lengths are unequal. We designed the setup deliberately to demonstrate that the strong performance of IRFL does not depend on segments being evenly spaced. The underlying model is the same as in Scenarios 1 and 2 given in \eqref{scenario1 equation}, and the corresponding design and structured difference matrices, as well as the parameter vector, remain unchanged.  

\noindent
The true signal consists of six contiguous segments with mean levels alternating between zero, positive, and negative shifts. Specifically, the mean is
\[
\mu =
\begin{cases}
0, & 1 \le i \le 150, \\
2, & 151 \le i \le 200, \\
0, & 201 \le i \le 500, \\
-2, & 501 \le i \le 750, \\
0, & 751 \le i \le 800, \\
2, & 801 \le i \le 1000
\end{cases}
\]
corresponding to six segments of lengths $(150, 50, 300, 250, 50, 200)$ with mean levels $(0, 2, 0, -2, 0, 2)$. This configuration more accurately represents many real-life applications of mean shift detection. By recovering the changepoints corresponding to these segments accurately, we show that IRFL adapts effectively to irregular segmentation structures while maintaining stability in estimation. 

The results in Figure \ref{fig:MS Uneven Lengths} reveal that IRFL remains competitive with BS, WBS, and PELT, exhibiting both accurate detection of the true number of changepoints and lower distance to the true changepoint vector $\bm\tau$, even when segment lengths vary considerably. 

\begin{figure}[H]
    \centering
    \textbf{White Noise with Five Uneven Changepoints}\\ 
    \vspace{0.5em}
    \includegraphics[width=1\textwidth]{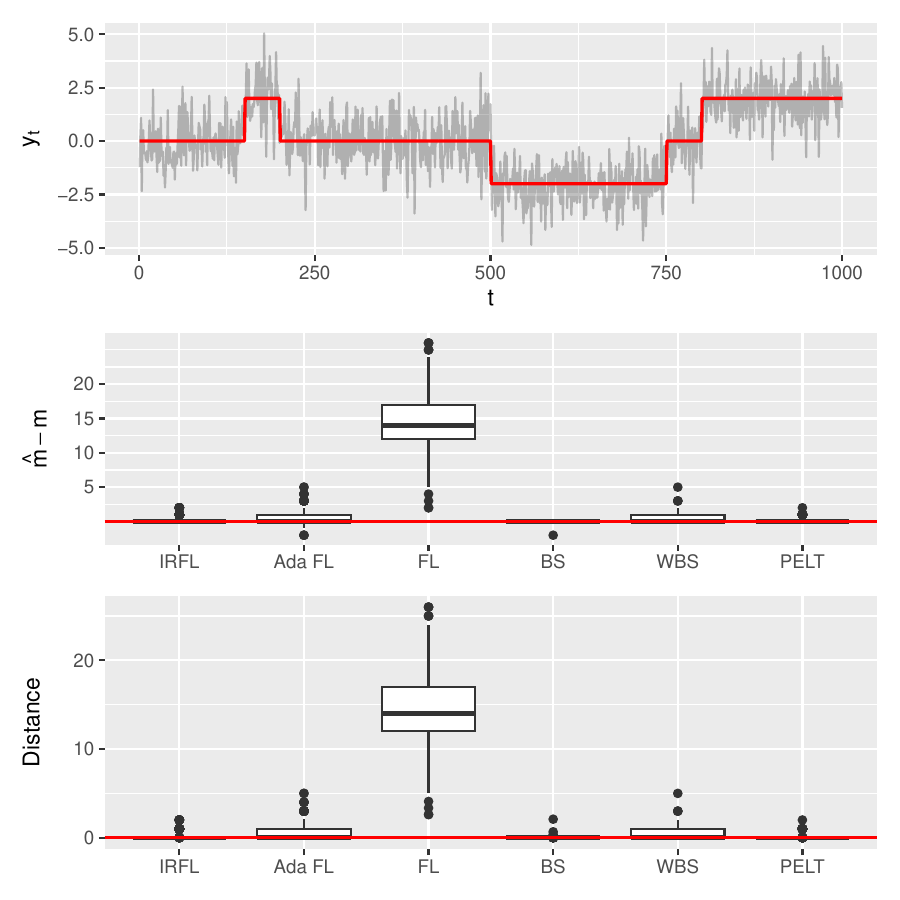}
    \caption[Comparison of mean shift detection methods with five changepoints at uneven intervals]{Top: Representative data set with five uneven changepoints with added white noise. Middle: Boxplot of sampling distribution of number of changepoints found using the competing methods. Bottom: Boxplot of sampling distribution of the distance between the estimated changepoint vector $\bm{\hat\tau}$ and the true vector $\bm\tau$ by method.}
    \label{fig:MS Uneven Lengths}
\end{figure}

\subsubsection*{Summary}
At the conclusion of Group 1, the results demonstrate that IRFL exhibits consistently strong performance in mean-shift settings without global parameters. In the null model with no changepoints, IRFL avoided detecting changepoints where none were present, outperforming both WBS and PELT in this regard. When true changepoints were present, IRFL not only detected the correct number of shifts with high accuracy but also achieved lower distance to the true changepoint vector than competing methods, reflecting greater precision in localization. Importantly, this advantage persisted even under unevenly spaced changepoints, indicating that the method’s stability is not dependent on regular segment lengths. Taken together, these findings establish IRFL as both parsimonious and reliable in detecting mean shifts, providing a robust foundation before extending analysis to more complex models with global parameters.

\subsection{Group 2: Mean Shifts with Global Parameters}
\noindent This group investigates the impact of global dependence structures—specifically, autoregressive correlation and linear trend—on changepoint detection. Such features commonly confound mean-shift procedures by mimicking local level changes. Autoregressive correlation introduces runs of similar residuals that can resemble stepwise jumps in the mean, while linear trend produces gradual level changes that can be falsely attributed to mean shifts.

Within the generalized lasso design, we extend the mean-shift framework by including a global trend regressor and modeling AR(1) dependence (that is, a model in which $y_t$ is dependent on $y_{t-1}$). Scenarios include (i) trend-only and AR(1) null models with no changepoints, and (ii) those same models with genuine mean shifts superimposed on these global structures. 

Throughout this work, the dependence structure of the noise is modeled using an AR(1) process rather than a more general ARMA$(p,q)$ or more complicated ARIMA, SARIMA, or even state-space specification. The AR(1) model serves as a parsimonious catch-all representation of short-range serial correlation while remaining tractable within changepoint frameworks, where higher-order dependence models are difficult to estimate reliably. Indeed, there are few published methods for estimating models more complicated than AR$(p)$.

\subsubsection*{Scenario 4: Mean Shift with AR(1) and No Changepoints}
In the fourth scenario, an AR(1) dependence structure is present across the entire series, with no changepoints. The model for such a scenario is 
\begin{equation}\label{ar1 model}
    y_t = \begin{cases}
        \mu_1 + \epsilon_t & \qquad t=1\\
        \mu_{t-1} + \phi y_{t-1}  + \epsilon_t & \qquad 2\leq t\leq n\\
    \end{cases}
    \end{equation}
where $\epsilon_t$ is a white-noise process with mean 0 and variance $\sigma^2$, and the accompanying design, parameter, and difference matrices are:\begin{equation*}
X  = 
\begin{bmatrix}
0   & 1 & 0 & 0 & \dots & 0 \\
y_1 & 1 & 0 & 0 & \dots & 0 \\
y_2 & 0 & 1 & 0 & \dots & 0 \\
y_3 & 0 & 0 & 1 & \ldots & 0\\
\vdots & \vdots & \vdots & \vdots & \ddots & \vdots \\
y_{n-1}& 0 & 0 & 0 & \dots & 1
\end{bmatrix}_{n \times n},\quad
{\beta}  = \begin{bmatrix}
    \phi\\
    \mu_1\\
    \mu_2\\
    \vdots\\
    \mu_{n-1}
\end{bmatrix}_{n \times 1}\end{equation*}
\begin{equation*}
D = 
\begin{bmatrix}
0 & -1 & 1  & 0  & 0  & \dots & 0 & 0 \\
0 & 0  & -1 & 1  & 0  & \dots & 0 & 0 \\
0 & 0  & 0  & -1 & 1  & \dots & 0 & 0 \\
\vdots & \vdots & \vdots & \vdots & \vdots & \ddots & \vdots & \vdots \\
0 & 0  & 0  & 0  & 0  & \dots & -1 & 1
\end{bmatrix}_{(n-1) \times n}.
\end{equation*}
To test the effectiveness of the IRFL procedure for this scenario, 1000 datasets were generated for each of $\phi=-0.7$, $-0.3$, $0.0$, $0.3$, and $0.7$ for a null model with zero changepoints. The effectiveness of the procedure was measured by BIC, number of (falsely) detected changepoints, and estimate of $\phi$, and comparisons are made to the other methods for estimating changepoints in this scenario (WCM and AR1Seg, see their respective subsections in Section 2). In Figure~\ref{fig:AR1 Group Null BIC}, boxplots of the sampling distributions of the BIC for the methods capable of estimating \eqref{ar1 model} demonstrate that IRFL targets the true model effectively, beating both WCM and AR1Seg, especially at the higher values of $\phi$. The vertical axis shows $\mathrm{BIC}(\widehat{\text{Model}}) - \mathrm{BIC}(\text{True Model})$, so values below zero indicate that the estimated model achieves a lower (and therefore better) BIC score than the true model, while positive values indicate worse fit relative to the truth.

\begin{figure}[H]
    \centering
        \textbf{BIC of AR(1) Model with No Changepoints}\\ 
    \vspace{0.5em}
    \includegraphics[width=.7\textwidth]{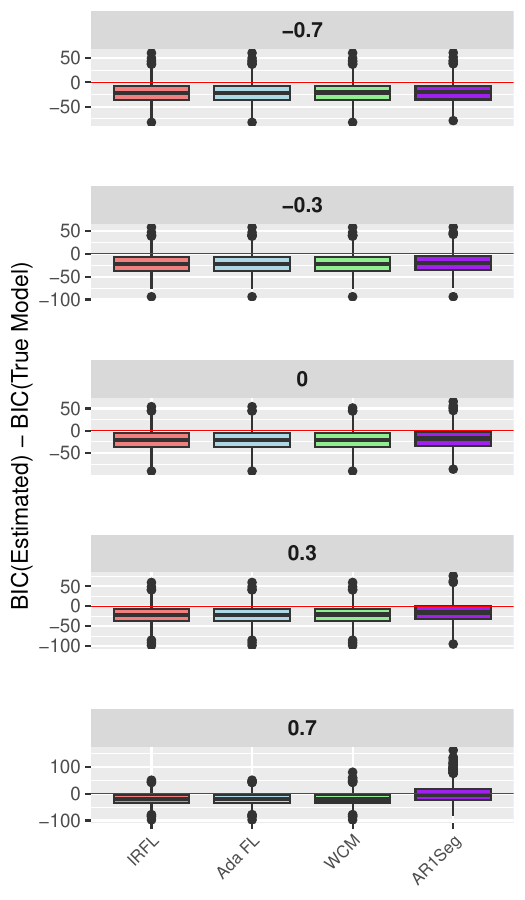}
        \caption[Comparison of BIC score for mean shift detection methods with no changepoints and various AR(1) parameter values]{Comparison of BIC score for mean shift detection methods with no changepoints and AR(1) parameter $\phi=-0.7$, $-0.3$, $0.0$, $0.3$, and $0.7$}
    \label{fig:AR1 Group Null BIC}
\end{figure}

In Figure \ref{fig:AR1 Group Null nknot}, the boxplots of the sampling distribution of the detected number of changepoints $\hat m$ demonstrate that IRFL has stronger false-detection properties than both WCM and AR1Seg, irrespective of the value of $\phi$. Note that while truly boxplots, because the overwhelming bulk of results are ``0'', all outliers appear as single dots.

\begin{figure}[H]
    \centering
    \textbf{Number of Changepoints Detected in AR(1) Model with No Changepoints}\\
    \vspace{0.5em}
    \includegraphics[width=.7\textwidth]{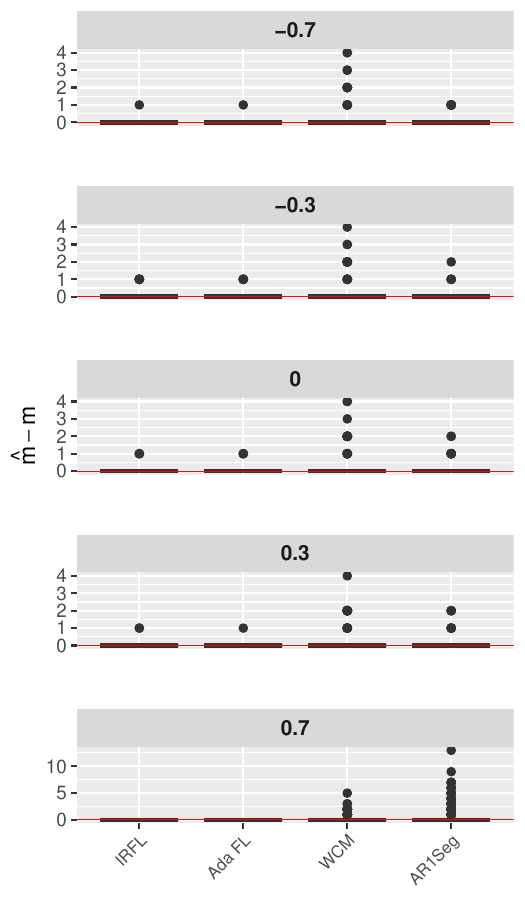}
        \caption[Comparison of number of changepoints found for mean shift detection methods with no changepoints and various AR(1) parameter values]{Comparison of number of changepoints found for mean shift detection methods with no changepoints and AR(1) parameter $\phi=-0.7$, $-0.3$, $0.0$, $0.3$, and $0.7$}
    \label{fig:AR1 Group Null nknot}
\end{figure}

It's also important to be able to effectively estimate the autocorrelation parameter $\phi$. In Figure \ref{fig:AR1 Group Null Phi}, boxplots displaying the sampling distribution of $\phi$ for all five of the $\phi$ values demonstrate that IRFL performs about the same as WCM, while having much lower deviation than AR1Seg. Formally, each boxplot represents the sampling distribution of 
$\hat{\phi} - \phi$, where 
$\hat{\phi}$ is the estimated AR(1) autocorrelation parameter 
and $\phi$ is its true value.

\begin{figure}[H]
    \centering
    \textbf{Estimation of AR(1) Parameter in Model with No Changepoints}\\
    \vspace{0.5em}
    \includegraphics[width=.7\textwidth]{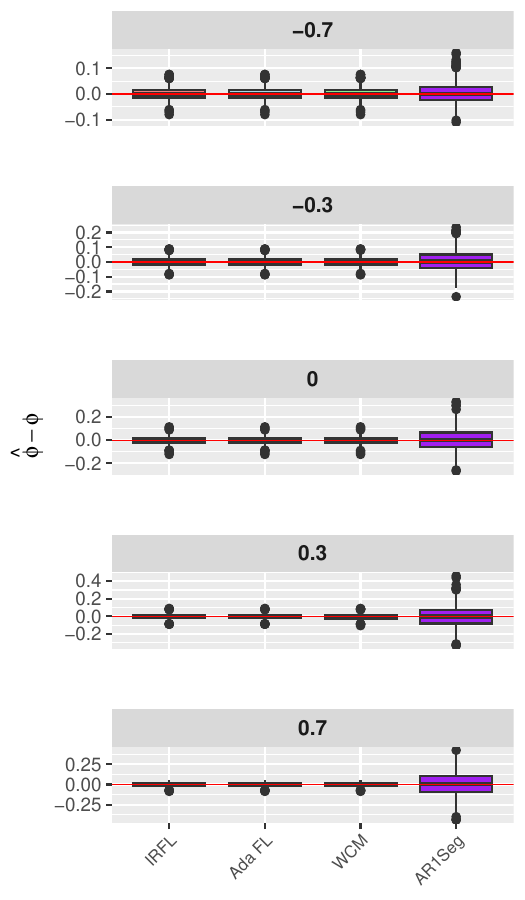}
        \caption[Comparison of $\phi$ estimate for mean shift detection methods with no changepoints and various AR(1) parameter values]{Comparison of $\phi$ estimate for mean shift detection methods with no changepoints and AR(1) parameter $\phi=-0.7$, $-0.3$, $0.0$, $0.3$, and $0.7$}
    \label{fig:AR1 Group Null Phi}
\end{figure}
From the plots, it is evident that adafused lasso and IRFL have superior false detection properties compared with WCM and AR1Seg for mean shifts, and they perform markedly better than AR1Seg in the estimation of $\phi$, for which an unbiased estimate is found.

\subsubsection*{Scenario 5: Mean Shift with AR(1) and Three Changepoints}

The fifth simulation considers an AR(1) model with changepoints at $t=251, 501,$ and $751$, where the mean alternates between $0$ and $2$. The model estimated is the same as in Scenario 4 \eqref{ar1 model}, and likewise the design and structured difference matrices, and parameter vector remain unchanged. 

Figure \ref{fig:AR1 Group BIC} highlights the primary distinction between IRFL and WCM. At $\phi = 0.7$, IRFL (and Ada FL) attain substantially lower BIC values than WCM and AR1Seg, indicating a more favorable balance between model fit and complexity. This outcome reflects IRFL’s selection principle: it minimizes BIC directly, even when doing so leads to a simpler segmentation that underestimates the true number of changepoints.

\begin{figure}[H]
    \centering
    \textbf{BIC of AR(1) Model with Three Changepoints}\\ 
    \vspace{0.5em}
    \includegraphics[width=.7\textwidth]{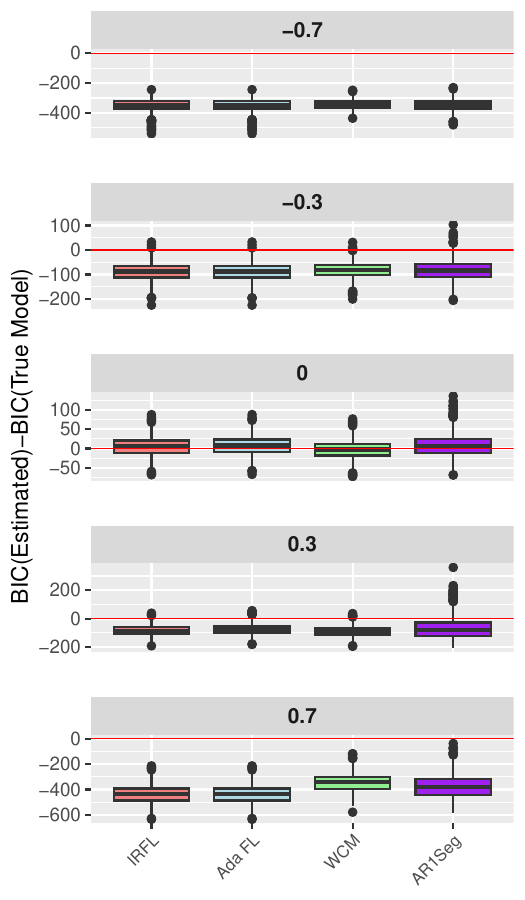}
    \caption[Comparison of BIC score for mean shift detection methods with three changepoints and various AR(1) parameter values]{Comparison of BIC scores for mean shift detection methods with three changepoints and AR(1) parameter $\phi = -0.7, -0.3, 0.0, 0.3,$ and $0.7$.}
    \label{fig:AR1 Group BIC}
\end{figure}

As shown in Figure~\ref{fig:AR1 Group nknot}, both IRFL and Ada FL collapse to zero changepoints when $\phi = 0.7$, whereas WCM continues to favor models that recover the correct number of changepoints. This contrast illustrates the difference between a criterion-driven approach that optimizes predictive performance (IRFL) and one designed for structural accuracy (WCM).

\begin{figure}[H]
    \centering
    \textbf{Number of Changepoints Detected in AR(1) Model with Three Changepoints}\\
    \vspace{0.5em}
    \includegraphics[width=.7\textwidth]{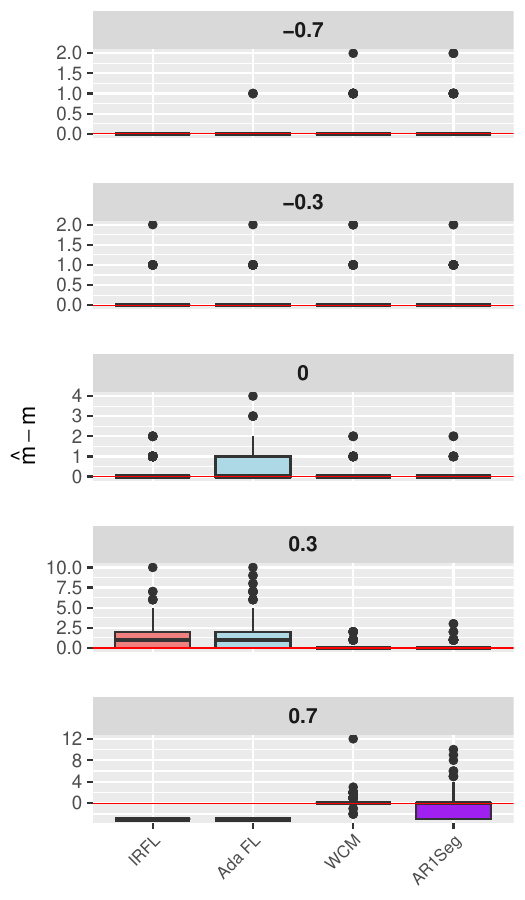}
    \caption[Comparison of number of changepoints found for mean shift detection methods with three changepoints and various AR(1) parameter values]{Comparison of number of changepoints found for mean shift detection methods with three changepoints and AR(1) parameter $\phi=-0.7$, $-0.3$, $0.0$, $0.3$, and $0.7$}
    \label{fig:AR1 Group nknot}
\end{figure}

The mechanism behind this choice is illustrated in Figure \ref{fig:AR1 Group Phi}. IRFL raises its estimate of $\phi$ above $0.7$, attributing more persistence to the series and thereby explaining the structure without changepoints. This substitution of autoregressive persistence for structural breaks avoids the need to estimate additional segment means and changepoint locations, yielding a lower BIC.

\begin{figure}[H]
    \centering
    \textbf{Estimation of AR(1) Parameter in Model with Three Changepoints}\\
    \vspace{0.5em}
    \includegraphics[width=.7\textwidth]{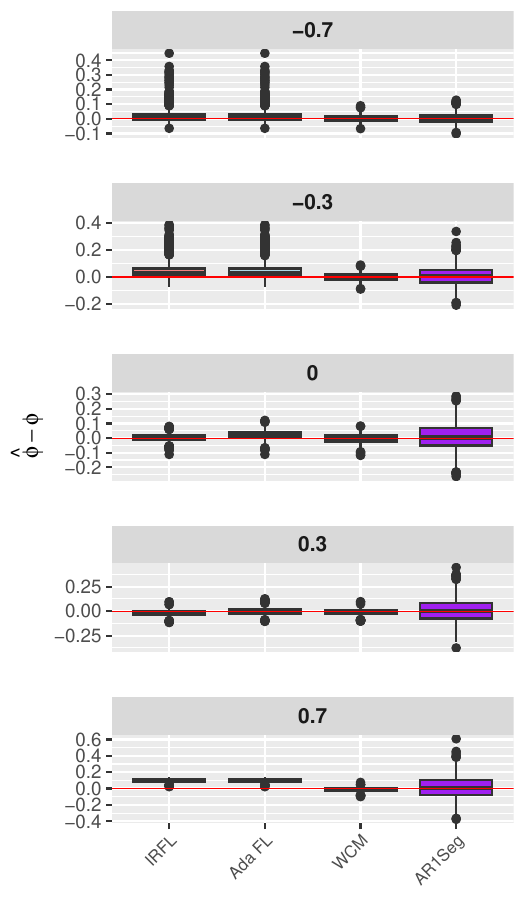}
    \caption[Comparison of $\phi$ estimate for mean shift detection methods with three changepoints and various AR(1) parameter values]{Comparison of $\phi$ estimate for mean shift detection methods with three changepoints and AR(1) parameter $\phi=-0.7$, $-0.3$, $0.0$, $0.3$, and $0.7$}
    \label{fig:AR1 Group Phi}
\end{figure}

To get an idea of how this can happen, Figure~\ref{fig:AR1 BIC vs AR1} presents a simple example in which data are generated from a stationary AR(1) process ($\phi=0.7$) with three equally spaced mean shifts, represented in blue. Despite the presence of these changepoints, an AR(1) mean-shift model attains a slightly higher BIC than the incorrectly specified AR(1) model with no changepoints, effectively level shifts for serial correlation.

\begin{figure}[H]
    \centering
    \textbf{Example: Misattribution of Mean Shifts as Autocorrelation}
    \begin{subfigure}{1\textwidth}
        \includegraphics[width=\linewidth]{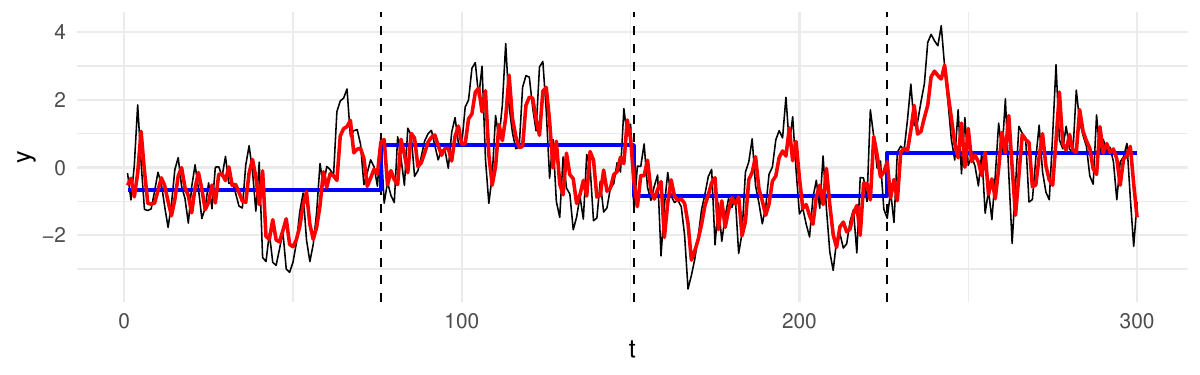}
        \caption*{(a) Three Mean-shift interpretation (BIC = 849.362, $\hat\phi=0.689$)}
    \end{subfigure}
    \vspace{0.75em}
    \begin{subfigure}{1\textwidth}
        \includegraphics[width=\linewidth]{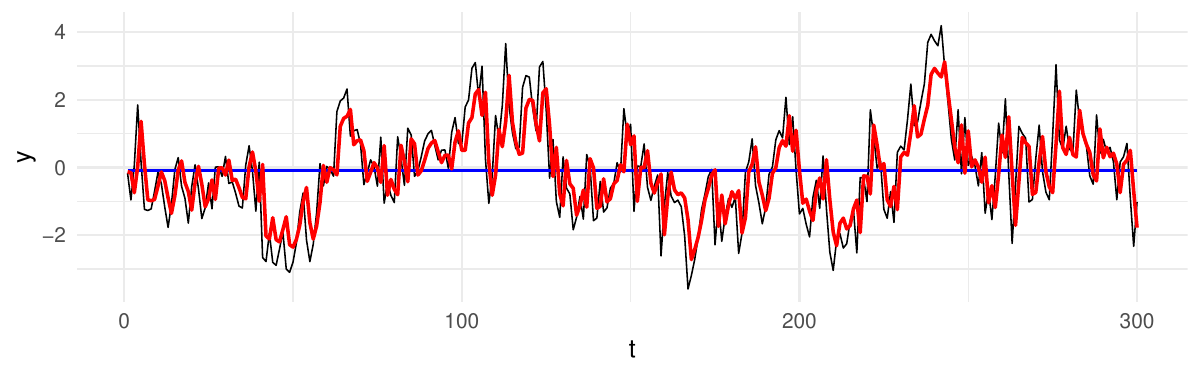}
        \caption*{(b) AR(1) interpretation (BIC=847.971, $\hat\phi=0.751$)}
    \end{subfigure}
    \caption[Missing Mean Shifts under AR(1) dependence]{
        Comparison of two interpretations of the same AR(1) data.
        The top panel shows the true mean in blue with AR(1) noise ($\phi^\text{true}=0.7$), with estimated AR(1) parameter $\hat\phi=0.689$. The bottom panel shows that despite not modeling true changepoints, the model achieves a slightly lower BIC by inflating the estimate of the AR(1) parameter ($\hat\phi=0.751$),
        illustrating how autoregressive persistence can mimic structural breaks.}
    \label{fig:AR1 BIC vs AR1}
\end{figure}

The consequence of choosing a higher degree of autocorrelation over the presence of a mean shift is evident in Figure \ref{fig:AR1 Group Distance}: IRFL minimizes BIC but fails to recover the true changepoint structure, whereas WCM---though less parsimonious---remains closer to the ground truth. Readers are encouraged to cross-reference this issue with the discussion in Section 1.5, the WCM description in Section 2, and the real-world application in Section 5.

\begin{figure}[H]
    \centering
    \textbf{Distance to True $\bm\tau$ in AR(1) Model with Three Changepoints}\\
    \vspace{0.5em}
    \includegraphics[width=.7\textwidth]{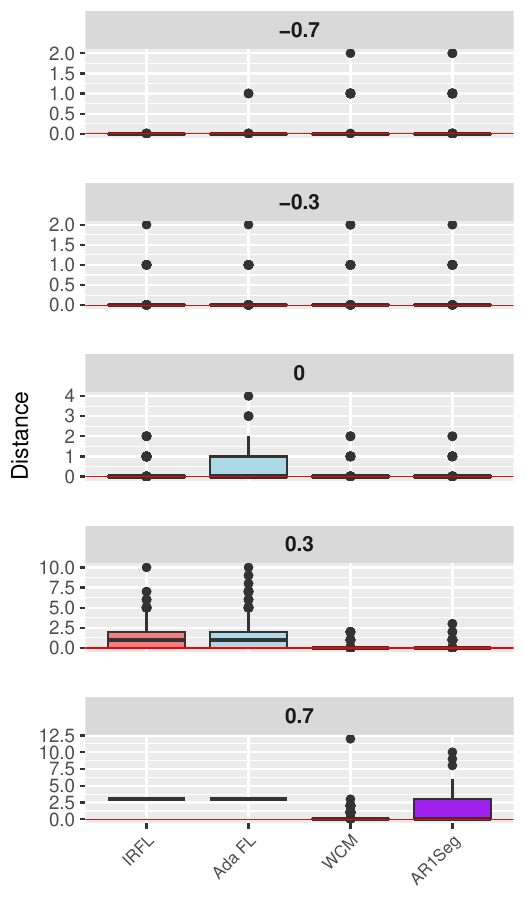}
    \caption[Comparison of distance to true changepoint locations found for mean shift detection methods with three changepoints and various AR(1) parameter values]{Comparison of distance to true changepoint locations found for mean shift detection methods with three changepoints and AR(1) parameter $\phi=-0.7$, $-0.3$, $0.0$, $0.3$, and $0.7$}
    \label{fig:AR1 Group Distance}
\end{figure}

In Figure \ref{fig:AR1 Group Distance}, IRFL and Ada FL do quite well in estimating the true vector of changepoints $\bm\tau$ when $\phi \in \{-0.7,-0.3,0.0\}$. However, when $\phi\in\{0.3,0.7\}$, they increase their estimate of $\phi$ (Figure \ref{fig:AR1 Group Phi}) and decrease their estimate of $m$ (Figure \ref{fig:AR1 Group nknot}), resulting in a lower BIC (Figure \ref{fig:AR1 Group BIC}), but at the cost of a greater distance between the estimated changepoint vector $\hat{\bm\tau}$ and the true vector $\bm\tau$.

In summary, while IRFL performs similarly to WCM and AR1Seg when $\phi$ is small, at higher persistence ($\phi=0.7$) it explicitly prioritizes model simplicity and BIC reduction, even at the expense of correctly identifying changepoints. WCM takes the opposite stance, emphasizing recovery of the true segmentation.

\subsubsection*{Scenario 6: Constant Trend with No Mean Shifts}
In the sixth scenario, a fixed linear trend is present across the entire series, with no changepoints. This configuration tests how each method performs in estimating the trend without detecting spurious changepoints. Here, the true mean is a linear function of time $t$, and accurate trend estimation is required without mistakenly identifying changepoints. The model for such a scenario is \begin{equation}\label{fixed trend model}
    y_t  = \begin{cases}
        \alpha + \mu_1 + \epsilon_t & \qquad t=1\\
        \alpha t + \mu_{t-1} + \epsilon_t & \qquad 2\leq t \leq n,
    \end{cases}
\end{equation}
where $\epsilon_t$ is a white noise process and the accompanying design, parameter, and difference matrices are as follows:
\begin{align*}
    X = 
\begin{bmatrix}
1 & 1 & 0 & 0 & \dots & 0 \\
2 & 1 & 0 & 0 & \dots & 0 \\
3 & 0 & 1 & 0 & \dots & 0 \\
4 & 0 & 0 & 1 & \dots & 0 \\
\vdots & \vdots & \vdots & \vdots & \ddots & \vdots \\
n & 0 & 0 & 0 & \dots & 1
\end{bmatrix}_{n \times n}, \quad
{\beta} = \begin{bmatrix}
    \alpha\\
    \mu_1\\
    \mu_2\\
    \vdots\\
    \mu_{n-1}
\end{bmatrix}_{n \times 1}\\
\\
D = 
\begin{bmatrix}
0 & -1 & 1  & 0  & 0  & \dots & 0 & 0 \\
0 & 0  & -1 & 1  & 0  & \dots & 0 & 0 \\
0 & 0  & 0  & -1 & 1  & \dots & 0 & 0 \\
\vdots & \vdots & \vdots & \vdots & \vdots & \ddots & \vdots & \vdots \\
0 & 0  & 0  & 0  & 0  & \dots & -1 & 1
\end{bmatrix}_{(n-1) \times n}
\end{align*}
This model cannot be estimated with BS, WBS, nor PELT. These methods are designed to find mean shifts without the presence of a trend, as the bottom panel of Figure \ref{fig:FTWI-Null-Group} indicates. The bottom panel of Figure \ref{fig:FTWI-Null-Group} also demonstrates that IRFL has an outstanding rate of false changepoint detection, while BS, WBS and PELT fail to estimate the model correctly, instead fitting a step function (not shown). The middle panel of  Figure \ref{fig:FTWI-Null-Group} demonstrates that IRFL is unbiased for $\alpha$.

\begin{figure}[H]
    \centering
    \textbf{Fixed Trend Simulation with No Changepoints}\\ 
    \vspace{0.5em}
    \includegraphics[width=1\textwidth]{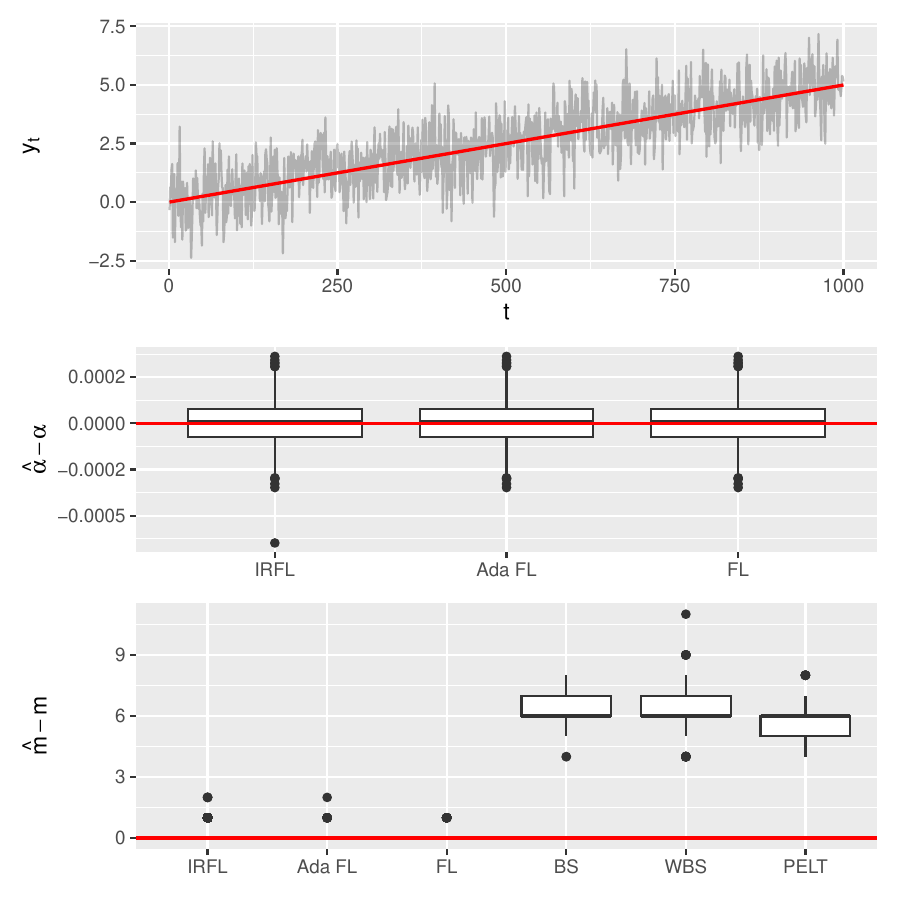}
    \caption[Comparison of mean shift detection methods with no changepoints and constant trend]{Top: Representative dataset of IID white noise model with fixed trend and no changepoints. Middle: Boxplot of sampling distribution of trend parameter $\alpha$. There is no entry for BS, WBS, or PELT because they do not have the ability to estimate trend.
    Bottom: Boxplot of sampling distribution of number of changepoints found using the competing methods.}
    \label{fig:FTWI-Null-Group}
\end{figure}

\subsubsection*{Scenario 7: Constant Trend with Three Mean Shifts}
The seventh scenario combines a fixed trend with three changepoints at the same locations as in scenarios two and five ($t=251$, $t=501$, and $t=751$). This setup introduces complexity by requiring simultaneous detection of changepoints and accurate trend estimation. The model being estimated is the same as in Scenario 6 \eqref{fixed trend model}, as are the design and structured difference matrices as well as parameter vector.

The middle panel of Figure \ref{fig:FTWI-Group}, shows that IRFL is an unbiased estimator of the trend parameter $\alpha$. It also does a superb job estimating the true number of changepoints and the proximity of the estimated changepoint vector $\hat{\bm\tau}$ to the true vector $\bm\tau$, as indicated by the boxplots in the second from bottom and bottom panels in Figure \ref{fig:FTWI-Group}, respectively.

\begin{figure}[H]
    \centering
    \textbf{Fixed Trend Simulation with Three Changepoints}\\ 
    \vspace{0.5em}
    \includegraphics[width=.8\textwidth]{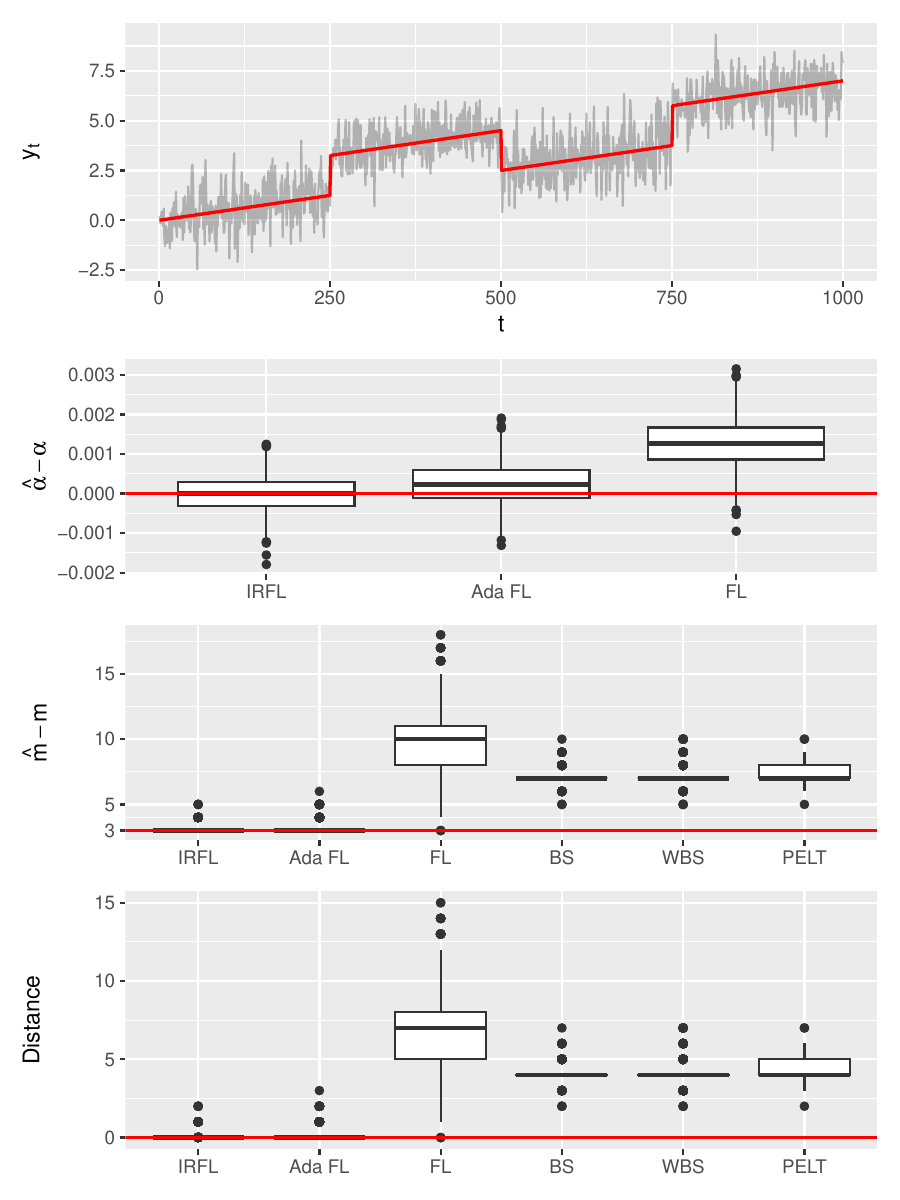}
    \caption[Comparison of mean shift detection methods with three changepoints and constant trend]{Top: Representative simulation of a fixed trend model with three changepoints at $t = 251$, $501, $ and $ 751 $. Second from top: Boxplot of centered sampling distribution of trend estimate $\hat\alpha$. Third from top: Boxplot of centered sampling distribution of the number of estimated changepoints, $\hat m$. Bottom: Boxplot of sampling distribution of the distance from the estimated changepoint vector $\hat{\bm\tau}$ to the true vector $\bm\tau$.}
    \label{fig:FTWI-Group}
\end{figure}

The fixed–trend scenarios above highlight two important points: first, that IRFL does not spuriously detect changepoints when none are present (Scenario 6), and second, that it can recover both the trend and changepoint structure when these coexist (Scenario 7). However, both of these settings impose a regular, deterministic structure on the changepoint vector. To fully evaluate the stability of IRFL, it is necessary to consider models in which the number, locations, and magnitudes of changepoints are not fixed in advance but vary randomly across datasets. 

This randomization serves two purposes. It avoids overfitting conclusions to a single carefully chosen segmentation pattern, and it provides evidence that IRFL continues to perform reliably under more realistic and irregular conditions. In practice, changepoints rarely occur at evenly spaced intervals with uniform jump sizes; instead, their frequency, placement, and magnitude are stochastic. By simulating under such random configurations, we can assess whether IRFL still recovers both the global trend and the local mean shifts effectively, and how its performance compares to other fused-lasso methods under these less structured circumstances.

\subsubsubsection*{Exploring Random Changepoint Locations}
To evaluate performance under more irregular conditions, 1000 datasets were generated with a random number and size of intercept shifts, combined with varying values of $\alpha$. The data-generating process proceeded as follows. First, the number of changepoints $m$ was drawn uniformly from $\{1,\ldots,7\}$. Next, an ordered vector of changepoints $\bm{\tau}=\{\tau_i\}_{i=1}^m$ was selected at random, subject to the spacing constraint
\begin{equation*}
    \min_{1\leq i\leq m+1}(\tau_i-\tau_{i-1}) \geq 100,
\end{equation*}
ensuring that adjacent changepoints were at least 100 indices apart. For each changepoint, a jump size was drawn independently from $\mathcal{U}(1.5,2.5)$, and its sign ($+$ or $-$) was assigned with equal probability, forming the mean vector $\mu$. The noise was taken to be white noise $WN(0,1)$, so these jump magnitudes are moderate relative to the noise level. This places the problem in a nontrivial signal-to-noise regime in which changepoints are detectable but not immediately obvious, allowing differences among estimation procedures to be meaningfully expressed. Next, one of four trend parameters $\alpha \in {-0.01, -0.005, 0.005, 0.1}$ was added to the mean trajectory, completing the data-generating process. An example simulation for the four trend values is shown in Figure \ref{fig:FTWI-random-results}.

\begin{figure}[H]
    \centering
    \textbf{Representative Simulation Dataset}\\ 
    \vspace{0.5em}
    \includegraphics[width=.8\textwidth]{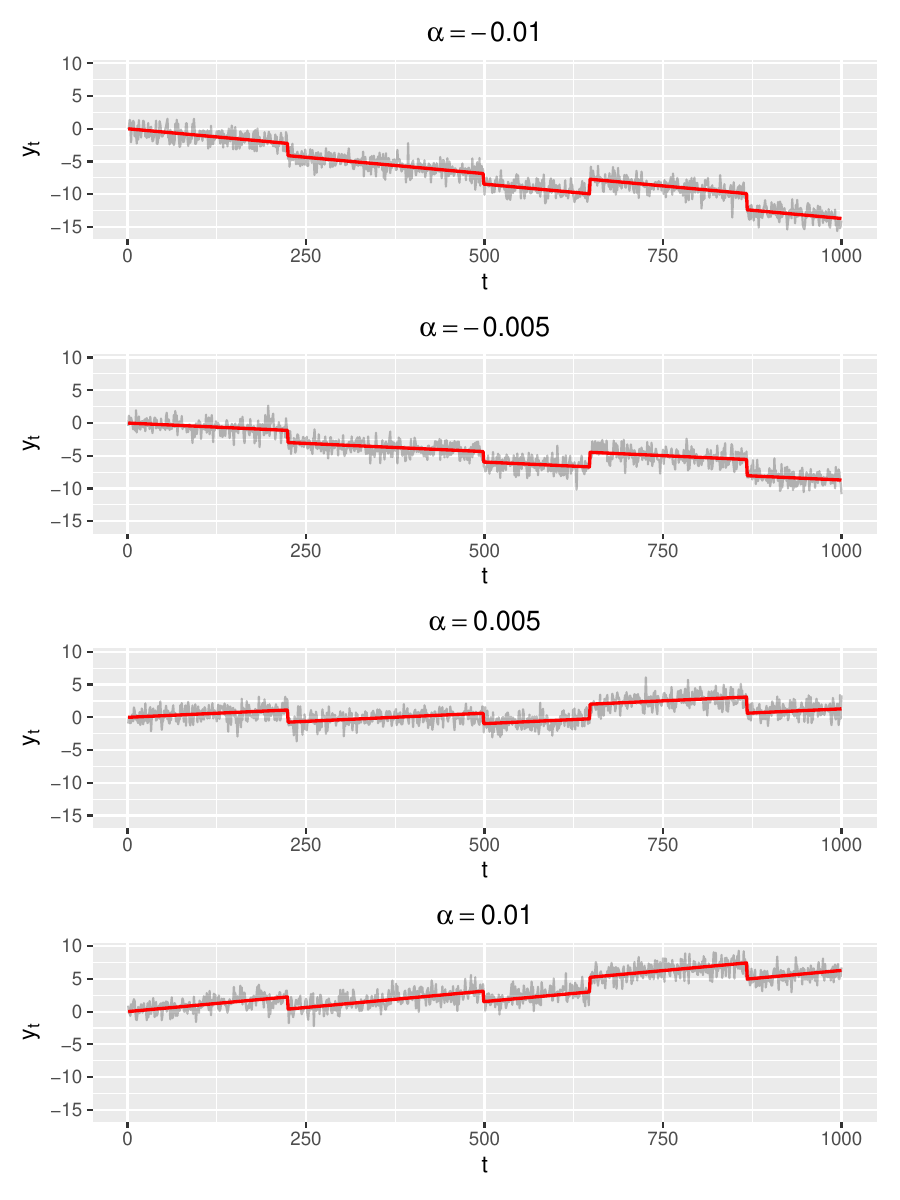}
    \caption[Simulation of fixed trend model with various trend parameter $\alpha$]{Simulation of a fixed trend model with random changepoints spaced at least 100 apart, jump sizes sampled from a $\mathcal{U}(1.5,2.5)$ distribution, and jump directions chosen randomly with equal probability. The model includes a constant trend and IID white noise errors. The panels illustrate the effect of different values of $\alpha$.}
    \label{fig:FTWI-random-results}
\end{figure}

Because only the fused lasso, adafused lasso, and IRFL can accommodate changepoint detection in the presence of a constant trend, these methods were used for comparison (none of the other methods explored in Section 2 have the ability to estimate mean shift models with a global trend). In addition, since the second-from-top panel in Figure \ref{fig:FTWI-Group} demonstrated bias in the direct estimation of the slope parameter $\alpha$, OLS estimates of $\alpha$ were obtained using the changepoints detected by adafused lasso and IRFL.

The top panel in Figure \ref{fig:FTWI-random-simulations} shows that the OLS estimates of $\alpha$ are nearly identical for the two methods---and both are unbiased as seen from the boxplot of the centered sampling distribution being contained within a small neighborhood of $0$. 

In the second panel of Figure \ref{fig:FTWI-random-simulations}, we see that adafused lasso frequently introduces spurious changepoints, as its iterative reweighting procedure has not yet eliminated small jumps after a single iteration, whereas IRFL achieves more parsimonious segmentation. This demonstrates that the lone reweighting step in adafused lasso is insufficient to isolate the true changepoint number, and some extra number of reweighting steps done in IRFL is necessary (\emph{iterate to isolate!}).

\begin{figure}[H]
    \centering
    \textbf{Results of Fixed Trend with Random Changepoints}\\ 
    \vspace{0.5em}
    \includegraphics[width=.8\textwidth]{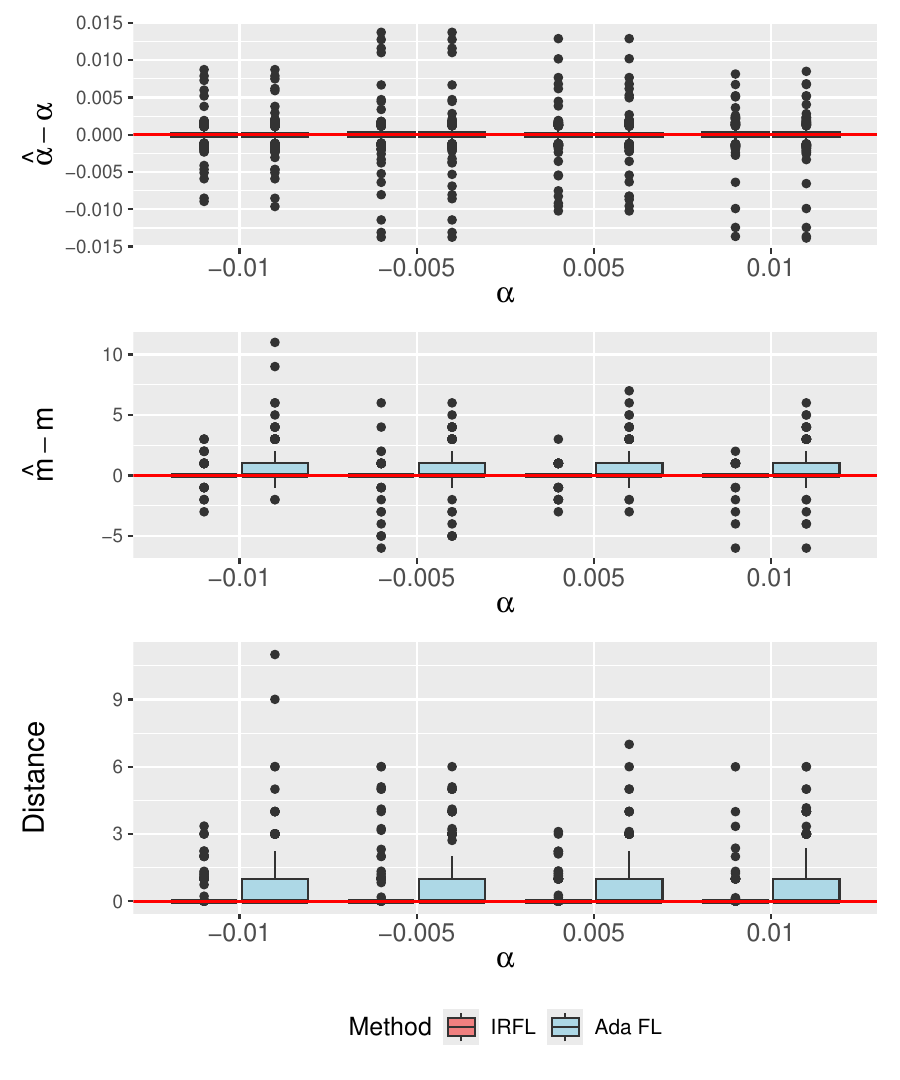}
    \caption[Comparison of random changepoint and fixed trend results for various slope $\alpha$]{Top: Boxplot of centered sampling distribution of trend estimate $\hat\alpha$. Middle: Boxplot of centered sampling distribution of the number of estimated changepoints, $\hat m$. Bottom: Boxplot of sampling distribution of the distance from the estimated changepoint vector $\hat{\bm\tau}$ to the true vector $\bm\tau$.}
    \label{fig:FTWI-random-simulations}
\end{figure}

Additional simulations below further clarify why this behavior arises. Because $\alpha$ is a global parameter, the inclusion of an extra (spurious) changepoint has little effect on its estimation. By contrast, omitting a true changepoint substantially biases the estimate of $\alpha$.

To illustrate the effect of changepoint specification on the estimation of global trend parameters, consider the model
\begin{equation*}
    y_t = 
    \begin{cases}
        \alpha t + \epsilon_t, & t \in \{1,\ldots,250\} \cup \{501,\ldots,750\},\\[4pt]
        \alpha t + 2 + \epsilon_t, & t \in \{251,\ldots,500\} \cup \{751,\ldots,1000\},
    \end{cases}
\end{equation*}
where $\epsilon_t \overset{\text{IID}}{\sim} N(0,1)$ and $\alpha = 0.005$. The true changepoint vector is $\tau = (251, 501, 751)$. For each of 1000 simulated datasets, $\alpha$ was first estimated via ordinary least squares (OLS) under the correct changepoint specification. Two perturbations were then applied. First, an additional changepoint was inserted at a uniformly random location, and $\alpha$ was re-estimated. Second, one of the true changepoints was removed, and $\alpha$ was again re-estimated.

This experiment highlights a general principle. Introducing an extra changepoint subdivides an existing segment, but the underlying mean function within the segment remains unchanged. As a result, the OLS estimate of $\alpha$ changes very little when a spurious changepoint is added. In contrast, removing a true changepoint forces two segments with different mean levels to be treated as a single segment. The resulting misspecification induces a systematic distortion in the fitted linear trend, producing a biased estimate of $\alpha$.

Figure~\ref{fig:OLS Fixed Trend Comparison} demonstrates this contrast. The distribution of estimated slopes remains concentrated when a spurious changepoint is added, whereas removing a true changepoint leads to substantial variability and bias. This explains why both the adaptive fused lasso and IRFL recover the global trend parameter with similar accuracy, even though IRFL provides more reliable recovery of the underlying changepoint configuration.

\begin{figure}[H]
    \centering
    \textbf{OLS Estimation of $\alpha$ with True and Faulty $\tau$}\\ 
    \vspace{0.5em}
    \includegraphics[width=1\textwidth]{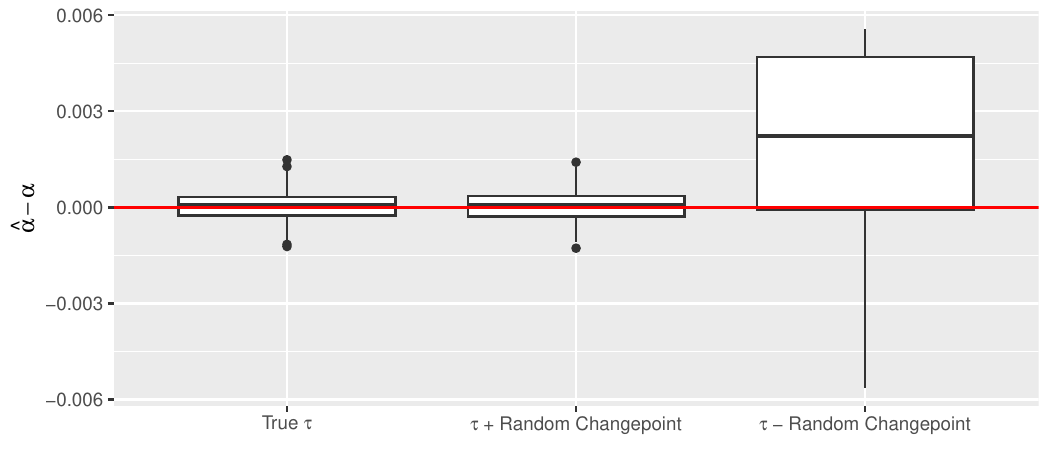}
    \caption[Comparison of ``too many'' versus ``too few'' changepoints for trend estimation]{A group box plot demonstrating the accuracy of the estimate of $\alpha$ is relatively unaffected by spurious changepoints, but is badly altered by missing changepoints.}
    \label{fig:OLS Fixed Trend Comparison}
\end{figure}
In plain words: If your goal is to estimate the slope $\hat\alpha$, it is far better to incorrectly include an extra spurious changepoint in your model than to incorrectly neglect a true changepoint.

\subsubsection*{Summary}
At the conclusion of Group 2, the findings show that IRFL continues to outperform competing methods when global linear trends are present. In the null case with no changepoints, IRFL again avoided detecting changepoints where none were present, performing substantially better than WBS and PELT in this regard. When true changepoints were introduced, IRFL accurately identified both their number and locations, demonstrating resilience to the confounding effect of an underlying trend. Notably, even when segment lengths were deliberately varied to create heterogeneous structures, IRFL retained its accuracy and stability, achieving closer alignment to the true changepoint vector than the alternative approaches. These results underscore that the strong performance observed in pure mean-shift models extends naturally to scenarios with linear trends, confirming IRFL’s robustness across both stationary and trend-influenced regimes.

\subsection{Group 3: Trend Shifts}
\noindent This group focuses on piecewise-linear signals with continuity at changepoints, where the underlying structure is defined by slope changes rather than level jumps. The generalized lasso is specified with a second-difference penalty to enforce long, continuous linear segments separated by sparse slope changes. Scenarios include a trend-only null case (no slope changes) as well as settings containing a small number of true slope shifts. Methods capable of estimating trend-shift models—fused lasso, adafused, IRFL, and CPOP—are compared on their ability to avoid detecting slope changes when none are present, on the number of slope changes they recover when they do occur, and on the accuracy with which estimated slope-change locations align to the true changepoint vector.

\subsubsection*{Scenario 8: Trend Shift Model with No Changepoints}

In the eighth scenario, a trend is present in the model. However, the trend can shift over time, while continuity between the piecewise linear segments is enforced. The model used for such a regression is therefore \begin{equation}\label{trend shift model}
    y_t = \mu_t + \epsilon_t,\qquad \epsilon_t \overset{IID}{\sim} WN(0,\sigma^2).
\end{equation} 
This is perhaps unexpected as it resembles the mean shift model in \eqref{scenario1 equation}. However, the structured difference matrix enforces the piecewise linear structure. Accordingly, the accompanying design, parameter, and difference matrices are as follows: \begin{equation*}
X  = 
\begin{bmatrix}
1  & 0 & 0 & \dots & 0 \\
0 & 1 & 0 & \dots & 0 \\
0 & 0 & 1 & \dots & 0 \\
\vdots & \vdots & \vdots & \ddots & \vdots \\
0 & 0 & 0 & \dots & 1
\end{bmatrix}_{n \times n},\quad
{\beta}  = \begin{bmatrix}
    \mu_1\\
    \mu_2\\
    \mu_3\\
    \vdots\\
    \mu_{n}
\end{bmatrix}_{n \times 1}\end{equation*}
\begin{equation*}
D = 
\begin{bmatrix}
1 & -2 & 1  & 0  & 0  & \dots &0& 0 & 0 \\
0 & 1  & -2 & 1  & 0  & \dots &0& 0 & 0 \\
0 & 0  & 1 & -2 & 1  & \dots &0& 0 & 0 \\
\vdots & \vdots & \vdots & \vdots&\vdots & \ddots & \vdots & \vdots & \vdots \\
0 & 0  & 0 &0 & 0  & \dots &1& -2 & 1
\end{bmatrix}_{(n-2) \times n}
\end{equation*}
Here, $D$ is the second–difference matrix. Each row of $D$ applies the operator 
\[
(\Delta^2 \mu)_t = (\mu_{t+1} - \mu_t )-(\mu_t- \mu_{t-1}),
\]
which measures the change in the first differences of $\mu_t$—that is, the change in slope between consecutive time points. In contrast to the first–difference matrix used in Scenario 1, which penalized deviations in level (enforcing long stretches of constant mean), the second–difference matrix penalizes deviations in slope. 

When applied within the generalized lasso, we will see that this penalty drives most second differences to zero. The result is that the estimated signal $\mu_t$ consists of long continuous segments with constant slope, separated by a small number of slope changes (trend shifts). In other words, $D$ enforces piecewise linearity with continuity at the changepoints: the fitted function does not jump in level, but only changes direction where the penalty allows a nonzero second difference.

Figure~\ref{fig:TSFL-Group-Null} shows that all four methods perform equivalently in the null case: the estimated number of changepoints is almost always zero, with only occasional isolated instances in which a single changepoint is detected. These occurrences reflect rare sampling fluctuations rather than systematic differences between the procedures.

\begin{figure}[H]
    \centering
    \textbf{Trend Shift Results with No Changepoints}\\ 
    \vspace{0.5em}
    \includegraphics[width=.8\textwidth]{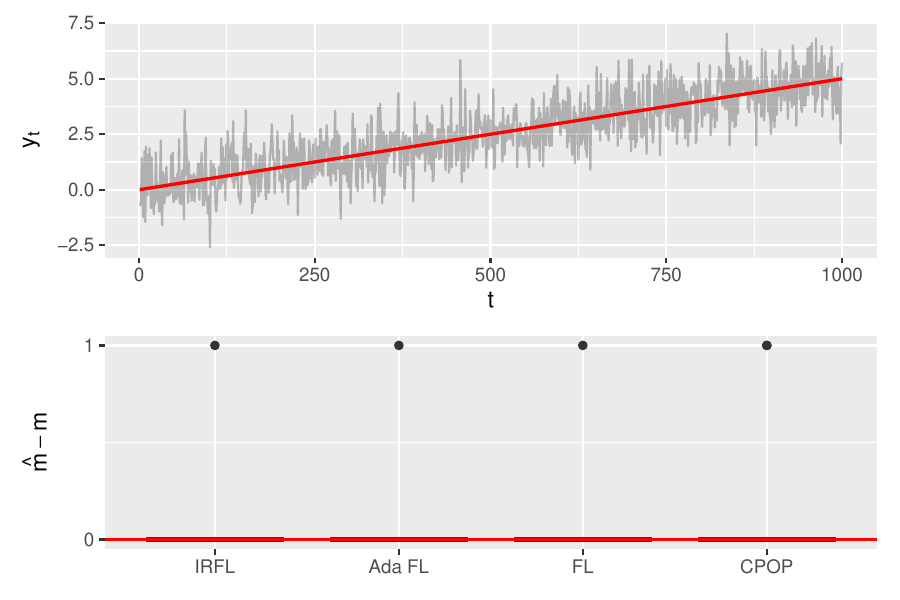}
    \caption[Comparison of trend shift detection methods with no changepoints]{Top: Representative simulation of a trend shift model with no changepoints. Bottom: Boxplot of centered sampling distribution of the number of estimated changepoints, $\hat m$.}
    \label{fig:TSFL-Group-Null}
\end{figure}

The bottom panel of Figure \ref{fig:TSFL-Group-Null} is actually a boxplot, with the entirety of the box contained at $\hat m-m=0$; the dots at $\hat m-m=1$ are outliers, showing that the estimated number of changepoints $\hat m$ is almost always equal to $m$--in this case, 0.

\subsubsection*{Scenario 9: Trend Shift Model with Three Changepoints}

The ninth scenario builds on the previous case by introducing changepoints into the trend shift model. Here, the underlying mean is piecewise linear, with slope shifts occurring at $t=252$, $t=502$, and $t=752$. Between changepoints, the slope remains constant, and continuity at the changepoint locations is preserved by construction. This setting is more demanding than the pure trend model of Scenario 8, as it requires methods not only to avoid spurious detections but also to identify the correct locations of genuine slope changes. The model is the same as in Scenario 8 \eqref{trend shift model}, as are the design and structured difference matrices, as well as the parameter vector.

The comparison in Figure \ref{fig:TSFL-Group} highlights several key points. First, the middle panel shows that IRFL consistently identifies the correct number of changepoints across all simulations, showing both robustness and parsimony. Second, the bottom panel indicates that the estimated changepoint vectors $\hat\tau$ from IRFL lie very close to the true vector $\bm\tau = (252,502,752)^\top$. The distance from $\hat{\bm\tau}$ to $\bm\tau$ is very close to zero, reflecting accurate localization. By contrast,  fused lasso overestimates the number of changepoints. IRFL, Ada FL and CPOP have the strongest overall performance in terms of the estimation of the number and location of changepoints.

\begin{figure}[H]
    \centering
    \textbf{Trend Shift Results with Three Changepoints}\\ 
    \vspace{0.5em}
    \includegraphics[width=.8\textwidth]{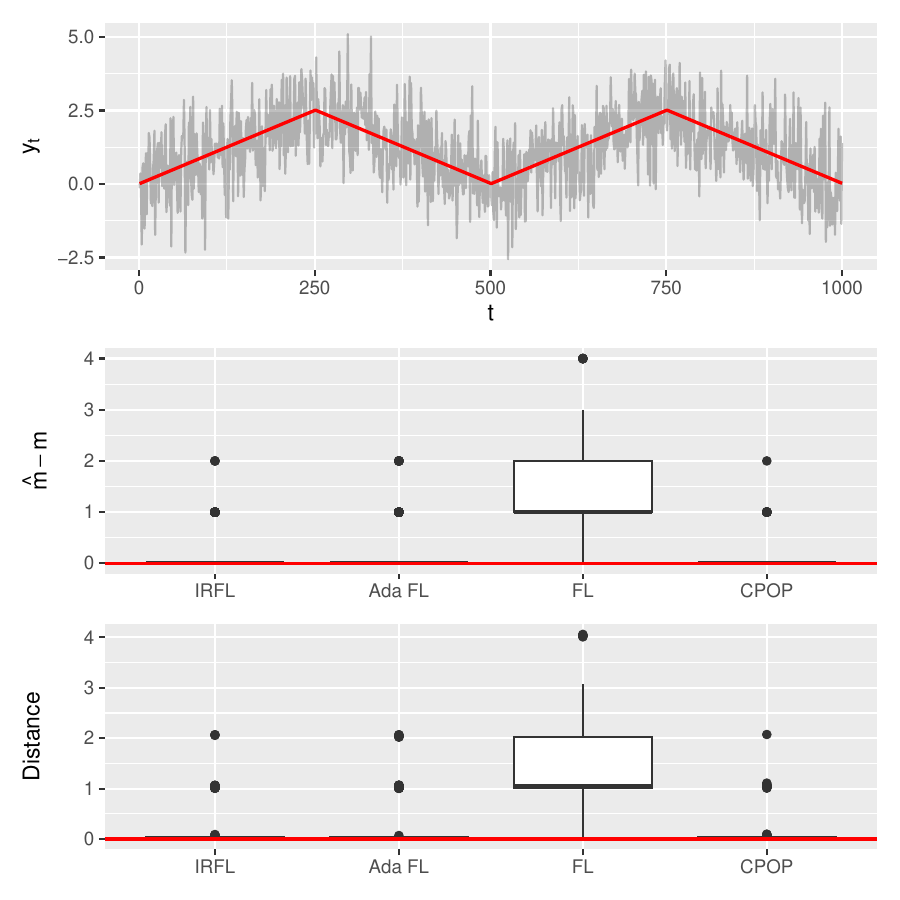}
    \caption[Comparison of trend shift detection methods with three changepoints]{Top: Representative simulation of trend shift model with changepoints at $t=252$, $502$, and $752$. Middle: Boxplot of centered sampling distribution of the number of estimated changepoints, $\hat m$. Bottom: Boxplot of sampling distribution of the distance from the estimated changepoint vector $\hat{\bm\tau}$ to the true vector $\bm\tau$.}
    \label{fig:TSFL-Group}
\end{figure}

Together, these results demonstrate that IRFL achieves the strongest balance: it avoids falsely detecting changepoints while maintaining excellent accuracy in locating true slope changes. This complements the findings of Scenario 8 by showing that IRFL not only resists spurious detections when no changepoints exist but also excels when genuine changepoints are present.

\subsubsection*{Scenario 10: Trend Shift Model with Random Changepoints}

The tenth scenario extends the trend shift model by introducing randomly located changepoints rather than fixing them in advance. In each simulation, the number of changepoints is drawn uniformly between one and seven, and their locations are sampled with a minimum spacing of 100 observations to ensure adequate segment length. The resulting mean function is piecewise linear, with alternating slope direction and continuous joins at each changepoint, so that both the magnitude and location of slope changes vary across replicates. The sign alternates in at changepoints, and in all cases the slope was $\pm0.01$. This setup reflects a substantially more challenging regime than the deterministic three‐changepoint design of Scenario 9, as estimation accuracy must now hold uniformly over diverse configurations of changepoint counts and positions.

The comparison in Figure \ref{fig:TSFL-Group-Random} highlights the resulting performance differences across methods. IRFL remains the most consistent overall, accurately recovering both the number and approximate locations of the true changepoints while exhibiting low dispersion in its sampling distribution as seen in the middle panel of \ref{fig:TSFL-Group-Random}. Adafused lasso shows comparable performance but with slightly greater variability, whereas fused lasso tends to oversegment, producing an inflated number of changepoints. CPOP performs not quite as well as IRFL well in localization as seen in the bottom panel, but also tends to undersegment as seen by the middle panel.

\begin{figure}[H]
\centering
\textbf{Trend Shift Results with Random Changepoints}\
\vspace{0.5em}
\includegraphics[width=.9\textwidth]{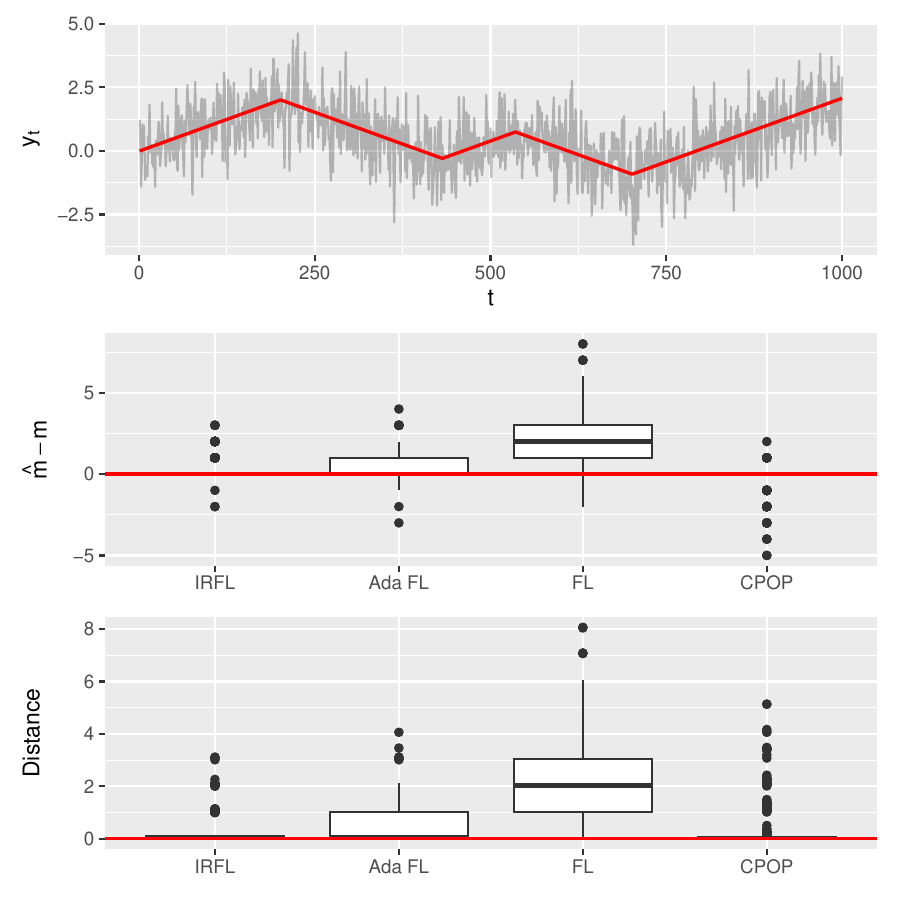}
\caption[Comparison of trend shift detection methods with random changepoints]{Top: Representative simulation of a trend‐shift model with randomly generated changepoint locations. Middle: Boxplot of centered sampling distribution of the number of estimated changepoints, $\hat m$. Bottom: Boxplot of sampling distribution of the distance between the estimated changepoint vector $\hat{\bm\tau}$ and the true vector $\bm\tau$ by method.}
\label{fig:TSFL-Group-Random}
\end{figure}

Together, these results demonstrate that IRFL provides the best overall balance—resisting false detections while accurately localizing genuine slope changes—even under heterogeneous and randomly varying changepoint structures. This generalizes the conclusions from Scenario 9, showing that IRFL’s robustness extends from smooth, globally linear trends to complex, piecewise linear environments with stochastic changepoint configurations.

\subsubsection*{Summary}
At the conclusion of Group 3, the results show that all methods designed for trend–shift models perform similarly in the null case: the estimated number of changepoints is almost always zero, with only occasional isolated detections arising from sampling variation. When true slope changes are present, adaptive fused lasso, IRFL, and CPOP each recover the changepoint structure with comparable accuracy, with no method showing systematic over- or under-segmentation. Small differences arise primarily in the precision of localized changepoint placement, but these remain modest across all procedures. Overall, the results indicate that IRFL's strong performance observed in mean-shift settings extends naturally to piecewise-linear trend settings, and that the methods considered here demonstrate broadly similar behavior in detecting and localizing slope changes.

\subsection{Group 4: Mean Shifts with Seasonality and Trend}

\noindent This group evaluates simultaneous estimation of seasonal structure and mean shifts, with an additional set of scenarios that include a global linear trend. Seasonality poses a particular challenge because recurring fluctuations can obscure changes in the underlying mean. In particular, seasonal peaks and troughs may either mimic a shift in level or drown out a genuine shift, thereby reducing detection power \citep{liu-2018-change}. As with any structured component, the difficulty of the problem depends on its scale relative to the noise: when the seasonal pattern is large compared with the noise variance, it can be estimated reliably, whereas when the noise is of comparable size, the seasonal effect becomes harder to distinguish, and errors in seasonal estimation can propagate directly into changepoint recovery. To address this, the design matrices explicitly encode the seasonal pattern (period $p$) and, when present, a global trend term, while the structured difference matrix penalizes only differences among segment means. This construction preserves unbiased estimation of seasonal coefficients while still allowing changepoints in the underlying mean.

The included design matrices explicitly encode the seasonal pattern (period $p$) and, when present, a global trend term; the structured difference matrix penalizes only differences among segment means so that seasonal coefficients remain unbiased. Scenarios cover seasonal nulls (no changepoints), seasonal models with true mean shifts, and seasonal-plus-trend models with and without changepoints. Because there are no published methods capable of handling multiple mean shifts with the presence of seasonality only fused lasso, adafused lasso, and IRFL are compared on (i) rate of changepoint detection, (ii) unbiasedness of seasonal parameters $\hat s_i$ (and trend, when applicable) estimates, and (iii) changepoint counting and localization via the established distance metric.

\subsubsection*{Scenario 11: Mean Shift Model with Seasonality and No Changepoints}
In the 11th scenario, seasonality is present in the data but there are no mean shifts. The model used in such a regression is 
\begin{equation}\label{seasonality model}
y_t = \begin{cases}
    \mu_1 + s_t +\epsilon_t & 1\le t\le p\\
    \mu_{t-p+1} + s_t + \epsilon_t &  p+1\le t\le n
    \end{cases}
\end{equation}
where $\sum_{i=1}^{p-1}s_i=0$ and $s_t=s_{t-p}$ for $p+1 \le t \le n$. This parametrization means that a mean shift cannot occur within the first $p$ indices, so that the number of estimated parameters is maximally $n$ (by placing a changepoint at every index $t$ between $p+1$ and $n$ inclusive) and the model is estimated uniquely. 

\begin{remark}
    For clarity and emphasis, the first $p$ indices are ``pre-fused" in this parameterization so that the design matrix $X$ has full rank and the model can be uniquely estimated. Under this parameterization, a changepoint can then be located anywhere from index $p+1$ to $n$---importantly, \emph{not} only at the end of a season.
\end{remark}

The design, coefficient, and difference matrices will be given shortly. However, as they can be difficult to describe for large dimensions, an example will be provided for $p=7$ and $n=21$ to build intuition. This might correspond to three weeks' worth of data, where the seasonal component (period $p=7$) is one week long. In such a scenario, the design matrix $X$ is given as follows:
\begin{figure}[H]
\centering
\vspace{.1cm}
\[
\scalebox{0.8}{$
X=\left[
\begin{array}{ccccccccccccccccccccc}
1 & \cdot & \cdot & \cdot & \cdot & \cdot & 1 & \cdot & \cdot & \cdot & \cdot & \cdot & \cdot & \cdot & \cdot & \cdot & \cdot & \cdot & \cdot & \cdot & \cdot \\
\cdot & 1 & \cdot & \cdot & \cdot & \cdot & 1 & \cdot & \cdot & \cdot & \cdot & \cdot & \cdot & \cdot & \cdot & \cdot & \cdot & \cdot & \cdot & \cdot & \cdot \\
\cdot & \cdot & 1 & \cdot & \cdot & \cdot & 1 & \cdot & \cdot & \cdot & \cdot & \cdot & \cdot & \cdot & \cdot & \cdot & \cdot & \cdot & \cdot & \cdot & \cdot \\
\cdot & \cdot & \cdot & 1 & \cdot & \cdot & 1 & \cdot & \cdot & \cdot & \cdot & \cdot & \cdot & \cdot & \cdot & \cdot & \cdot & \cdot & \cdot & \cdot & \cdot \\
\cdot & \cdot & \cdot & \cdot & 1 & \cdot & 1 & \cdot & \cdot & \cdot & \cdot & \cdot & \cdot & \cdot & \cdot & \cdot & \cdot & \cdot & \cdot & \cdot & \cdot \\
\cdot & \cdot & \cdot & \cdot & \cdot & 1 & 1 & \cdot & \cdot & \cdot & \cdot & \cdot & \cdot & \cdot & \cdot & \cdot & \cdot & \cdot & \cdot & \cdot & \cdot \\
-1 & -1 & -1 & -1 & -1 & -1 & 1 & \cdot & \cdot & \cdot & \cdot & \cdot & \cdot & \cdot & \cdot & \cdot & \cdot & \cdot & \cdot & \cdot & \cdot \\
1 & \cdot & \cdot & \cdot & \cdot & \cdot & \cdot & 1 & \cdot & \cdot & \cdot & \cdot & \cdot & \cdot & \cdot & \cdot & \cdot & \cdot & \cdot & \cdot & \cdot \\
\cdot & 1 & \cdot & \cdot & \cdot & \cdot & \cdot & \cdot & 1 & \cdot & \cdot & \cdot & \cdot & \cdot & \cdot & \cdot & \cdot & \cdot & \cdot & \cdot & \cdot \\
\cdot & \cdot & 1 & \cdot & \cdot & \cdot & \cdot & \cdot & \cdot & 1 & \cdot & \cdot & \cdot & \cdot & \cdot & \cdot & \cdot & \cdot & \cdot & \cdot & \cdot \\
\cdot & \cdot & \cdot & 1 & \cdot & \cdot & \cdot & \cdot & \cdot & \cdot & 1 & \cdot & \cdot & \cdot & \cdot & \cdot & \cdot & \cdot & \cdot & \cdot & \cdot \\
\cdot & \cdot & \cdot & \cdot & 1 & \cdot & \cdot & \cdot & \cdot & \cdot & \cdot & 1 & \cdot & \cdot & \cdot & \cdot & \cdot & \cdot & \cdot & \cdot & \cdot \\
\cdot & \cdot & \cdot & \cdot & \cdot & 1 & \cdot & \cdot & \cdot & \cdot & \cdot & \cdot & 1 & \cdot & \cdot & \cdot & \cdot & \cdot & \cdot & \cdot & \cdot \\
-1 & -1 & -1 & -1 & -1 & -1 & \cdot & \cdot & \cdot & \cdot & \cdot & \cdot & \cdot & 1 & \cdot & \cdot & \cdot & \cdot & \cdot & \cdot & \cdot \\
1 & \cdot & \cdot & \cdot & \cdot & \cdot & \cdot & \cdot & \cdot & \cdot & \cdot & \cdot & \cdot & \cdot & 1 & \cdot & \cdot & \cdot & \cdot & \cdot & \cdot \\
\cdot & 1 & \cdot & \cdot & \cdot & \cdot & \cdot & \cdot & \cdot & \cdot & \cdot & \cdot & \cdot & \cdot & \cdot & 1 & \cdot & \cdot & \cdot & \cdot & \cdot \\
\cdot & \cdot & 1 & \cdot & \cdot & \cdot & \cdot & \cdot & \cdot & \cdot & \cdot & \cdot & \cdot & \cdot & \cdot & \cdot & 1 & \cdot & \cdot & \cdot & \cdot \\
\cdot & \cdot & \cdot & 1 & \cdot & \cdot & \cdot & \cdot & \cdot & \cdot & \cdot & \cdot & \cdot & \cdot & \cdot & \cdot & \cdot & 1 & \cdot & \cdot & \cdot \\
\cdot & \cdot & \cdot & \cdot & 1 & \cdot & \cdot & \cdot & \cdot & \cdot & \cdot & \cdot & \cdot & \cdot & \cdot & \cdot & \cdot & \cdot & 1 & \cdot & \cdot \\
\cdot & \cdot & \cdot & \cdot & \cdot & 1 & \cdot & \cdot & \cdot & \cdot & \cdot & \cdot & \cdot & \cdot & \cdot & \cdot & \cdot & \cdot & \cdot & 1 & \cdot \\
-1 & -1 & -1 & -1 & -1 & -1 & \cdot & \cdot & \cdot & \cdot & \cdot & \cdot & \cdot & \cdot & \cdot & \cdot & \cdot & \cdot & \cdot & \cdot & 1 \\
\end{array}
\right]
$}
\]
\vspace{.1cm}
\end{figure}
The first 6 columns correspond to the first $p-1$ independently estimated seasonal components $s_1,s_2,\ldots,s_6$. Recall that $s_7=-\sum_{i=1}^6s_i$, and so every 7th row, the first six entries are all $-1$. As the seasonal component repeats every $7$ observations, this pattern of rows repeat three times making $n=21$ rows. The coefficient vector $\beta$ is given as:
\begin{equation*}
    \beta=\begin{bmatrix}
        s_1 & s_2 & \ldots & s_6 & \mu_1 & \mu_2 & \ldots & \mu_{15}
    \end{bmatrix}^\top.
\end{equation*}
Note that though the seasonality has period $p=7$, only $p-1=6$ parameters are required to estimate this seasonality---leaving $n-(p-1)=21-(7-1)=15$ mean parameters.

Inspection will reveal that one can recover the full model with a changepoint at every time index $t\ge p$ in \eqref{seasonality model} by writing the matrix product $y=X\beta+\epsilon$. The desire when using the genlasso framework is to penalize differences between segment means, and so in the $D$ matrix, the first 6 columns are omitted from differences so that their estimation is unbiased. The $D$ matrix is therefore:
\begin{figure}[H]
\centering
\vspace{.1cm}
\[
\scalebox{0.8}{$
D =
\left[\begin{array}{ccccccccccccccccccccc}
\cdot & \cdot & \cdot & \cdot & \cdot & \cdot & -1 & 1 & \cdot & \cdot & \cdot & \cdot & \cdot & \cdot & \cdot & \cdot & \cdot & \cdot & \cdot & \cdot & \cdot \\
\cdot & \cdot & \cdot & \cdot & \cdot & \cdot & \cdot & -1 & 1 & \cdot & \cdot & \cdot & \cdot & \cdot & \cdot & \cdot & \cdot & \cdot & \cdot & \cdot & \cdot \\
\cdot & \cdot & \cdot & \cdot & \cdot & \cdot & \cdot & \cdot & -1 & 1 & \cdot & \cdot & \cdot & \cdot & \cdot & \cdot & \cdot & \cdot & \cdot & \cdot & \cdot \\
\cdot & \cdot & \cdot & \cdot & \cdot & \cdot & \cdot & \cdot & \cdot & -1 & 1 & \cdot & \cdot & \cdot & \cdot & \cdot & \cdot & \cdot & \cdot & \cdot & \cdot \\
\cdot & \cdot & \cdot & \cdot & \cdot & \cdot & \cdot & \cdot & \cdot & \cdot & -1 & 1 & \cdot & \cdot & \cdot & \cdot & \cdot & \cdot & \cdot & \cdot & \cdot \\
\cdot & \cdot & \cdot & \cdot & \cdot & \cdot & \cdot & \cdot & \cdot & \cdot & \cdot & -1 & 1 & \cdot & \cdot & \cdot & \cdot & \cdot & \cdot & \cdot & \cdot \\
\cdot & \cdot & \cdot & \cdot & \cdot & \cdot & \cdot & \cdot & \cdot & \cdot & \cdot & \cdot & -1 & 1 & \cdot & \cdot & \cdot & \cdot & \cdot & \cdot & \cdot \\
\cdot & \cdot & \cdot & \cdot & \cdot & \cdot & \cdot & \cdot & \cdot & \cdot & \cdot & \cdot & \cdot & -1 & 1 & \cdot & \cdot & \cdot & \cdot & \cdot & \cdot \\
\cdot & \cdot & \cdot & \cdot & \cdot & \cdot & \cdot & \cdot & \cdot & \cdot & \cdot & \cdot & \cdot & \cdot & -1 & 1 & \cdot & \cdot & \cdot & \cdot & \cdot \\
\cdot & \cdot & \cdot & \cdot & \cdot & \cdot & \cdot & \cdot & \cdot & \cdot & \cdot & \cdot & \cdot & \cdot & \cdot & -1 & 1 & \cdot & \cdot & \cdot & \cdot \\
\cdot & \cdot & \cdot & \cdot & \cdot & \cdot & \cdot & \cdot & \cdot & \cdot & \cdot & \cdot & \cdot & \cdot & \cdot & \cdot & -1 & 1 & \cdot & \cdot & \cdot \\
\cdot & \cdot & \cdot & \cdot & \cdot & \cdot & \cdot & \cdot & \cdot & \cdot & \cdot & \cdot & \cdot & \cdot & \cdot & \cdot & \cdot & -1 & 1 & \cdot & \cdot \\
\cdot & \cdot & \cdot & \cdot & \cdot & \cdot & \cdot & \cdot & \cdot & \cdot & \cdot & \cdot & \cdot & \cdot & \cdot & \cdot & \cdot & \cdot & -1 & 1 & \cdot \\
\cdot & \cdot & \cdot & \cdot & \cdot & \cdot & \cdot & \cdot & \cdot & \cdot & \cdot & \cdot & \cdot & \cdot & \cdot & \cdot & \cdot & \cdot & \cdot & -1 & 1
\end{array}\right].
$}
\]
\vspace{.1cm}
\end{figure}
\noindent The last $n-p+1=15$ columns correspond to differences in consecutive $\mu_t$ so that sparsity in changepoints is enforced.

It is also useful to write these matrices in block form for general $p$ and $n$ for ease of construction. Accordingly, let the matrix $P$ be defined as
\[
P = 
\begin{bmatrix}
1 & 0 & 0 & \ldots & 0 \\
0 & 1 & 0 & \ldots & 0 \\
0 & 0 & 1 & \ldots & 0 \\
\vdots & \vdots & \vdots & \ddots & \vdots \\
0 & 0 & 0 & \ldots & 1 \\
-1 & -1 & -1 & \ldots & -1
\end{bmatrix}_{p \times (p - 1)}
\]
so that $P$ captures one full period. Then define $\mathcal{P} \in \mathbb{R}^{n \times (p - 1)}$ be the matrix formed by repeating $P$ vertically and truncating to $n$ rows. In this way, the columns of $X$ corresponding to the seasonal component of $y_t$ are constructed. Then, define the first-order difference matrix over $n - p + 1$ variables as 
\[
\Delta =
\begin{bmatrix}
-1 & 1 & 0 & \cdots & 0 \\
0 & -1 & 1 & \cdots & 0 \\
\vdots & \ddots & \ddots & \ddots & \vdots \\
0 & \cdots & 0 & -1 & 1
\end{bmatrix}_{(n - p) \times (n - p + 1)}.
\]
Further define $\mathbf{1}^{[1:(p-1)]}_{n\times 1}$ to be a vector whose first $p-1$ entries are 1, and 0 thereafter. Then the design matrix $X$, coefficient vector ${\beta}$, and structured difference matrix $D$ are given by
\begin{equation*}
\begin{array}{c}
X = 
\begin{bmatrix}
\mathcal{P}_{n \times (p - 1)} \,\, \bigg| \,\,\mathbf{1}^{[1:(p-1)]}_{n\times 1}\,\, \bigg| \,\,
\begin{array}{c}
\mathbf{0}_{p\times (n - p )} \\
\hline
I_{n - p }
\end{array}
\end{bmatrix}_{n \times n},\\
\\
\\
D = 
\left[
\begin{array}{c|c}
\mathbf{0}_{(n - p) \times (p - 1)} & \Delta
\end{array}
\right]_{(n - p) \times n},
\end{array}
\qquad 
{\beta} = 
\begin{bmatrix}
s_1 \\
s_2 \\
\vdots \\
s_{p - 1} \\
\mu_1 \\
\mu_2 \\
\vdots \\
\mu_{n - p + 1}
\end{bmatrix}_{n \times 1}.
\end{equation*}

With this setup in place, we can now examine how the fused lasso, adafused lasso, and IRFL perform when applied to purely seasonal data without any mean shifts. In the following simulation, 1000 datasets were generated with \(n=1200\), all with seasonality of period 12 and no mean shifts.

The third-from-top panel in Figure~\ref{fig:Season-Group-Null} shows that fused lasso, adafused lasso, and IRFL almost never detect changepoints when none are present. The bottom panel further shows that all three methods yield unbiased estimates of the seasonal components in the absence of mean shifts. In other words---all three methods are approximately equally effective.

\begin{figure}[H]
    \centering
    \textbf{Mean Shift With No Changepoints and Seasonality Results}\\ 
    \vspace{0.5em}
    \includegraphics[width=.9\textwidth]{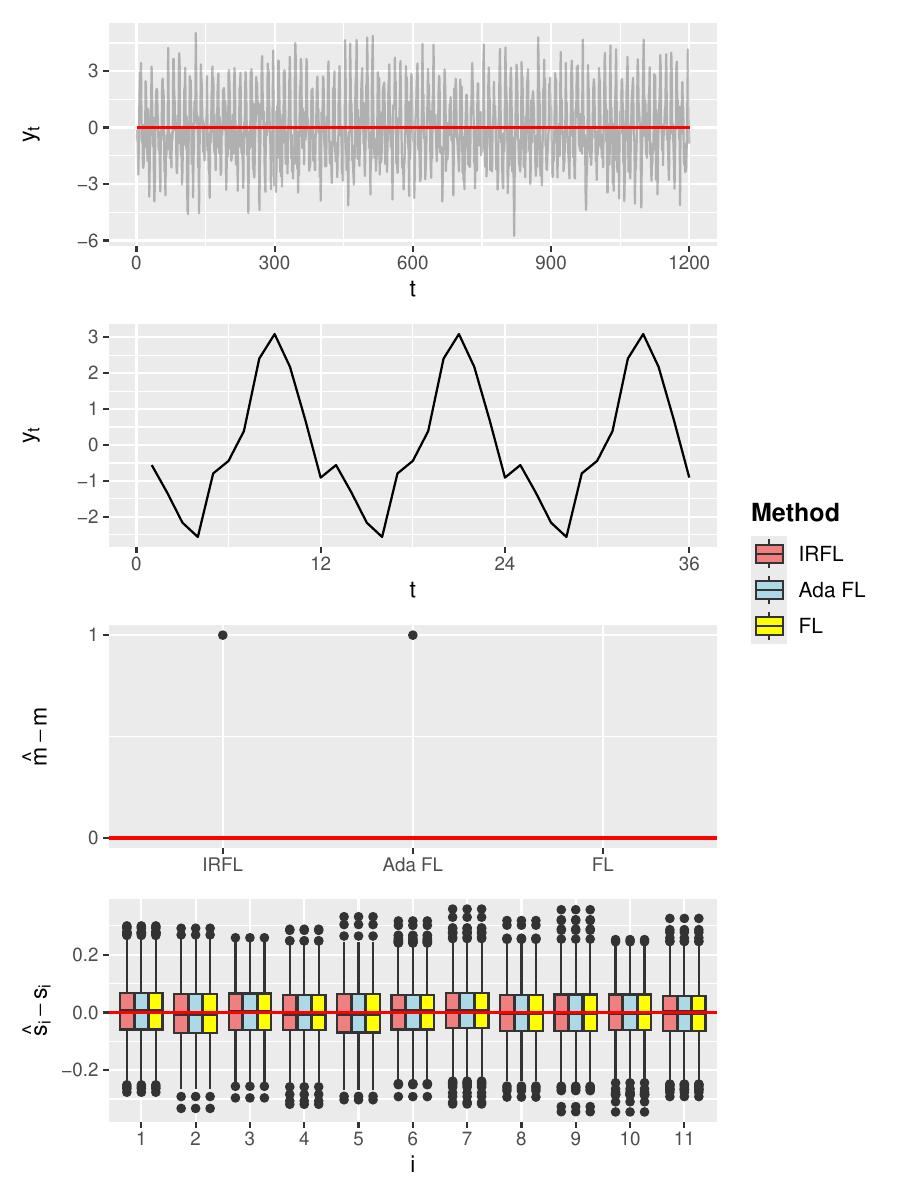}
    \caption[Simulation and results of IRFL for no mean shifts + seasonality estimation]{Top: Simulation of time series with seasonality and no changepoints. Second from Top: Three periods of an example season. Third from Top: Boxplot of centered sampling distribution of the number of estimated changepoints, $\hat m$. Bottom: Boxplot of centered sampling distribution of the seasonality components $\hat s_i$ for the three methods.}
    \label{fig:Season-Group-Null}
\end{figure}

\subsubsection*{Scenario 12: Mean Shift Model with Seasonality and Three Changepoints}
In the 12th scenario, there is once again seasonality, and the model is identical to the 11th scenario \eqref{seasonality model}. Likewise, the design and structured difference matrices and parameter vector are identical. However, mean shifts of size $2$ with alternating signs are present at $t=301$, $t=601$, and $t=901$. A representative dataset is given in Figure \ref{fig:Season-Sim}.

\begin{figure}[H]
    \centering
    \textbf{Mean Shift With Three Changepoints and Seasonality Simulation}\\ 
    \vspace{0.5em}
    \includegraphics[width=1\textwidth]{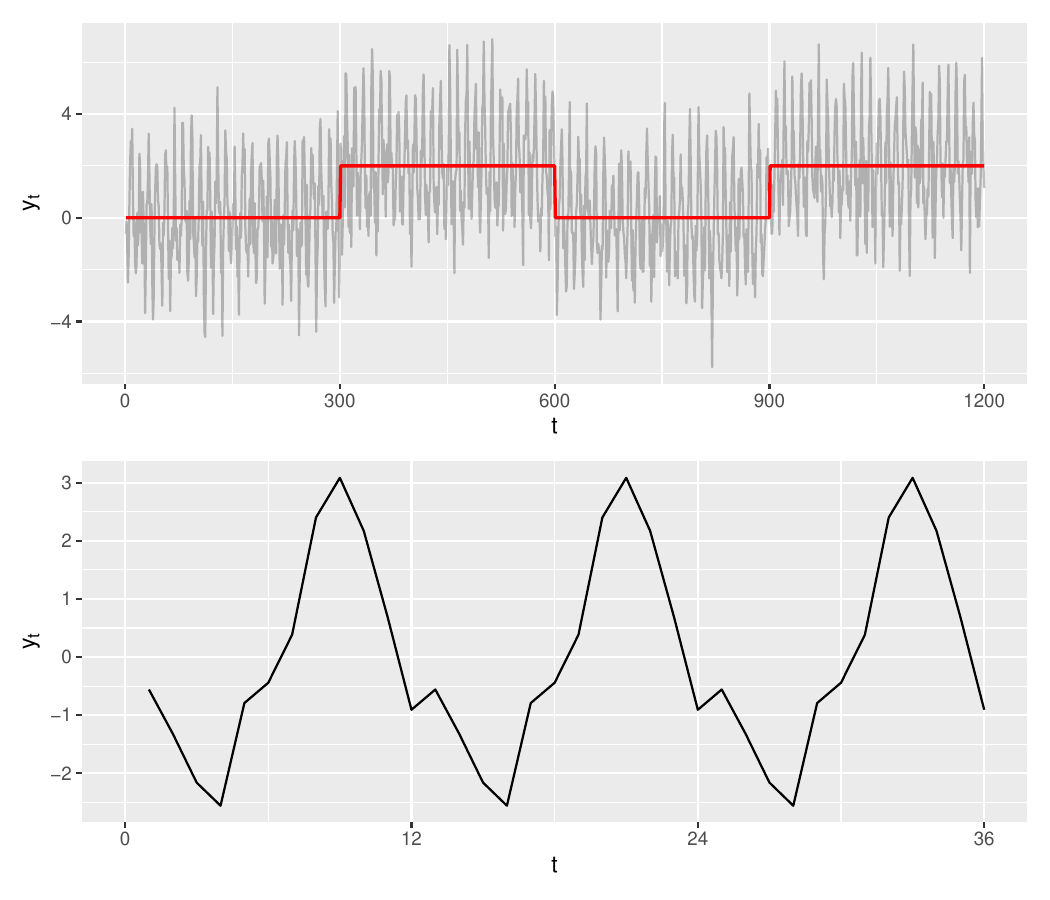}
    \caption[Simulation of dataset with three mean shifts + seasonality]{Top: Representative simulation of a mean shift model with seasonality and three changepoints at $t=301$, $t=601$, and $t=901$. Bottom: Three periods of an example season.}
    \label{fig:Season-Sim}
\end{figure}

As demonstrated in Figure \ref{fig:Season-Group-Plots} and in keeping with previous modeled scenarios, the fused lasso finds far too many changepoints in comparison to adafused lasso and IRFL, which both tend to choose the correct number and placement of changepoints. Despite this, fused lasso is competitive with both adafused lasso and IRFL for the unbiased estimation of the seasonal components. This is a wonderful result---to the best of our knowledge at the time of this writing, no existing penalized-likelihood or optimization-based method recovers both seasonal structure and multiple changepoints in a unified framework.

\begin{figure}[H]
    \centering
    \textbf{Mean Shift With Three Changepoints and Seasonality Results}\\ 
    \vspace{0.5em}
    \includegraphics[width=1\textwidth]{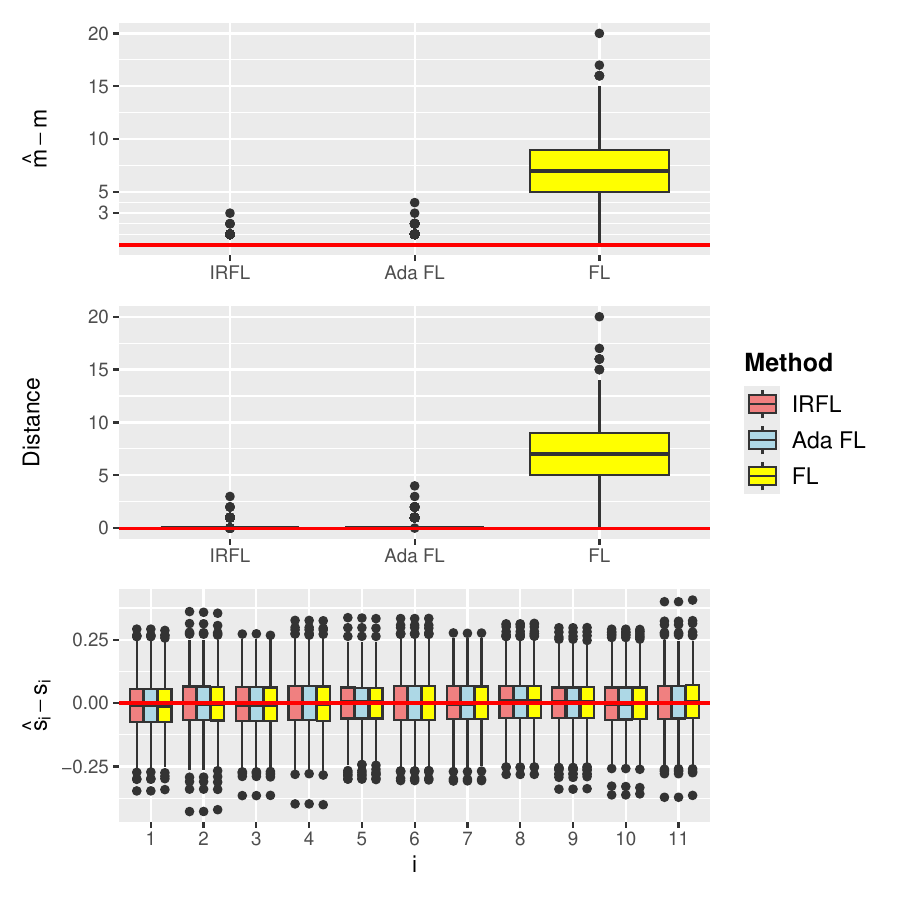}
    \caption[Results of IRFL for three mean shifts + seasonality estimation]{Top: Boxplot of centered sampling distribution of the number of estimated changepoints, $\hat m$. Second from Top: Boxplot of sampling distribution of the distance from the estimated changepoint vector $\hat{\bm\tau}$ to the true vector $\bm\tau$. Bottom: Boxplot of centered sampling distribution of the seasonality components $\hat s_i$ for the three methods.}
    \label{fig:Season-Group-Plots}
\end{figure}

\subsubsection*{Scenario 13: Mean Shift Model with Seasonality, Trend and No Changepoints}
In the 13th scenario, seasonality is once again present in the data, but there is also a constant trend. The model for such a scenario is 
\begin{equation}\label{full model with trend}
y_t = \begin{cases}
    \alpha t+ s_t +\mu_1 + \epsilon_t &  1\le t\le p+1\\
    \alpha t + s_t +\mu_{t-p} + \epsilon_t & p+2 \leq t \leq n
    \end{cases}
\end{equation}
where once again $\sum_{i=1}^{p-1}s_i=0$ and $s_t=s_{t-p}$ for $p+1\le t\le n$. This parametrization enforces that the maximal number of parameters to be $n$. Thus, even full models (with the maximal number of changepoints) are identifiable. Following the convention established in the 10th scenario, an example design, coefficient, and structured difference matrix is given for $n=21$ and $p=7$ here:
\begin{figure}[H]
\centering
\vspace{.1cm}
\[
\scalebox{0.8}{$
X=\left[
\begin{array}{ccccccccccccccccccccc}
1 & 1 & \cdot & \cdot & \cdot & \cdot & \cdot  & 1 & \cdot & \cdot & \cdot & \cdot & \cdot & \cdot & \cdot & \cdot & \cdot & \cdot & \cdot & \cdot & \cdot \\
2 & \cdot & 1 & \cdot & \cdot & \cdot & \cdot  & 1 & \cdot & \cdot & \cdot & \cdot & \cdot & \cdot & \cdot & \cdot & \cdot & \cdot & \cdot & \cdot & \cdot \\
3 & \cdot & \cdot & 1 & \cdot & \cdot & \cdot  & 1 & \cdot & \cdot & \cdot & \cdot & \cdot & \cdot & \cdot & \cdot & \cdot & \cdot & \cdot & \cdot & \cdot \\
4 & \cdot & \cdot & \cdot & 1 & \cdot & \cdot  & 1 & \cdot & \cdot & \cdot & \cdot & \cdot & \cdot & \cdot & \cdot & \cdot & \cdot & \cdot & \cdot & \cdot \\
5 & \cdot & \cdot & \cdot & \cdot & 1 & \cdot  & 1 & \cdot & \cdot & \cdot & \cdot & \cdot & \cdot & \cdot & \cdot & \cdot & \cdot & \cdot & \cdot & \cdot \\
6 & \cdot & \cdot & \cdot & \cdot & \cdot & 1  & 1 & \cdot & \cdot & \cdot & \cdot & \cdot & \cdot & \cdot & \cdot & \cdot & \cdot & \cdot & \cdot & \cdot \\
7 & -1 & -1 & -1 & -1 & -1 & -1 & 1 & \cdot & \cdot & \cdot & \cdot & \cdot & \cdot & \cdot & \cdot & \cdot & \cdot & \cdot & \cdot & \cdot \\
8 & 1 & \cdot & \cdot & \cdot & \cdot  & \cdot & 1 & \cdot & \cdot & \cdot & \cdot & \cdot & \cdot & \cdot & \cdot & \cdot & \cdot & \cdot & \cdot & \cdot \\
9 & \cdot & 1 & \cdot & \cdot & \cdot  & \cdot & \cdot & 1 & \cdot & \cdot & \cdot & \cdot & \cdot & \cdot & \cdot & \cdot & \cdot & \cdot & \cdot & \cdot \\
10 & \cdot & \cdot & 1 & \cdot & \cdot & \cdot & \cdot & \cdot & 1 & \cdot & \cdot & \cdot & \cdot & \cdot & \cdot & \cdot & \cdot & \cdot & \cdot & \cdot \\
11 & \cdot & \cdot & \cdot & 1 & \cdot & \cdot & \cdot & \cdot & \cdot & 1 & \cdot & \cdot & \cdot & \cdot & \cdot & \cdot & \cdot & \cdot & \cdot & \cdot \\
12 & \cdot & \cdot & \cdot & \cdot & 1 & \cdot & \cdot & \cdot & \cdot & \cdot & 1 & \cdot & \cdot & \cdot & \cdot & \cdot & \cdot & \cdot & \cdot & \cdot \\
13 & \cdot & \cdot & \cdot & \cdot & \cdot & 1  & \cdot & \cdot & \cdot & \cdot & \cdot & 1 & \cdot & \cdot & \cdot & \cdot & \cdot & \cdot & \cdot & \cdot \\
14 & -1 & -1 & -1 & -1 & -1 & -1 & \cdot & \cdot & \cdot & \cdot & \cdot & \cdot & 1 & \cdot & \cdot & \cdot & \cdot & \cdot & \cdot & \cdot \\
15 & 1 & \cdot & \cdot & \cdot & \cdot & \cdot & \cdot & \cdot & \cdot & \cdot & \cdot & \cdot & \cdot & 1 & \cdot & \cdot & \cdot & \cdot & \cdot & \cdot \\
16 & \cdot & 1 & \cdot & \cdot & \cdot & \cdot & \cdot & \cdot & \cdot & \cdot & \cdot & \cdot & \cdot & \cdot & 1 & \cdot & \cdot & \cdot & \cdot & \cdot \\
17 & \cdot & \cdot & 1 & \cdot & \cdot & \cdot & \cdot & \cdot & \cdot & \cdot & \cdot & \cdot & \cdot & \cdot & \cdot & 1 & \cdot & \cdot & \cdot & \cdot \\
18 & \cdot & \cdot & \cdot & 1 & \cdot & \cdot & \cdot & \cdot & \cdot & \cdot & \cdot & \cdot & \cdot & \cdot & \cdot & \cdot & 1 & \cdot & \cdot & \cdot \\
19 & \cdot & \cdot & \cdot & \cdot & 1 & \cdot & \cdot & \cdot & \cdot & \cdot & \cdot & \cdot & \cdot & \cdot & \cdot & \cdot & \cdot & 1 & \cdot & \cdot \\
20 & \cdot & \cdot & \cdot & \cdot & \cdot & 1 & \cdot & \cdot & \cdot & \cdot & \cdot & \cdot & \cdot & \cdot & \cdot & \cdot & \cdot & \cdot & 1 & \cdot \\
21 & -1 & -1 & -1 & -1 & -1 & -1 & \cdot & \cdot & \cdot & \cdot & \cdot & \cdot & \cdot & \cdot & \cdot & \cdot & \cdot & \cdot & \cdot & 1 \\
\end{array}
\right]
$}
\]
\vspace{.1cm}
\end{figure}
The coefficient vector $\beta$ is given as:
\begin{equation*}
    \beta=\begin{bmatrix}
        \alpha & s_1 & s_2 & \ldots & s_6 & \mu_1 & \mu_2 & \ldots & \mu_{14}
    \end{bmatrix}^\top.
\end{equation*}
and finally the structured difference matrix $D$ is 
\begin{figure}[H]
\centering
\vspace{.1cm}
\[
\scalebox{0.8}{$
D =
\left[\begin{array}{ccccccccccccccccccccc}
\cdot &\cdot & \cdot & \cdot & \cdot & \cdot & \cdot & -1 & 1 & \cdot & \cdot & \cdot & \cdot & \cdot & \cdot & \cdot & \cdot & \cdot & \cdot & \cdot & \cdot  \\
\cdot &\cdot & \cdot & \cdot & \cdot & \cdot & \cdot & \cdot & -1 & 1 & \cdot & \cdot & \cdot & \cdot & \cdot & \cdot & \cdot & \cdot & \cdot & \cdot & \cdot  \\
\cdot &\cdot & \cdot & \cdot & \cdot & \cdot & \cdot & \cdot & \cdot & -1 & 1 & \cdot & \cdot & \cdot & \cdot & \cdot & \cdot & \cdot & \cdot & \cdot & \cdot  \\
\cdot &\cdot & \cdot & \cdot & \cdot & \cdot & \cdot & \cdot & \cdot & \cdot & -1 & 1 & \cdot & \cdot & \cdot & \cdot & \cdot & \cdot & \cdot & \cdot & \cdot  \\
\cdot &\cdot & \cdot & \cdot & \cdot & \cdot & \cdot & \cdot & \cdot & \cdot & \cdot & -1 & 1 & \cdot & \cdot & \cdot & \cdot & \cdot & \cdot & \cdot & \cdot  \\
\cdot &\cdot & \cdot & \cdot & \cdot & \cdot & \cdot & \cdot & \cdot & \cdot & \cdot & \cdot & -1 & 1 & \cdot & \cdot & \cdot & \cdot & \cdot & \cdot & \cdot  \\
\cdot &\cdot & \cdot & \cdot & \cdot & \cdot & \cdot & \cdot & \cdot & \cdot & \cdot & \cdot & \cdot & -1 & 1 & \cdot & \cdot & \cdot & \cdot & \cdot & \cdot  \\
\cdot &\cdot & \cdot & \cdot & \cdot & \cdot & \cdot & \cdot & \cdot & \cdot & \cdot & \cdot & \cdot & \cdot & -1 & 1 & \cdot & \cdot & \cdot & \cdot & \cdot  \\
\cdot &\cdot & \cdot & \cdot & \cdot & \cdot & \cdot & \cdot & \cdot & \cdot & \cdot & \cdot & \cdot & \cdot & \cdot & -1 & 1 & \cdot & \cdot & \cdot & \cdot  \\
\cdot &\cdot & \cdot & \cdot & \cdot & \cdot & \cdot & \cdot & \cdot & \cdot & \cdot & \cdot & \cdot & \cdot & \cdot & \cdot & -1 & 1 & \cdot & \cdot & \cdot  \\
\cdot &\cdot & \cdot & \cdot & \cdot & \cdot & \cdot & \cdot & \cdot & \cdot & \cdot & \cdot & \cdot & \cdot & \cdot & \cdot & \cdot & -1 & 1 & \cdot & \cdot  \\
\cdot &\cdot & \cdot & \cdot & \cdot & \cdot & \cdot & \cdot & \cdot & \cdot & \cdot & \cdot & \cdot & \cdot & \cdot & \cdot & \cdot & \cdot & -1 & 1 & \cdot  \\
\cdot &\cdot & \cdot & \cdot & \cdot & \cdot & \cdot & \cdot & \cdot & \cdot & \cdot & \cdot & \cdot & \cdot & \cdot & \cdot & \cdot & \cdot & \cdot & -1 & 1 
\end{array}\right].
$}
\]
\vspace{.1cm}
\end{figure}
These matrices can also be established in blocks for generalization to higher dimensions. Let $
    \textbf{t}=\left(
        1, 2, \ldots, n
    \right)^\top
$ and $\mathbf{1}^{[1:p]}_{n\times 1}$ be a vector of length $n$ whose first $p$ elements are 1 and 0 thereafter. Then the design matrix $X$, coefficient vector ${\beta}$, and structured difference matrix $D$ are given by
\[
\begin{array}{c}
X = 
\begin{bmatrix}\textbf{t}_{n\times 1}\ \bigg| \ 
\mathcal{P}_{n \times (p - 1)} \,\, \bigg| \,\,\mathbf{1}^{[1:p]}_{n\times 1}\,\, \bigg| \,\,
\begin{array}{c}
\mathbf{0}_{(p+1)\times (n - p -1)} \\
\hline
 I_{n - p -1}
\end{array}
\end{bmatrix}_{n \times n},
\\
\\
\\D = 
\left[
\begin{array}{c|c}
\mathbf{0}_{(n - p-1) \times p } & \Delta
\end{array}
\right]_{(n - p-1) \times n}\hspace{1pt},
\end{array}
\qquad 
{\beta} = 
\begin{bmatrix}
\alpha\\
s_1 \\
s_2 \\
\vdots \\
s_{p - 1} \\
\mu_1 \\
\mu_2 \\
\vdots \\
\mu_{n - p}
\end{bmatrix}_{n \times 1}
\] where here
\[
\Delta =
\begin{bmatrix}
-1 & 1 & 0 & \cdots & 0 \\
0 & -1 & 1 & \cdots & 0 \\
\vdots & \ddots & \ddots & \ddots & \vdots \\
0 & \cdots & 0 & -1 & 1
\end{bmatrix}_{(n - p-1) \times (n - p )}.
\]
In the following simulation 1000 datasets were generated with $n=1200$, all with seasonality of period 12 and no mean shifts, and a constant trend of $\alpha=0.005$. A representative dataset with three full representative seasons is given in Figure \ref{fig:Season-Group-Trend-Null-Simulation}.

\begin{figure}[H]
    \centering
        \textbf{Mean Shift With Trend, Seasonality, and No Changepoints}\\ \includegraphics[width=1\linewidth]{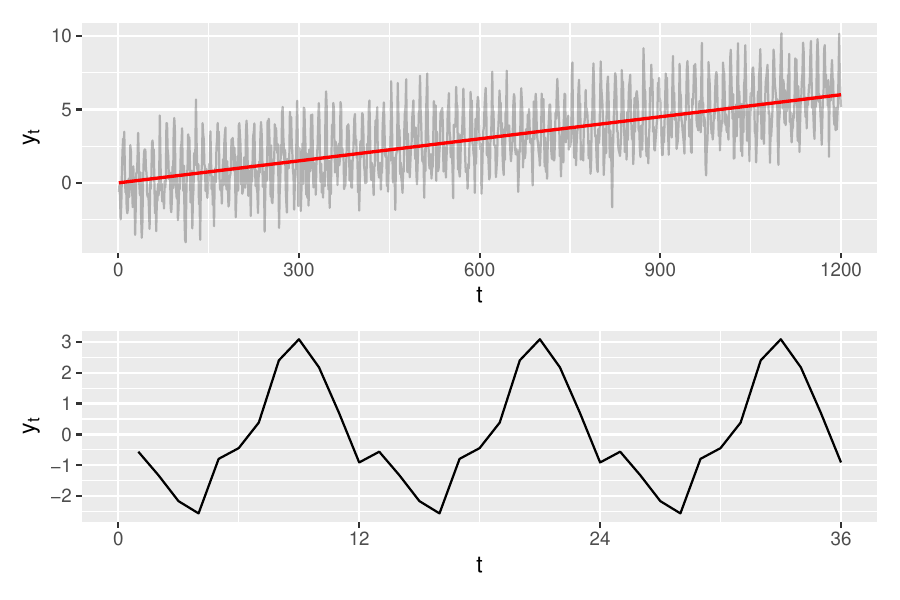}
    \caption[Simulation of IRFL for global trend + no mean shifts + seasonality estimation]{Top: Representative simulation of a mean shift model with global trend, seasonality and no changepoints. Bottom: Three periods of an example season.}
    \label{fig:Season-Group-Trend-Null-Simulation}
\end{figure}

As seen in the top panel of Figure~\ref{fig:Season-Group-Trend-Null}, fused lasso, adaptive fused lasso, and IRFL rarely introduce changepoints when none are present. The middle and bottom panels show that each method produces unbiased estimates of both the seasonal and trend components in the absence of changepoints.

\begin{figure}[H]
    \centering
        \textbf{Mean Shift With Trend, Seasonality, and No Changepoints}\\ \includegraphics[width=1\linewidth]{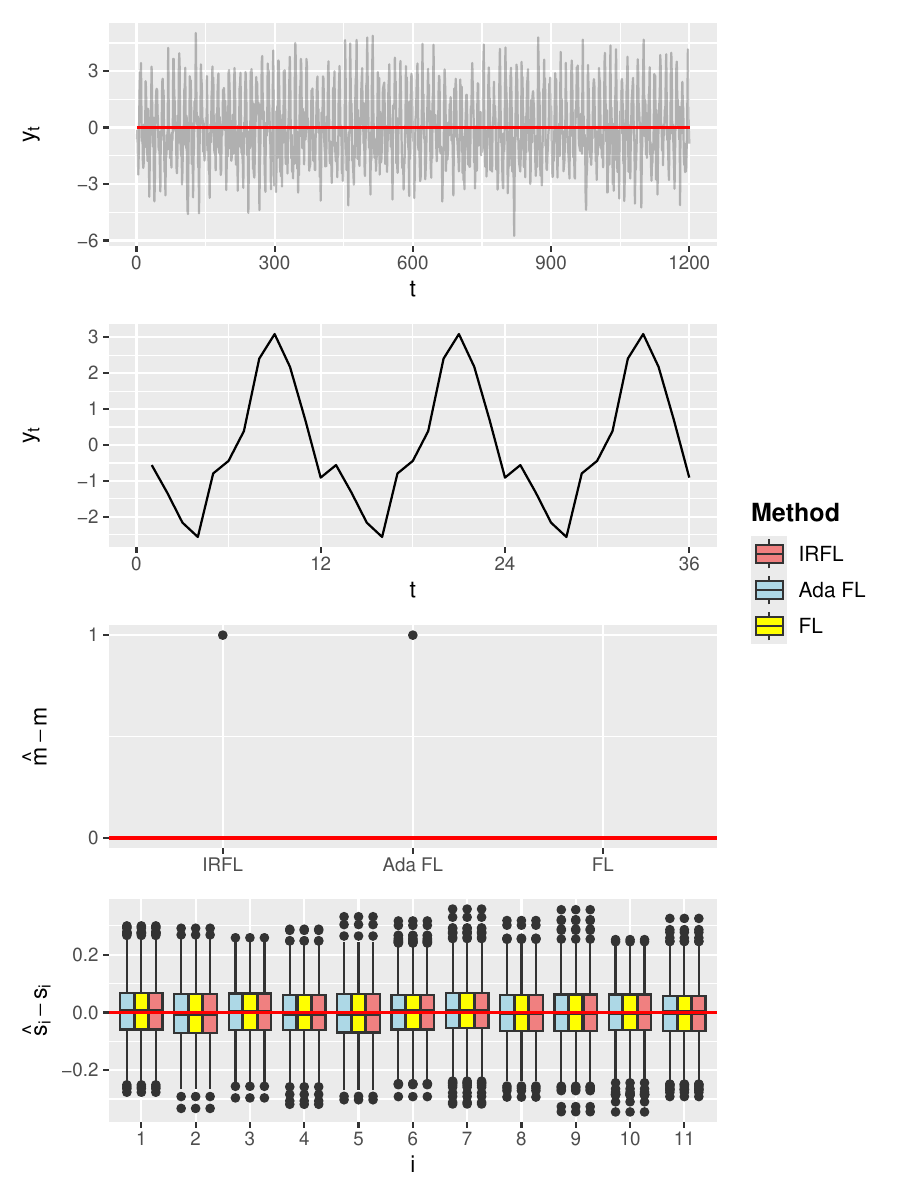}
    \caption[Results of IRFL for global trend + no mean shifts + seasonality estimation]{Top: Boxplot of centered sampling distribution of the number of estimated changepoints, $\hat m$. Middle: Boxplot of centered sampling distribution of the seasonality components $\hat s_i$ for the three methods. Bottom:  Boxplot of centered sampling distribution of trend estimate $\hat\alpha$ for the three methods.}
    \label{fig:Season-Group-Trend-Null}
\end{figure}

\subsubsection*{Scenario 14: Mean Shift Model with Seasonality, Trend and Three Changepoints}
In keeping with the established pattern, the model for the 14th scenario is the same as the 13th in \eqref{full model with trend}, but there are mean shifts present at $t=301, 601, 901$ with magnitude $2$ and alternating sign. The design and structured difference matrices and parameter vector are also unchanged from Scenario 13. 

A representative dataset along with three full periods of a representative seasonality component is given in Figure \ref{fig:Season-Group-Trend-Simulation}.

\begin{figure}[H]
    \centering
            \textbf{Mean Shift With Trend, Seasonality, and Three Changepoints}\\ \includegraphics[width=1\linewidth]{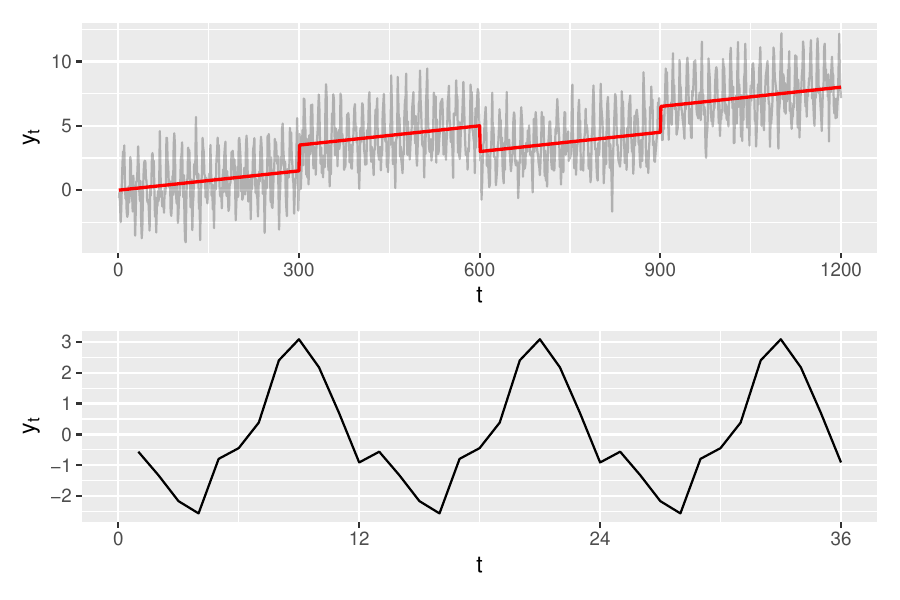}
    \caption[Simulation of IRFL for global trend + three mean shifts + seasonality estimation]{Top: Representative simulation of a mean shift model with global trend, seasonality and three changepoints. Bottom: Three periods of an example season.}
    \label{fig:Season-Group-Trend-Simulation}
\end{figure}

The results in Figure~\ref{fig:Season-Group-Trend-Plots} show that fused lasso, adaptive fused lasso, and IRFL remain unbiased estimators of the seasonal components, while IRFL provides the strongest performance for estimating changepoint locations. In addition, only IRFL yields an unbiased estimate of the global trend parameter. 

The top panel shows that the number of detected changepoints converges toward the true value as the IRFL iterations proceed, reducing the over-segmentation exhibited by the fused lasso at the initial iteration. The second-from-bottom panel confirms that all three methods recover the seasonal coefficients accurately. The bottom panel highlights the key distinction: the fused and adaptive fused lasso produce biased estimates of the trend when mean shifts and seasonality are present, whereas the reweighting steps in IRFL progressively remove this bias. This demonstrates that the iterative refinement in IRFL not only improves changepoint recovery but also yields empirically more accurate for larger $n$ estimation of global trend components.

\begin{figure}[H]
    \centering
            \textbf{Mean Shift With Trend, Seasonality, and Three Changepoints}\\ \includegraphics[width=.9\linewidth]{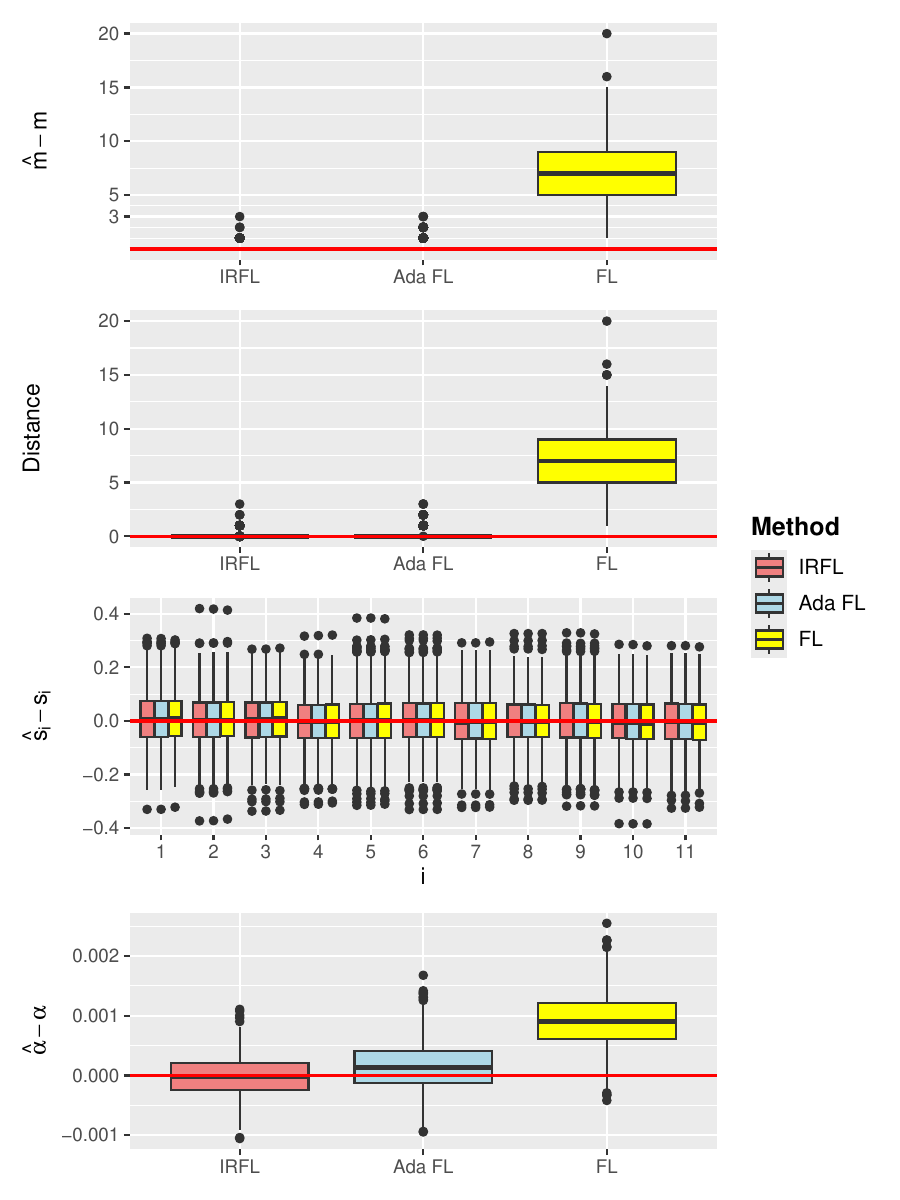}
    \caption[Results of IRFL for global trend + three mean shifts + seasonality estimation]{Top: Boxplot of centered sampling distribution of the number of estimated changepoints, $\hat m$. Second from Top: Boxplot of sampling distribution of the distance from the estimated changepoint vector $\hat{\bm\tau}$ to the true vector $\bm\tau$. Third from Top: Centered sampling distribution of the seasonality components $\hat s_i$ for the three methods. Bottom: Boxplot of centered sampling distribution of trend estimate, $\hat\alpha$.}
    \label{fig:Season-Group-Trend-Plots}
\end{figure}

\subsubsection*{Summary}
At the conclusion of Group 4, the evidence indicates that IRFL provides the most reliable overall performance when seasonality and (optionally) a global linear trend are present alongside mean shifts. In the seasonal–null setting (no mean shifts), fused lasso, adaptive fused lasso, and IRFL all exhibit excellent false‐detection control and yield unbiased estimates of the seasonal components; with an added global trend, all three also estimate the trend accurately in the absence of changepoints. When mean shifts are introduced, fused lasso tends to oversegment, whereas both adaptive fused lasso and IRFL typically recover the correct number and locations of changepoints, with IRFL achieving the closest alignment to the true changepoint vector (\emph{iterate to isolate!}). Moreover, under the joint presence of trend, seasonality, and changepoints, IRFL uniquely remains unbiased for the trend while retaining accurate seasonal estimation and superior localization of changepoints. Collectively, these results demonstrate that IRFL balances parsimony and accuracy under realistic seasonal and trending structures, and it is especially effective at the simultaneous estimation of seasonal patterns, trend, and mean‐shift locations. These findings inform the methodological conclusions that follow and set the stage for their deployment on real-world datasets, an example of which follows immediately in Section 5.

\subsection{Conclusions}
In this section, the Iteratively Reweighted Fused Lasso (IRFL) was introduced as a flexible and unifying framework for addressing changepoint location and estimation problems. By iteratively refining weights within a generalized lasso structure, IRFL applies to a wide range of signal types and underlying structures, offering improved performance in detecting both the number and location of changepoints over its unadapted counterparts. The method accommodates various assumptions about dependence structure and signal complexity, positioning it as a general-purpose tool within the growing changepoint analysis corpus.

Extensive empirical comparisons were conducted against leading methods, including dynamic programming approaches, penalized likelihood criteria, and other heuristic methods. Across these benchmarks, IRFL consistently demonstrated competitive or superior performance in terms of both the number and placement of changepoint locations as well as the estimation of global parameters. Taken together, these results support the value of IRFL as a unifying and effective approach to changepoint problems, with potential for further extension and integration into existing methods. In the remainder of this dissertation, these principles are instantiated on real data: first a long climate record with trend, seasonality, autocorrelation, and regime shifts; then two-dimensional imaging problems where “changepoints” appear as edges.

\section{Applications to Real-world Data}
The empirical results of the previous section motivate two complementary applications. 
First, a univariate time series with known seasonality and suspected structural breaks (Mauna Loa CO$_2$) illustrates how IRFL, combined with prewhitening, jointly estimates trend, curvature, and mean-shift locations under autocorrelation. 

Second, two-dimensional image denoising and segmentation demonstrate that the same regularization logic carries over when changepoints become edges on a grid. 
Across both settings, the emphasis is on faithful structure recovery with transparent model selection.

\subsection{\texorpdfstring{Analysis of Mauna Loa CO\textsubscript{2} Concentration Data}{Analysis of Mauna Loa CO2 Concentration Data}}

The Mauna Loa Observatory in Hawaii provides the longest continuous record of directly measured atmospheric carbon dioxide \citep{noaa_maunaloa_co2_trends}. Initiated in 1958 by Charles David Keeling, this dataset is foundational in climate science and can be found on the National Oceanic and Atmospheric Administration (NOAA) website. It contains monthly CO\textsubscript{2} data for a total of $n=768$ data points over 64 years. A plot of the CO\textsubscript{2} in parts per million (ppm) over time is given in Figure \ref{fig: Mauna Loa Data} below:
\begin{figure}[H]
    \centering
    \textbf{Mauna Loa CO$_2$ Dataset}
    \vspace{.2cm}
    \includegraphics[width=1\textwidth]{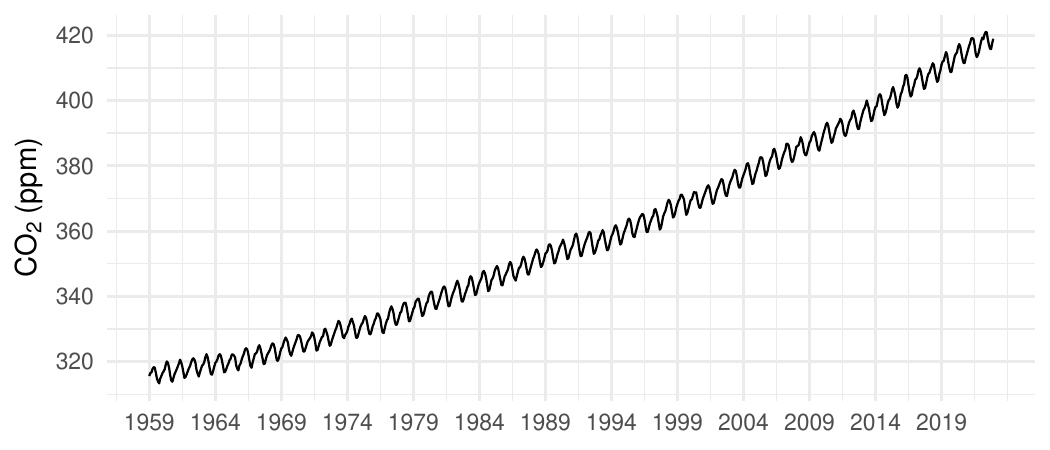}
    \caption[Plot of Mauna Loa Dataset]{Plot of Mauna Loa Dataset}
    \label{fig: Mauna Loa Data}
\end{figure}
Standard analyses, including the NOAA's own analysis \citep{noaa_crvfit}, commonly apply ordinary least squares (OLS) regression, often incorporating a quadratic term or higher to assess acceleration (curvature) and, in higher-frequency settings, seasonal terms \citep{beaulieu-2012-changepoint_CO2, robbins-2020-flexible_changepoint}. That is, earlier analyses fit models resembling
\begin{equation}
    \label{naive model}
    y_t = \mu + \beta_1 t + \beta_2 t^2 + s_t + \epsilon_t \qquad \epsilon_t\sim WN(0,\sigma^2)
\end{equation}
where $\mu$ represents the intercept, $\beta_1$ and $\beta_2$ represent the linear and quadratic trend, and $s_t$ represents the seasonality component with the usual sum-to-zero and periodicity constraints; $\sum_{t=1}^{12}s_t=0$ and $s_t=s_{t+p}$ for $1\le t\le n-12$ where $n$ is the number of observations. However, such modeling efforts do not account for serial correlation or structural breaks—both of which are known to bias trend estimation, and both of which are known to be common in climate data \citep{lund-1995-climatological}.

The climate statistics literature has long cautioned against ignoring serial correlation and structural breaks. It was shown in \citet{lund-2023-practices} that failing to model changepoints in autocorrelated data can result in incorrect inference such as inaccurate trend estimates. Nevertheless, many influential climate trend analyses still lean on naive OLS—downplaying internal variability, serial dependence, and model uncertainty---which can yield misleading inferences \citep{poppick-2016-trends, schotz-2025-rethinking}.

IRFL has the ability to estimate global parameters in the presence of mean shifts (see Section 4), and so it is better qualified to analyze this dataset than more simplistic methods.
Accordingly, we begin our analysis using IRFL under the naive assumption of no serial correlation, and then refine the analysis by incorporating an autoregressive (AR(1)) error structure to capture the temporal dependence explicitly. The model is

\begin{equation}\label{mauna loa model no correlation}
y_t = \begin{cases}
    \mu_1 + \beta_1 t + \beta_2 t^2 + s_t +\epsilon_t & \text{for } 1\le t\le 13\\
    \mu_1 + \beta_1 t + \beta_2 t^2 + s_t +\epsilon_t & \text{for } \tau_0 \leq t \leq \tau_1 - 1 \\
    \mu_2 + \beta_1 t + \beta_2 t^2 + s_t +\epsilon_t& \text{for } \tau_1 \leq t \leq \tau_2 - 1 \\
    \hspace{2cm}\vdots \\
    \mu_m + \beta_1 t + \beta_2 t^2 + s_t +\epsilon_t & \text{for } \tau_m \leq t \leq \tau_{m+1} - 1
    \end{cases}
\end{equation}
where the errors ${\epsilon_t}$ are independent and identically distributed with zero mean and variance $\sigma^2$. To ensure that the number of parameter estimated never exceeds $n=768$, the first 14 data points are ``fused" and cannot admit a changepoint. This ensures that all parameters can be uniquely estimated or ``identified". This is enforced by forcing the first 14 datapoints to share the same $\mu$, $\mu_1$, so that $\tau_0=14$. Likewise, to ensure the model ends at the $n^\text{th}$ data point, $\tau_{m+1}=n+1$ is another boundary condition. To model the seasonality with period $12$, it is required that $s_t=s_{t+12}$ for $1\leq t \leq 768-12$, and for a mean-zero seasonal effect the condition $\sum_{t=1}^{12}s_t=0$ is required. The design matrix $X$, difference matrix $D$ and parameter vector $\mu$ for IRFL are given below:

\begin{figure}[H]
\centering
\vspace{.2cm}
\[
\scalebox{0.85}{$
X = \left[
\begin{array}{ccccccccccccccccc}
1  & 1   & 1 & \cdot & \cdot & \cdot & \cdot & \cdot & \cdot & \cdot & \cdot & \cdot & \cdot & 1 & \cdot & \ldots & \cdot\\
2  & 4   & \cdot & 1 & \cdot & \cdot & \cdot & \cdot & \cdot & \cdot & \cdot & \cdot & \cdot & 1 & \cdot & \ldots & \cdot\\
3  & 9   & \cdot & \cdot & 1 & \cdot & \cdot & \cdot & \cdot & \cdot & \cdot & \cdot & \cdot & 1 & \cdot & \ldots & \cdot\\
4  & 16  & \cdot & \cdot & \cdot & 1 & \cdot & \cdot & \cdot & \cdot & \cdot & \cdot & \cdot & 1 & \cdot & \ldots & \cdot\\
5  & 25  & \cdot & \cdot & \cdot & \cdot & 1 & \cdot & \cdot & \cdot & \cdot & \cdot & \cdot & 1 & \cdot & \ldots & \cdot\\
6  & 36  & \cdot & \cdot & \cdot & \cdot & \cdot & 1 & \cdot & \cdot & \cdot & \cdot & \cdot & 1 & \cdot & \ldots & \cdot\\
7  & 49  & \cdot & \cdot & \cdot & \cdot & \cdot & \cdot & 1 & \cdot & \cdot & \cdot & \cdot & 1 & \cdot & \ldots & \cdot\\
8  & 64  & \cdot & \cdot & \cdot & \cdot & \cdot & \cdot & \cdot & 1 & \cdot & \cdot & \cdot & 1 & \cdot & \ldots & \cdot\\
9  & 81  & \cdot & \cdot & \cdot & \cdot & \cdot & \cdot & \cdot & \cdot & 1 & \cdot & \cdot & 1 & \cdot & \ldots & \cdot\\
10 & 100 & \cdot & \cdot & \cdot & \cdot & \cdot & \cdot & \cdot & \cdot & \cdot & 1 & \cdot & 1 & \cdot & \ldots & \cdot\\
11 & 121 & \cdot & \cdot & \cdot & \cdot & \cdot & \cdot & \cdot & \cdot & \cdot & \cdot & 1 & 1 & \cdot & \ldots & \cdot\\
12 & 144 & -1 & -1 & -1 & -1 & -1 & -1 & -1 & -1 & -1 & -1 & -1 & 1 & \cdot & \ldots & \cdot\\
13 & 169 & 1  & \cdot & \cdot & \cdot & \cdot & \cdot & \cdot & \cdot & \cdot & \cdot & \cdot & 1 & \cdot & \ldots & \cdot\\
14 & 196 & \cdot & 1  & \cdot & \cdot & \cdot & \cdot & \cdot & \cdot & \cdot & \cdot & \cdot & 1 & \cdot & \ldots & \cdot\\
15 & 225 & \cdot & \cdot & 1 & \cdot & \cdot & \cdot & \cdot & \cdot & \cdot & \cdot & \cdot & \cdot & 1 & \ldots & \cdot \\
\vdots & \vdots & \vdots & \vdots & \vdots & \vdots & \vdots & \vdots & \vdots & \vdots & \vdots & \vdots & \vdots & \vdots & \vdots & \ddots & \vdots\\
768 & 768^2 & -1 & -1 & -1 & -1 & -1 & -1 & -1 & -1 & -1 & -1 & -1 & \cdot  & \cdot & \ldots & 1 \\
\end{array}
\right]_{768\times768}
$}
\]
\vspace{.2cm}
\end{figure}

\begin{figure}[H]
\centering
\vspace{.1cm}
\[
\scalebox{0.9}{$
D =
\left[\begin{array}{ccccccccccccc|cccccccc}
\cdot & \cdot & \cdot & \cdot & \cdot & \cdot & \cdot & \cdot & \cdot & \cdot & \cdot & \cdot & \cdot & -1 & 1 & \cdot & \cdot & \cdot & \cdot & \cdot & \cdot \\
\cdot & \cdot & \cdot & \cdot & \cdot & \cdot & \cdot & \cdot & \cdot & \cdot & \cdot & \cdot & \cdot & \cdot & -1 & 1 & \cdot & \cdot & \cdot & \cdot & \cdot \\
\cdot & \cdot & \cdot & \cdot & \cdot & \cdot & \cdot & \cdot & \cdot & \cdot & \cdot & \cdot & \cdot & \cdot & \cdot & -1 & 1 & \cdot & \cdot & \cdot & \cdot \\
   &        &        &        &        &        &    \vdots    &        &        &        &        &        &        &        &        & \ddots & \ddots & \ddots &        &        &        \\
\cdot & \cdot & \cdot & \cdot & \cdot & \cdot & \cdot & \cdot & \cdot & \cdot & \cdot & \cdot & \cdot & \cdot & \cdot & \cdot & \cdot & -1 & 1 & \cdot & \cdot \\
\cdot & \cdot & \cdot & \cdot & \cdot & \cdot & \cdot & \cdot & \cdot & \cdot & \cdot & \cdot & \cdot & \cdot & \cdot & \cdot & \cdot & \cdot & -1 & 1 & \cdot \\
\cdot & \cdot & \cdot & \cdot & \cdot & \cdot & \cdot & \cdot & \cdot & \cdot & \cdot & \cdot & \cdot & \cdot & \cdot & \cdot & \cdot & \cdot & \cdot & -1 & 1 \\
\end{array}\right]_{754\times768}
$}
\]
\end{figure}
\[
\mu =
\left[
\beta_1 \;\; \beta_2 \;\; s_1 \;\; s_2 \;\; \cdots \;\; s_{11} \;\; \mu_1 \;\; \mu_2 \;\; \cdots \;\; \mu_{755}
\right]^\top.
\]
IRFL was run on the above model, and a parameter vector was found under the naive assumption that there was no autocorrelation. However, a test for the significance of an AR(1) parameter was highly significant, indicating the presence of serial dependence in the data. 

Accordingly, the model to be estimated is not \eqref{mauna loa model no correlation} but the following:

\begin{equation}\label{mauna loa model with correlation}
y_t = \begin{cases}
\mu_1 + \beta_1 t + \beta_2 t^2 + s_t +\eta_t & \text{for } t=1\\
    \mu_1 + \beta_1 t + \beta_2 t^2 + s_t + \eta_t, & \text{for } 2\le t\le 13\\
    \mu_1 + \beta_1 t + \beta_2 t^2 + s_t +\eta_t, & \text{for } \tau_0 \leq t \leq \tau_1 - 1 \\
    \mu_2 + \beta_1 t + \beta_2 t^2 + s_t +\eta_t, & \text{for } \tau_1 \leq t \leq \tau_2 - 1 \\
    \hspace{2.25cm}\vdots \\
    \mu_m + \beta_1 t + \beta_2 t^2 + s_t +\eta_t, & \text{for } \tau_m \leq t \leq \tau_{m+1} - 1
    \end{cases}
\end{equation}
where $\eta_t=\phi\eta_{t-1}+\epsilon_t$ for $t\ge2$, $\eta_1=\epsilon_1$ and the $\epsilon_t$ are IID white noise. Because IRFL does not natively accommodate autoregressive error structures outside of the basic mean-shift case, we adopt a prewhitening approach (though it should be said, treating the above as an LS model rather than an IO model similar to Scenarios 4 and 5 from Section 2 in which the AR(1) error structure is estimated by incorporating a $\phi y_{t-1}$ term, would be a valid approach). For each $\phi$ in a grid over $(0,1)$, the data are transformed via a Cholesky factorization:  
\begin{equation*}
    \tilde y_{\phi} = L_\phi y,
\end{equation*}
where $L_\phi$ is the unique lower-triangular Cholesky factor satisfying
\begin{equation*}
    L_\phi L_\phi^{\top} = R_\phi^{-1}.
\end{equation*}
Here, $R_\phi$ denotes the AR(1) correlation matrix with entries
\begin{equation}
\label{AR1_corr}
    [R_\phi]_{ij} = \phi^{|i-j|}, \qquad 1 \leq i,j \leq n.
\end{equation}
This transformation decorrelates the error structure so that subsequent estimation can proceed under approximate independence.

Using the true value of $\phi$, this transformation ``whitens'' the autocorrelated errors, rendering them approximately independent and identically distributed. As a result, changepoint detection can be performed using IRFL on the transformed data \( \tilde y_{\phi} \), under the assumption of uncorrelated residuals. As a technical note, the prewhitening transformation inevitably alters the apparent magnitude and timing of changepoints, since it introduces dependence among adjacent observations. A changepoint corresponds to a discontinuity in the mean, \( \mu_t - \mu_{t-1} \neq 0 \). Under prewhitening, the transformed process satisfies \( \tilde{y}_t = y_t - \phi y_{t-1} \), so a mean shift of size \( \Delta_t = \mu_t - \mu_{t-1} \) becomes \( \tilde{\Delta}_t = \Delta_t - \phi \Delta_{t-1} \) in the whitened series. This transformation scales the jump magnitude by roughly \( 1 - \phi \) and slightly diffuses its effect across neighboring time points, which can blur the boundaries between adjacent regimes. Nevertheless, the discontinuities remain centered near the same indices \( t = \tau_k \), so the changepoints estimated from \( \tilde{y}_{\phi} \) typically occur in close proximity to those in the original data \( y \). In practice, a brief local search around each detected changepoint is sufficient to recover the configuration that minimizes \(\mathrm{BIC}_{\phi}\).

The criterion \(\mathrm{BIC}_{\phi}\) arises from the Gaussian log-likelihood of the model with AR(1) errors. Specifically, when the error covariance is given by \( \sigma^2 R_{\phi} \), the log-likelihood includes an additional term \(\log \det(R_{\phi})\) that penalizes the correlation structure. After prewhitening, the usual sum of squared errors term, \(\text{SSE}_{\phi}\), corresponds to the residual sum of squares computed on the decorrelated data, and \(k\) denotes the effective degrees of freedom associated with the fitted parameters (seasonal, linear, quadratic, and segmental). Combining these components yields

\begin{equation}
\label{BIC AR1}
\text{BIC}_\phi = n \log\left(\frac{\text{SSE}_\phi}{n}\right) + k \log(n) + \log \det(R_\phi),
\end{equation}
where \( R_\phi \) is the AR(1) correlation matrix. The full derivation of \eqref{BIC AR1}, including its connection to the generalized least squares likelihood, is provided in Appendix~\ref{app:bic}.

Using $\tilde y_{\phi}$ in place of $y_t$, we fit the model \eqref{mauna loa model no correlation} via IRFL, with the hope that the whitened series $\{\tilde y_\phi\}$ will yield approximately the same changepoints as the original $\{y_t\}$. For each value of $\phi$ in the grid, IRFL was run across iterations to generate candidate models. Among these, the model with the lowest $\text{BIC}_\phi$ was retained, and the overall optimal model was chosen as the one attaining the lowest $\text{BIC}_\phi$ across all $\phi$. Because this procedure involves a grid search, the effective standard error of $\text{BIC}_\phi$ is inflated, so the results should be regarded as suggestive rather than definitive. Figure \ref{fig: Mauna Loa Regression} gives a plot of the data, de-seasoned trend (the ``signal" in red), and estimated changepoints from the selected model as dashed vertical lines corresponding to the index where a changepoint occurs.

\begin{figure}[H]
    \centering
    \textbf{Mauna Loa CO$_2$ dataset esimated by IRFL}
    \vspace{.2cm}
    \includegraphics[width=1\textwidth]{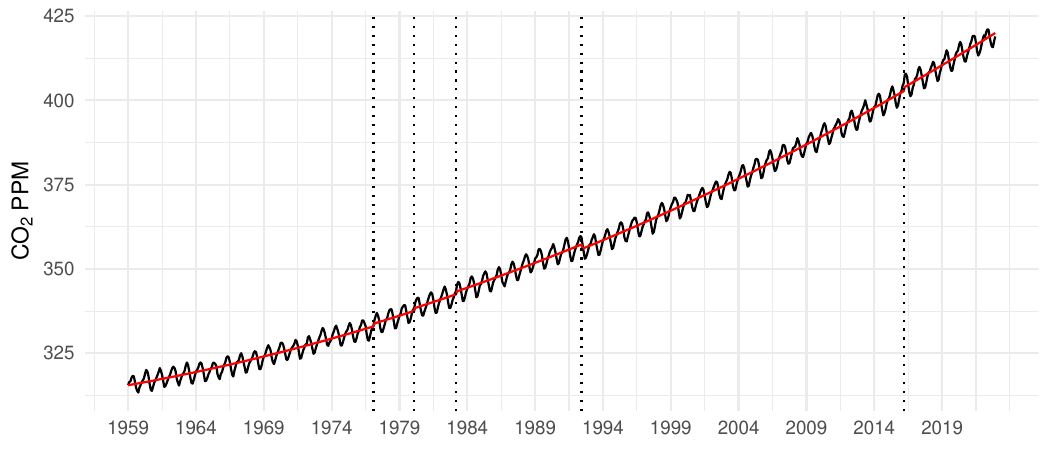}
    \caption[Plot of Mauna Loa Dataset with changepoints]{Plot of Mauna Loa dataset with changepoints and deseasoned regression line in red}
    \label{fig: Mauna Loa Regression}
\end{figure}
\noindent There are five estimated changepoints in the above plot, and all five of them are well-supported in the climatology literature, most notably through the two-way link between warming and CO\textsubscript{2} levels \citep{IPCC-2001-Chapter3}:
\begin{itemize}
    \item \textbf{February 1977} – A heating event which coincides with the widely studied 1976/77 Pacific climate regime shift, likely driven by El Ni\~no conditions \citep{hare-2000-empirical,miller-1994-climate}.
    \item \textbf{February 1980} – A heating event associated with Autumn 1979 to Winter 1979/1980 El Ni\~no \citep{li-2019-elNino1979_80}.
    \item \textbf{March 1983} – A heating event following one of the strongest El Ni\~no events in recorded history \citep{hoerling-1997-ENSO}.
    \item \textbf{June 1992} – A cooling event which occurs one year after the June 1991 eruption of Mount Pinatubo; 1992–1994 is widely regarded as a volcanic disturbance interval \citep{hansen-1992-pinatubo}.
    \item \textbf{March 2016} – A heating event aligned with a major El Ni\~no event, among the strongest on record \citep{lheureux-2017-ENSO}.  
\end{itemize}
To ensure that the reported configuration corresponds to a true local minimum of \(\mathrm{BIC}_{\phi}\), a grid search was conducted in the neighborhood of each estimated changepoint. Specifically, the indices of all detected changepoints were permuted by \(\pm 12\) months in all possible combinations, and the model was refitted for each configuration. None of these perturbations yielded a lower \(\mathrm{BIC}_{\phi}\), confirming that the identified changepoints represent the locally optimal solution.

Having identified the optimal changepoint configuration and corresponding AR(1) model, the next step is to evaluate the statistical significance and robustness of the estimated trend components. In particular, we seek to determine whether the observed acceleration in CO\textsubscript{2} concentration is statistically significant and to quantify the uncertainty in the estimated rate of change at both the beginning and end of the record. To this end, the following subsection details the inferential procedures used to test model parameters and construct confidence intervals under the fitted AR(1) error structure.

\subsubsection*{Methods}

The purpose of this subsection is to detail the inferential procedures used to quantify uncertainty in the estimated model components. In particular, we assess the statistical significance of the linear and quadratic trend terms and derive confidence intervals for the instantaneous rate of change at both the beginning and end of the record. These derivations and their implementation under an AR(1) error model are presented below, followed by the corresponding parameter estimates and inference results in the next subsection.

Both the estimation of changepoint locations in the above grid search and subsequent parameter estimation and inference were performed using the \texttt{gls()} function from the \texttt{nlme} package \citep{pinheiro-bates-nlme} in \texttt{R}, which implements generalized least squares under a user-specified error correlation structure. In this case, the residual correlation was modeled as AR(1) with parameter $\phi$, estimated jointly with $\beta$ and $\sigma^2$ by maximizing the Gaussian likelihood. For each fixed-effect coefficient, \texttt{gls()} reports a Wald $t$-statistic
\[
t_j = \frac{\hat{\beta}_j}{\mathrm{SE}(\hat{\beta}_j)},
\]
where $\mathrm{SE}(\hat{\beta}_j)$ is obtained from the estimated covariance matrix
\[
\widehat{\mathrm{Var}}(\hat{\beta}^{\mathrm{GLS}}) = \hat{\sigma}^2 (X^\top \hat{R}_\phi^{-1} X)^{-1}.
\]
For the slope ($j=1$) and quadratic acceleration ($j=2$) terms, the null and alternative hypotheses tested are
\[
H_0: \beta_j = 0, \qquad H_A: \beta_j \neq 0.
\]
The AR(1) parameter \(\phi\) was first selected by minimizing \(\mathrm{BIC}_{\phi}\) across a grid of candidate values, as described above. Conditional on the selected value, \(\phi\) was then re-estimated jointly with the regression coefficients \(\beta\) and variance \(\sigma^2\) by maximizing the Gaussian likelihood within the \texttt{gls()} framework, holding the changepoint configuration fixed at that obtained under the selected \(\phi\). The standard error of \(\hat{\phi}\) was computed from the inverse observed information matrix (the negative inverse Hessian evaluated at the joint MLE) for all model parameters.

 The test for serial correlation used
\[
H_0: \phi = 0, \qquad H_A: \phi \neq 0.
\]
The $95\%$ confidence interval for $\phi$ is then $\hat{\phi} \pm 1.96\,\mathrm{SE}(\hat{\phi})$ under the assumption of asymptotic normality of the MLE. Because the time index $t$ in \eqref{naive model} is in months, the deseasoned trend is $m(t) = \mu + \beta_1 t + \beta_2 t^2$, giving an instantaneous rate of change with respect to months $m'(t) = \beta_1 + 2\beta_2 t$. The yearly rate is obtained by multiplying by 12:
\[
r(t) = 12\big(\beta_1 + 2\beta_2 t\big).
\]
Let $\hat b = (\hat{\beta}_1,\hat{\beta}_2)^\top$ and let $\widehat{\mathrm{Var}}(\hat b)$ be the $2\times2$ submatrix of $\widehat{\mathrm{Var}}(\hat{\beta}^{\mathrm{GLS}})$ corresponding to $\beta_1,\beta_2$. The gradient vector of $r(t)$ with respect to $(\beta_1,\beta_2)$ is $g(t) = [\,12,\,24t\,]^\top$, and by the variance of a linear combination of random variables the standard error of $\hat r(t)$ is
\[
\mathrm{SE}\{\hat r(t)\} = \sqrt{\,g(t)^\top\,\widehat{\mathrm{Var}}(\hat b)\,g(t)\,} = 12\,\sqrt{\widehat{\mathrm{Var}}(\hat{\beta}_1) + 4t^2\,\widehat{\mathrm{Var}}(\hat{\beta}_2) + 4t\,\widehat{\mathrm{Cov}}(\hat{\beta}_1,\hat{\beta}_2)}.
\]
A $95\%$ confidence interval for the yearly rate is then $\hat r(t) \pm 1.96\,\mathrm{SE}\{\hat r(t)\}$.

\subsubsection*{Results}
The estimated AR(1) parameter was highly significant, $\hat{\phi} = 0.765 \pm 0.054$ (95\% CI). The estimated linear trend was $0.710$ ppm CO$_2$/year, and the estimated quadratic coefficient was $0.0134$ ppm/year$^2$, both with $p \ll 10^{-16}$. This corresponds to a yearly rate of $0.7101 \pm 0.0353$ ppm/year at the start of the dataset and $2.5432 \pm 0.0319$ ppm/year in 2025, closely matching NOAA’s decadal averages of approximately 0.7 ppm/year in the 1960s and 2.4 ppm/year from 2010–2020 \citep{noaa_maunaloa_co2_trends}.

\subsubsection*{Conclusions}

The changepoints detected by IRFL coincide with well-documented climatic events such as El Ni\~no episodes and volcanic disturbances. By explicitly modeling these structural breaks, the analysis avoids misattributing abrupt shifts in the mean level to long-run acceleration, thereby ensuring that the estimated curvature in the trend reflects genuine acceleration rather than regime artifacts.  

Three broader conclusions follow. First, naive OLS regression—even when augmented with seasonal and quadratic terms—cannot adequately capture the statistical realities of climate data, which typically exhibit both autocorrelation and structural breaks. Second, the IRFL framework, extended with prewhitening to account for AR(1) dependence, provides a principled means of jointly estimating seasonality, piecewise means, and long-run trend. Third, the resulting acceleration estimates are not only statistically sound but also consistent with NOAA’s established benchmarks, strengthening confidence in their interpretation.  

In sum, this segmented and autocorrelated modeling approach delivers both improved interpretability and more faithful inference. It demonstrates that rigorous treatment of climate time series structure can reconcile trend and acceleration estimates with consensus findings while offering a more transparent account of the mechanisms driving observed changes in atmospheric CO$_2$.  

The Mauna Loa analysis shows IRFL’s capacity to separate long-run trend from regime shifts in the presence of seasonality and autocorrelation. An analogous separation problem arises in two dimensional image analysis: edges (discontinuities) must be preserved while within-region noise is suppressed. The next subsection recasts one-dimensional changepoints as edges on a lattice and applies the same reweighting logic to total-variation regularization for images.

\subsection{IRFL on Image Denoising and Segmentation}

While previous results have focused on estimating changepoint models in univariate time series, the same underlying principle—that structure can be recovered by penalizing differences across neighboring values—extends naturally to two-dimensional images and even higher-dimensional datasets. In this context, the goal is to recover a piecewise smooth (or constant) image while preserving sharp transitions that correspond to edges. In this paradigm, the observed image is first vectorized. Suppose the image is an \( R \times C \) grayscale matrix ($R$ rows and $C$ columns), where each entry corresponds to the intensity of a pixel. It is represented as a vector \( y \in \mathbb{R}^{RC} \) by stacking the columns of the image matrix on top of one another. This is known as column-major order, where the first \( R \) entries of \( y \) correspond to the first column of the image, the next \( R \) entries to the second column, and so on. Formally, if the original image is denoted \( Y = [Y_{rc}] \) with \( r = 1, \dots, R \) and \( c = 1, \dots, C \), then the vectorized form \( y \) is defined by
\begin{equation}\label{column major order}
y_{(c-1)R + r} = Y_{rc}.
\end{equation}
Because $c$ is the column index and each column has $R$ rows, $y_{(c-1)R}$ always corresponds to the last pixel of the $c^\text{th}$ column of the image. Therefore, to proceed to the $r^\text{th}$ row of the $(c+1)^\text{st}$ column, one must progress $r$ more indices, giving \eqref{column major order}.

Representing the image in column-major vectorized form allows us to express the total variation penalty. In this case, it is an \( \ell_1 \) norm applied to the differences between adjacent pixels in the generalized lasso framework. This is achieved by introducing a structured difference matrix \( D  \), where each row corresponds to either a horizontal or vertical difference between neighboring pixels. Because there are $R(C-1)$ horizontal differences and $C(R-1)$ vertical differences, and because $y$ is an $RC\times1$ vector, $D\in \mathbb{R}^{(R(C-1) + C(R-1)) \times RC}$. When multiplied by the image vector \( \hat\mu \), the product \( D\hat\mu \) yields a vector of all first-order horizontal and vertical differences in the image. 

It is perhaps easiest to get a sense of this $D$ matrix using an example. Let $Y$ be a $4$ pixel $\times 4$ pixel image, and let $y$ be the vector resulting from vectorizing $Y$ in column-major order. There are $4 \times 3 = 12$ vertical differences and $4 \times 3 = 12$ horizontal differences, for a total of 24 rows and 16 columns in $D$. The $D$ matrix can be written as
\[
D = \left[
\begin{array}{c}
D_{{\text{vert}}} \\
\hline
D_{{\text{horiz}}}
\end{array}
\right]
\]
where $D_{\text{vert}}$ and $D_{\text{horiz}}$ refer to sub-matrices whose rows correspond to vertical and horizontal differences of $y$, respectively. Below is the complete $D$ matrix composed of both vertical and horizontal difference blocks:
\begin{equation}\label{Dmatrix}
D = \left[
\begin{array}{cccccccccccccccc}
-1 &  1 &  \cdot & \cdot & \cdot & \cdot & \cdot & \cdot & \cdot & \cdot & \cdot & \cdot & \cdot & \cdot & \cdot & \cdot \\
\cdot & -1 &  1 &  \cdot & \cdot & \cdot & \cdot & \cdot & \cdot & \cdot & \cdot & \cdot & \cdot & \cdot & \cdot & \cdot \\
\cdot & \cdot & -1 &  1 &  \cdot & \cdot & \cdot & \cdot & \cdot & \cdot & \cdot & \cdot & \cdot & \cdot & \cdot & \cdot \\
\cdot & \cdot & \cdot & \cdot & -1 &  1 &  \cdot & \cdot & \cdot & \cdot & \cdot & \cdot & \cdot & \cdot & \cdot & \cdot \\
\cdot & \cdot & \cdot & \cdot & \cdot & -1 &  1 &  \cdot & \cdot & \cdot & \cdot & \cdot & \cdot & \cdot & \cdot & \cdot \\
\cdot & \cdot & \cdot & \cdot & \cdot & \cdot & -1 &  1 &  \cdot & \cdot & \cdot & \cdot & \cdot & \cdot & \cdot & \cdot \\
\cdot & \cdot & \cdot & \cdot & \cdot & \cdot & \cdot & \cdot & -1 &  1 &  \cdot & \cdot & \cdot & \cdot & \cdot & \cdot \\
\cdot & \cdot & \cdot & \cdot & \cdot & \cdot & \cdot & \cdot & \cdot & -1 &  1 &  \cdot & \cdot & \cdot & \cdot & \cdot \\
\cdot & \cdot & \cdot & \cdot & \cdot & \cdot & \cdot & \cdot & \cdot & \cdot & -1 &  1 &  \cdot & \cdot & \cdot & \cdot \\
\cdot & \cdot & \cdot & \cdot & \cdot & \cdot & \cdot & \cdot & \cdot & \cdot & \cdot & \cdot & -1 &  1 &  \cdot & \cdot \\
\cdot & \cdot & \cdot & \cdot & \cdot & \cdot & \cdot & \cdot & \cdot & \cdot & \cdot & \cdot & \cdot & -1 &  1 & \cdot \\
\cdot & \cdot & \cdot & \cdot & \cdot & \cdot & \cdot & \cdot & \cdot & \cdot & \cdot & \cdot & \cdot & \cdot & -1 & 1 \\
\hline
-1 & \cdot & \cdot & \cdot & 1 & \cdot & \cdot & \cdot & \cdot & \cdot & \cdot & \cdot & \cdot & \cdot & \cdot & \cdot \\
\cdot & -1 & \cdot & \cdot & \cdot & 1 & \cdot & \cdot & \cdot & \cdot & \cdot & \cdot & \cdot & \cdot & \cdot & \cdot \\
\cdot & \cdot & -1 & \cdot & \cdot & \cdot & 1 & \cdot & \cdot & \cdot & \cdot & \cdot & \cdot & \cdot & \cdot & \cdot \\
\cdot & \cdot & \cdot & -1 & \cdot & \cdot & \cdot & 1 & \cdot & \cdot & \cdot & \cdot & \cdot & \cdot & \cdot & \cdot \\
\cdot & \cdot & \cdot & \cdot & -1 & \cdot & \cdot & \cdot & 1 & \cdot & \cdot & \cdot & \cdot & \cdot & \cdot & \cdot \\
\cdot & \cdot & \cdot & \cdot & \cdot & -1 & \cdot & \cdot & \cdot & 1 & \cdot & \cdot & \cdot & \cdot & \cdot & \cdot \\
\cdot & \cdot & \cdot & \cdot & \cdot & \cdot & -1 & \cdot & \cdot & \cdot & 1 & \cdot & \cdot & \cdot & \cdot & \cdot \\
\cdot & \cdot & \cdot & \cdot & \cdot & \cdot & \cdot & -1 & \cdot & \cdot & \cdot & 1 & \cdot & \cdot & \cdot & \cdot \\
\cdot & \cdot & \cdot & \cdot & \cdot & \cdot & \cdot & \cdot & -1 & \cdot & \cdot & \cdot & 1 & \cdot & \cdot & \cdot \\
\cdot & \cdot & \cdot & \cdot & \cdot & \cdot & \cdot & \cdot & \cdot & -1 & \cdot & \cdot & \cdot & 1 & \cdot & \cdot \\
\cdot & \cdot & \cdot & \cdot & \cdot & \cdot & \cdot & \cdot & \cdot & \cdot & -1 & \cdot & \cdot & \cdot & 1 & \cdot \\
\cdot & \cdot & \cdot & \cdot & \cdot & \cdot & \cdot & \cdot & \cdot & \cdot & \cdot & -1 & \cdot & \cdot & \cdot & 1
\end{array}
\right].
\end{equation}
The first three rows of $D$ correspond to the vertical differences in the first column of the image (differences between the pixels of the image in the first column, rows $1$--$2$, $2$--$3$, and $3$--$4$). The next three rows correspond to the vertical differences in the second column, and so on across all four columns. 

In contrast, the rows of $D_{\text{horiz}}$ encode horizontal differences. Because the image is vectorized in column-major order, moving from a pixel to its horizontal neighbor (same row, next column) requires advancing by the number of rows---in this case, $4$. Thus each row of $D_{\text{horiz}}$ places a $-1$ and $+1$ exactly four indices apart, creating the offset pattern visible in the lower half of $D$.

Using \eqref{Dmatrix} and vectorized $y$, one can find 
\begin{equation}\label{genlasso2d}
    \hat{\mu} = \underset{\mu\in\mathbbm{R}^{RC}}{\text{arg min }} \left\{\|y-\mu\|_2^2 + \lambda\|D\mu\|_1\right\}
\end{equation}
where $\hat{\mu}$ represents a vectorized estimate of the denoised image. This is just the familiar genlasso, and therefore the IRFL can be used to enforce sharp edges and piecewise constant segments through iteration.

\subsubsection*{Adapting IRFL to the Images}

In univariate time series, IRFL efficiently computes a full regularization path using the generalized lasso framework, iteratively reweighting the $\ell_1$ penalty to recover changepoint locations and enforce piecewise constancy. Extending this principle to images, such as $128 \times 128$ images (16,384 pixels), introduces substantial computational challenges: the difference matrix $D$ has tens of thousands of rows, and computing an entire solution path for all $\lambda$ values becomes prohibitive. As a practical concession, we solve \eqref{genlasso2d} for individual values of $\lambda$ using the total variation alternating direction method of multipliers (TV-ADMM) \citep{boyd-2011-admm}, detailed in Appendix \ref{app:pseudocode}. The 2D IRFL proceeds by combining this TV-ADMM solver with iterative reweighting: at iteration $i$, the objective is
\[
\hat\mu^{(i)} = \underset{\mu \in \mathbb{R}^{RC}}{\text{arg min }} \left\{ \frac{1}{2} \|y - \mu\|_2^2 + \lambda \sum_{j=1}^{p} \hat w_j^{(i)} |(D\mu)_j| \right\},
\]
where $p = R(C-1) + C(R-1)$ is the number of edges in the 2D grid and the weights are updated as
\[
\hat w_j^{(i)} = \frac{1}{|(D \hat\mu^{(i-1)})_j| + \varepsilon_n}, \qquad \text{small }\varepsilon_n > 0.
\]
As $(D\hat\mu^{(i-1)})_j$ represents a vertical or horizontal pixel difference, the weight $\hat w_j^{(i)}$ directly reflects the local contrast. When $(D\hat\mu^{(i-1)})_j \approx 0$, the corresponding weight becomes large $(\hat w_j^{(i)} \approx \varepsilon_n^{-1})$, so the next iteration strongly penalizes any deviation from equality across that edge, enforcing local smoothness. Conversely, when $(D\hat\mu^{(i-1)})_j$ is large, the associated weight is small, reducing the penalty and allowing a distinct vertical or horizontal boundary to persist.

By \eqref{0-1 convergence}, the iteratively reweighted penalty converges to a 0/1 penalty on changes in $\mu$:
\[
\sum_{j=1}^{p} \hat w_j^{(i)} |(D\mu)_j|
\;\longrightarrow\;
\sum_{j=1}^{p} \mathbf{1}\{(D\mu)_j \neq 0\}.
\]
In one dimension, $\mathbf{1}\{(D\mu)_j \neq 0\}$ identifies changepoints. In two dimensions, each index $j$ corresponds to a horizontal or vertical adjacency in the image grid, so $\mathbf{1}\{(D\mu)_j \neq 0\}$ flags a local jump across that grid edge. Summing over $j$ therefore measures the \emph{total length of all boundaries} between constant regions, i.e., the discrete perimeter of the segmentation. Since the penalty drives $(D\mu)_j = 0$ within each region, the resulting estimate is piecewise constant with sharp transitions only where supported by the data, yielding a characteristic ``stained glass window'' structure.

To compute these estimates in practice, the TV-ADMM solver is applied across a small grid of regularization values,
\[
\lambda_1 > \lambda_2 > \dots > \lambda_L,
\]
selected to balance edge detection and smoothness (for grayscale images with intensities in $[0,1]$, it is usually enough to set $0\le\lambda_\ell\le0.5$). For each $\lambda_\ell$, IRFL is iterated to convergence, yielding a weighted TV-ADMM solution at every iteration. A standard implementation uses $L = 5$ values spaced between $\lambda_{\max}$ and $\lambda_{\min}$.

Once the solution set \( \{ \mu_{\lambda_\ell} \}_{\ell=1}^L \) has been computed, model selection can be performed by visual inspection.  The full 2D IRFL algorithm is summarized in Appendix \ref{app:pseudocode}.

\subsubsection*{2D IRFL on Grayscale Images}
To demonstrate the effectiveness of IRFL for denoising and edge detection, we tested the 2D IRFL algorithm on several images. Figure \ref{fig:man with camera} features the well-known ``man with a camera" image, commonly used in image processing libraries such as scikit-image. This example illustrates the algorithm’s ability to preserve sharp edges while significantly reducing noise. The second example in Figure \ref{fig:S image} uses the ``S" image from \citet{tibshirani-2011-solutiongenlasso}. In both cases, a grid of images is presented to visualize the denoising progression across IRFL iterations. The first column is the original image and the second is the image with injected white noise. The first iteration of IRFL is done without any adaptive reweighting, and so its representation in Figures \ref{fig:man with camera} and \ref{fig:S image} in the third column is denoted ``FL", for the fused lasso. In a similar way, the second iteration of IRFL is the adafused lasso and so the fourth column is labelled ``Ada FL". The rows correspond to increasing $\lambda$.
\begin{figure}[H]
    \centering
    \includegraphics[width=1\textwidth]{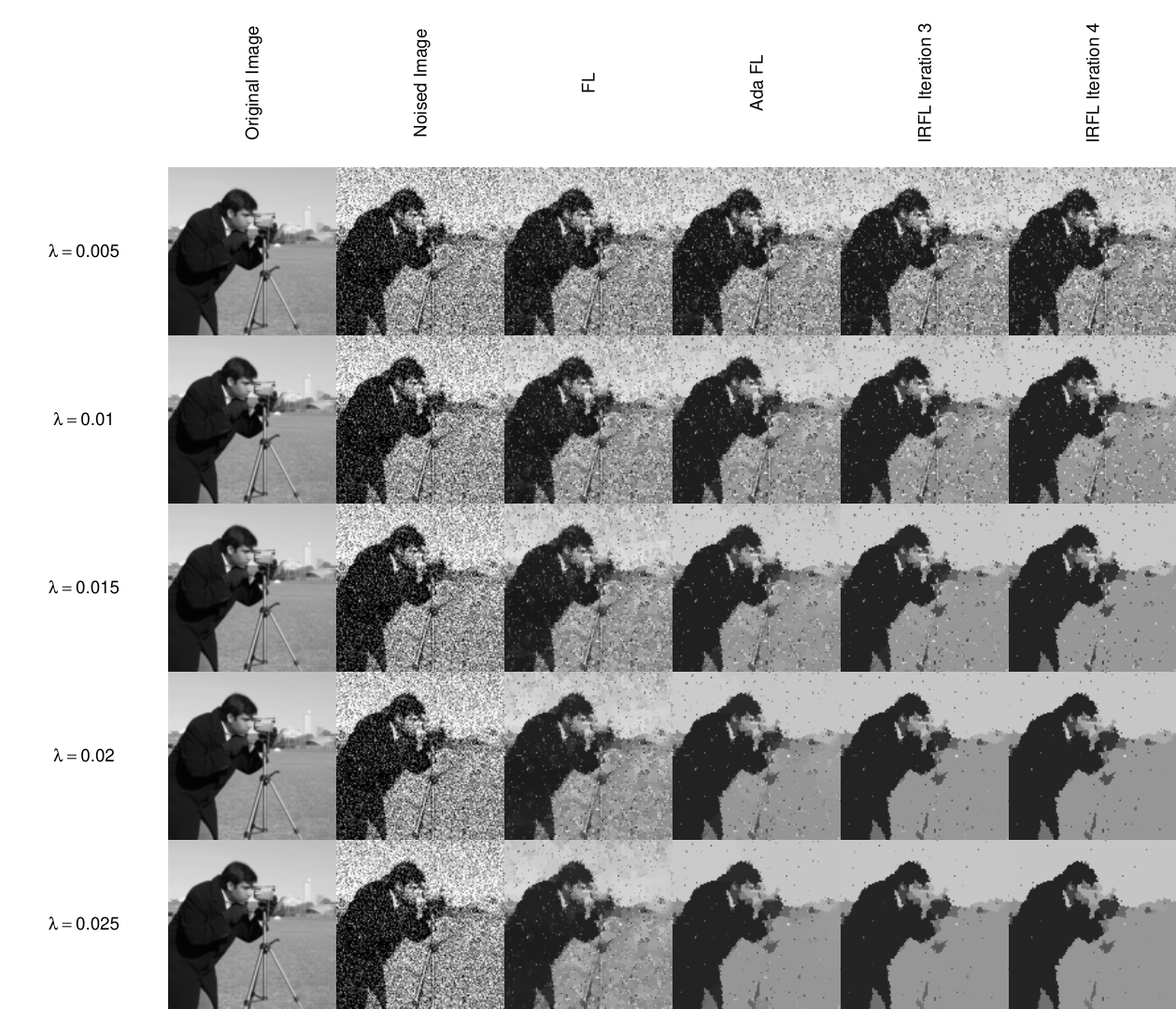}
    \caption[Man with a camera denoised using IRFL]{Columns from left to right: Original image, noised image, FL, Ada FL, IRFL iterations 3 and 4. Rows from top to bottom: $\lambda=0.005$, $0.01$, $0.015$, $0.02$, $0.025$.}
    \label{fig:man with camera}
\end{figure}
One can see that for various values of $\lambda$, there is a clear tradeoff between the amount of detail one wishes to preserve and the amount of noise one wishes to remove. Regularization increases both as $\lambda$ increases (down the grid) and as the number of IRFL iterations increases (left to right on the grid, beginning with the third column from the left). Small values of $\lambda$ are ineffective at removing noise (see, for example, the panel corresponding to $\lambda = 0.005$ at IRFL iteration 4), whereas values of $\lambda$ that are too large tend to remove excessive detail (see, for example, the panel corresponding to $\lambda = 0.025$ at IRFL iteration 4). The precise choice of $\lambda$ that best balances this tradeoff between detail preservation and noise reduction is ultimately a judgment call for the user.

IRFL can detect edges. Figure \ref{fig:man with camera edges} shows the noised image on the left and the same image with the red edges—identified by the fourth IRFL iteration with $\lambda = 0.025$ (using the nonzero entries of $D\hat{\mu}^{(4)}$).
\begin{figure}[H]
    \centering
    \textbf{Detected Edges on ``Man with a Camera" image using IRFL}
    \includegraphics[width=1\textwidth]{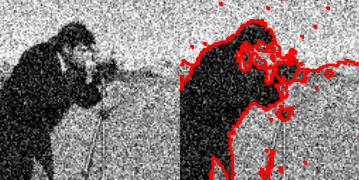}
    \caption[Edges detected on ``Man with A Camera" denoised using IRFL]{From left to right: noised image, and noised image with edges detected by IRFL with $\lambda=0.025$.}
    \label{fig:man with camera edges}
\end{figure}

\subsubsection*{2D IRFL on Color Images}
As was noted in \citet{tibshirani-2011-solutiongenlasso}, the genlasso can be used to estimate segmentwise constant images which are in color by encoding a color scale. An example in \citet{tibshirani-2011-solutiongenlasso} was given denoising an ``S" image, contained below. Following this example, the colors green, blue, purple, and red correspond to the numeric values 1, 2, 3, and 4. The image was then vectorized and IRFL was run as if it were black and white. The results are shown in Figure \ref{fig:S image}.
\begin{figure}[H]
    \centering
    \includegraphics[width=.8\textwidth]{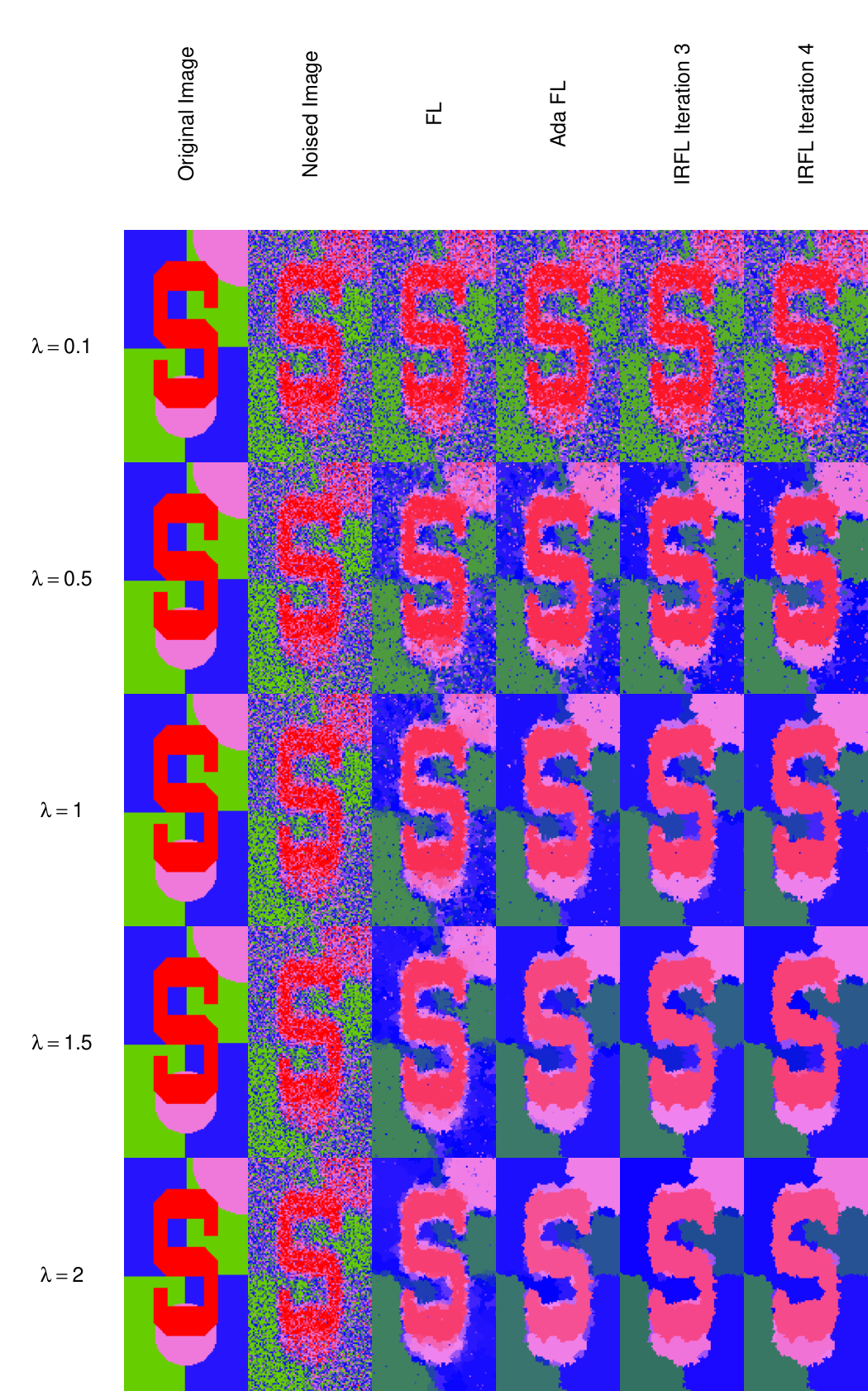}
    \caption[Denoised ``S" image for various $\lambda$, iterations]{Columns from left to right: Original image, noised image, FL, Ada FL, IRFL iterations 3 and 4. Rows from top to bottom: $\lambda=0.1$, $0.5$, $1$, $1.5$, $2$.}
    \label{fig:S image}
\end{figure}
As in the previous grid, larger values of $\lambda$ correspond to harsher regularization, and in such cases the tradeoff between detail and noise reduction is more strongly in favor of noise reduction. 

It seems that perhaps the fourth iteration of IRFL with $\lambda=1$ does the best job of segmenting the image while simultaneously preserving the underlying ``S" shape. To visualize the edges detected by this value of $\lambda$ on the fourth iteration of IRFL, Figure \ref{fig:S edges} includes from left to right the original ``S" image, the noised ``S" image, and the noised ``S" image with edges found by the fourth iteration of IRFL with $\lambda=1$ superimposed.
\begin{figure}[H]
    \centering
    \textbf{Edges Detected on ``S" image using IRFL}
    \vspace{.2cm}
    \includegraphics[width=1\textwidth]{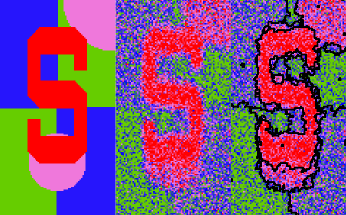}
    \caption[Edges detected on ``S" image denoised using IRFL]{From left to right: Original Image, Noised Image, and Noised Image with edges as detected by IRFL with $\lambda=1$.}
    \label{fig:S edges}
\end{figure}

\subsubsection*{Applications to Biomedical Imaging}
Image denoising is now widely recognized as an important pre-processing step in medical imaging pipelines \citep{kaur-2021-review}. This importance arises because magnetic resonance imaging (MRI), despite its central role in medical diagnostics and treatment planning, is frequently corrupted by acquisition noise, intensity inhomogeneity, and low tissue contrast. Such imperfections not only impair visual interpretation but also degrade the performance of machine learning systems, which rely on clean inputs to achieve clinically acceptable results \citep{song-2021-mri}.

The consequences are particularly acute in the task of image segmentation—partitioning an image into anatomically or pathologically meaningful regions. Segmentation underpins essential clinical applications such as tumor delineation, organ boundary detection, and lesion tracking, which in turn guide radiation therapy planning, surgical navigation, and disease monitoring \citep{chupetlovska-2025-esr}. Yet, when noise corrupts MRI data, traditional segmentation algorithms often become unstable and error-prone. For this reason, effective denoising methods are indispensable, especially those that preserve edge structure: excessive smoothing may obscure lesions, distort anatomical boundaries, or alter radiomic feature distributions \citep{paton-2021-mr}.

The 2D IRFL method helps address these concerns. The $\ell_1$ penalty and data fidelity term used together penalty on within-segment errors in (\ref{genlasso2d}) enforce boundaries between disparate regions and denoise segments within those regions, respectively. The result is a denoised image that maintains the fidelity of clinical features—especially edges.

We applied the IRFL technique to brain MRI scans, encompassing both healthy and pathological cases, to demonstrate its effectiveness in this context. Figure \ref{fig:MRI Y} illustrates the performance of IRFL on a brain MRI scan exhibiting a pathological feature. The original MRI scan of the brain with a cancerous lesion is on the left, while the image on the right has the same image, with red edges superimposed as found by IRFL.

\begin{figure}[H]
    \centering
    \textbf{Edges Detected on MRI using IRFL: Cancer Present}
    \vspace{.2cm}
    \includegraphics[width=1\textwidth]{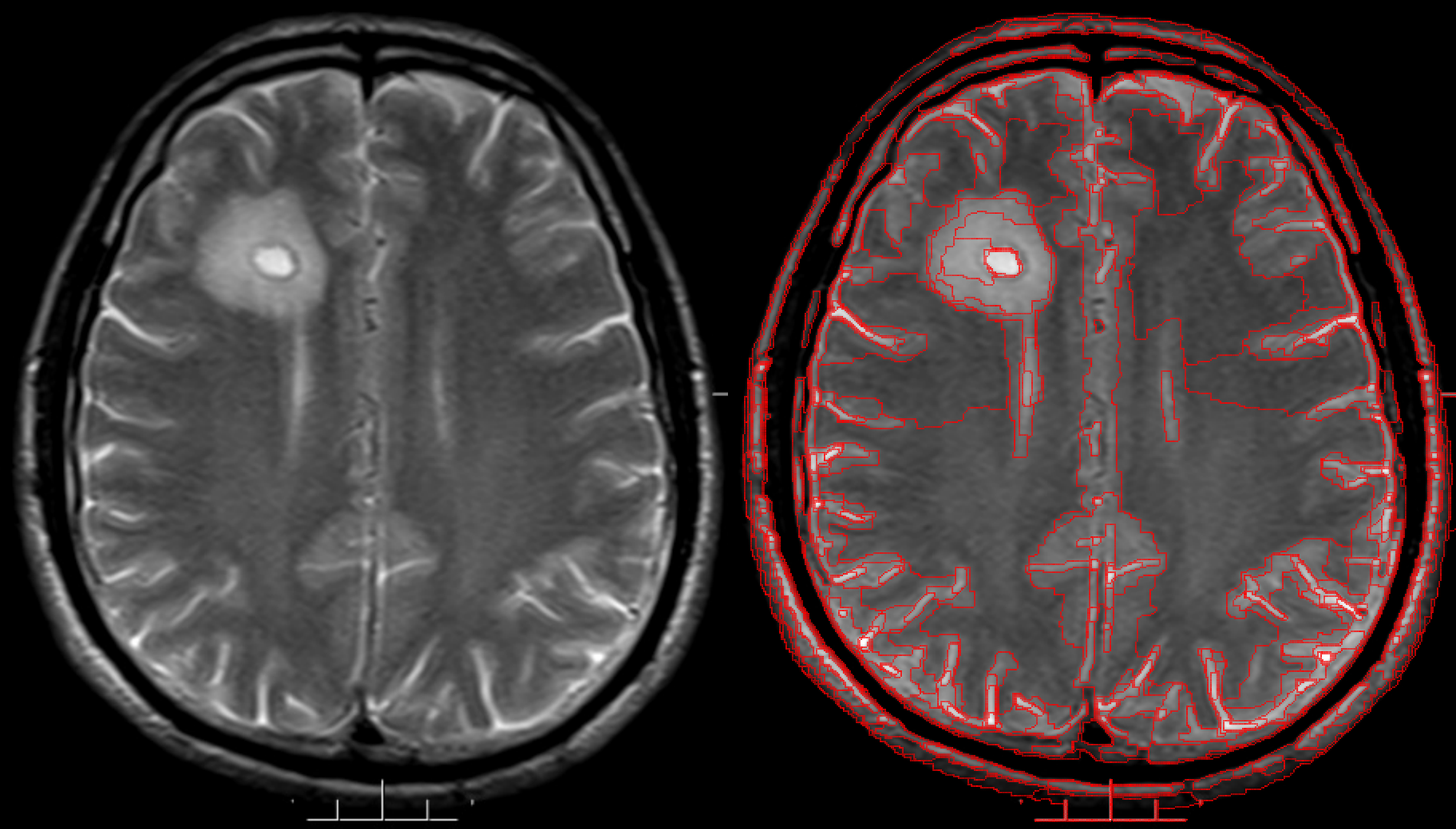}
    \caption[Edges detected on Brain MRI with cancer]{From left to right: MRI image, MRI image with edges detected by IRFL}
    \label{fig:MRI Y}
\end{figure}

\begin{figure}[H]
    \centering
    \textbf{Edges Detected on MRI using IRFL: Cancer Not Present}
    \vspace{.2cm}
    \includegraphics[width=1\textwidth]{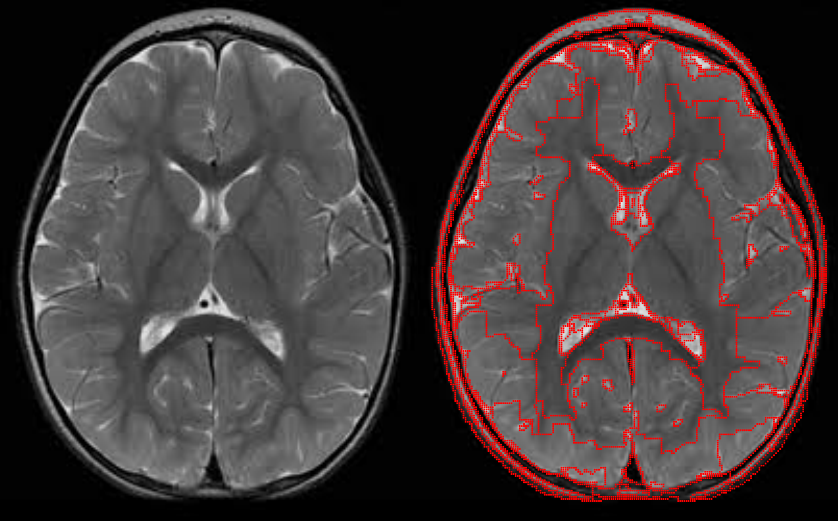}
    \caption[Edges detected on Brain MRI with cancer]{From left to right: MRI image, MRI image with edges detected by IRFL}
    \label{fig:MRI N}
\end{figure}
Taken together, the pathological and non-pathological cases highlight the robustness of IRFL: it suppresses noise and corrects intensity irregularity while preserving the fine-scale anatomical detail necessary for clinical interpretation. This property is especially important in medical imaging workflows, where excessive smoothing can obscure diagnostically relevant structure, and over-segmentation can introduce false boundaries. 

\vspace{0.5em}

\paragraph{Comparison to Other Approaches}
The problem of segmentation and identification in medical imaging is a well-studied field, and there are a number of cutting-edge methods already in wide use. Among these are convolutional neural networks (CNNs) and related deep learning models. However, their deployment in clinical settings is limited by several practical considerations, including infrastructural requirements, sensitivity to training data, and challenges in interpretability. By contrast, the IRFL framework provides a principled alternative that better aligns with the transparency and reproducibility standards expected in clinical imaging analysis.

Specifically, the proposed approach differs from CNN-based pipelines in the following ways:

\begin{itemize}\setlength{\itemsep}{0.35em}
    \item \textbf{Interpretability.} The inference procedure is explicit and produces boundaries that can be directly examined and justified, rather than relying on latent, non-linear feature abstractions.
    \item \textbf{Model Dependence.} The method is grounded in a well-defined statistical model, reducing dependence on training distributions that may not generalize across scanners, institutions, or patient populations.
    \item \textbf{Lightweight and Efficient.} IRFL runs efficiently on standard hardware without requiring GPUs or large-scale computational infrastructure.
    \item \textbf{No Training Data Required.} The method avoids the need for curated, annotated datasets, which are often labor-intensive to generate and can encode institutional or demographic biases.
    \item \textbf{Auditable and Transparent.} Each step of the procedure can be reviewed, validated, and documented, supporting regulatory compliance and clinical trust.
\end{itemize}

These characteristics make IRFL particularly suitable for clinical environments where reproducibility, interpretability, and clear decision pathways are prioritized. At the same time, the framework is flexible: a pre-trained CNN may still be incorporated as an upstream feature extractor if desired, with IRFL providing the final inference. This allows integration of richer feature representations while preserving transparency in the final decision-making stage.

The boundaries recovered by IRFL are not merely visually sharper; they align with core clinical tasks that depend on reliable and reproducible delineation of anatomical structures:

\begin{itemize}\setlength{\itemsep}{0.35em}
  \item \textbf{Clinical measurement.} Precise boundary delineation enables reproducible quantification of tumor size, shape, and margins, which is critical for staging and treatment planning.
  \item \textbf{Longitudinal monitoring.} Stable edge recovery allows objective comparison over time (e.g., lesion growth, shrinkage, or displacement), supporting evidence-based assessment of therapeutic response.
  \item \textbf{Robustness in practice.} IRFL suppresses noise and corrects intensity inhomogeneity while preserving anatomical fidelity, reducing inter-reader variability.
  \item \textbf{Pipeline compatibility.} The resulting clean, structurally coherent images improve the performance of downstream segmentation algorithms and AI-based diagnostic systems.
\end{itemize}
Taken together, these properties position IRFL as both scientifically rigorous and practically deployable within medical imaging workflows.

A similar need arises in other imaging domains where fine structural detail must be recovered from noisy observations. For example, in plant phenotyping, large-scale root imaging experiments generate high-dimensional data that are difficult to process reliably. Traditional image processing pipelines often struggle to separate root structures from background noise, and manual annotation is labor-intensive, subjective, and not scalable \citep{Dadi-2025-RootEx, Handy-2024-RootAnnotationVariation}. As in medical imaging, robust preprocessing methods are essential for producing clean and interpretable inputs for segmentation and quantitative trait extraction.

IRFL directly addresses this challenge. By suppressing irrelevant variation while preserving boundary information, the method produces a piecewise-smooth representation of root architecture that retains biologically meaningful structure. This facilitates accurate and reproducible downstream analyses, including binary masking, skeletonization, and morphological feature extraction. In machine learning workflows, such denoised representations improve the stability and reliability of thresholding and segmentation steps, particularly in heterogeneous or low-contrast imaging conditions.

As an illustration, Figure \ref{fig:Root Hairs} shows an example image found in \citet{pietrzyk-2025-dirtmu_dataset}, the supplemental materials to \citet{pietrzyk-2024-dirtmu}. The image on the left is an image of a root with its many root hairs. On the right, the boundaries of the root hairs have been cleanly delineated using IRFL and marked in red:

\begin{figure}[H]
    \centering
    \textbf{Root Hairs Detected using IRFL}
    \vspace{.2cm}
    \includegraphics[width=1\textwidth]{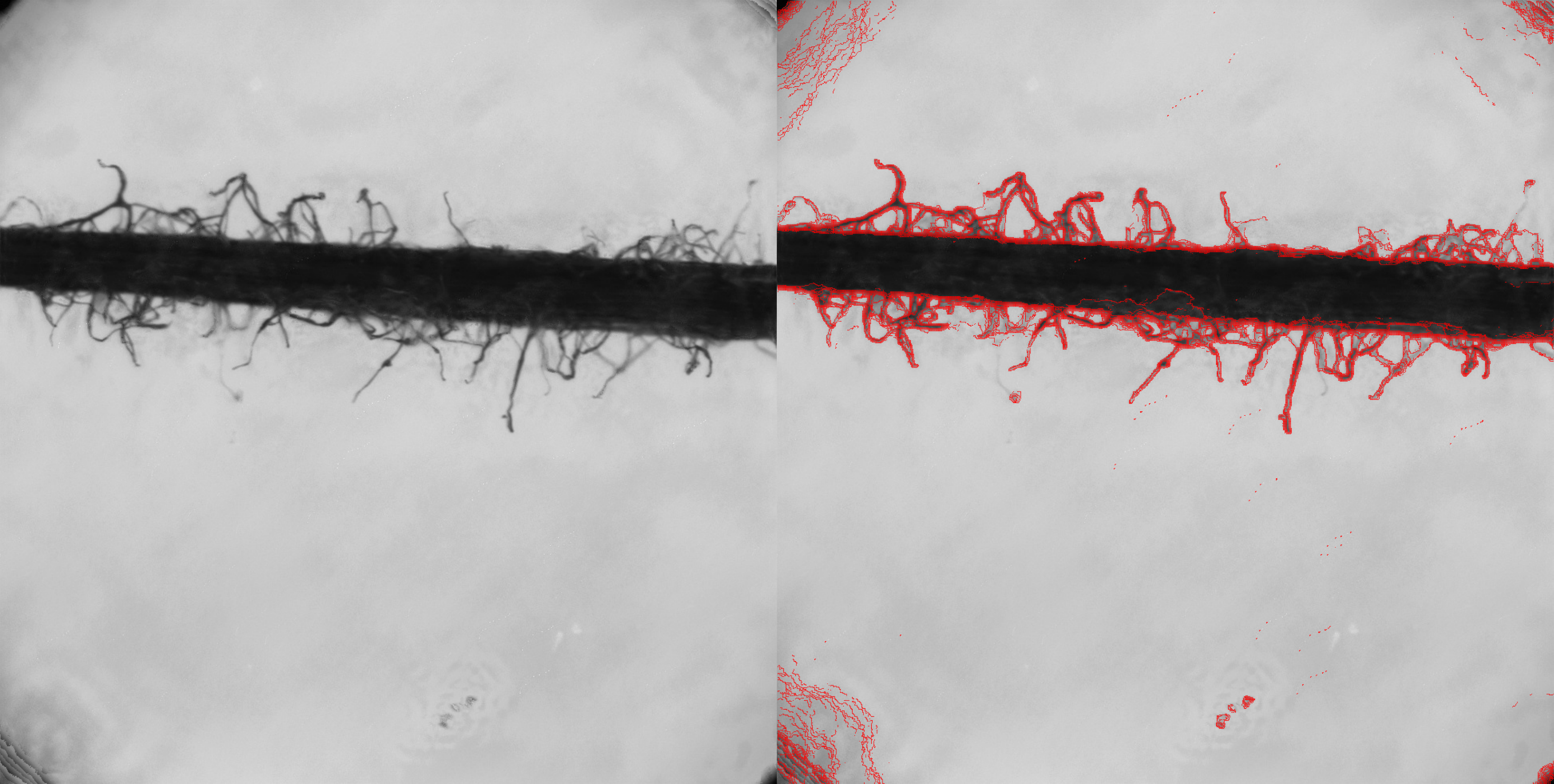}
    \caption[Root Hairs detected using IRFL]{From left to right: Root image, root image with edges detected by IRFL}
    \label{fig:Root Hairs}
\end{figure}
\subsubsection*{Conclusions}
The Iteratively Reweighted Fused Lasso (IRFL) demonstrates potential value as a preprocessing step in machine learning pipelines, particularly for applications such as early cancer detection where interpretability and geometric fidelity are paramount. It may also serve a complementary postprocessing role, aiding in the localization and quantification of tumors once identified.

Beyond biomedical imaging, IRFL has potential applicability across domains as varied as medical diagnostics and agricultural phenotyping, operating successfully in both grayscale and color contexts. Whether segmenting malignant lesions or quantifying fine-scale root hairs, it consistently preserves structural boundaries while suppressing noise—showing that precision, adaptability, and interpretability can be achieved simultaneously rather than traded off.

Ultimately, IRFL is more than a specialized algorithm. It embodies a principled framework for inference under uncertainty: recovering signal where it exists, eliminating spurious variation where it does not, and doing so in a manner that respects the underlying geometry of the data. This unifying perspective highlights that changepoint detection and image segmentation are not separate problems, but instances of a broader challenge—model selection over structure—where carefully designed regularization can reliably recover meaningful patterns even in delicate and noisy settings.


%
%
\subsection*{Data statement}
All simulations of IRFL used in this dissertation may be found at \url{https://github.com/michael-grantham/IRFL}.  The Mauna Loa CO\textsubscript{2} dataset can be located at \url{https://gml.noaa.gov/data/dataset.php?item=mlo-co2-observatory-monthly}.

\newpage
\section{Appendix}
This appendix collects supplementary material referenced in earlier chapters. Section \ref{app:crit} provides critical values for the CUSUM test (Chapter 2.1). Section \ref{IO LS} contains a derivation of the difference between an LS and IO model. Section \ref{app:pseudocode} gives pseudocode for heuristic and exact segmentation algorithms (Chapter 2). Section \ref{Genlasso soln path} outlines supporting derivations for the genlasso solution path (Chapter 3). Section \ref{app:tvadmm} presents the TV-ADMM algorithm (Chapter 5). Section \ref{app:bic} contains the derivation for BIC in AR(1) models.
\subsection{Asymptotic Critical Values of CUSUM}\label{app:crit}
\vspace{1cm}
\begin{table}[H]\label{BB Crit Values}
\begin{center}
\label{tab:CUSUM-crit-values}
\begin{tabular}{ lc l l l l l }
\hline
 & 90\% & 95\% & 97.5\% & 99\% & 99.9\%  \\
\hline
 Critical Values &1.224 & 1.358 & 1.480 & 1.628 & 1.949\\
\hline
\end{tabular}
\caption{Asymptotic critical values for CUSUM test statistics}
\end{center}
\end{table}
These values correspond to the limiting distribution of the scaled CUSUM statistic \eqref{CUSUM-scaled} under the null hypothesis, derived from the supremum of a Brownian Bridge (see Chapter 2.1).

\subsection{Geometric Adjustment of Mean Shifts in AR(1) Models} \label{IO LS}

In the WCM section of Chapter 2, a rough equivalence between the Level Shift (LS) and Innovation Outlier (IO) models was demonstrated in \eqref{Innovation Outlier}. However, there are key differences between the two models that warrant explanation. The LS model can be written
\begin{equation}\label{LS model}
y_t = \phi y_{t-1} + \mu_t + \epsilon_t, \qquad |\phi| < 1,
\end{equation}
where the autoregressive coefficient $\phi$ induces serial dependence within each stationary segment.  
Expanding this recursion gives
\begin{align}
y_t
&= \phi(\phi y_{t-2} + \mu_{t-1} + \epsilon_{t-1}) + \mu_t + \epsilon_t \nonumber\\
\vspace{1cm}&= \phi^2 y_{t-2} + \phi \mu_{t-1} + \mu_t + \phi \epsilon_{t-1} + \epsilon_t \nonumber\\
&= \sum_{j=0}^{\infty} \phi^{j} \mu_{t-j} + \sum_{j=0}^{\infty} \phi^{j} \epsilon_{t-j}.
\label{eq:LS-sum}
\end{align}
The first term in \eqref{eq:LS-sum} represents the weighted contribution of past mean levels, while the second term represents the accumulated autoregressive effect of past errors. In the absence of mean shifts (so that $\mu_{t-j}=\mu_{t-j-1}$ for $j\ge0$), \eqref{eq:LS-sum} simplifies to 
\begin{align}
y_t
&= \sum_{j=0}^{\infty} \phi^{j} \mu_{t-j} + \sum_{j=0}^{\infty} \phi^{j} \epsilon_{t-j}\nonumber\\[4pt]
&=\mu_t\sum_{j=0}^{\infty} \phi^{j} + \sum_{j=0}^{\infty} \phi^{j} \epsilon_{t-j}\nonumber\\[4pt]
&= \frac{\mu_t}{1-\phi} + \sum_{j=0}^{\infty} \phi^{j} \epsilon_{t-j}\nonumber \\[4pt]
&=\mu_t^\prime + \sum_{j=0}^{\infty} \phi^{j} \epsilon_{t-j}
\label{eq:LS-sum-simplified}
\end{align}
where $\mu_t^\prime := \mu_t(1-\phi)^{-1}$.

The IO model can be expressed so that its expected means are comparable to those of the LS model.  Under serial dependence, the expected mean of the process satisfies $\mathbb{E}[y_t] = \mu_t / (1 - \phi) = \mu_t^\prime$.  
Using this notation, the IO model can be written
\begin{equation}
y_t = \mu_t' + \eta_t, \qquad \eta_t = \phi \eta_{t-1} + \epsilon_t,
\end{equation}
where $\epsilon_t$ is white noise with mean zero and variance $\sigma^2$.  
Expanding $\eta_t$ recursively yields
\begin{align}
\eta_t
&= \phi(\phi \eta_{t-2} + \epsilon_{t-1}) + \epsilon_t
  = \phi^2 \eta_{t-2} + \phi \epsilon_{t-1} + \epsilon_t \nonumber\\
&= \sum_{j=0}^{\infty} \phi^{j} \epsilon_{t-j},
\end{align}
so that
\begin{equation}
y_t = \mu_t' + \sum_{j=0}^{\infty} \phi^{j} \epsilon_{t-j}.
\label{eq:IO-sum}
\end{equation}

Comparing \eqref{eq:LS-sum-simplified} and \eqref{eq:IO-sum} shows that the two models share the same autoregressive innovation structure. Asymptotically, in the absence of changepoints, they are pathwise identical.

However, in the presence of mean shifts, they differ. To see the difference explicitly, consider a single changepoint at $t = \tau_1$ where $\mu_t$ jumps from $\mu_1$ to $\mu_2$.  
For the first $k$ points in the new segment ($t = \tau_1 + k - 1$, $k \ge 1$), the mean sequence satisfies $\mu_{t-j} = \mu_2$ for $j \le k-1$ and $\mu_{t-j} = \mu_1$ for $j \ge k$.  
Substituting this pattern into the first term of \eqref{eq:LS-sum},
\begin{align}
\sum_{j=0}^{\infty} \phi^{j} \mu_{t-j}
&= \sum_{j=0}^{k-1} \phi^{j} \mu_2 + \sum_{j=k}^{\infty} \phi^{j} \mu_1 \nonumber\\
&= \mu_2 \frac{1 - \phi^{k}}{1 - \phi}
   + \mu_1 \frac{\phi^{k}}{1 - \phi} \nonumber\\
&= \mu_2' - \phi^{\,k} (\mu_2' - \mu_1').
\label{eq:mean-sum-eq}
\end{align}
Substituting \eqref{eq:mean-sum-eq} back into \eqref{eq:LS-sum} gives the LS model in its adjusted form:
\begin{equation}
y_{\tau_1 + k - 1}
  = \mu_2' - \phi^{\,k} (\mu_2' - \mu_1')
    + \sum_{j=0}^{\infty} \phi^{j} \epsilon_{\tau_1 + k - 1 - j}, \qquad k \ge 1.
\label{eq:LS-final}
\end{equation}
By comparison, the IO model in \eqref{eq:IO-sum} has
\begin{equation}
y_{\tau_1 + k - 1}
  = \mu_2' + \sum_{j=0}^{\infty} \phi^{j} \epsilon_{\tau_1 + k - 1 - j},\qquad k\ge1.
  \label{eq:IO-final}
\end{equation}

Comparing \eqref{eq:LS-final} and \eqref{eq:IO-final} shows that the LS model differs from the IO model only by the additional geometric term
$-\phi^{\,k}(\mu_2' - \mu_1')$, which decays to zero as $k$ increases because $|\phi|<1$  \eqref{LS model}.  
Thus, the LS model can be viewed as an IO model with an additional transient component reflecting the gradual geometric adjustment of the mean following each changepoint.
\subsection{Pseudo Code}\label{app:pseudocode}
Included for reproducibility and comparison, these pseudo code listings formalize algorithms discussed informally in Chapter 2.
\begin{algorithm}[H]
\caption{Binary Segmentation (BS)}
\begin{algorithmic}[1]
\State \textbf{Input:} A time series $\{y_t\}_{t=1}^n$; start value $s$; end value $e$;   AMOC test of choice with test statistic $\mathcal{C}^b_{s,e}$ for $s<b<e$ and critical value $\zeta_n$
\vspace{5pt}
\State Set $\mathcal{C}_{s,e}=\underset{s<b<e}{\text{max }}|\mathcal{C}^b_{s,e}|$\vspace{5pt}
    \State Set $k=\underset{s<b<e}{\text{arg max }}\mathcal{C}^b_{s,e}$\vspace{5pt}
    \If{$\mathcal{C}_{s,e} > \zeta_n$}
        \State Add $k$ to the set of estimated change-points
        \State Recursively run BS on both $\{y_t\}_{t=s}^{k-1}$ and $\{y_t\}_{t=k}^{e}$
    \Else
        \State \textbf{STOP}
    \EndIf
\end{algorithmic}
\end{algorithm}

\begin{algorithm}[H]
\caption{Wild Binary Segmentation (WBS)}
\begin{algorithmic}[1]
\State \textbf{Input:} A time series $\{y_t\}_{t=1}^n$; start value $s$; end value $e$; threshold $\zeta_n$
\State Set $\mathcal{M}_{s,e} =$ the set of indices $m$ for which $[s_m,e_m] \in F_n^M$ and $[s_m,e_m]\subseteq[s,e]$.
\State Augment $\mathcal{M}_{s,e}$ with $\{0\}$, where $[s_0,e_0] = [s,e]$.
\vspace{5pt}
\State Set $(k,m_0) = \underset{m\in\mathcal{M}_{s,e},\hspace{.15cm} s_m<b< e_m}{\text{arg max}}|\mathcal{C}_{s_m,e_m}^b|$
\vspace{5pt}
    \If{$|\mathcal{C}_{s_{m_0},e_{m_0}}^k| > \zeta_n$}
    \vspace{5pt}
        \State Add $k$ to the set of estimated change-points
        \State Recursively run WBS on both $\{y_t\}_{t=s}^{k}$ and $\{y_t\}_{t=k}^{e}$
    \Else
        \State \textbf{STOP}
    \EndIf
\end{algorithmic}
\end{algorithm}

There are also derivations in \citet{fryzlewicz-2014-WBS} which give guidance on the optimal values to choose for $M$ and $\zeta_n$. As $M$ controls the probability of capturing a small changepoint through the random selection of subintervals, it should be chosen to be maximally large without incurring a prohibitive time cost. Anecdotally, it is reported that a value of $M=5000$ for data of length $n\approx5\times10^{3}$ runs in seconds on modern hardware with minimal variability among the detected changepoints from run to run. There is also a derivation of the threshold value, $\zeta_n$, which shows it must be proportional to $\sqrt{2\log(n)}$ (a common threshold in wavelet thresholding literature). Rather than derive the exact constant of proportionality $C$ for a given cost function, a large-scale simulation study was conducted to determine the value of $C$ that minimizes the absolute deviation between the number of detected changepoints and the true number of changepoints.
Rather than derive the exact constant of proportionality $C$ for a given cost function, a large-scale simulation study was conducted to determine the value of $C$ that minimizes the absolute deviation between the number of detected changepoints and the true number of changepoints. In this way, it was determined that an appropriate value of $C$ is merely $\hat\sigma$, where $\hat\sigma$ is the Median Absolute Deviation estimator of $Var(\epsilon_t)$ from \citet{hampel-1974-MAD}. That is, the default value of $\zeta_n$ is given as follows:
\begin{equation*}
    \zeta_n = \hat\sigma\sqrt{2\log(n)}.
\end{equation*}In addition to thresholding with $\zeta_n$, one can also use a BIC stopping criterion based on the number of changepoints located.

\begin{algorithm}[H]
\caption{Wild Contrast Maximization (WCM)}
\begin{algorithmic}[1]
\State \textbf{Input:} Time series $\{y_t\}_{t=1}^{n}$, the number of intervals $M$, the number of possible models $R$ to consider.
\vspace{5pt}\State \textbf{Phase 1: Gappy Sequence Generation}
\vspace{5pt}\State Generate up to $M$ intervals, whether randomly or deterministically, contained in $[1,n]$. If it's not possible to generate $M$ distinct intervals because $e-s$ is too small, then generate as many as possible. Call this number of intervals $\Tilde{M}$, and call the collection of these intervals $\mathcal{M}_{s,e} = \{(s_m,e_m):1\leq m\leq\Tilde{M}\}$
    \State Identify the interval $(s_m,e_m)\in\mathcal{M}_{s,e}$ and changepoint $k_0$ that maximizes $\big|\mathcal{C}_{s_m,e_m}^{k_0}\big|$.
    \State Split the time series at $k_0$ and recursively search the two segments in like manner to locate up to $n-1$ changepoints and corresponding CUSUM scores.
    \vspace{5pt}\State \textbf{Phase 2: Model Selection}
    \vspace{5pt}\State Choose the models corresponding to the $R$ largest gaps in BIC scores of the nested models, ending with a null model (no changepoints). Call these models $\Theta=\{\Theta_0,\Theta_1,\ldots,\Theta_R\}$.
    \State Find the order of serial dependence $\hat p$ of $\Theta_R$ and estimate of autocorrelation parameter(s) $\hat{{\phi}}_R$ which minimizes BIC using OLS.
    \State Over all the segments which introduce $q\geq1$ new changepoints in $\Theta_R$ which are not present in the same segment in $\Theta_{R-1}$, compare the BIC score \textit{with} the new changepoint(s) to the BIC of the same segment \textit{without} such. In the one case, the BIC score of a segment starting at time $s$ and ending at time $e$ is \begin{equation*}
        \text{BIC}^\text{with}_{s,e,\hat\phi} = \frac{e-s+1}{2}\log(\text{SSE}_{q,\hat{p}}/(e-s+1)) + (q+\hat p)\log(e-s+1)
    \end{equation*} and the BIC score of the same segment with no changepoints is \begin{equation*}
        \text{BIC}^\text{without}_{s,e,\hat\phi} = \frac{e-s+1}{2}\log(\text{SSE}_{0,\hat{p}}/(e-s+1)) + (\hat p)\log(e-s+1).
    \end{equation*}  Critically, the estimation of $\phi$ used in both models is estimated in the model \textit{with changepoints}. If any of the segments have a lower BIC without the new changepoint(s), prefer model $\Theta_{R-1}$. Otherwise, retain $\Theta_R$.
    \State Continue recursively backward-selecting until either a local minimum BIC is achieved or the null model is selected.
\end{algorithmic}
\end{algorithm}

Note that in Step 9, the BIC comparisons are done over the segment, and $\hat\phi$ is estimated \textit{with} changepoints. This is done deliberately because very often comparisons of raw BIC score will result in the estimation of models without changepoints--opting instead for a stronger $\phi$--when in fact changepoints are present. This is also confirmed in simulation studies (see Chapter 4). As a technical note, WCM as presented in \citet{cho-2024-WCM} estimates models with AR(p) errors, to which the pseudocode given above can be adapted with minimal and intuitive modifications.

\begin{algorithm}[H]
\begin{algorithmic}[1]
\caption{Narrowest-Over-Threshold (NOT)}
\State \textbf{Input:} A time series $\{y_t\}_{t=1}^n$; start value $s$; end value $e$; threshold $\zeta_n$
\State Set $\mathcal{M}_{s,e} =$ the set of indices $m$ for which $[s_m,e_m] \in F_n^M$ and $[s_m,e_m]\subseteq[s,e]$.
\State Augment $\mathcal{M}_{s,e}$ with $\{0\}$, where $[s_0,e_0] = [s,e]$.
\vspace{5pt}
\State Set $\mathcal{O}_{[s,e]}=\{m\in\mathcal{M}_{s,e}:\underset{ s_m<b< e_m}{\text{max }}|\mathcal{C}_{s_m,e_m}^b|>\zeta_n\}$\vspace{5pt}
\State Set $m_0 = \underset{m\in\mathcal{O}_{[s,e]}}{\text{arg min }}|e_m-s_m|$
\vspace{5pt}
\State Set $k_0 = \underset{ b\in(s_{m_0},e_{m_0})}{\text{arg max }}|\mathcal{C}_{s_{m_0},e_{m_0}}^b|$
\vspace{5pt}
        \State Add $k_0$ to the set of estimated change-points
        \State Recursively run NOT on both $\{y_t\}_{t=s}^{k-1}$ and $\{y_t\}_{t=k}^{e}$
        \State \textbf{STOP}
    \end{algorithmic}
\end{algorithm}

\begin{algorithm}[H]
\caption{Optimal Partition (OP)}
\begin{algorithmic}[1]
\State \textbf{Input:} A time series $\{y_t\}_{t=1}^n$; a cost function $\mathcal{C}(\cdot)$
\State Initialize: $F(1) = -\beta$, $\hat{\tau}(1) = \text{NULL}$
\For{$s = 2,3, \dots, n+1$}
    \State Set $F(s) = \underset{1 \leq r < s}{\text{min }} \big(F(r) + \mathcal{C}(y_{r:(s-1)}) + \beta\big)$
    \State Set $t = \underset{1 \leq r < s}{\text{arg min }} \big(F(r) + \mathcal{C}(y_{r:(s-1)}) + \beta\big)$
    \State Set $\hat{\tau}(s)$ equal to the set of optimal changepoints prior to time $t$ (which is $\hat{\tau}(t)$), union the new best changepoint prior to $s$ (which is $t$).
\EndFor
\State \Return $\hat\tau({n+1})$
\end{algorithmic}
\end{algorithm}

\begin{algorithm}[H]
\caption{Pruned Exact Linear Time (PELT)}
\begin{algorithmic}[1]
\State \textbf{Input:} A time series $\{y_t\}_{t=1}^n$; a cost function $\mathcal{C}(\cdot)$; constant $K$ satisfying (\ref{criterion}).
\State Initialize: $F(1) = -\beta$, $\hat\tau(1) = \text{NULL}$, $R_1=\{1\}$
\For{$s = 2,3, \dots, n+1$}
    \State Set $F(s) = \underset{r\in R_s}{\text{min }} \big(F(r) + \mathcal{C}(y_{r:(s-1)}) + \beta\big)$
    \State Set $t = \underset{r\in R_s}{\text{arg min }} \big(F(r) + \mathcal{C}(y_{r:(s-1)}) + \beta\big)$
    \State Set $\hat\tau(s)$ equal to the set of optimal changepoints at time $t$ (which is stored in $\hat\tau(t)$), union the new best changepoint prior to $s$ (which is $t$).
    \State Set $R_{s+1}$ equal to the set of $r\in R_s$ which satisfy (\ref{criterion}), along with $s$. That is, $R_{s+1} = \big\{r \in R_{s} : F(r) + \mathcal{C}(y_{r:(s-1)}) + K < F(s)\big\} \cup \{s\}$
\EndFor
\State \Return $\hat\tau({n+1})$
\end{algorithmic}
\end{algorithm}

\begin{algorithm}[H]
\caption{Conjugate Pruned Optimal Partition (CPOP)}
\textbf{Input:} Time series $ y = \{y_t\}_{t=1}^n $, positive penalty constant $\beta$, non-negative, non-decreasing penalty function $h(\cdot).$
\begin{algorithmic}[1]
    \State Set $\hat{\mathcal{T}}_1=\{0\}$
    \vspace{5pt}
    \State Set $K=2\beta+h(1)+h(n)$
    \vspace{5pt}
    \For{$t=1,\ldots,n$}
        \vspace{5pt}
        \For{$\tau\in\hat{\mathcal{T}}_t$}
            \If{$\tau = \{0\}$} 
                \State $f^t _\tau(\phi) = \underset{\phi^\prime}{\text{min }}\mathcal{C}(y_{1:t},\phi^\prime,\phi)+h(t)$
            \Else
                \State $f^t _\tau(\phi) = \underset{\phi^\prime}{\text{min }}\left\{f_{\tau_1,\ldots,\tau_{k-1}}^{\tau_k}(\phi^\prime)+\mathcal{C}(y_{(\tau_k+1):t},\phi^\prime,\phi)+h(t-\tau_k)+\beta\right\}$
            \EndIf
        \EndFor
        \vspace{5pt}
        \For{$\tau\in\hat{\mathcal{T}}_t$}
            \State $Int^t _\tau = \left\{\phi:f_\tau^t (\phi)=\underset{\tau^\prime\in\hat{\mathcal{T}}_t}{\text{min }}[f_{\tau^\prime}^t (\phi)]\right\}$
            \vspace{5pt}
            \State $\overset{*}{\mathcal{T}_t} = \{\tau:Int_\tau^t \neq\emptyset\}$
            \vspace{5pt}
            \State $\hat{\mathcal{T}}_{t+1} = \hat{\mathcal{T}}_t\hspace{3pt}\cup\left\{(\tau,t):\tau\in\overset{*}{\mathcal{T}}_t\right\}$
        \EndFor
        \vspace{5pt}
        \State $\hat{\mathcal{T}}_{t+1}=\left\{\tau\in\hat{\mathcal{T}}_{t+1}:\underset{\phi}{\text{min }}f^t _\tau(\phi)\leq\underset{\phi^\prime,\tau^\prime}{\text{min }}[f^t _{\tau^\prime}(\phi^\prime)]+K\right\}$
    \EndFor
    \vspace{5pt}
    \State $\hat f=\underset{\tau,\phi}{\text{min }}f_\tau^n(\phi)$
    \vspace{5pt}
    \State $\hat\tau = \underset{\tau}{\text{arg min }}\left[\underset{\phi}{\text{min }}f_\tau^n(\phi)\right]$
    \vspace{5pt}
    \State \Return $\hat f$, $\hat\tau$
\end{algorithmic}
\end{algorithm}

Similar to the recursion strategy outlined in OP and refined in PELT, on the $s^\text{th}$ iteration the \textbf{for} loop on line 3 loops through all previously computed (un-pruned) changepoint vectors $\tau$ and calculates for each the cost of adding a new changepoint at time $s$. On line 12, the intervals corresponding to each of the possible y-values $\phi$ are created for each candidate changepoint vector $\tau$. Lines 13 and 14 are then the functional pruning step, and line 16 is the inequality-based pruning step. Once this list of potential changepoint vectors and corresponding y-values has been built and populated up to time $n$, all that remains is to calculate the minimum loss $\hat f$ and the changepoint vector corresponding to it, $\hat\tau$.

On line 12, the process for generating intervals of $\phi$ over which a given candidate vector $\tau$ is optimal is non-trivial and requires attention. The original algorithm is given in the supplemental materials to \citet{fearnhead-2019-cpop} and recreated here.

\begin{algorithm}[H]
\caption{Algorithm for calculation of $ Int^t _{\tau} $ at time $t$}
\begin{algorithmic}[1]
\State \textbf{Input: } Set of changepoint candidate vectors $ \hat{\mathcal{T}}_t $ for current timestep $ t $, optimal segmentation functions $ f^t _{\tau}(\phi) $ for $ t $ and $ \tau \in \hat{\mathcal{T}}_t $
\State $ \mathcal{T}_{\text{temp}} = \hat{\mathcal{T}}_t $
\State $ Int^t _{\tau} = \emptyset $ for $ \tau \in \hat{\mathcal{T}}_t $
\State $ \phi_{\text{curr}} = -\infty $
\State $ \tau_{\text{curr}} = \underset{\tau \in \mathcal{T}_{\text{temp}}}{\text{arg min }}[f^t _{\tau}(\phi_{\text{curr}})] $
\vspace{5pt}
\While{$ \mathcal{T}_{\text{temp}} \setminus \{\tau_{\text{curr}}\} \neq \emptyset $}
    \For{$ \tau \in \mathcal{T}_{\text{temp}} \setminus \{\tau_{\text{curr}}\} $}
        \State $ x_{\tau} = \min \{ \phi : f^t _{\tau}(\phi) - f^t _{\tau_{\text{curr}}}(\phi) = 0 \text{ and } \phi > \phi_{\text{curr}} \} $
        \If{$ x_{\tau} = \emptyset $}
            \State $ \mathcal{T}_{\text{temp}} = \mathcal{T}_{\text{temp}} \setminus \{\tau\} $
        \EndIf
    \EndFor
    \State $ \tau_{\text{new}} = \underset{\tau}{\text{arg min }} (x_{\tau}) $
    \vspace{5pt}
    \State $ \phi_{\text{new}} = \min (x_{\tau}) $
    \State $ Int^t _{\tau_{\text{curr}}} = [\phi_{\text{curr}}, \phi_{\text{new}}] \cup Int^t _{\tau_{\text{curr}}} $
    \State $ \tau_{\text{curr}} = \tau_{\text{new}} $
    \State $ \phi_{\text{curr}} = \phi_{\text{new}} $
\EndWhile
\State \Return The intervals $ Int^t _{\tau} $ for $ \tau \in \hat{\mathcal{T}}_t $
\end{algorithmic}
\end{algorithm}
The idea of this algorithm is to partition the values of $\phi$ from $-\infty$ to $\infty$ into intervals over which some changepoint vector $\tau\in\hat{\mathcal{T}}_t$ is optimal, with the understanding that if there is no interval of $\phi$ over which a candidate changepoint vector $\tau\in\hat{\mathcal{T}}_t$ is optimal then it can be discarded from future consideration. On line 6, a \textbf{while}-loop constructs these intervals by initializing $\phi_\text{curr}=-\infty$ and $\tau_\text{curr}$ as the changepoint configuration which minimizes the cost for $\phi=-\infty$. It then recursively partitions intervals of $\phi$ by iterating over all of the candidate changeopint vectors $\tau\in\hat{\mathcal{T}}_t$, and for each of them finding the first value of $\phi$ greater than $\phi_\text{curr}$ where $\tau$ becomes better at partitioning the data than $\tau_\text{curr}$. The minimum of these $\phi$ values is the right endpoint of the presently-determined interval, and the $\tau$ leading to that $\phi$ becomes the new $\tau_\text{curr}$. The loop terminates when there is no $\phi>\phi_\text{curr}$ for which some $\tau\in\hat{\mathcal{T}}_t$ is better at partitioning the data than $\tau_\text{curr}$.

\begin{algorithm}[H]
\caption{\textbf{Pseudocode:} Segment Neighborhood}
\begin{algorithmic}[1]
\State \textbf{Input:} A time series $\{y_t\}_{t=1}^n$; a cost function $\mathcal{C}(\cdot)$
\For{$s=3,4,\ldots,n+1$}
\State $F_1(s) = \underset{r<s}{\text{min }}\left[\mathcal{C}(y_{1:(r-1)}) + \mathcal{C}(y_{r:(s-1)})\right]$
\State $\tau(1,s) = \underset{2\le r<s}{\text{arg min }}\left[\mathcal{C}(y_{1:(r-1)}) + \mathcal{C}(y_{r:(s-1)})\right]$

\EndFor
\For{$m = 2,3,\ldots,m_\text{max}$}\vspace{5pt}
\For{$s = m+2,m+3,\ldots,n+1$}\vspace{5pt}
\State  $F_m(s) = \underset{m<r<s}{\text{min }}\left[F_{m-1}(r) + \mathcal{C}(y_{r:(s-1)})\right]$\vspace{5pt}
\State $\tau(m,s) = \underset{m<r<s}{\text{arg min }}\left[F_{m-1}(r) + \mathcal{C}(y_{r:(s-1)})\right]$\vspace{5pt}
\EndFor
\EndFor
\For{$m=1,2,\ldots,m_\text{max}$}\vspace{5pt}
\State $\tau_m^m =\tau(m,n+1)$\vspace{5pt}
\For{$i = m-1, m-2, \dots, 1$}\vspace{5pt}
\State $\hat\tau_m^i =\tau(i,\hat\tau_m^{i+1})$\vspace{5pt}
\EndFor
\EndFor
\State \Return $\hat\tau_1({n+1})$, $\hat\tau_2({n+1})$, $\ldots$, $\hat\tau_{m_\text{max}}({n+1})$
\end{algorithmic}
\end{algorithm}
In lines 2-11, a recursion is performed to both find the cost associated with the optimal segmentation as well as to store the information necessary to reconstruct the changepoint configurations for each $1\leq m\leq m_\text{max}$. The reconstruction for each of these configurations occurs in lines 12-17.

\begin{algorithm}[H]
\caption{Pruned Segment Neighborhood}
\begin{algorithmic}[1]
\State \textbf{Input:} A time series $\{y_t\}_{t=1}^n$; a cost function $\mathcal{C}(\cdot)$
\For{$s=3,4,\ldots,n+1$}
\State $F_1(s) = \underset{r<s}{\text{min }}\left[\mathcal{C}(y_{1:(r-1)}) + \mathcal{C}(y_{r:(s-1)})\right]$
\State $\tau(1,s) = \underset{r<s}{\text{arg min }}\left[\mathcal{C}(y_{1:(r-1)}) + \mathcal{C}(y_{r:(s-1)})\right]$

\EndFor
\vspace{5pt}
\For{$m = 2,3,\ldots,K$}
\State \textbf{Initialize: } list of potential (unpruned) changepoints  $\mathcal{P} = \{m-1\}$
\vspace{5pt}
\State \textbf{Initialize: } $\mathcal{F}_{m-1}^m(\mu) = F_{m-1}(m)$\vspace{5pt}
\For{$s = m+2,m+3,\ldots,n+1$}\vspace{5pt}
\State \textbf{Initialize: }$\mathcal{F}_m^s(\mu) = F_{m-1}(s)$\vspace{5pt}
\State \textbf{Initialize: }$Set_m^s = D$\vspace{5pt}
\For{$ r\in\mathcal{P}$}\vspace{5pt}
\State $\mathcal{F}_m^r(\mu) = \mathcal{F}_m^{r}(\mu) + \gamma(y_{r+1},\mu)$\vspace{5pt}
\State $I_m^s(r) = \{\mu:\mathcal{F}_m^s(r,\mu)\leq F_{m-1}(s)\}$\vspace{5pt}
\State $    Set_m^s(r) = Set_m^{s}(r)\quad \cap \quad I_m^s(r)$\vspace{5pt}
\If{$Set_m^s(r)=\emptyset$}
\State $\mathcal{P}=\mathcal{P}\hspace{2pt}\setminus\hspace{2pt}\{r\}$
\EndIf
\State $Set_m^s(s) = Set_m^s(s) \hspace{2pt}\setminus\hspace{2pt}I_m^s(r)$
\EndFor
\If{$Set_m^s(r)\neq\emptyset$}
\State $\mathcal{P} = \mathcal{P}\hspace{2pt}\cup\hspace{2pt}\{s\}$
\EndIf
\State $F_m(s) = \underset{r\in\mathcal{P}}{\text{min }}\left[\underset{\mu}{\text{min }}\mathcal{F}_m^s(\mu)\right]$
\State $\tau(m,s)=\underset{r\in\mathcal{P}}{\text{arg min }}\left[\underset{\mu}{\text{min }}\mathcal{F}_m^s(r,\mu)\right]$
\EndFor
\EndFor
\For{$m=1,2,\ldots,m_\text{max}$}\vspace{5pt}
\State $\tau_m^m =\tau(m,n+1)$\vspace{5pt}
\For{$i = m-1, m-2, \dots, 1$}\vspace{5pt}
\State $\hat\tau_m^i =\tau(i,\hat\tau_m^{i+1})$\vspace{5pt}
\EndFor
\EndFor
\State \Return $\hat\tau_1({n+1})$, $\hat\tau_2({n+1})$, $\ldots$, $\hat\tau_{m_\text{max}}({n+1})$
\end{algorithmic}
\end{algorithm}
As in SN, a first-pass through the data is performed on lines 2-5 to build $F_1(s)$ and $\tau(1,s)$ for all $3\leq s\leq n+1$. In the recursion on lines 6-27 as $s$ loops from $m+2$ to $n+1$, intervals of $\mu$ are found for each $r\in\mathcal{P}$ for which $r$ is the optimal changepoint prior to $s$. If no valid interval exists, then $r$ is pruned from consideration. This happens on line 17. On line 22, the current value of $s$ is added to the unpruned changepoints as long as there is some value of $\mu$ such that $F_{m-1}(s)<\mathcal{F}_m^s(\mu)$. On lines 24-25, $F_m(s)$ and $\tau(m,s)$ are updated, and on lines 28-33 a recursion is used to re-build $\hat\tau_m(n+1)$ for $1\leq m\leq K$. 

\begin{algorithm}[H]
\caption{AR1Seg Algorithm}
\textbf{Input:} Time series $ y = \{y_t\}_{t=1}^n $, maximum number of changepoints $ K $.  

\begin{algorithmic}[1]
    \State Estimate $ \tilde{\phi} $ using  
    \[
    \tilde{\phi} = \frac{(\text{med}_{1\leq t\leq n-2}|y_{t+2} - y_t|)^2}{(\text{med}_{1\leq t\leq n-1}|y_{t+1} - y_t|)^2} - 1
    \]

    \State Use pSN on the decorrelated series to find all changepoint models up to $K$ changepoints.

    \State Remove any second changepoint of a consecutive pair introduced spuriously by decorrelation.

    \For{each segmentation model considered}
        \State Compute  
        \[
        SS(y, \tilde{\phi}) = \min_{\delta, \tau} SS(y, \tilde{\phi}, \delta, \tau)
        \]
        \State Compute the corresponding MLE of variance
        \[
        \hat{\sigma}^2 = \frac{1}{n} 
        \sum_{k=1}^{m} \sum_{t=\tau_{k-1}}^{\tau_k-1}
        (y_t - \tilde{\phi} y_{t-1} - \delta_k)^2
        \]
    \EndFor

    \State Select optimal $ \hat{m} $ using the BIC criterion.

    \State \Return Estimated changepoints $\hat{\tau}$, segment means $\hat{\delta}$, and estimated AR(1) parameter $\tilde{\phi}$.
\end{algorithmic}
\end{algorithm}

\begin{algorithm}[H]
\caption{2D Iteratively Reweighted Fused Lasso (2D IRFL)}
\begin{algorithmic}[1]
\State \textbf{Input:} Image vector $y \in \mathbb{R}^{RC}$; structured difference matrix $D_0 \in \mathbb{R}^{p \times RC}$; weight offset $\varepsilon > 0$; number of IRFL iterations $N$; regularization grid $\{\lambda_\ell\}_{\ell=1}^L$.
 \vspace{5pt}
\For{$\ell = 1$ \textbf{to} $L$}
    \State \textbf{Reset per-$\lambda$ state:} $\hat W^{(1)} \gets I_{p \times p}$; \; $\mathbf{D}^{(1)} \gets D_0$
    \For{$i = 1$ \textbf{to} $N$}
        \State \textbf{Solve (TV-ADMM):}
        \begin{equation*}
            \hat{\mu}^{(i)}_{\lambda_\ell}
            = \arg\min_{\mu \in \mathbb{R}^{RC}}
            \frac{1}{2}\,\|y - \mu\|_2^2 \;+\; \lambda_\ell\,\|\mathbf{D}^{(i)}\mu\|_1
        \end{equation*}
        \State \textbf{Reweight:} For $j=1,\ldots,p$ set
        \[
            \hat w^{(i+1)}_j \;=\; \frac{1}{\,\big|(D_0\,\hat{\mu}^{(i)}_{\lambda_\ell})_j\big| + \varepsilon\,},
            \qquad \hat W^{(i+1)} \;=\; \mathrm{diag}\!\big(\hat w^{(i+1)}_1,\ldots,\hat w^{(i+1)}_p\big).
        \]
        \State \textbf{Update $D$:} \quad $\mathbf{D}^{(i+1)} \;\gets\; \hat W^{(i+1)}\, D_0$
    \EndFor
\EndFor
\State \textbf{Output:} Denoised estimates $\big\{\,\hat{\mu}^{(N)}_{\lambda_\ell}\,\big\}_{\ell=1}^{L}$
\end{algorithmic}
\end{algorithm}

\subsection{Estimation of the Genlasso Solution Path}\label{Genlasso soln path}
Ryan Tibshirani in \citet{tibshirani-2011-solutiongenlasso} lays out a method by which (\ref{genlasso-definition}) may be estimated down to a normalization factor. It is derived using the Lagrange dual problem and is largely recreated here, along with clarifying detail. First, the simpler fused lasso problem is rescaled for convenience with the $\ell_1$ penalty on the first differences of $\mu$ written as the $\ell_1$ norm of a matrix product:\begin{equation}
    \label{fusedlasso_2}
    \min_{\mu\in\mathbbm{R}^n}\frac{1}{2}\sum_{t=1}^n(y_t-\mu_t)^2 + \lambda||D\mu||_1
\end{equation} where $D$ is the matrix of first differences, though  more general penalty matrices $D$ are also considered. In order to solve (\ref{fusedlasso_2}), the formula is rewritten as \begin{equation}
    \label{fusedlasso_3}
    \min_{\mu\in\mathbbm{R}^n,z\in\mathbbm{R}^m} \frac{1}{2}\sum_{t=1}^n(y_t-\mu_t)^2 + \lambda||z||_1 \qquad\text{subject to } D\mu=z.
\end{equation}
The Lagrangian is therefore \begin{equation}
    \frac{1}{2}\sum_{t=1}^n(y_t-\mu_t)^2 + \lambda||z||_1 + \nu^\top (D\mu-z)
\end{equation}
which can be minimized over $\mu$ and $z$ separately. The terms involving $\mu$ can be minimized as follows:\begin{align*}
    \min_{\mu}\bigg(\frac{1}{2}\sum_{t=1}^n(y_t-\mu_t)^2 + \nu^\top D\mu\bigg)&=\min_{\mu}\bigg(\frac{1}{2}(y-\mu)^\top (y-\mu) + \nu^\top D\mu\bigg)\\
    &=\min_{\mu}\bigg(\frac{1}{2}y^\top y-y^\top \mu + \frac{1}{2}\mu^\top \mu+\nu^\top D\mu\bigg)
\end{align*}
Taking the partial derivative of the above equation with respect to $\mu$ and setting equal to 0 yields\begin{equation}
\label{primal-dual correspondence}
    -y^\top  + \mu^\top  + \nu^\top D =0 \implies\mu = y-D^\top \nu.
\end{equation}
Plugging this value in for $\mu$ yields
\begin{align}
        \min_{\mu}\bigg(\sum_{t=1}^n(y-\mu_t)^2 + \nu^\top D\mu\bigg) &=\min_{\mu}\bigg(\frac{1}{2}(y-\mu)^\top (y-\mu) + \nu^\top D\mu\bigg)\nonumber\\
        &=\frac{1}{2}\bigg(D^\top \nu\bigg)^\top \bigg(D^\top \nu\bigg) + \nu^\top D\left(y-D^\top \nu\right)\nonumber\\
        &=-\frac{1}{2}\bigg\{- 2y^\top \left(D^\top \nu\right) + \bigg(D^\top \nu\bigg)^\top \bigg(D^\top \nu\bigg)\bigg\}\nonumber\\
        &=-\frac{1}{2}\bigg\{y^\top y- 2y^\top \left(D^\top \nu\right) + \bigg(D^\top \nu\bigg)^\top \bigg(D^\top \nu\bigg)\bigg\}+\frac{1}{2}y^\top y\nonumber\\
        &=-\frac{1}{2}\bigg(y-D^\top \nu\bigg)^\top \bigg(y-D^\top \nu\bigg) + y^\top y\nonumber\\
        &=-\frac{1}{2}\|y-D^\top\nu\|_2^2 + \frac{1}{2}y^\top y.\label{min_mu}
\end{align}
While minimizing over $z$ yields \begin{equation}
\label{minz}
    \min_{z}\left(\lambda||z||_1-\nu^\top z\right)=\begin{cases}
        0& ||\nu||_\infty\leq\lambda\\
        -\infty &\text{else,}
    \end{cases}
\end{equation}
where $||\cdot||_\infty$ is the $\ell_\infty$ norm. That is, \begin{equation*}
    ||\nu||_\infty = \max_i |\nu_i|.
\end{equation*}
To see that (\ref{minz}) is true, observe that if $\nu_i>\lambda$ for any $i$, \begin{align*}
    \min_z\left(\lambda|z_i|-\nu_iz_i\right)\to\left(\lambda-\nu_i\right)z_i\to-\infty\text{ as }z_i\to\infty.
\end{align*}
By contrast, if $\nu_i\leq\lambda$ for all $i$, then $||\nu||_\infty\leq\lambda$ and for all $i$, \begin{align*}
        &\min_z\left(\lambda|z_i|-\nu_iz_i\right)\to\left(\lambda-\nu_i\right)z_i\to\infty\text{ as }z_i\to\infty&\text{and}\\
        &\min_z\left(\lambda|z_i|-\nu_iz_i\right)\to-\left(\lambda+\nu_i\right)z_i\to\infty\text{ as }z_i\to-\infty&\text{and}\\
        &\min_z\left(\lambda|z_i|-\nu_iz_i\right)\to 0 \text{ as }z_i\to 0
\end{align*}
and thus (\ref{minz}) follows. Combining only the portion of (\ref{min_mu}) which depends on $\nu$ along with (\ref{minz}), the dual of (\ref{fusedlasso_2}) can be written as \begin{equation}
    \label{dual}\min_{\nu\in\mathbbm{R}^m}\frac{1}{2}\|y-D^\top\nu\|_2^2\qquad\text{subject to }||\nu||_\infty\leq\lambda.
\end{equation}
It is noted in \citet{tibshirani-2011-solutiongenlasso} that the dual therefore has a very convenient and intuitive set of constraints. Namely, the only $\nu$ of interest are those for which $||\nu||_\infty\leq\lambda$. This is simply a box in $m-$dimensional space, centered at the origin, and with each $\nu_i$ between $-\lambda$ and $\lambda$.
Matrix algebra will reveal that if rank$(D)<m$, (\ref{dual}) is not strictly convex. However, (\ref{fusedlasso_3}) is always strictly convex, always has a unique solution, and is strictly feasible as a result of being unconstrained. Therefore, strong duality holds and for any $\lambda$, by (\ref{primal-dual correspondence}) the relationship between the primal and dual solutions is:\begin{equation}
\label{primal-dual correspondence estimate}
    \hat{\mu} = y-D^\top \hat\nu
\end{equation}
As well, because of the constraint in (\ref{dual}), for any $\lambda$,\begin{equation}
\label{nu_i}
    \hat\nu_i \in\begin{cases}
        \{\lambda\}, & \text{if }(D\hat{\mu})_i>0,\\
        \{-\lambda\}, & \text{if }(D\hat{\mu})_i<0,\\
        [-\lambda,\lambda], & \text{if } (D\hat{\mu})_i=0.
    \end{cases}
\end{equation}
To see this, consider the Lagrangian of (\ref{dual}):\begin{equation}
    \label{lagrangian_of_dual}\mathcal{L} = \frac{1}{2}\|y-D^\top\nu\|_2^2 + \sum_{i=1}^n\alpha_i(\nu_i-\lambda) + \sum_{i=1}^n\beta_i(-\nu_i-\lambda)
\end{equation}
The KKT conditions ensure that the gradient of (\ref{lagrangian_of_dual}) vanishes at stationarity. That is, \begin{align*}
    0 &= \nabla_\nu\mathcal{L}\nonumber\\ &= -D(y-D^\top \nu) + \alpha-\beta \nonumber\\
    &\overset{(\ref{primal-dual correspondence estimate})}{=}-D\hat{\mu} + \alpha-\beta
\end{align*}Of course this means at stationarity for all $i$,\begin{equation*}
    (D\hat{\mu})_i = \alpha_i-\beta_i.
\end{equation*}Complementary slackness for these constraints states for all $i$:\begin{equation*}
    \alpha_i(\nu_i-\lambda)=0\qquad \text{and}\qquad  \beta_i(-\nu_i-\lambda)=0.
\end{equation*} 
Therefore, because $-\lambda\leq\nu_i\leq\lambda$ by (\ref{dual}), if $(D\hat{\mu})_i=0$ then $\alpha_i=\beta_i$ and all that is known about $\nu_i$ is that $\nu_i\in[-\lambda,\lambda]$. However, if $(D\hat{\mu})_i<0$ then $\alpha_i<\beta_i$ which implies $\alpha_i=0$ and $\beta_i>0$ (they cannot both be positive because this would imply that $\nu_i=\pm\lambda$ simultaneously, which is impossible), and this in turn implies that $\nu_i=-\lambda$. An analogous argument shows that if $(D\hat{\mu})_i>0$ then $\nu_i=\lambda$, and so (\ref{nu_i}) is shown.
The implication of $(\ref{nu_i})$ is that the coordinates of $\nu$ which are equal to $\pm\lambda$ correspond precisely to the coordinates of $D\hat{\mu}$ which are not forced to be zero. That is, if $\hat\nu_i=\pm\lambda$ then $|\mu_{i+1}-\mu_i|>0$. For future reference, the set of these indices is denoted by $\mathcal{B}$. That is,\begin{equation}
\label{B}
    \mathcal{B} = \{i:|\hat\nu_i|=\lambda\}.
\end{equation}
As well, the vector $s$ is defined in such a way as to capture the sign of the $\hat\nu_i$. If  $\hat\nu_i=\lambda$, then the entry in $s$ corresponding to it will be $1$, and if $\hat\nu_i=-\lambda$, then the entry in $s$ corresponding to it will be $-1$.
What follows in \citet{tibshirani-2011-solutiongenlasso} is a description of an algorithm used to solve (\ref{fusedlasso_2}) first in the case that $D=D_{1d}$ is the $(n-1)\times n$ matrix of first differences given in (\ref{D Matrix}), before generalizing the algorithm to solve the related equation,\begin{equation}
    \label{General X and D}    \min_{\beta\in\mathbbm{R}^n}\frac{1}{2}||y-X\beta||_2^2+\lambda||D\beta||_1,
\end{equation}
for general $X$ and $D$. First, a lemma called $\textit{The Boundary Lemma}$ is given and proven in supplemental materials \citet{tibshirani-2011-supplement_to_solutiongenlasso}, which states that for any coordinate $i$:
\begin{align*}
&\text{If }\hat\nu_i=\lambda_0, \text{ then }\hat\nu_i=\lambda\text{ for all }\lambda\in[0,\lambda_0]\text{,}\qquad\text{ and}\\
&\text{If }\hat\nu_i=-\lambda_0, \text{ then }\hat\nu_i=-\lambda\text{ for all }\lambda\in[0,\lambda_0].  
\end{align*}
To grasp this intuitively, consider that the constraint set $\{\nu:||\nu||_\infty\leq\lambda\}$ from (\ref{dual}) is a hypercube of side length $2\lambda\subseteq\mathbbm{R}^{n-1}$, and so it is natural to think of the $\nu_i$ with $i\in\mathcal{B}$ (that is, those $\nu_i$ such that $\nu_i=\pm\lambda$) as being on the boundary of this box. The Boundary Lemma states that if some coordinate of $\nu$ is on the boundary for some value of $\lambda$, then it will remain on the boundary as the box shrinks (as $\lambda\to0$). By The Boundary Lemma and because of the primal-dual correspondence articulated in (\ref{primal-dual correspondence estimate}) for the case of $D=D_{1d}$, inspection will reveal that  \begin{equation}
    \label{fusion of mu}
    \text{If }\hat\mu_i = \hat\mu_{i+1}\text{ for }\lambda=\lambda_0,\text{ then }\hat\mu_i=\hat\mu_{i+1}\text{ for }\lambda\geq\lambda_0.
\end{equation}
Stated plainly, this means that if for any value of $\lambda$, two consecutive coordinates of $\mu$ become ``fused", then they can never become ``unfused" for a greater value of $\lambda$. 
Next, some useful notation is introduced. If $\mathcal{B} = [i_1,\ldots,i_k]$, then for a matrix $A$ and vector $x$:
\begin{equation*}
    A_\mathcal{B} = \begin{bmatrix}
        A_{i_1}\\
        \vdots\\
        A_{i_k}
    \end{bmatrix}\qquad\text{and}\qquad x_\mathcal{B} = \begin{bmatrix}
        x_{i_1}\\
        \vdots\\
        x_{i_k}
    \end{bmatrix}
\end{equation*}
where $A_i$ is the $i^{th}$ row of $A$. That is, $A_\mathcal{B}$ is a matrix composed of the rows of $A$ whose indices are in $\mathcal{B}$. Similarly, $A_{-\mathcal{B}}$ and $x_{-\mathcal{B}}$ denotes the matrix or vector formed from the rows of $A$ or entries in $x$ respectively whose indices are not in $\mathcal{B}$. Note then that for any matrix $A$ and vector $x$ for which $A^\top x$ is defined,\begin{align*}
    [A^\top x]_i &=\sum_{j}A_{ji}x_j\\
    &=\sum_{j\in\mathcal{B}}A_{ji}x_j + \sum_{j\not\in\mathcal{B}}A_{ji}x_j\\
    &=[(A_\mathcal{B})^\top x_\mathcal{B}]_i + \left[(A_\mathcal{-B})^\top x_\mathcal{-B}\right]_i
\end{align*}
and therefore,\begin{equation}
\label{decomposition}
    A^\top x = (A_\mathcal{B})^\top x_\mathcal{B} + (A_\mathcal{-B})^\top x_\mathcal{-B}.
\end{equation}
The algorithm used to estimate a solution path for any $\lambda$ to (\ref{dual}) is now described below. Initially $\lambda$ is set to $\infty$. Clearly the solution is just the OLS solution to (\ref{dual}) as the problem is unconstrained. But suppose a solution had been found for some value of $\lambda_k$, so that $\mathcal{B} = \mathcal{B}_{\lambda_k}$ and $s=s_{\lambda_k}$. Then \begin{equation}
\label{dual fusion}
    \hat\nu_\mathcal{B} = \lambda s\qquad\text{for all}\qquad \lambda\in[0,\lambda_k].
\end{equation} Thus as $\lambda$ decreases, (\ref{dual}) the following reduction holds: \begin{align}
    \label{reduced dual}
    &\min_{\nu}\frac{1}{2}||y-D^\top \nu||_2^2\qquad\text{subject to }||\nu||_{\infty}\leq\lambda\nonumber\\
    &\overset{(\ref{decomposition})}{=}\min_{\nu}\frac{1}{2}||y-D^\top _\mathcal{B}\nu_\mathcal{B}-D^\top _\mathcal{-B}\nu_\mathcal{-B}||_2^2\qquad\text{subject to }||\nu||_{\infty}\leq\lambda\nonumber\\
    &\overset{(\ref{dual fusion})}=\min_{\nu_{-\mathcal{B}}}\frac{1}{2}||y-\lambda D^\top _\mathcal{B}s-D^\top _\mathcal{-B}\nu_\mathcal{-B}||_2^2\qquad\text{subject to }||\nu_{-\mathcal{B}}||_{\infty}\leq\lambda
\end{align}
which involves solving only for those coordinates which are in the interior; that is, the coordinates $\nu_i$ satisfying $\nu_i\in[-\lambda,\lambda]$. Because for $\lambda=\lambda_k$ these coordinates $\hat\nu_{-\mathcal{B}}$ must satisfy the previous statement, the constraint in (\ref{reduced dual}) is satisfied for free and the solution is the familiar OLS solution:\begin{equation}
    \hat\nu_{-\mathcal{B}} = \left(D_{-\mathcal{B}}(D_{-\mathcal{B}})^\top \right)^{-1}D_{-\mathcal{B}}\left(y-\lambda_k(D_\mathcal{B})^\top s\right).
\end{equation}
Tibshirani et al. then set the right side of this equality equal to $a-\lambda_kb$ for some vectors $a$ and $b$, and observe that for $\lambda\leq\lambda_k$, $\hat\nu_{-\mathcal{B}}$ will continue to equal $a-\lambda b$ until the moment when one of the coordinates hits the boundary, at which point that coordinate (say, at index $i$) will satisfy the equation $a_i-\lambda b_i=\pm\lambda$. Through simple algebraic manipulation, the value of $\lambda$ for each $i$ which satisfies this equation (denoted $t_i^\text{(hit)}$ for ``hitting time" in \citet{tibshirani-2011-solutiongenlasso}) is:\begin{equation}
    t_i^\text{(hit)} = \frac{a_i}{b_i\pm 1} = \frac{\left[\left(D_{-\mathcal{B}}(D_{-\mathcal{B}})^\top \right)^{-1}D_{-\mathcal{B}}\hspace{.1cm}y\right]_i}{\left[\left(D_{-\mathcal{B}}(D_{-\mathcal{B}})^\top \right)^{-1}D_{-\mathcal{B}}(D_\mathcal{B})^\top s\right]_i\pm 1}.
\end{equation}
Because only one of the above solutions will yield a value $t_i^\text{(hit)}\in[0,\lambda_k]$, the value satisfying this set membership is taken to be $t_i$ for each $i$. The next hitting time is denoted $h_{k+1} = \underset{i}{\text{max}}\hspace{3pt}t_i^\text{(hit)}$, and the next value for $\lambda$, $\lambda_{k+1}$ is chosen as $h_{k+1}$:\begin{equation*} \label{update}
    \lambda_{k+1} = h_{k+1}\quad,\quad i_{k+1}^\text{(hit)} = \underset{i}{\text{arg max }} t_i^\text{(hit)}\quad,\quad s_{k+1} = \text{sign}\left(\hat\nu_{\lambda_{k+1},i_{k+1}^\text{(hit)}}\right)
\end{equation*}
where $\hat\nu_{\lambda_{k+1},i_{k+1}^\text{(hit)}}$ denotes the estimated value of $\nu_{i_{k+1}^\text{(hit)}}$ when $\lambda=\lambda_{k+1}$. Then, $i_{k+1}^\text{(hit)}$ is appended to $\mathcal{B}$ and $s_{k+1}^\text{(hit)}$ is appended to $s$. The process iterates until every coordinate is on the boundary ($\mathcal{B}=\{i\}_{i=1}^{rank(D)}$). In the case of the fused lasso, $D=D_{1d}$, which results in $n-1$ iterations.

For a general $D$, the Boundary Lemma does not hold and so it isn't the case that once a coordinate of $\nu$ is on the boundary for some $\lambda$, it will always be on the boundary as $\lambda\to0$. Therefore, the ``leaving times" of the $\nu_i$ on the boundary must be calculated as well. The leaving time of the $i^{th}$ boundary coordinate is \begin{equation}
    \label{leaving time}
    t_i^\text{(leave)} = \begin{cases}
        c_i/d_i,&\text{if $c_i<0$ and $d_i<0$,}\\
        0& \text{otherwise,}
    \end{cases}
\end{equation}
where \begin{align*}
    c_i &=s_i\cdot\left[D_\mathcal{B}\left[I-\left(D_{-\mathcal{B}}\right)^\top\left(D_{-\mathcal{B}}\left(D_{-\mathcal{B}}\right)^\top\right)^+D_{-\mathcal{B}}\right]y\right]_i\\
    \\
    d_i &= s_i\cdot\left[D_\mathcal{B}\left[I-\left(D_{-\mathcal{B}}\right)^\top\left(D_{-\mathcal{B}}\left(D_{-\mathcal{B}}\right)^\top\right)^+D_{-\mathcal{B}}\right]\left(D_\mathcal{B}\right)^\top s\right]_i\hspace{4pt}.
\end{align*}
The derivation of the above is omitted here but can be found in \citet{tibshirani-2011-supplement_to_solutiongenlasso}. The next leaving time is therefore taken as \begin{equation*}
    l_{k+1} = \max_i t_i^\text{(leave)}\quad,\quad i_{k+1}^{\text{(leave)}} = \underset{i}{\text{arg max }}t_{i}^{\text{(leave)}}\quad,\quad s_{k+1}^{\text{(leave)}} = \text{sign}\left(\hat\nu_{\lambda_{k+1},i_{k+1}^{\text{(leave)}}}\right).
\end{equation*}
In this case, at the $k^{th}$ step the next $\lambda$ is taken to be $\lambda_{k+1} = \max\{h_{k+1}, l_{k+1}\}$. If $h_{k+1}>l_{k+1}$ then the hitting coordinate $i_{k+1}^{\text{(hit)}}$ is added to $\mathcal{B}$ and its sign $s_{k+1}^{\text{(hit)}}$ is appended to $s$. Otherwise, the leaving coordinate $i_{k+1}^{\text{(leave)}}$  and its sign $s_{k+1}^{\text{(leave)}}$ are deleted from $\mathcal{B}$ and $s$, respectively. As in the case that $D=D_{1d}$ above, the process iterates until a full model is calculated. The derivation and statement of the estimation procedure in the case of a general design matrix $X$, while abridged here, is also contained in \citet{tibshirani-2011-supplement_to_solutiongenlasso}.

\subsection{TV-ADMM}\label{app:tvadmm}

TV-ADMM, introduced in \citet{boyd-2011-admm} as a more general ADMM algorithm before specializing for the total variation (generalized lasso) variant, begins solving (\ref{genlasso2d}) by first introducing an auxiliary variable $z$: 
\begin{equation}\label{auxiliary}
\min_{\mu, z} \quad \frac{1}{2} \|y - \mu\|_2^2 + \lambda \|z\|_1 \quad \text{subject to} \quad D\mu - z = 0.
\end{equation}
As is common practice in convex optimization, the Lagrangian of (\ref{auxiliary}) is then found:
\begin{equation}\label{lagrangian}
\mathcal{L}(\mu, z, \nu) = \frac{1}{2} \|y - \mu\|_2^2 + \lambda \|z\|_1 + \nu^\top (D\mu - z).
\end{equation}
As an added step to improve numerical stability and convergence, a quadratic penalty term is also added to the constraint:
\begin{equation}\label{augmented lagrangian}
\mathcal{L}_\rho(\mu, z, \nu) = \frac{1}{2} \|y - \mu\|_2^2 + \lambda \|z\|_1 + \nu^\top (D\mu - z) + \frac{\rho}{2} \|D\mu - z\|_2^2.
\end{equation}
In order to reduce the number of variables over which one must optimize, the
scaled dual variable $u := \frac{1}{\rho} \nu$ is defined, so that $\nu = \rho u$. Substituting this into the augmented Lagrangian gives:
\begin{equation}\label{scaled augmented Lagrangian}
\mathcal{L}_\rho(\mu, z, u) = \frac{1}{2} \|y - \mu\|_2^2 + \lambda \|z\|_1 + \rho u^\top (D\mu - z) + \frac{\rho}{2} \|D\mu - z\|_2^2.
\end{equation}
There is now a convenient way to write the last two terms by completing the square:
\begin{align}
\rho u^\top (D\mu - z) + \frac{\rho}{2} \|D\mu - z\|_2^2 
&= \frac{\rho}{2} \left( 2 u^\top (D\mu - z) + \|D\mu - z\|_2^2 \right) \nonumber \\
&=\frac{\rho}{2}\left(\|u\|_2^2+2u^\top(D\mu-z)+\|D\mu-z\|_2^2\right) - \frac{\rho}{2}\|u\|_2^2\nonumber \\
&= \frac{\rho}{2}  \|D\mu - z + u\|_2^2 - \frac{\rho}{2}\|u\|_2^2 .\label{last two terms}
\end{align}
Substituting the above back into equation (\ref{scaled augmented Lagrangian}), the scaled augmented Lagrangian becomes:
\begin{equation}\label{augmented - lagrangian}
\mathcal{L}_\rho(\mu, z, u) = \frac{1}{2} \|y - \mu\|_2^2 + \lambda \|z\|_1 + \frac{\rho}{2} \|D\mu - z + u\|_2^2 - \frac{\rho}{2} \|u\|_2^2.
\end{equation}
The penalty parameter $\rho$ is user-specified and affects convergence speed. A typical choice is $\rho = 1$. Equation (\ref{augmented - lagrangian}) is then minimized by a repeating sequence of updates until convergence. Variables $\mu$ and $z$  are initialized to zero, while $\mu$ is initialized to $y$:
\[
\mu^{(0)} = y, \quad z^{(0)} = 0, \quad u^{(0)} = 0
\]
with updates as follows:
\begin{align*}
\mu^{(k+1)} &\leftarrow \underset{\mu}{\text{arg min }} \ \frac{1}{2}\|y - \mu\|_2^2 + \frac{\rho}{2}\|D \mu - z^{(k)} + u^{(k)}\|_2^2 \\
z^{(k+1)} &\leftarrow \underset{z}{\text{arg min }} \ \lambda\|z\|_1 + \frac{\rho}{2}\|D \mu^{(k+1)} - z + u^{(k)}\|_2^2 \\
u^{(k+1)} &\leftarrow u^{(k)} + D \mu^{(k+1)} - z^{(k+1)}.
\end{align*}
These updates are derived below.
\subsubsection*{Updates: $\mu$}
The $\mu$-update in TV-ADMM is obtained by minimizing the scaled augmented Lagrangian with respect to $\mu$, keeping $z$ and $u$ fixed at their current iterates. From the augmented Lagrangian:
\begin{equation*}
\mathcal{L}_\rho(\mu, z, u) = \frac{1}{2} \|y - \mu\|_2^2 + \lambda \|z\|_1 + \frac{\rho}{2} \|D\mu - z + u\|_2^2 - \frac{\rho}{2} \|u\|_2^2,
\end{equation*}
the terms involving $\mu$ are isolated:
\begin{equation*}
 \frac{1}{2} \|y - \mu\|_2^2 + \frac{\rho}{2} \|D\mu - z^{(k)} + u^{(k)}\|_2^2.
\end{equation*}
This is a strictly convex quadratic problem in $\mu$, so the optimality condition can be derived by taking the gradient with respect to $\mu$ and setting it to zero. Differentiating the augmented Lagrangian with respect to  $\mu$ yields:
\[
\nabla_\mu \hspace{2pt}\mathcal{L} = (\mu - y) + \rho D^\top (D\mu - z^{(k)} - u^{(k)})
\]
while setting it equal to zero yields
\[
\mu - y + \rho D^\top (D\mu - z^{(k)} - u^{(k)}) = 0.
\]
Upon rearrangement, this becomes
\[
\mu + \rho D^\top D \mu = y + \rho D^\top z^{(k)} - u^{(k)}
\]
which in turn yields the normal equations
\[
( I + \rho D^\top D ) \mu = y + \rho D^\top ( z^{(k)} - u^{(k)} ).
\]
Therefore, $\mu^{(k+1)}$ is determined to be 

\begin{equation*}
\mu^{(k+1)} = ( I + \rho D^T D ) ^{-1}\left(y + \rho D^T ( z^{(k)} - u^{(k)} )\right).
\end{equation*}
\subsubsection*{Updates: $z$}
Next, the $z$ auxiliary variable is updated. By the KKT conditions, at stationarity it must be the case that for all $1\leq i\leq m$ \begin{equation}
    0 \in\nabla_{z_i}\hspace{2pt}\mathcal{L} 
\end{equation}where the subgradient is used because the derivative of $z_i$ is undefined if $z_i=0$. Applying this condition to (\ref{augmented - lagrangian}) yields \begin{equation}\label{inclusion}
    0\in-\rho([D\mu^{(k+1)}]_i-z_i+u^{(k)}_i)+\partial\lambda|z_i|
\end{equation}
where the  subgradient $\partial|z_i|$ is defined as \begin{equation*}
    \partial\lambda|z_i| = \begin{cases}
        \{\lambda\cdot\text{sign}(z_i)\} & z_i\neq 0 \\
        [-\lambda,\lambda] & z_i=0.
    \end{cases}
\end{equation*}
Equation (\ref{inclusion}) means the same thing as 
\begin{equation*}
    \rho([D\mu^{(k+1)}]_i-z_i+u^{(k)}_i)\in\partial\lambda|z_i|,
\end{equation*}
and so, for $z_i\neq0$, \begin{equation*}
    z_i = [D\mu^{(k+1)}+u^{(k)}]_i-\frac{\lambda}{\rho}\cdot\text{sign}(z_i).
\end{equation*}
This is nothing but  $[D \mu^{(k+1)} + u^{(k)}]_i$ after soft-thresholding by $\lambda/\rho$. Therefore,\begin{align*}
    z^{(k+1)}
    &=S_{\lambda/\rho}( D \mu^{(k+1)} + u^{(k)} )
\end{align*}
where $S_{\delta}(\cdot)$ is the elementwise soft-thresholding operator:
\begin{equation*}
S_{\delta}(a) = \text{sign}(a)\cdot ( |a| - \delta)_+
\end{equation*}
\subsubsection*{Updates: $u$}
To update $u$, note that the gradient of the augmented Lagrangian with respect to the dual variable is
\begin{equation*}
    \nabla_u \mathcal{L}_\rho = \rho(D \mu^{(k+1)} - z^{(k+1)}).
\end{equation*}
Thus, the update on $u$ corresponds to a gradient ascent step in the dual variable:
\begin{align*}
    u^{(k+1)} &= u^{(k)} + \eta \, \nabla_u \mathcal{L}_\rho \\
    &= u^{(k)} + \eta \, \rho (D \mu^{(k+1)} - z^{(k+1)}).
\end{align*}
Substituting $\eta = 1/\rho$---which arises naturally from the scaling $u = \nu / \rho$ of the unscaled dual variable---yields the standard scaled ADMM update:
\begin{equation*}
    u^{(k+1)} = u^{(k)} + D \mu^{(k+1)} - z^{(k+1)}.
\end{equation*}
This update represents a unit step in the scaled variable formulation of ADMM, which can be interpreted as a gradient ascent step with effective step size $1/\rho$. 
Convergence of this canonical update is established in \citet{boyd-2011-admm} for any fixed $\rho > 0$ under the usual convexity and saddle-point assumptions.
\subsection{\texorpdfstring{$\mathrm{BIC}_\phi$}{BIC-phi} Derivation}\label{app:bic}

The following derivation outlines how the standard BIC is modified by incorporating the Cholesky factorization of $R_\phi^{-1}$, leading to an explicit form of $\mathrm{BIC}_\phi$ used in the empirical studies of this dissertation.

Recall per \eqref{BIC}, for a fitted model with maximized likelihood $\mathcal{L}(\hat\theta)$, number of estimated parameters $k$, and sample size $n$, the BIC is defined as
\begin{equation}
\label{appendix BIC}
\text{BIC} = -2 \log \mathcal{L}(\hat\theta) + k \log n.
\end{equation}

In the case of Gaussian AR(1) errors with parameter $\phi$ and design matrix $X$,
\begin{equation}
\label{distribution of y}
y\sim\mathcal{N}(X\beta,\sigma^2R_\phi)
\end{equation}
where $R_\phi$ is the correlation matrix from \eqref{AR1_corr}. The full likelihood of $y$ is then
\begin{equation}
\label{full likelihood}
\mathcal{L}(\theta) = (2\pi\sigma^2|R_\phi|)^{-1/2}\exp\left( -\frac{1}{2\sigma^2} (y - X\beta)^\top R_\phi^{-1} (y - X\beta) \right).
\end{equation}
Therefore, the log-likelihood of $y$ is
\begin{align}
\label{loglikelihood}
&\log \mathcal{L}(\theta) = \nonumber\\
\nonumber\\
&-\frac{n}{2}\log(2\pi) - \frac{n}{2}\log \sigma^2 - \frac{1}{2} \log |R_\phi| - \frac{1}{2\sigma^2}(y - X\beta)^\top R_\phi^{-1} (y - X\beta).
\end{align}
The MLE of $\beta$ is found by the generalized least-squares formula:
\begin{equation}
\label{gls formula}
\hat\beta^{\text{GLS}} = (X^\top R_\phi^{-1}X)^{-1}X^\top R_\phi^{-1}y,
\end{equation}
while the MLE of $\sigma^2$ is given by
\begin{equation}
\label{sigma2 MLE}
\hat\sigma^2 = \frac{1}{n}(y - X\hat\beta^{\text{GLS}})^\top R_\phi^{-1} (y - X\hat\beta^{\text{GLS}}) = \frac{1}{n}\text{SSE}_\phi.
\end{equation}
Substituting $\hat\beta^\text{GLS}$ and $\hat\sigma^2$ into the log-likelihood in equation \eqref{loglikelihood} yields

\begin{equation}
\label{maximized loglikelihood}
\log \mathcal{L}(\hat\theta) = -\frac{n}{2} \log(2\pi) - \frac{n}{2} \log\left( \frac{\text{SSE}_\phi}{n} \right) - \frac{1}{2} \log |R_\phi| - \frac{n}{2}.
\end{equation}
As the constant terms do not vary between models and hence do not affect model selection, they are omitted when computing BIC. Thus, after multiplication by $-2$ and adding $k\cdot \log n$, the resulting BIC simplifies to \eqref{BIC AR1}:

\begin{equation*}
\text{BIC}_\phi = n \log\left(\frac{\text{SSE}_\phi}{n}\right) + k \log(n) + \log |R_\phi|.
\end{equation*}

\bibliographystyle{plainnat}
\bibliography{references.bib}

@article{fryzlewicz-2014-WBS,
  title={Wild Binary Segmentation for Multiple Change-Point Detection},
  author={Fryzlewicz, Piotr},
  journal={The Annals of Statistics},
  volume={42},
  number={6},
  pages={2243--2281},
  year={2014},
  publisher={Institute of Mathematical Statistics},
  DOI={10.1214/14-AOS1245}
}

@techreport{Perron-1991_ERPmemo,
  author       = {Perron, Pierre and Vogelsang, Tim},
  title        = {Nonstationarity and level shifts with an application to purchasing power parity},
  institution  = {Princeton University, Department of Economics, ERP Memo},
  number       = {359},
  year         = {1991},
  url          = {https://www.princeton.edu/~erp/ERParchives/archivepdfs/M359.pdf}
}

@article{chen-1993-joint,
  title     = {Joint estimation of model parameters and outlier effects in time series},
  author    = {Chen, Changhu and Liu, Lien-Te},
  journal   = {Journal of the American Statistical Association},
  volume    = {88},
  number    = {421},
  pages     = {284--297},
  year      = {1993},
  publisher = {Taylor \& Francis},
  doi       = {10.1080/01621459.1993.10594321}
}

@article{box-1975-intervention,
  author       = {Box, George E. P. and Tiao, George C.},
  title        = {Intervention Analysis with Applications to Economic and Environmental Problems},
  journal      = {Journal of the American Statistical Association},
  year         = {1975},
  volume       = {70},
  number       = {349},
  pages        = {70--79},
  doi          = {10.1080/01621459.1975.10480264},
  url          = {https://www.jstor.org/stable/2285379}
}

@book{burnham-2002-model,
  title     = {Model Selection and Multimodel Inference: A Practical Information-Theoretic Approach},
  author    = {Burnham, Kenneth P. and Anderson, David R.},
  year      = {2002},
  edition   = {2nd},
  publisher = {Springer},
  address   = {New York}
}

@article{schwarz-1978-bic,
  author  = {Schwarz, Gideon E.},
  title   = {Estimating the Dimension of a Model},
  journal = {The Annals of Statistics},
  year    = {1978},
  volume  = {6},
  number  = {2},
  pages   = {461--464},
  doi     = {10.1214/aos/1176344136}
}

@article{baranowski-2019-narrowest,
  title={Narrowest-Over-Threshold Detection of Multiple Change Points and Change-Point-Like Features},
  author={Baranowski, Rafal and Chen, Yining and Fryzlewicz, Piotr},
  journal={Journal of the Royal Statistical Society: Series B (Statistical Methodology)},
  volume={81},
  number={3},
  pages={649--672},
  year={2019},
  publisher={Royal Statistical Society},
  DOI={10.1111/rssb.12322}
}

@book{csorgo-1997-CUSUM,
  title={Limit Theorems in Change-Point Analysis},
  author={Cs{\"o}rg{\"o}, Mikl{\'o}s and Horv{\'a}th, Lajos},
  publisher={John Wiley \& Sons},
  year={1997},
  series={Wiley Series in Probability and Statistics},
  isbn={978-0-471-95522-1}
}

@article{auger-1989-op,
  title={Algorithms for the Optimal Identification of Segment Neighborhoods},
  author={Auger, Ivan E and Lawrence, Charles E},
  journal={Bulletin of Mathematical Biology},
  volume={51},
  number={1},
  pages={39--54},
  year={1989},
  publisher={Elsevier}
}

@article{scott-1974-cluster,
  title={A Cluster Analysis Method for Grouping Means in the Analysis of Variance},
  author={Scott, A. J. and Knott, M.},
  journal={Biometrics},
  volume={30},
  number={3},
  pages={507--512},
  year={1974},
  publisher={International Biometric Society},
  doi={10.2307/2529204}
}

@article{fearnhead-2019-cpop,
  title={Detecting Changes in Slope with an \(L_0\) Penalty},
  author={Fearnhead, Paul and Maidstone, Robert and Letchford, Adam},
journal = {Journal of Computational and Graphical Statistics},
volume = {28},
number = {2},
pages = {265--275},
year = {2019},
publisher = {ASA Website},
URL ={https://doi.org/10.1080/10618600.2018.1512868},
eprint = { https://doi.org/10.1080/10618600.2018.1512868},
note={Supplementary material available at \url{https://www.tandfonline.com/doi/suppl/10.1080/10618600.2018.1512868}}
}

@article{killick-2012-pelt,
  title={Optimal Detection of Changepoints With a Linear Computational Cost},
  author={Killick, Rebecca and Fearnhead, Paul and Eckley, Idris A},
  journal={Journal of the American Statistical Association},
  volume={107},
  number={500},
  pages={1590--1598},
  year={2012},
  publisher={Taylor \& Francis},
  url={https://doi.org/10.1080/01621459.2012.737745}
}

@article{jackson-2005-op,
  author    = {Brad Jackson and Jeffrey D. Scargle and David Barnes and Sundararajan Arabhi and Alina Alt and Peter Gioumousis and Elyus Gwin and Paungkaew Sangtrakulcharoen and Linda Tan and Tun Tao Tsai},
  title     = {An algorithm for optimal partitioning of data on an interval},
  journal   = {IEEE Signal Processing Letters},
  volume    = {12},
  number    = {2},
  pages     = {105--108},
  year      = {2005},
  doi       = {10.1109/LSP.2004.840679},
  note      = {Available at: \url{https://ieeexplore.ieee.org/document/1381461}}
}

@article{fan-2001-variable,
  title={Variable Selection via Nonconcave Penalized Likelihood and its Oracle Properties},
  author={Fan, Jianqing and Li, Runze},
  journal={Journal of the American Statistical Association},
  volume={96},
  number={456},
  pages={1348--1360},
  year={2001},
  url= {https://doi.org/10.1198/016214501753382273}
}

@article{zou-2006-adaptive,
  title={The Adaptive Lasso and its Oracle Properties},
  author={Zou, Hui},
  journal={Journal of the American Statistical Association},
  volume={101},
  number={476},
  pages={1418--1429},
  year={2006},
  url={https://doi.org/10.1198/016214506000000735}
}

@article{tibshirani-1996-lasso,
  title={Regression Shrinkage and Selection Via the Lasso},
  author={Tibshirani, Robert},
  journal={Journal of the Royal Statistical Society: Series B (Statistical Methodology)},
  volume={58},
  number={1},
  pages={267--288},
  year={1996},
  url={https://doi.org/10.1111/j.2517-6161.1996.tb02080.x}
}

@article{rojas-2014-fusedlasso,
  title={On change point detection using the fused lasso method},
  author={Rojas, Cristian R and Wahlberg, Bo},
  journal={arXiv preprint arXiv:1401.5408},
  year={2014},
  url={https://arxiv.org/abs/1401.5408}
}

@article{gallagher-2014-adalasso,
  title={Detection of multiple undocumented change-points using adaptive Lasso},
  author={Shen, Jie and Gallagher, Colin M. and Lu, QiQi},
  journal={Journal of Applied Statistics},
  volume={41},
  number={6},
  pages={1161--1173},
  year={2014},
  url={https://doi.org/10.1080/02664763.2013.862220}
}

@article{harchaoui-2010-tv,
  title={Multiple Change-Point Estimation With a Total Variation Penalty},
  author={Harchaoui, Zaïd and Lévy-Leduc, Céline},
  journal={Journal of the American Statistical Association},
  volume={105},
  number={492},
  pages={1480--1493},
  year={2010},
  url={https://doi.org/10.1198/jasa.2010.tm09181}
}

@article{ribeiro-2016-homogenisation,
  title={Review and discussion of homogenisation methods for climate data},
  author={Ribeiro, Sara and Caineta, Júlio and Costa, Ana Cristina},
  journal={Physics and Chemistry of the Earth, Parts A/B/C},
  volume={94},
  pages={167--179},
  year={2016},
  url={https://doi.org/10.1016/j.pce.2015.08.007}
}

@article{domonkos-2021-homogenisation,
  title={Efficiency of Time Series Homogenization: Method Comparison with 12 monthly Temperature Test Datasets},
  author={Domonkos, Peter and Guijarro, José A. and Venema, Victor and Brunet, Manola and Sigró, Javier},
  journal={Journal of Climate},
  volume={34},
  number={8},
  pages={2877--2891},
  year={2021},
  url={https://doi.org/10.1175/JCLI-D-20-0611.1}
}

@article{page-1955-test,
  title={A test for a change in a parameter occurring at an unknown point},
  author={Page, E. S.},
  journal={Biometrika},
  volume={42},
  number={3/4},
  pages={523--527},
  year={1955},
  url={https://doi.org/10.1093/biomet/42.3-4.523}
}

@phdthesis{kirch-2006-SCUSUM,
  title={Resampling Methods for the Change Analysis of Dependent Data},
  author={Kirch, Claudia},
  year={2006},
  school={Universität zu Köln},
  type={PhD Thesis},
  url={https://kups.ub.uni-koeln.de/1795/} 
}

@article{sen-1975-tests,
  title={On Tests for Detecting Change in Mean},
  author={Sen, Ashish and Srivastava, Muni S.},
  journal={The Annals of Statistics},
  volume={3},
  number={1},
  pages={98--108},
  year={1975},
  doi= {10.1214/aos/1176343001}
}

@mastersthesis{ghassany-adafused-2010,
  author    = {Ghassany, Mohammad},
  title     = {The adaptive fused lasso regression and its application on microarrays CGH data},
  school    = {Università degli Studi di Milano},
  year      = {2010},
  url       = {https://www.mghassany.com/bibliography/files/ghassany-master-thesis.pdf},
}

@article{bloch-1946-nmr,
  title={Nuclear Induction},
  author={Bloch, Felix},
  journal={Physical Review},
  volume={70},
  number={7-8},
  pages={460},
  year={1946},
  url={https://doi.org/10.1103/PhysRev.70.460}
}

@article{kleinberg-1992-nmr,
  title={Novel NMR apparatus for investigating an external sample},
  author={Kleinberg, Robert L. and Sezginer, A. and Griffin, D.D. and Fukuhara, M.},
  journal={Journal of Magnetic Resonance (1969)},
  volume={97},
  number={3},
  pages={466--485},
  year={1992},
  url={https://doi.org/10.1016/0022-2364(92)90028-6}
}

@article{tibshirani-2011-solutiongenlasso,
  title={The solution path of the generalized lasso},
  author={Tibshirani, Ryan J. and Taylor, Jonathan},
  journal={Annals of Statistics},
  volume={39},
  number={3},
  pages={1335--1371},
  year={2011},
  url={https://doi.org/10.1214/11-AOS878}
}

@article{tibshirani-2011-supplement_to_solutiongenlasso,
  author  = {Tibshirani, Ryan J. and Taylor, Jonathan},
  title   = {Supplement A: Proofs and technical details for ``The solution path of the generalized lasso''},
  journal = {The Annals of Statistics},
  volume  = {39},
  number  = {3},
  pages   = {1335--1371},
  year    = {2011},
  url     = {https://doi.org/10.1214/11-AOS878SUPP},
  note    = {Supplement}
}

@article{fearnhead-2019-detecting,
  title={Changepoint Detection in the Presence of Outliers},
  author={Fearnhead, Paul and Rigaill, Guillem},
  journal={Journal of the American Statistical Association},
  volume={114},
  number={525},
  pages={169--183},
  year={2019},
  doi={10.1080/01621459.2017.1385466}
}

@article{Maidstone-2017-algorithms,
  title = {On optimal multiple changepoint algorithms for large data},
  author = {Maidstone, Robert and Hocking, Toby and Rigaill, Guillem and Fearnhead, Paul},
  journal = {Statistics and Computing},
  volume = {27},
  number = {2},
  pages = {519--533},
  year = {2017},
  url = {https://doi.org/10.1007/s11222-016-9636-3}
}

@article{Rigaill-2012-criteria,
  title = {Exact posterior distributions and model selection criteria for multiple change-point detection problems},
  author = {Rigaill, Guillem and Lebarbier, Émilie and Robin, Stéphane},
  journal = {Statistics and Computing},
  volume = {22},
  number = {4},
  pages = {917--929},
  year = {2012},
  url = {https://doi.org/10.1007/s11222-011-9258-8}
}

@article{Chakar-2017-AR1Seg,
  title = {A robust approach for estimating change-points in the mean of an AR(1) process},
  author = {Chakar, Souhil and Lebarbier, Émilie and Lévy-Leduc, Céline and Robin, Stéphane},
  journal = {Bernoulli},
  volume = {23},
  number = {2},
  pages = {1408--1447},
  year = {2017},
  url = {https://doi.org/10.3150/15-BEJ782}
}

@article{hampel-1974-MAD,
  author = {Hampel, Frank R.},
  title = {The Influence Curve and Its Role in Robust Estimation},
  journal = {Journal of the American Statistical Association},
  volume = {69},
  number = {346},
  pages = {383--393},
  year = {1974},
  url = {https://doi.org/10.1080/01621459.1974.10482962}
}

@article{cho-2024-WCM,
  title={Multiple change point detection under serial dependence: Wild contrast maximisation and gappy Schwarz algorithm},
  author={Cho, Haeran and Fryzlewicz, Piotr},
  journal={Journal of Time Series Analysis},
  volume={45},
  number={3},
  pages={479--494},
  year={2024},
  url={https://doi.org/10.1111/jtsa.12722}
}

@article{lund-2023-practices,
  title={Good Practices and Common Pitfalls in Climate Time Series Changepoint Techniques: A Review},
  author={Lund, Robert B. and Beaulieu, Claudie and Killick, Rebecca and Lu, Qiqi and Shi, Xueheng},
  journal={Journal of Climate},
  volume={36},
  number={23},
  pages={8041--8057},
  year={2023},
  publisher={American Meteorological Society},
  doi={10.1175/JCLI-D-22-0954.1}
}

@article{beaulieu-2024-globalwarmingsurge,
  title={A recent surge in global warming is not detectable yet},
  author={Beaulieu, Claudie and Gallagher, Colin and Killick, Rebecca and Lund, Robert and Shi, Xueheng},
  journal={Communications Earth \& Environment},
  volume={5},
  number={576},
  year={2024},
  month={10},
  doi={10.1038/s43247-024-01711-1},
  url={https://www.nature.com/articles/s43247-024-01711-1}
}

@online{nasa_gis,
  author    = {NASA GISS},
  title     = {GISS Surface Temperature Analysis (GISTEMP)},
  year      = {2025},
  url       = {https://data.giss.nasa.gov/gistemp/},
  note      = {Accessed: 2025-03-13}
}

@article{natarajan-1995-sparse,
  author = {Natarajan, B. K.},
  title = {Sparse Approximate Solutions to Linear Systems},
  journal = {SIAM Journal on Computing},
  volume = {24},
  number = {2},
  pages = {227--234},
  year = {1995},
  doi = {10.1137/S0097539792240406}
}

@book{garey-1979-computers,
  author = {Garey, Michael R. and Johnson, David S.},
  title = {Computers and Intractability: A Guide to the Theory of NP-Completeness},
  publisher = {W. H. Freeman and Company},
  year = {1979},
  isbn = {978-0-7167-1045-5}
}

@misc{tibshirani-2014-closerlook,
  title        = {A Closer Look at Sparse Regression},
  author       = {Tibshirani, Ryan},
  note         = {Amended by Larry Wasserman},
  year         = 2014,
  url          = {https://www.stat.cmu.edu/~larry/=sml/sparsity.pdf}
}

@article{Rigaill-2010-PrunedSN,
  author    = {Guillem Rigaill},
  title     = {Pruned dynamic programming for optimal multiple change-point detection},
  journal   = {arXiv preprint arXiv:1004.0887},
  year      = {2010},
  url       = {https://arxiv.org/abs/1004.0887}
}

@phdthesis{gao-2020-variance,
  title     = {Variance Change Point Detection under A Smoothly-changing Mean Trend with Application to Liver Procurement},
  author    = {Zhenguo Gao},
  school    = {Virginia Polytechnic Institute and State University},
  year      = {2018},
  type      = {Ph.D. dissertation},
  url       = {https://vtechworks.lib.vt.edu/items/f3138cd3-684a-4a87-aaa3-a97a99cd1988}
}

@article{Rinaldo-2009-Refinements,
  author    = {Alessandro Rinaldo},
  title     = {Properties and Refinements of the Fused Lasso},
  journal   = {The Annals of Statistics},
  volume    = {37},
  number    = {5B},
  pages     = {2922--2952},
  year      = {2009},
  doi       = {10.1214/08-AOS665},
  eprint    = {0805.0234},
  archivePrefix = {arXiv},
  primaryClass = {math.ST}
}

@article{kuhn-1955-hungarian,
  author = {Kuhn, Harold W.},
  title = {The Hungarian Method for the Assignment Problem},
  journal = {Naval Research Logistics Quarterly},
  volume = {2},
  number = {1-2},
  pages = {83--97},
  year = {1955},
  doi = {10.1002/nav.3800020109},
  publisher = {Wiley Online Library}
}

@article{Lund-2007-seasonality,
  author       = {Lund, Robert and Wang, Xiaolan L. and Lu, Qi Qi and Reeves, Jaxk and Gallagher, Colin and Feng, Yang},
  title        = {Changepoint Detection in Periodic and Autocorrelated Time Series},
  journal      = {Journal of Climate},
  volume       = {20},
  number       = {10},
  pages        = {5178--5190},
  month        = oct,
  year         = {2007},
  doi          = {10.1175/JCLI4291.1},
}

@article{boyd-2011-admm,
  title={Distributed optimization and statistical learning via the alternating direction method of multipliers},
  author={Boyd, Stephen and Parikh, Neal and Chu, Eric and Peleato, Borja and Eckstein, Jonathan},
  journal={Foundations and Trends in Machine Learning},
  volume={3},
  number={1},
  pages={1--122},
  year={2011},
  publisher={Now Publishers},
  doi={10.1561/2200000016}
}

@misc{pietrzyk-2025-dirtmu_dataset,
  author       = {Pietrzyk, Przemysław and others},
  title        = {{DIRT/$\mu$} root hair image dataset (All\_Supplement\_zip)},
  howpublished = {\url{https://figshare.com/articles/dataset/All_Supplement_zip/24886500/2}},
  year         = {2025},
  note         = {Dataset associated with Pietrzyk et al. "DIRT/$\mu$: automated extraction of root hair traits using combinatorial optimization"},
  doi          = {10.6084/m9.figshare.24886500.v2}
}

@article{pietrzyk-2024-dirtmu,
  author       = {Pietrzyk, Peter and Phan‑Udom, Neen and Chutoe, Chartinun and Pingault, Lise and Roy, Ankita and Libault, Marc and Saengwilai, Patompong Johns and Bucksch, Alexander},
  title        = {{DIRT/$\mu$}: Automated Extraction of Root Hair Traits Using Combinatorial Optimization},
  journal      = {Journal of Experimental Botany},
  volume       = {76},
  number       = {2},
  pages        = {285--298},
  year         = {2025},
  doi          = {10.1093/jxb/erae385},
  url          = {https://academic.oup.com/jxb/article/76/2/285/7756301}
}

@article{chupetlovska-2025-esr,
  title        = {ESR Essentials: a step‑by‑step guide of segmentation for radiologists—practice recommendations by the European Society of Medical Imaging Informatics},
  author       = {Chupetlovska, Kalina and Akinci D’Antonoli, Tuğba and Bodalal, Zuhir and Abdelatty, Mohamed A. and Erenstein, Hendrik and Santinha, João and Huisman, Merel and Visser, Jacob J. and Trebeschi, Stefano and Groot Lipman, Kevin B.W.},
  journal      = {European Radiology},
  year         = {2025},
  note         = {Invited review, Open Access},
  doi          = {10.1007/s00330-025-11621-1}
}

@article{song-2021-mri,
  title        = {Magnetic Resonance Imaging Segmentation via Weighted Level Set Model Based on Local Kernel Metric and Spatial Constraint},
  author       = {Song, Jianhua and Zhang, Zhe},
  journal      = {Entropy},
  volume       = {23},
  number       = {9},
  pages        = {1196},
  year         = {2021},
  doi          = {10.3390/e23091196}
}

@inproceedings{kaur-2021-review,
  title       = {Review on Medical Image Denoising Techniques},
  author      = {Kaur, Simarjeet and Singla, Jimmy and Nikita and Singh, Amar},
  booktitle   = {2021 International Conference on Innovative Practices in Technology and Management (ICIPTM)},
  pages       = {61--66},
  month       = feb,
  year        = {2021},
  publisher   = {IEEE},
  doi         = {10.1109/ICIPTM52218.2021.9388367}
}

@article{paton-2021-mr,
  title   = {MR Denoising Increases Radiomic Biomarker Precision and Reproducibility in Oncologic Imaging},
  author  = {Patón, M.F. and Iglesias, J.E. and Batchelor, P.G. and others},
  journal = {Journal of Digital Imaging},
  volume  = {34},
  number  = {5},
  pages   = {1134--1145},
  year    = {2021},
  doi     = {10.1007/s10278-021-00512-8}
}

@article{Dadi-2025-RootEx,
  title     = {RootEx: An automated method for barley root system extraction and evaluation},
  author    = {Dadi, Maichol and Lumini, Alessandra and Franco, Annalisa},
  journal   = {Computers and Electronics in Agriculture},
  volume    = {230},
  pages     = {110030},
  year      = {2025},
  doi       = {10.1016/j.compag.2025.110030},
  url       = {https://doi.org/10.1016/j.compag.2025.110030}
}

@article{Handy-2024-RootAnnotationVariation,
  title        = {Variation in forest root image annotation by experts, novices, and a CNN},
  author       = {Handy, G. and others},
  journal      = {Plant Methods},
  year         = {2024},
  volume       = {20},
  pages        = {12779},  
  doi          = {10.1186/s13007-024-01279-z},
}

@article{miller-1994-climate,
  title = {The 1976–77 Climate Shift of the Pacific Ocean},
  author = {Miller, Arthur J. and Cayan, Daniel R. and Barnett, Tim P. and Graham, Nicholas E. and Oberhuber, Josef M.},
  journal = {Oceanography},
  volume = {7},
  number = {1},
  pages = {21--26},
  year = {1994},
  doi = {10.5670/oceanog.1994.11}
}

@article{hare-2000-empirical,
  title = {Empirical Evidence for North Pacific Regime Shifts in 1977 and 1989},
  author = {Hare, Steven R. and Mantua, Nathan J.},
  journal = {Progress in Oceanography},
  volume = {47},
  number = {2-4},
  pages = {103--145},
  year = {2000},
  doi = {10.1016/S0079-6611(00)00033-1}
}

@incollection{IPCC-2001-Chapter3,
  author       = {{Intergovernmental Panel on Climate Change (IPCC)}},
  title        = {Chapter 3: The Carbon Cycle and Atmospheric Carbon Dioxide},
  booktitle    = {Climate Change 2001: The Scientific Basis},
  series       = {Contribution of Working Group I to the Third Assessment Report of the IPCC},
  publisher    = {Cambridge University Press},
  address      = {Cambridge, UK / New York, USA},
  year         = {2001},
}

@article{li-2019-elNino1979_80,
  author    = {Li, X. and others},
  title     = {Contributions of atmosphere–ocean interaction and low‑frequency variability to El Ni\~no events during 1979–2017},
  journal   = {Journal of Climate},
  volume    = {32},
  number    = {5},
  pages     = {1379--1398},
  year      = {2019}
}

@article{hoerling-1997-ENSO,
  author  = {Hoerling, Martin P. and Kumar, Arun},
  title   = {Origins of Extreme Climate States During the Exceptional 1982--83 El Ni\~no Event},
  journal = {Journal of Climate},
  volume  = {10},
  number  = {11},
  pages   = {2859--2870},
  year    = {1997}
}

@article{hansen-1992-pinatubo,
  author  = {Hansen, James and Lacis, Andrew and Ruedy, Reto and Sato, Makiko},
  title   = {Potential Climate Impact of Mount Pinatubo Eruption},
  journal = {Geophysical Research Letters},
  volume  = {19},
  number  = {2},
  pages   = {215--218},
  year    = {1992},
  doi     = {10.1029/91GL02788}
}

@article{lheureux-2017-ENSO,
  author  = {L’Heureux, Michelle L. and Stein, Alfred F. and Kumar, Arun and others},
  title   = {Observing and Predicting the 2015/16 El Ni\~no in the United States},
  journal = {Bulletin of the American Meteorological Society},
  volume  = {98},
  number  = {7},
  pages   = {1363--1382},
  year    = {2017},
  doi     = {10.1175/BAMS-D-16-0009.1}
}

@article{lund-1995-climatological,
  author    = {Lund, Robert B. and Hurd, Harry L. and Bloomfield, Peter and Smith, Richard L.},
  title     = {Climatological Time Series with Periodic Correlation},
  journal   = {Journal of Climate},
  year      = {1995},
  volume    = {8},
  number    = {11},
  pages     = {2787--2809},
  doi       = {10.1175/1520-0442(1995)008<2787:CTSWPC>2.0.CO;2}
}

@article{beaulieu-2012-changepoint_CO2,
  author  = {Beaulieu, Claudie and Chen, Jie and Sarmiento, Jorge L.},
  title   = {Change-point analysis as a tool to detect abrupt climate variations},
  journal = {Philosophical Transactions of the Royal Society A},
  year    = {2012},
  volume  = {370},
  number  = {1962},
  pages   = {1228--1249},
  doi     = {10.1098/rsta.2011.0383}
}

@article{robbins-2020-flexible_changepoint,
  author  = {Robbins, Michael W.},
  title   = {A Fully Flexible Change-Point Test for Climate Time Series Applied to the Mauna Loa CO$_2$ Record},
  journal = {Statistica Sinica},
  year    = {2020},
  volume  = {30},
  pages   = {1671--1692}
}

@misc{noaa_maunaloa_co2_trends,
  author       = {National Oceanic and Atmospheric Administration Global Monitoring Laboratory},
  title        = {Trends in atmospheric carbon dioxide at Mauna Loa Observatory},
  howpublished = {\url{https://gml.noaa.gov/ccgg/trends/}},
  note         = {Data record initiated in March 1958 by C. D. Keeling},
  year         = {2025},
}

@Manual{pinheiro-bates-nlme,
  title        = {nlme: Linear and Nonlinear Mixed Effects Models},
  author       = {Jos{\'e} Pinheiro and Douglas Bates and Saikat DebRoy and Deepayan Sarkar and R Core Team},
  year         = {2025},
  note         = {R package version 3.1-168},
  url          = {https://CRAN.R-project.org/package=nlme},
  doi          = {10.32614/CRAN.package.nlme},
}

@article{zhao-2006-modelselection,
  title   = {On Model Selection Consistency of {Lasso}},
  author  = {Zhao, Peng and Yu, Bin},
  journal = {Journal of Machine Learning Research},
  volume  = {7},
  pages   = {2541--2563},
  year    = {2006},
  url     = {https://www.jmlr.org/papers/volume7/zhao06a/zhao06a.pdf}
}

@article{wainwright-2009-sharp-thresholds,
  title   = {Sharp Thresholds for High-Dimensional and Noisy Recovery of Sparsity Using $\ell_{1}$-Constrained Quadratic Programming ({Lasso})},
  author  = {Wainwright, Martin J.},
  journal = {Annals of Statistics},
  volume  = {37},
  number  = {5A},
  pages   = {2465--2493},
  year    = {2009},
  doi     = {10.1214/08-AOS653}
}

@article{bickel-2009-lasso-dantzig,
  title   = {Simultaneous Analysis of {Lasso} and Dantzig Selector},
  author  = {Bickel, Peter J. and Ritov, Ya'acov and Tsybakov, Alexandre B.},
  journal = {Annals of Statistics},
  volume  = {37},
  number  = {4},
  pages   = {1705--1732},
  year    = {2009},
  doi     = {10.1214/08-AOS620}
}

@book{buhlmann-2011-hdstat,
  title     = {Statistics for High-Dimensional Data: Methods, Theory and Applications},
  author    = {B{\"u}hlmann, Peter and van de Geer, Sara},
  year      = {2011},
  publisher = {Springer},
  series    = {Springer Series in Statistics},
  isbn      = {978-3-642-20191-2},
  doi       = {10.1007/978-3-642-20192-9}
}

@article{liu-2018-change,
  title={Change-point detection in panel data with applications in mortality forecasting},
  author={Liu, Shu and Yao, Weixin and Yu, Yan},
  journal={Lifetime Data Analysis},
  volume={24},
  number={4},
  pages={569--593},
  year={2018},
  publisher={Springer},
  doi={10.1007/s10985-018-9418-9}
}

@article{poppick-2016-trends,
  title={Estimating trends in the global mean temperature record},
  author={Poppick, Andrew and Moyer, Elisabeth and Stein, Michael L.},
  journal={arXiv preprint arXiv:1607.03855},
  year={2016},
  url={https://arxiv.org/abs/1607.03855}
}

@article{schotz-2025-rethinking,
  title={Rethinking Climate Econometrics},
  author={Sch{\"o}tz, Martin and Hassel, Malte and Otto, Christian},
  journal={arXiv preprint arXiv:2505.18033},
  year={2025},
  url={https://arxiv.org/abs/2505.18033}
}

@book{hamilton-1994-series,
  author    = {Hamilton, James D.},
  year      = {1994},
  title     = {Time Series Analysis},
  publisher = {Princeton University Press},
  address   = {Princeton, NJ}
}

@book{brockwell-davis-2016-forecasting,
  author    = {Brockwell, Peter J. and Davis, Richard A.},
  year      = {2016},
  title     = {Introduction to Time Series and Forecasting},
  edition   = {3rd},
  publisher = {Springer},
  address   = {New York}
}

@article{Ma-2000-RobustAR,
  author = {Ma, Yan and Genton, Marc G.},
  title = {Highly robust estimation of the autocorrelation coefficient},
  journal = {Journal of Time Series Analysis},
  year = {2000},
  volume = {21},
  number = {6},
  pages = {663--684}
}

@book{hastie-2015-sparsity,
  author    = {Trevor Hastie and Robert Tibshirani and Martin Wainwright},
  title     = {Statistical Learning with Sparsity: The Lasso and Generalizations},
  year      = {2015},
  publisher = {Chapman and Hall/CRC},
  address   = {Boca Raton, FL}
}

@article{efron-2004-lar,
  author    = {Bradley Efron and Trevor Hastie and Iain Johnstone and Robert Tibshirani},
  title     = {Least Angle Regression},
  journal   = {Annals of Statistics},
  year      = {2004},
  volume    = {32},
  number    = {2},
  pages     = {407--499}
}

@article{hastie-2020-bestsubset,
  author    = {Trevor Hastie and Ryan J. Tibshirani and Martin Wainwright},
  title     = {Best Subset, Forward Stepwise, or Lasso? Analysis and Recommendations Based on Extensive Comparisons},
  journal   = {Statistical Science},
  year      = {2020},
  volume    = {35},
  number    = {4},
  pages     = {579--592}
}

@article{Shi-2022-comparison,
  title={A comparison of single and multiple changepoint techniques for time series data},
  author={Shi, Xueheng and Gallagher, Colin and Lund, Robert and Killick, Rebecca},
  journal={Computational Statistics \& Data Analysis},
  volume={170},
  pages={107433},
  year={2022},
  publisher={Elsevier}
}

@article{Lund-2020-WBScomments,
  title={Detecting possibly frequent change-points: wild binary segmentation 2 and steepest-drop model selection},
  author={Lund, Robert and Shi, Xueheng},
  journal={Journal of the Korean Statistical Society},
  volume={49},
  number={4},
  pages={1090--1095},
  year={2020},
  publisher={Springer}
}

@article{Shi-2022-autocovariance,
  title={Autocovariance estimation in the presence of changepoints},
  author={Gallagher, Colin and Killick, Rebecca and Lund, Robert and Shi, Xueheng},
  journal={Journal of the Korean Statistical Society},
  volume={51},
  number={4},
  pages={1021--1040},
  year={2022},
  publisher={Springer}
}

@article{Shi-2022-CET,
  title={Changepoint detection: An analysis of the Central England temperature series},
  author={Shi, Xueheng and Beaulieu, Claudie and Killick, Rebecca and Lund, Robert},
  journal={Journal of Climate},
  volume={35},
  number={19},
  pages={6329--6342},
  year={2022}
}

@misc{noaa_crvfit,
  author       = {{NOAA Global Monitoring Laboratory}},
  title        = {Curve Fitting to Atmospheric CO\textsubscript{2} Time Series},
  year         = {2025},
  howpublished = {\url{https://gml.noaa.gov/ccgg/mbl/crvfit/crvfit.html}},
  note         = {Accessed: 2025-11-20}
}

@article{candes-2008-reweighted,
  title        = {Enhancing Sparsity by Reweighted {$\ell_1$} Minimization},
  author       = {Cand{\`e}s, Emmanuel J. and Wakin, Michael B. and Boyd, Stephen P.},
  journal      = {Journal of Fourier Analysis and Applications},
  volume       = {14},
  number       = {5},
  pages        = {877--905},
  year         = {2008},
  doi          = {10.1007/s00041-008-9045-x}
}

\end{document}